\pdfoutput=1

\documentclass[sigplan, screen]{acmart}

\newif\iffull\fulltrue

\usepackage{xspace}

\usepackage{typingjudgements}
\usepackage{tikz}
\usetikzlibrary{arrows,positioning}
\usepackage{mathpartir}
\usepackage{booktabs}   
\usepackage{subcaption} 
\usepackage{typingjudgements}
\usepackage{trfrac}
\usepackage{tcolorbox}
\usepackage{syntax}
\usepackage{xcolor}

\usepackage{listings}
\usepackage{comment}

\definecolor{commcolor}{rgb}{0, 0.6, 0}

\definecolor{pblue}{rgb}{0.2, 0, 1}
\definecolor{pred}{rgb}{1, 0, 0}

\lstdefinelanguage{pfour}
{
    keywords=[1]{
        low,high,top, bot,A,B
    },
    keywordstyle=[1]\color{pred},
    keywords=[2]{table, enum, action, apply, header, control, struct},
    keywordstyle=[2]\color{pblue}
}

\makeatletter
\@addtoreset{equation}{enumi}
\makeatother

\makeatletter
\let\old@lstKV@SwitchCases\lstKV@SwitchCases
\def\lstKV@SwitchCases#1#2#3{}
\makeatother
\usepackage{lstlinebgrd}
\makeatletter
\let\lstKV@SwitchCases\old@lstKV@SwitchCases

\lst@Key{numbers}{none}{%
    \def\lst@PlaceNumber{\lst@linebgrd}%
    \lstKV@SwitchCases{#1}%
    {none:\\%
    left:\def\lst@PlaceNumber{\llap{\normalfont
                \lst@numberstyle{\thelstnumber}\kern\lst@numbersep}\lst@linebgrd}\\%
    right:\def\lst@PlaceNumber{\rlap{\normalfont
                \kern\linewidth \kern\lst@numbersep
                \lst@numberstyle{\thelstnumber}}\lst@linebgrd}%
    }{\PackageError{Listings}{Numbers #1 unknown}\@ehc}}
\makeatother
\usepackage{lstlinebgrd}

\usepackage{balance}

\usepackage[capitalize]{cleveref}

\newcommand{\nt}[1]{\textit{#1}}

\lstdefinestyle{topper}{
  float=tp,
  floatplacement=tbp,
  abovecaptionskip=-5pt
}

\def\name{\textsc{P4BID}\xspace}

\setcopyright{rightsretained}
\acmPrice{}
\acmDOI{10.1145/3519939.3523717}
\acmYear{2022}
\copyrightyear{2022}
\acmSubmissionID{pldi22main-p469-p}
\acmISBN{978-1-4503-9265-5/22/06}
\acmConference[PLDI '22]{Proceedings of the 43rd ACM SIGPLAN International Conference on Programming Language Design and Implementation}{June 13--17, 2022}{San Diego, CA, USA}
\acmBooktitle{Proceedings of the 43rd ACM SIGPLAN International Conference on Programming Language Design and Implementation (PLDI '22), June 13--17, 2022, San Diego, CA, USA}

\begin{document}

\title[\name: Information Flow Control in P4]{\name: Information Flow Control in P4}
\author{Karuna Grewal}
\affiliation{%
  \institution{Cornell University}
  \country{USA}
}

\author{Loris D'Antoni}
\affiliation{%
  \institution{University of Wisconsin}
  \country{USA}
}

\author{Justin Hsu}
\affiliation{%
  \institution{Cornell University}
  \country{USA}
}

\begin{abstract}
Modern programmable network switches can implement custom applications using efficient packet processing hardware, and the programming language P4 provides high-level constructs to program such switches.
The increase in speed and programmability has inspired research in \emph{dataplane programming}, where many complex functionalities, e.g., key-value stores and load balancers, can be implemented entirely in network switches.
However, dataplane programs may suffer from novel security errors that are not traditionally found in network switches.

To address this issue, we present a new information-flow control type system for P4.
We formalize our type system in a recently-proposed core version of P4, and we prove a soundness theorem: well-typed programs satisfy non-interference.
We also implement our type system in a tool, \name, which extends the type checker in the \textsc{p4c} compiler, the reference compiler for the latest version of P4.
We present several case studies showing that natural security, integrity, and isolation properties in networks can be captured by non-interference, and our type system can detect violations of these properties while certifying correct programs.
\end{abstract}



\maketitle

\section{Introduction}
\label{sec:intro}

The last two decades have seen an ongoing shift in how networks are programmed.
The task of programming a network once consisted of manually setting configurations in specialized switch hardware that provided limited customization; low-level programming was the only way to achieve performance.
Today, switches are highly programmable and provide rich functionalities for processing network packets.
This increased programmability is enabling complex network functionalities, which traditionally run on slower dedicated devices, to run directly on switches and other networking hardware~\cite{netcache,d2r}.
Furthermore, new programming models and languages make it easier for network operators to define complex functionalities~\citep{p4}.

While the advent of programmable network switches has inspired a large number of practitioners and researchers to write complex functionalities that can run on switches, it has also brought a new level of complexity in a world where bugs can be costly.
As is well known, network configuration errors have led to widespread and costly outages~(e.g., \citep{fb-outage,gcp-outage}).
The problem of preventing these, and other, types of bugs has received a lot of attention in the programming languages and verification communities.
For example, researchers have developed formal tools for verifying that switch configurations guarantee desirable network properties, such as node reachability, the absence of black holes, and resilience to link failures~(e.g., \citep{netkat,minesweeper,zeppelin}).
While these tools are extremely useful for network operators, applications running on programmable switches may exhibit errors that are not traditionally associated with networks.
In particular, there has been little work on verifying security properties for dataplane programs.

\paragraph*{Our work.}
We develop a new information-flow control (IFC) type system for the network programming language P4~\citep{p4}, a leading language for programming network switches.
P4 is an attractive target: it is actively developed by researchers from academia and industry, and can compile to a variety of networking hardware.
\emph{Information flow control} (IFC) is a well-studied, language-based approach to verifying security properties where variables in the program are tagged with security labels, and the type system ensures that no information can flow from high-security variables (secret) to low-security ones (public).
IFC is
(i) \textit{flexible}: by changing the label usage one can model security properties, like confidentiality and integrity;
(ii) \textit{general}: it can accommodate complex programming constructs; and
(iii) \textit{lightweight}: the analysis is simple, type-based, and requires minimal annotations from the programmer.
Owing to these strengths, IFC has found wide adoption and has been deployed in real languages~\citep{jif,flowcaml}.

Designing an IFC type system for P4 involves both technical and conceptual challenges.
On the technical side, while P4 resembles a standard imperative language, it has a number of features to target the restricted computational model of networking switches.
For instance, much of the computation in P4 programs happens via \emph{tables}, which match on data in packet headers and select which actions to run.
While a P4 program implements the actions, the table itself is not known until it is installed at runtime by the network controller.
A second technical challenge is the size and complexity of the language.
Like many languages in real-world use, P4 does not have a formal specification.
To firm up the foundations of P4, \citet{petr4} developed a formal version of P4, called Core P4, as part of the broader \Petr\ project.
The formal operational model of Core P4 makes it possible to develop type systems that provably guarantee program properties.
However, Core P4 is still quite large---P4 is a language intended for real-world use, with a wide variety of declarations, statements, and expressions, and Core P4 models almost all the features of P4.
Our work develops an IFC system that can handle the principal features of Core P4.

On the conceptual side, IFC for dataplane programming has been little-studied and it is not know what useful properties network properties an IFC system can enforce.
As part of our work, we present case studies showing that standard properties guaranteed by IFC, like \emph{confidentiality} and \emph{integrity}, are useful security properties for networking applications.
We also show how natural network \emph{isolation} properties can also be guaranteed with an IFC system, by adjusting the lattice of security labels.

\paragraph*{Outline.}
After overviewing our approach in \Cref{sec:overview} and providing the
necessary background on P4 and Core P4 in \Cref{sec:corep4}, we present our
central contributions:
\begin{enumerate}
  \item An \textit{Information Flow Control (IFC) type system for Core
    P4}~\citep{petr4}, a core calculus modeling the P4 language, together with
    a soundness theorem: well-typed programs satisfy \emph{non-interference}
    (\Cref{sec:ifc}).
  \item \name: a \textit{type-checker} implemented on top of \textsf{p4c}, the
    reference compiler for P4. We evaluate our system through four case studies,
    demonstrating how properties enforced by IFC, like confidentiality and
    integrity, can be useful in a networking context. We implement our case
    studies in P4 and show that \name can automatically detect when these
    properties are violated, while correctly type-checking versions of these
    programs where the problems are removed (\Cref{sec:eval}).
\end{enumerate}
We conclude by surveying related work (\Cref{sec:rw}) and outlining possible
future directions (\Cref{sec:conc}).

\section{Overview}
\label{sec:overview}

\paragraph*{A quick introduction to P4.}
P4 is an actively-developed language for programming the network data plane.
Computation is divided into three phases: \emph{parser}, \emph{pipeline}, and
\emph{deparser}. The packet processing starts at the parser, where the input
packet is extracted into a typed representation given by headers using a finite
state machine. The pipeline phase executes the primary logic of the switch by
transforming the parsed representation of the input packet. The deparser
serializes the parsed typed representation of the input packet into the output
packet. Our work focuses on P4 \emph{control blocks}, which implement the
pipeline phase.
To get a feel for the language, we consider a P4 program for a basic
task: converting virtual addresses to physical addresses when packets enter a
local network. \Cref{lst:topology} begins by declaring the types of the
\emph{headers} which are carried by packets; P4 programs manipulate the state of
packets by modifying the headers. In our case, there are three headers:
\textsf{ipv4} and \textsf{ethernet} carry the routing information in the original packet, while
\textsf{local_hdr} carries information specific to the local network.

\begin{lstlisting}[caption={Translating virtual to physical addresses.},label={lst:topology}, numbers=left, breaklines=true, escapechar=|, style=topper]
header local_hdr_t {
    bit<32> phys_dstAddr;
    bit<8> phys_ttl;
    bit<48> next_hop_MAC_addr;
}

header ipv4_t {
    bit<8> ttl;
    bit<8> protocol;
    bit<32> srcAddr;
    bit<32> dstAddr;
}

header eth_t {
  bit<48> srcAddr;
  bit<48> dstAddr;
}

struct headers {
  ipv4_t ipv4;
  eth_t eth;
  local_hdr_t local_hdr;
}

control Obfuscate_Ingress(inout headers hdr,
        inout standard_metadata_t std_metadata) {
    table virtual2phys_topology {  |\label{ln:topo-v2p-tbl}|
        key = { hdr.ipv4.dstAddr: exact; }
        actions = { update_to_phys; }
    }
    action update_to_phys(bit<32> phys_dstAddr,
                          bit<8> phys_ttl) { |\label{ln:topo-update-to-phys}|
        hdr.local_hdr.phys_dstAddr = phys_dstAddr;
        hdr.ipv4.ttl = phys_ttl; |\label{ln:topo:problem}|
    }
    table ipv4_lpm_forward {
        key = { hdr.ipv4.dstAddr: lpm; }
        actions = { ipv4_forward; drop; }
    }
    action ipv4_forward(bit<48> dstAddr, bit<9> port) { |\label{ln:topo-ipv4-forward}|
        hdr.eth.dstAddr = dstAddr;
        standard_metadata.egress_spec = port;
    }
    action drop() { mark_to_drop(standard_metadata); }
    apply {
        virtual2phys_topology.apply();
        ipv4_lpm_forward.apply();
    }
}
\end{lstlisting}%

\Cref{lst:topology} shows the code for the control block, which implements the
core part of the logic. (The full P4 program also describes other stages of the
packet-processing pipeline like parsing and deparsing, which we do not consider
in our work.) The switch behavior is organized into \emph{tables} and
\emph{actions}. Tables match data in headers (the \emph{keys}) and apply
actions. For instance, the table \mbox{\textsf{ipv4_lpm_forward}} inspects the
value of the header \textsf{hdr.ipv4.dstAddr} and then decides whether to run
action \textsf{ipv4_forward} or \textsf{drop} the packet. The concrete mapping
is not specified by the P4 program; instead, the switch controller installs
these mappings at runtime. Actions can inspect and modify packet headers.
Actions can also be parameterized by arguments, which are supplied by the table
when the action is applied. For example, the action \mbox{\textsf{ipv4_forward}}
accepts a destination address and port as arguments, and then proceeds to update
headers. Finally, the apply block specifies the overall behavior of the control
block: here, the switch applies table \textsf{virt2phys} to translate virtual
addresses to physical addresses, and then \mbox{\textsf{ipv4_lpm_forward}} to
forward the packet.

\paragraph*{A potential security vulnerability.}
\Cref{lst:topology} is designed to process a packet as it enters a local
network. The incoming packet refers to a \emph{virtual address}, which must be
translated to a physical address. Furthermore, the switch adjusts other packet
fields, like the maximum number of hops (time-to-live, \textsf{ttl}), to reflect
the topology of the local network. To preserve privacy, details of the local
network should not leak into fields that are visible when the packet leaves the
network. To accomplish this goal, the program uses a separate header of
type \textsf{local_hdr_t} to store local information (\Cref{ln:local:hdr}). As
the packet is routed in the local network, the switches do not touch the public
\textsf{ipv4} and \textsf{ethernet} headers; instead, they parse
\textsf{local_hdr} and update it with the next hop route information. When the
packet exits the local network, the header \textsf{local_hdr} is removed.

While the intended behavior is simple to describe, the program in
\Cref{lst:topology} has an error: \Cref{ln:topo:problem} incorrectly stores the
local ttl in the \textsf{ipv4} header, rather than the \textsf{local_hdr}
header. Even when the local header is removed, the \textsf{ipv4} header will
carry private information about the local network. This kind of error
unintentionally leaks local information into public headers, but it can be easy
to overlook.

\begin{lstlisting}[caption={Security-Annotated Version of \Cref{lst:topology}},label={lst:fixed-topology}, numbers=left, breaklines=true, escapechar=|, style=topper]
header local_hdr_t {  |\label{ln:local:hdr}|
    <bit<32>, high> phys_dstAddr;
    <bit<8>, high> phys_ttl;
    // ...
}

header ipv4_t {
    <bit<8>, low> ttl;
    // ...
}

struct headers {
  ipv4_t ipv4;
  local_hdr_t local_hdr;
  // ...
}

control Obfuscate_Ingress(inout headers hdr,
        inout standard_metadata_t std_metadata) {
    action update_to_phys(<bit<32>, high> phys_dstAddr,
                          <bit<8>, high> phys_ttl) {
        hdr.local_hdr.phys_dstAddr = phys_dstAddr;
        // !BUG!: low <- high
        hdr.ipv4.ttl = phys_ttl; |\label{ln:topo:annot:bug}|
        // *FIX*: high <- high
        hdr.local_hdr.phys_ttl = phys_ttl; |\label{ln:topo:annot:fix}|
    }
    // ...
}
\end{lstlisting}%

\paragraph*{Security types to the rescue.}
We design an information-flow control type system for P4 to catch such bugs. Like standard IFC type
systems, our system extends each P4 type with a \emph{security label}:
\textsf{high} if the data is secret, and \textsf{low} if the data is public.
\Cref{lst:fixed-topology} shows our example program annotated with security
types. All data specific to the local network (e.g., \textsf{phys_dstAddr},
\textsf{phys_ttl}) are marked as \emph{high} security. The publicly visible
headers (e.g., \textsf{ipv4}, \textsf{eth}) are marked as \emph{low} security.
Our type system guarantees that information from high-security data does not
influence low-security data. For instance, the information leak we saw before
can be flagged in our type system: \Cref{ln:topo:annot:bug} incorrectly assigns
a high-security data \textsf{phys_ttl} to a low-security field
\textsf{ipv4.ttl}. The problem is corrected by assigning \textsf{phys_ttl}
to \mbox{\textsf{local_hdr.ttl}} (\Cref{ln:topo:annot:fix}), which is a
high-security field.

While this kind of analysis is fairly straightforward, the design of our type
system must handle unusual features from P4's programming model (e.g., actions
and tables); we discuss these aspects in \Cref{sec:corep4} and \Cref{sec:ifc}.
Furthermore, while \Cref{lst:topology} demonstrates a basic information leak, we
will see more interesting applications of our type system to networking
applications in \Cref{sec:eval}.


\section{Syntax and Semantics of Core P4}
\label{sec:corep4}

This section briefly reviews the core P4 calculus presented in the recent work
on \Petr~\cite{petr4}, the representation of P4 programs in terms of the core
calculus syntax, and the operational semantics and typing judgements for the
core calculus.

\begin{figure*}
\centering
\begin{subfigure}[t]{.48\linewidth}
\centering
$\begin{array}{rclr}
exp & ::= & b & \text{Boolean}\\
                 &\OR & n_w  &\text{integers or bits of width w}\\
                 &\OR & x &\text{variable} \\
                 &\OR & exp_1[exp_2] &\text{array indexing} \\
                 &\OR & exp_1 \oplus exp_2 &\text{binary operation}\\
                 &\OR & \{\overline{f_i = exp_i}\} &\text{record} \\
                 &\OR & exp.f_i &\text{field projection}\\
                 &\OR & exp_1(\overline{exp_2}) &\text{function call}
\end{array}$
\caption{Expressions}
\end{subfigure}
\begin{subfigure}[t]{.48\linewidth}
\centering
$\begin{array}{rclr}
{stmt} & ::= & exp_1(\overline{exp_2}) &\text{function call}\\
               & \OR & exp_1 := exp_2 &\text{assignment} \\
              &  \OR & \terminal{if}~ (exp_1)~ stmt_1~ \terminal{else}~ stmt_2 &\text{conditional}\\
              &  \OR & \{\overline{stmt}\} &\text{sequencing}\\
               & \OR & \terminal{exit} &\text{exit}\\
                &\OR & \terminal{return}~ exp &\text{return}\\
               & \OR & var\_decl &\text{variable declaration}
\end{array}$
\caption{Statements}
\end{subfigure}
\\\vspace{2mm}
\begin{subfigure}[t]{.52\linewidth}
\centering
$\begin{array}{rclr}
{prg} & ::=&  \overline{typ\_decl}~ ctrl\_body  \\
{ctrl\_body} &::= & \overline{decl} ~stmt \\
{decl} & ::=&  var\_decl \OR obj\_decl \OR typ\_decl \\
{var\_decl} & ::=&  \tau~x:= exp
        \OR \tau~x
        \\
{typ\_decl} & ::=& \terminal{match\_kind}~\{\overline{f}\} \OR \terminal{typedef}~\tau~X\\
{obj\_decl} & ::=& \terminal{table}~x~\{\overline{key}~\overline{act}\} \\
                &\OR & \terminal{function}~ \tau_{ret}~x~(\overline{d~y:\tau}) \{stmt\}
\end{array}$
\caption{Declarations}
\end{subfigure}
\begin{subfigure}[t]{.47\linewidth}
\centering
$\begin{array}{rclr}
{d}  & ::= &  in \OR inout\\
{lval} & ::=&  x \\
        &\OR & lval.f \\
        &\OR & lval[n]\\
{key} & ::= & exp: x \\
{act} & ::=&  x(\overline{exp}, \overline{x: \tau})
\end{array}$
\caption{Other constructs}
\end{subfigure}
\caption{Core P4 Expressions (fragment)}
\label{P4IFC-syntax}
\end{figure*}

\subsection{Core P4 Syntax} \label{sec:petr4-background}
\Petr\ formalizes the semantics of various P4 primitives, like control blocks, match-action tables, and statements in a calculus called Core P4.
For our information-flow control type system, we focus on the fragment of Core P4 in \Cref{P4IFC-syntax}.
Expressions and statements are largely standard.

Core P4 programs (\textit{prg}) are represented as a sequence of variable, object, or type declarations followed by a control block.
The central construct in a P4 program is the control block, which describes how the switch processes packets in terms of table and action calls inside its \terminal{apply} block.
A control block body (\textit{ctrl\_body}) is a sequence of declarations and statements.
The \textit{stmt} in the control block corresponds to the \terminal{apply} block of a P4 program.

Variable and type declarations (\textit{var\_decl}, \textit{typ\_decl}) are largely standard; the \terminal{match\_kind} enum declares different ways tables can match on packet fields.
Object declarations (\textit{obj\_decl}) declare P4 objects: tables and actions.
These object declarations can have nested ordinary statements (\textit{stmt}) that allow usual imperative primitives like mutation and control flow statements.
To get a feel for these features, let's consider how they correspond to parts of the {\terminal{Obfuscate\_Ingress}} control block in \Cref{lst:topology}.
The example control block consists of three actions declarations
(\terminal{update\_to\_phys}, \terminal{ipv4\_forward}, and \terminal{drop}),
and two table declarations (\terminal{virtual2phys\_topology} and \terminal{ipv4\_lpm\_forward}).

\paragraph*{Tables}
A table declaration, $\terminal{table}~x~\{\overline{key}~\overline{act}\}$, is composed of a list of expressions (usually packet header fields) that specify the lookup key, $\overline{key}$,  and actions, $\overline{act}$, which the lookup table might execute.
A table application uses the key to lookup the entries in the table (installed by the control plane) and invokes the action from the matched entry.
For example, table \terminal{virtual2phys\_topology} in \Cref{ln:topo-v2p-tbl} contains the key \terminal{hdr.ipv4.dstAddr: exact} (where \terminal{exact} specifies the match pattern, in this case, exact match on the key), and the action \terminal{update\_to\_phys} action.
Applying this table, represented in Core P4 as \terminal{virtual2phys\_topology()}, matches the table entries installed by the control plane against the corresponding keys in the current packet and returns an appropriate action to run, with all its arguments.
Any optional arguments in the returned action will be supplied by the control plane. The match pattern determines the criterion for choosing a table entry based on the key. For instance, \terminal{lpm} specifies that a key is matched to the entry corresponding to its longest prefix; \terminal{exact} specifies that a key should be exactly matched to some table entry otherwise it is a match failure.

\paragraph*{Actions}
An action declaration is a special case of a function declaration, $\terminal{function}~ \tau_{ret}~x~(\overline{d~y:\tau} )\{stmt\}$, with no return type.
For example, the action \terminal{update\_to\_phys} on \Cref{ln:topo-update-to-phys} in \Cref{lst:topology}  has parameters \terminal{phys\_dstAddr} and \terminal{phys\_ttl}, of types \terminal{bit$\langle 32 \rangle$} and \terminal{bit$\langle 8\rangle$}.
Parameters can have a \emph{directionality}, $d$: an \emph{in} expression can only be read from, while an \emph{inout} expression can be both read and written to.
Omitted directions in parameters default to the \emph{in} direction; these directionless parameters are optional arguments that can be passed by the control plane.
Invoking the action, which can be done directly as a statement or indirectly from a table, runs the statement $stmt$ in the action body.
Actions, like all Core P4 functions, do not support recursion.

\paragraph*{Differences compared to Core P4.}
The language in \Cref{P4IFC-syntax} is a significant fragment of Core P4, but it does not handle some of its more specialized features (e.g., generics, constant declarations, slice operation, and native functions).
We consider this fragment for simplicity, but we do not foresee difficulties in extending our IFC analysis to full Core P4. We omitted some lesser-used features, like generics, because the core language is already quite large and we believe it is unlikely that omitted features lead to information-flow violations.
We focus on programs with a single control block because most P4 programs encode their main functionality in a single ingress control block.
Since our system already supports user-defined functions and closures, with all of their technical intricacies, we do not see any obstacle to handling multiple control blocks besides increasing the complexity of our type system.

\subsection{Core P4 Semantics}\label{cp4-semantics}
To understand the semantics of Core P4 programs, we will review the evaluation judgement forms  for expressions, statements, and declarations from \Petr~\cite{petr4}.
The main judgements are as follows:
\begin{align*}
   \langle \mathcal{C}, \Delta, \mu, \epsilon, \nt{exp} \rangle &\Downarrow \langle \mu', val \rangle \\
   \langle \mathcal{C}, \Delta, \mu, \epsilon, stmt \rangle &\Downarrow \langle \mu', \epsilon', sig \rangle \\
   \langle \mathcal{C}, \Delta, \mu, \epsilon, decl \rangle &\Downarrow \langle \Delta', \mu', \epsilon', sig \rangle
\end{align*}

The contexts used in these judgements are defined in \Cref{fig:contexts}.
Here, $\Delta$ is the partial map from type names to types; $\epsilon$ is the partial map between variables and their memory locations; $\mu$ is the memory store mapping variable locations to their values.
 $\mathcal{C}$ models the table lookup map provided by the control plane: given a table at location $l$ with $key=val$, and a list of actions described by a list of $PartialActionRef$ (actions with optional arguments missing),  $\mathcal{C}$ returns an action call expression with all the optional arguments of the action supplied ($ActionRef$).
\begin{figure}
\raggedright
  \begin{subfigure}[t]{.4\linewidth}
  \[
	\begin{array}{rlrl}
    \mathit{Var} &: \text{variables} & \mathit{Val} &: \text{values} \\
    \mathit{TypVar} &: \text{type variables} &\mathit{Typ} &: \text{types in Core P4}\\
    \mathit{Loc} &: \text{locations}
	 \end{array}
	 \]
	 \end{subfigure}
	 \newline
	 \begin{subfigure}[t]{.4\linewidth}
	 \[
	 \begin{array}{rlrl}
     \Gamma &: \mathit{Var} \to \mathit{Typ} &\hspace{10mm}
     \Delta &: \mathit{TypVar} \to \mathit{Typ} \\
	    \epsilon &: \mathit{Var} \to \mathit{Loc} &
	    \mu &: \mathit{Loc} \to \mathit{Val}\\
	    \mathcal{C} & \multicolumn{3}{l}{: \mathit{Loc} \times \mathit{Val} \times \overline{PartialActionRef} \to ActionRef}
	\end{array}
	\]
	\end{subfigure}
\caption{Typing and Evaluation Contexts}
\label{fig:contexts}
\end{figure}
The judgements use $val$ to denote a value; and $sig$ to denote a \emph{signal}, which indicates whether the program's control flow proceeds normally (\terminal{cont}), returns a value ($\terminal{return}\ val$), or errors (\terminal{exit}).

Since function calls are expressions, and a function's body can update the memory store, the evaluation judgement for expressions can modify the memory store.
Similarly, the statement evaluation judgement captures the updated memory store from evaluating a statement with side-effects and the environment extension on declaring a new variable.
A declaration evaluation can reduce to a new memory store and environment when evaluating a variable or object declaration.
Additionally, a declaration statement can update the type definition context by introducing a new type alias.
Both declarations and statements evaluate to a signal $sig$, representing the result of the control flow in their sequencing blocks.

\subsection{Core P4 Type System}\label{cp4-type-system}

\begin{figure}
$\begin{array}{rcl}
  \rho \quad &::=& bool
            \OR int
            \OR bit \langle n\rangle
            \OR unit \\
             &\OR& \{\overline{f: \rho }\}
             \OR header\{\overline{f: \rho }\}
             \OR \rho[n] \\
             &\OR& match\_kind\{\overline{f}\} \\
  \kappa \quad &::=& \rho 
            \OR table
            \OR \overline{d~\kappa} \rightarrow \kappa
\end{array}$
\caption{Core P4 types \label{fig:coreP4-types}}
\end{figure}

\Cref{fig:coreP4-types} recalls the types from Core P4.  Core P4 divides the P4
types into two categories: base types, $\rho$, and general types, $\kappa$. The
fields of headers and records must be base types.
The simplified Core P4 typing judgements for the fragment of Core P4 presented in \Cref{P4IFC-syntax} are as follows:
\[
\Gamma, \Delta \vdash \nt{exp}: \kappa~goes~d \quad
\Gamma, \Delta \vdash stmt \dashv \Gamma' \quad
\Gamma, \Delta \vdash decl \dashv \Gamma', \Delta'
\]
The expression typing judgement associates a directionality with expressions to indicate if the expression is read only (\terminal{in}) or is both readable and writable (\terminal{inout}).
Intuitively, the contexts on the left of $\vdash$ in the statement and
declaration typing rule describe the contexts before their execution, while the
contexts on the right  of $\dashv$ define the context after the execution of the
statement and declaration.\footnote{%
  The original Core P4 typing judgements also have a constant store, to model compile-time constants.
  We omit this store since our fragment does not include compile-time constants.}

\section{IFC Type System for P4}
\label{sec:ifc}
This section presents the security-type extension
for the Core P4 fragment presented in \Cref{P4IFC-syntax}. Before presenting the security-types for our fragment of Core P4, we describe the main idea behind security type systems.

\subsection{Background on Security Type Systems}
A security type system lifts ordinary types to security types by annotating them
with security labels~\citep{ifc-survey}. These security labels are drawn from a
security lattice, $(\mathbb{L}, \sqsubseteq)$, associated with the type system.
We illustrate the key ideas using a simple two point lattice $\{ \textsf{low},
\textsf{high} \}$. Here, \textsf{low} identifies publicly visible values and
\textsf{high} represents secure values, and $\textsf{low} \sqsubseteq
\textsf{high}$.

Consider a well-typed closed expression $\nt{exp}$ with type $\tau$, represented by an ordinary type system as $\vdash \nt{exp}: \tau$.
A security-type system will additionally assign a security label, $\chi \in \mathbb{L}$ to $\nt{exp}$.
This can be represented by the typing judgement $\vdash \nt{exp}: \type{\tau}{\chi}$, where the pair $\type{\tau}{\chi}$ is the \emph{security type}.
For instance, if $\nt{exp}$ evaluates to $val$ and $\chi=\textsf{high}$, then $val$ is considered to be a secure value.

For statements (or expressions) that can mutate variables, a security type system assigns a security label $pc \in \mathbb{L}$ to the typing judgements.
This label denotes the security context used to track the security level for variables that can be written at a given program point (\emph{program counter}).
Consider a conditional statement that branches on a \textsf{high} security guard expression:
\[
\textsf{if~} (h == 1)~\{~h:=set\_high();\}~\textsf{else}~\{~h:=1;\},
\]
where the security level of $h$ is \textsf{high} and the $set\_high$ function call in the true branch writes to only \textsf{high} security variables.
Since the guard is at \textsf{high} security level, the $pc$ for both the conditional branches becomes \textsf{high}.
Here, both branches need to be well-typed under the \textsf{high} security label, which implies that no variable at security level lower than \textsf{high} can be mutated in  either branch.
For instance, we must have $\Gamma \vdash_{\textsf{high}}h:=set\_high()$ and $\Gamma \vdash_{\textsf{high}}(h:=1)$.
Without this restriction, there can be an implicit flow of information from the conditional guard into the statement blocks of the conditional, for instance, if the function wrote to a \textsf{low} variable.

The utility of a security-type system lies in the \emph{non-interference} guarantee offered by a well-typed program.
To define non-interference, suppose that all \textsf{low} security variables are observable while any \textsf{high} security variable is unobservable.
Informally, non-interference can be understood as the property of a program where no unobservable input variable influences the value of any observable output.

\subsection{P4 IFC Type System}
This section describes our information-flow control type system for the language in \Cref{P4IFC-syntax}.
We assume the lattice $(\mathbb{L}, \sqsubseteq)$ of security labels has $\top$ and $\bot$ elements, representing the top and bottom elements of the lattice.
In our example lattice, $\bot=\textsf{low}$ and $\top=\textsf{high}$.

\Cref{security-types} summarizes the security types of our information-flow control system.
Core P4 types are lifted to security types using a security label, $\chi$, from the lattice $\mathbb{L}$.
We also use $pc$ to denote a security label when it is used as a security context.
As in Core P4, we distinguish between base security types $\rho$ and general security types $\kappa$.
For non-base types, the security label is tracked within the type itself, for instance, the fields of headers and records are assigned security labels instead of the header or record.
But to keep the shape of types uniform, we assign the  $\bot$ security label for such types.
We use the metavariable $\tau$ to denote a security type without its outer-most security label; thus, security types are of the form $\type{\tau}{\chi}$.

\begin{figure}
$\begin{array}{rcl}
  \rho \quad &::=& \type{bool}{\chi}
                \OR \type{int}{\chi}
                \OR \type{bit \langle n\rangle}{\chi}
                \OR \type{unit}{\bot} \\
             &\OR& \type{\{\overline{f: \rho }\}}{\bot}
                \OR \type{header\{\overline{f: \rho }\}}{\bot}
                \OR \type{\rho[n]}{\bot} \\
             &\OR& \type{match\_kind\{\overline{f}\}}{\bot} \\
  \kappa \quad &::=& \rho
                \OR \type{table(pc_{tbl})}{\bot}
                \OR \type{\overline{d~ \rho} \xrightarrow{pc} \rho_{ret}}{\bot}\\
  \tau \quad &::=& bool
                \OR int
                \OR bit\langle n \rangle
                \OR unit \\
             &\OR& \{\overline{f: \rho }\}
                \OR header\{\overline{f: \rho }\} \\
             &\OR& \rho[n]
            \OR match\_kind\{\overline{f}\} \\
             &\OR& table(pc_{tbl})
            \OR \overline{d~\rho} \xrightarrow{pc} \rho_{ret}
\end{array}$
\caption{IFC Types}
\label{security-types}
\end{figure}

Before describing the judgement forms of the security type system, we introduce the contexts used in the typing judgements.  The typing judgements use a typing context, $\Gamma$, a type definition context, $\Delta$, and a security context, $pc$, which are same as Core P4's contexts \Cref{fig:contexts}, with the difference that now $\mathit{Typ}$ is the set of security types of the form $\type{\tau}{\chi}$.

For a given security label $pc$, variables in a typing context $\Gamma$ at security level  $\chi \sqsubseteq pc$
will be referred as \textit{below-pc} variables, and variables at security level $\chi \nsqsubseteq pc$
will be referred as \textit{not below-pc} (or sometimes \emph{above-pc}) variables.

Our security type system has three forms of judgements for expressions,
statements, and declarations, respectively:
\begin{align*}
\textbf{Expressions}:~ &\ordinaryTyping[pc]{\Gamma}{\Delta}{\nt{exp}}{\type{\tau}{\chi}}~goes~d\\
\textbf{Statements}:~ &\stmtTyping{pc}{\Gamma}{\Delta}{stmt}{\Gamma'}\\
\textbf{Declarations}:~ &\declTyping{pc}{\Gamma}{\Delta}{decl}{\Gamma'}{\Delta'}
\end{align*}
The direction annotation $goes~d$ in the typing judgement for expressions is dropped when the direction is not important.
The complete security typing rules can be found in \Cref{P4IFC-typing-exp}
(expressions), \Cref{P4IFC-typing-stmt} (statements), and
\Cref{P4IFC-typing-decl} (declarations).
Expression typing assumes a typing oracle $\mathcal{T}$, giving the meaning of the binary operations.
In statement and declaration typing, the judgement $\Delta \vdash \tau \rightsquigarrow \tau'$ converts $\tau$ to a base type by unfolding type definitions~\citep{petr4}.
Below, we discuss the most interesting---and technically intricate---typing rules: those for functions, tables, and subtyping.

\paragraph*{Typing rules for functions}
Our system has rules for function declarations and function calls.
These are also the key rules for typing actions, which are functions with no return type.
The \textsc{T-FnDecl} rule in \Cref{P4IFC-typing-decl} typechecks the body of
the function to eliminate any leaks in the function body.
The $pc_{fn}$ security label on the function's arrow type records the lower
bound on the security labels of the variables that the function mutates.
For instance, in the following function:
\[
  \textsf{function}~\terminal{insecure}()\{ l:= 1; h:=2;\} ,
\]
where the security labels of $l$ and $h$ variable are \textsf{low} and
\textsf{high} respectively, $pc_{fn}$ will be \textsf{low}.
The \textsc{T-FnCall} rule in \Cref{P4IFC-typing-exp} enforces that a function
will not be invoked in a context that is higher than the function's $pc_{fn}$
because doing so, for instance in the example program, will implicitly flow
information from a \textsf{high} guard expression into a \textsf{low} variable.

\begin{figure*}
\begin{mathpar}
\inferrule*[right=T-Subtype-PC]
{
\ordinaryTyping[pc']{\Gamma}{\Delta}{exp}{\type{\tau}{\chi}} \\
pc \sqsubseteq pc'
}
{
\ordinaryTyping[pc]{\Gamma}{\Delta}{exp}{\type{\tau}{\chi}}
}

\inferrule*[right=T-SubType-In]
{
\ordinaryTyping[pc]{\Gamma}{\Delta}{exp}{\type{\tau}{\chi}~goes~ in} \\
\chi \sqsubseteq \chi'
}
{
\ordinaryTyping[pc]{\Gamma}{\Delta}{exp}{\type{\tau}{\chi'}~goes~ in}
}

\inferrule*[right=T-Bool]
{
}
{
\ordinaryTyping[pc]{\Gamma}{\Delta}{b}{\type{bool}{\bot}~goes~ in}
}

\inferrule*[right=T-Int]
{
}
{
\ordinaryTyping[pc]{\Gamma}{\Delta}{n_{\infty}}{\type{int}{\bot}~goes~ in}
}

\inferrule*[right=T-Var]
{   x \in \dom{\Gamma} \qquad
\Gamma(x) = \type{\tau}{\chi}
}
{
\ordinaryTyping[pc]{\Gamma}{\Delta}{x}{\type{\tau}{\chi} \text{~goes inout}}
}

\inferrule*[right=T-BinOP]
{
\ordinaryTyping[pc]{\Gamma}{\Delta} {exp_{1}}{\type{\rho_{1}}{\chi_{1}}} \\
\ordinaryTyping[pc]{\Gamma}{\Delta}{exp_{2}}{\type{\rho_{2}}{\chi_{2}}} \\\\
\mathcal{T}(\Delta; \oplus; \rho_{1}; \rho_{2}) = \rho_{3} \\
\chi_{1} \sqsubseteq \chi' \\
\chi_{2} \sqsubseteq \chi'
}
{
\ordinaryTyping[pc]{\Gamma}{\Delta}{exp_{1} \oplus exp_{2}}{\type{\rho_3}{\chi'} ~goes~ in}
}

\inferrule*[right=T-Rec]
{
\listOrdinaryTyping[pc]{\Gamma}{\Delta}{\{\overline{exp: \type{\tau_i}{\chi_i}}\}}
}
{
\ordinaryTyping[pc]{\Gamma}{\Delta}{\{ \overline{f: exp} \}}{\type{\{ \overline{f: \langle \tau_i, \chi_i \rangle} \}}{\bot}~goes~ in}
}

\inferrule*[right=T-MemRec]
{
\ordinaryTyping[pc]{\Gamma}{\Delta}{exp}{\type{\{ \overline{f_i: \langle \tau_i, \chi_i \rangle} \}}{\bot}}~goes~d
}
{
\ordinaryTyping[pc]{\Gamma}{\Delta}{exp.f_{i}}{\type{\tau_{i}}{\chi_{i}}~ goes~ d}
}

\inferrule*[right=T-Index]
{
\ordinaryTyping[pc]{\Gamma}{\Delta}{exp_{1}}{\type{\type{\tau}{\chi_1}[n]}{\bot}~goes~ d} \\\\
\ordinaryTyping[pc]{\Gamma}{\Delta}{exp_{2}}{\type{bit \langle 32 \rangle}{\chi_2}} \\
\chi_2 \sqsubseteq \chi_1
}
{
\ordinaryTyping[pc]{\Gamma}{\Delta} {exp_{1}[exp_{2}]}{\type{\tau}{\chi_1}~ goes~ d }
}

\inferrule*[right=T-MemHdr]
{
\ordinaryTyping[pc]{\Gamma}{\Delta}{exp} {\type{header \{ \overline{f_i: \langle \tau_i, \chi_i \rangle} \}}{\bot}~ goes~ d}
}
{
\ordinaryTyping[pc]{\Gamma}{\Delta}{exp.f_{i}}{\type{\tau_{i}}{\chi_{i}}~goes~ d}
}

\inferrule*[right=T-Call]
{
\Gamma, \Delta \vdash_{pc} exp_{1}: \langle \overline{d~\type{\tau_i}{\chi_i}}
\xrightarrow{pc_{fn}} \langle \tau_{ret}, \chi_{ret}\rangle, \bot \rangle \\\\
\Gamma, \Delta \vdash_{pc} \overline{exp_{2}:\type{\tau_i}{\chi_i}~ goes~d} \\
pc \sqsubseteq pc_{fn}
}
{
\Gamma, \Delta \vdash_{pc} exp_{1} (\overline{exp_{2}}): \type{\tau_{ret}}{\chi_{ret}} ~\text{goes in}
}
\end{mathpar}
    \caption{IFC Typing Rules for Expressions}
    \label{P4IFC-typing-exp}
    \end{figure*}

    \begin{figure*}
    \begin{mathpar}
\inferrule*[right=T-Empty]
{
~
}
{
\stmtTyping{pc}{\Gamma}{\Delta}{\{ \}}{\Gamma}
}

\inferrule*[right=T-Exit]
{
~
}
{
\stmtTyping{\bot}{\Gamma}{\Delta}{\terminal{exit}}{\Gamma}
}

\inferrule*[right=T-Seq]
{
\Gamma, \Delta \vdash_{pc} stmt_{1} \dashv \Gamma_{1} \qquad
\Gamma_{1}, \Delta \vdash_{pc} \{ \overline{stmt_{2}} \} \dashv \Gamma_{2}
}
{
\Gamma, \Delta \vdash_{pc} \{ stmt_{1}; \overline{stmt_{2}} \} \dashv \Gamma_2
}

\inferrule*[right=T-Assign]
{
\Gamma,\Delta \vdash_{pc} exp_{1}:\type{\tau}{\chi_{1}}~goes~ inout \\\\
\Gamma,\Delta \vdash_{pc} exp_{2} :\type{\tau}{\chi_{2}}\\
\chi_{2} \sqsubseteq \chi_{1}\\
pc \sqsubseteq \chi_{1}
}
{
\Gamma, \Delta \vdash_{pc} exp_{1} := exp_{2} \dashv \Gamma
}

\inferrule*[right=T-Cond]
{
\stmtTyping{\chi_2}{\Gamma}{\Delta}{stmt_{1}}{\Gamma_{1}}\\
\stmtTyping{\chi_2}{\Gamma}{\Delta}{stmt_{2}}{\Gamma_{2}}\\\\
\ordinaryTyping[pc]{\Gamma}{\Delta}{exp}{\type{bool}{\chi_1}}\\
\chi_1 \sqsubseteq \chi_2 \\
pc \sqsubseteq \chi_2
}
{
\stmtTyping{pc}{\Gamma}{\Delta}{\terminal{if}~(exp)~~stmt_{1}~\terminal{else}~stmt_{2}}{\Gamma}
}

\inferrule*[right=T-Return]
{
\Gamma,\Delta \vdash_{pc}  exp: \langle \tau, \chi_{ret} \rangle \\
\Gamma(\terminal{return}) = \type{\tau_{ret}}{\chi_{ret}} \\
\Delta    \vdash \tau_{ret} \rightsquigarrow \tau
}
{
\Gamma, \Delta \vdash_{\bot} \terminal{return}~ exp \dashv \Gamma
}

\inferrule*[right=T-Decl]
{
\Gamma, \Delta \vdash_{pc} var\_decl \dashv \Gamma_1, \Delta
}
{
\Gamma, \Delta \vdash_{pc} var\_decl \dashv \Gamma_1
}

\inferrule*[right=T-FnCallStmt]
{
\Gamma, \Delta \vdash_{pc} exp_1(\overline{exp_2}): \type{\tau_{ret}}{\chi_{ret}}
}
{
\Gamma, \Delta \vdash_{pc} exp_1(\overline{exp_2}) \dashv \Gamma
}

\inferrule*[right=T-TblCall]
{
\Gamma, \Delta \vdash_{pc} exp: \langle table(pc_{tbl}), \bot \rangle \\
pc \sqsubseteq pc_{tbl}
}
{
\Gamma, \Delta \vdash_{pc} exp() \dashv \Gamma
}
\end{mathpar}
    \caption{IFC Typing Rules for Statements}
    \label{P4IFC-typing-stmt}
    \end{figure*}

    \begin{figure*}
    \begin{mathpar}
    \inferrule*[right=T-VarDecl]
    {~}
    {\Gamma, \Delta \vdash_{pc} \type{\tau}{\chi}~ x \dashv \Gamma [x: \langle \tau, \chi \rangle], \Delta}

    \inferrule*[right=T-VarInit]
{
\ordinaryTyping[pc]{\Gamma}{\Delta}{exp}{\type{\tau'}{\chi}} \\
\Delta \vdash \tau \rightsquigarrow \tau'
}
{
\Gamma; \Delta \vdash_{pc} \type{\tau}{\chi}~ x:= exp \dashv \Gamma [x: \langle \tau', \chi\rangle]; \Delta
}

      \inferrule*[right=T-TblDecl]
    {
    \newordinaryTyping[pc_{tbl}]{\Gamma}{\Delta}{\overline{exp_k:\type{\tau_k}{\chi_k}}}{} \\
    \newordinaryTyping[pc_{tbl}]{\Gamma}{\Delta}{\overline{x_k:\type{match\_kind}{\bot}}}{} \\
    {\chi_k} \sqsubseteq {pc_{tbl}}~\text{for all}~k\\\\
    \newordinaryTyping[pc_{tbl}]{\Gamma}{\Delta}{act_{a_j}: \type{\overline{d\type{\tau_{a_{ji}}}{\chi_{a_{ji}}}}~;\overline{\type{\tau_{c_{ji}}}{\chi_{c_{ji}}}} \xrightarrow{pc_{fn_j}} \type{unit}{\bot}}{\bot}}{}
    \\
    pc_a \sqsubseteq pc_{fn_j}~\text{for all}~j \\\\
    \newordinaryTyping[pc_{tbl}]{\Gamma}{\Delta}{\overline{exp_{a_{ji}}:\type{\tau_{a_{ji}}}{\chi_{a_{ji}}}~goes~d}}{} \\
    {\chi_k} \sqsubseteq {pc_{fn_j}}~\text{for all}~j,k \\
    pc_{tbl} \sqsubseteq pc_a
    }
    {
    \declTyping{pc}{\Gamma}{\Delta}{\text{table}~x~ \{\overline{exp_k: x_k}~ \overline{act_{a_j}(\overline{exp_{a_{ji}}})}\}}{\Gamma[x: \type{table(pc_{tbl})}{\bot}]}{\Delta}
    }

\inferrule*[right=T-FuncDecl]
{
\Gamma_1 = \Gamma[\overline{x_{i}: \type{\tau_{i}'}{\chi_{i}}}, \terminal{return}: \type{\tau_{ret}'}{\chi_{ret}}]\\
\declTyping{pc_{fn}}{\Gamma_{1}}{\Delta}{stmt}{\Gamma_{2}} \\\\
\Delta \vdash \tau_i \rightsquigarrow \tau_i'~\text{for each}~\tau_i \\
\Delta \vdash \tau_{ret} \rightsquigarrow \tau_{ret}' \\
\Gamma' = \Gamma[x: \type{\overline{d~\type{\tau_{i}'}{\chi_{i} }} \xrightarrow[]{pc_{fn}} \type{\tau_{ret}'}{\chi_{ret}}}{\bot}]
}
{
\declTyping{pc}{\Gamma}{\Delta}{\terminal{function}~\type{\tau_{ret}}{\chi_{ret}}~x~(\overline{d~ x_{i}: \type{\tau_{i}}{\chi_{i}}}) \{stmt\}}{\Gamma'}{\Delta}
}

    \end{mathpar}

    \caption{IFC Typing Rules for Declaration}
    \label{P4IFC-typing-decl}
    \end{figure*}

\paragraph*{Typing rules for tables}
Since a table matches on the key to select an action to invoke, the key of a
table resembles the guard of a conditional. Thus, the value of a key can
implicitly leak in the action's body if the invoked action writes to variables
at security label lower that that of the key expression.
Therefore, to declare a table of type $\type{table(pc_{tbl})}{\bot}$,
the rule {\textsc{T-TblDecl}} in \Cref{P4IFC-typing-decl} ensures that the security
label of the most secure key, $\chi_k$, is lower than the label of the least
secure assignment, $pc_a$, in any action. Here, $pc_{tbl}$ records the lower bound on the write effects  associated with any keys, actions, or arguments.

The \textsc{T-TblCall} rule in \Cref{P4IFC-typing-stmt} prevents any implicit flow into any of the actions that a table might invoke by allowing a  table
 to be applied only in a $pc$ context lower than the least secure write effect
 associated with the table application, $pc_{tbl}$. This prevents implicit leaks during the evaluation of keys, arguments, or the action's body.

\paragraph*{Subtyping rule}
The \textsc{T-SubType-In} rule in \Cref{P4IFC-typing-exp} allows only read-only
($in$) expressions to increase their security label. It is not safe to allow
$inout$ expressions to be subtyped. To see why, consider the following function:
\[
  \terminal{write_to_high}~(inout~ \textsf{h} : \langle bool, high \rangle)~\{\textsf{h}:=true;\}
\]
Suppose we have a \textsf{low} variable $\textsf{l} : \langle bool, low \rangle$. Since variables are $inout$ expressions (\textsc{T-Var} in \Cref{P4IFC-typing-exp}), if $inout$ expressions were allowed to increase their label, $\terminal{write_to_high}(\textsf{l})$ call would have been valid. In this case, the function would have written to a \textsf{low} variable when it should have operated with only a \textsf{high} variable.

\subsection{Non-Interference}
\label{subsec:noninterference}

To define non-interference, consider two program states,
$ \langle \mathcal{C},\Delta, \mu_a, \epsilon_a\rangle$ and $ \langle \mathcal{C}, \Delta, \mu_b, \epsilon_b\rangle$,
where the environments have equal domains. Suppose every below-pc variable $x$ has equal value under both the
memory stores, $\mu_a(\epsilon_a(x)) = \mu_b(\epsilon_b(x))$, but the value of any variables that are not below-pc can differ between the two stores. Non-interference is satisfied if evaluating an expression, statement, or declaration
in the two program states results in two final program states that agree on below-pc variables.



The following definition formally describes a pair of below-pc equivalent memory stores and environments.
The store typing context $\Xi$ maps locations in a store to security types.
\begin{definition}
Consider two pairs of memory stores and environments $\langle \mu_a, \epsilon_a \rangle$ and
$\langle \mu_b, \epsilon_b \rangle$. Then
\[
  \semanticBelowPCState{pc}{\Xi_a}{\Xi_b}{\Delta}{\mu_{a}}{\epsilon_{a}}{\mu_{b}}{\epsilon_{b}}{\Gamma}
\]
is satisfied when
\[
  \semanticStoreEnv{\Xi_a}{\Delta}{\mu_a}{\epsilon_a}{\Gamma}
  \qquad\text{and}\qquad
  \semanticStoreEnv{\Xi_b}{\Delta}{\mu_b}{\epsilon_b}{\Gamma}
\]
and every below-pc variable $x$ in $\epsilon_a$ and $\epsilon_b$ has equal value
\ie $\mu_a(\epsilon_a(x)) = \mu_b(\epsilon_b(x))$.
\end{definition}
Intuitively, $\semanticStoreEnv{\Xi}{\Delta}{\mu}{\epsilon}{\Gamma}$ states that
the store and environment are well-typed: recalling that the location of every
variable is described by the environment $\epsilon$ and the value at valid
locations is described by the memory store $\mu$, the type assigned to a
variable using the store typing $\Xi$ must be the same as the type assigned by
the typing context $\Gamma$.
\iffull
The formal definition for this relation is provided in \Cref{def:mem-store-pair-semantic}.
\fi

The following definition of non-interference for statements requires that
evaluating a statement under below-pc equivalent pairs of memory stores and
environment can only reduce to pairs of final memory stores and environments
that are below-pc equivalent. Technically, this is a \emph{termination
insensitive} notion of non-interference, since it does not require that both
executions terminate. However, P4 programs do not allow recursion and
\citet{petr4} prove that all well-typed Core P4 programs terminate.

\begin{definition}[Non-interference for statements]
	For any security lable $l$, $\NIstmt{pc}{\Gamma}{\Delta}{stmt}{\Gamma'}$ holds
  for any 	 $\Xi_a$, $\Xi_b$, $\mu_{a}$, $\mu_{b}$, $\epsilon_{a}$,
  $\epsilon_{b}$, $\mu_{a}'$, $\mu_{b}'$, $\epsilon_{a}'$, $\epsilon_{b}'$ if
  whenever
\begin{enumerate}
\item  $\semanticBelowPCState{l}{\Xi_a}{\Xi_b}{\Delta}{\mu_{a}}{\epsilon_{a}}{\mu_{b}}{\epsilon_{b}}{\Gamma}$,
\item $\evalsto{\config[stmt]{\mathcal{C}; \Delta}{\mu_{a}}{\epsilon_{a}}}{\config{\mu_{a}'}{\epsilon_{a}'}{sig_{1}}}$,
\item $\evalsto{\config[stmt]{\mathcal{C}; \Delta}{\mu_{b}}{\epsilon_{b}}}{\config{\mu_{b}'}{\epsilon_{b}'}{sig_{2}}}$
  \end{enumerate}
  then there exists $\Xi_a'$, $\Xi_b'$, such that
  \begin{enumerate}
    \item $\stmtTyping{pc}{\Gamma}{\Delta}{stmt}{\Gamma'}$,
    \item $\semanticBelowPCState{l}{\Xi_a'}{\Xi_b'}{\Delta} {\mu_{a}'}{\epsilon_{a}'}{\mu_{b}'}{\epsilon_{b}'}{\Gamma'}$,
    \item $\semanticBelowPCState{l}{\Xi_a'}{\Xi_b'}{\Delta} {\mu_{a}'}{\epsilon_{a}}{\mu_{b}'}{\epsilon_{b}}{\Gamma}$,
    \item for any $l_a \in \dom{\mu_a}$ and $l_b \in \dom{\mu_b}$ such that $\ordinaryTyping[]{\Xi_a}{\Delta}{\mu_a(l_a)}{\type{\tau}{\chi}}$ and $\ordinaryTyping[]{\Xi_b}{\Delta}{\mu_b(l_b)}{\type{\tau}{\chi}}$ and $pc \nsqsubseteq \chi$, we have $\mu_{a}'(l_a) = \mu_{a}(l_a)$ and $\mu_{b}'(l_b) = \mu_{b}(l_b)$,
   \item for any $l_a \in \dom{\mu_a}$ such that
     $\ordinaryTyping[]{\Xi_a}{\Delta}{\mu_a(l_a)}{\tau_{clos}}$, where
     $\tau_{clos} \in \{\tau_{fn}, \tau_{tbl}\}$, we have $\mu_a'(l_a) = \mu_a(l_a)$,
   \item for any $l_b \in \dom{\mu_b}$ such that
     $\ordinaryTyping[]{\Xi_b}{\Delta}{\mu_b(l_b)}{\tau_{clos}}$, where
     $\tau_{clos} \in \{\tau_{fn}, \tau_{tbl}\}$, we have $\mu_b'(l_b) = \mu_b(l_b)$,
   \item one of the following holds:
     \begin{itemize}
       \item $sig_1=sig_2= cont$; or
       \item $sig_1=sig_2= exit$; or
       \item $sig_1=\textsf{return}~val_1$ and $sig_2=\textsf{return}~val_2$ such
         that $\Xi_a', \Xi_b', \Delta \models_{l} NI(val_{1}, val_2):
         \type{\tau_{ret}'}{\chi_{ret}}$, where $\Delta \vdash \tau_{ret}
         \rightsquigarrow \tau_{ret}'$ and $\Gamma[\terminal{return}] =
         \type{\tau_{ret}}{\chi_{ret}}$,
     \end{itemize}
	 \item we have the inclusions:
     \begin{itemize}
       \item $\Xi_a \subseteq \Xi_a'$ and $\Xi_b \subseteq \Xi_b'$;
       \item $\dom{\mu_a} \subseteq \dom{\mu_a'}$ and $\dom{\mu_b} \subseteq
         \dom{\mu_b'}$; and
       \item $\dom{\epsilon_a} \subseteq \dom{\epsilon_a'}$ and $\dom{\epsilon_b} \subseteq \dom{\epsilon_b'}$.
     \end{itemize}
 \end{enumerate}
\end{definition}
We present similar non-interference definitions for expressions and declarations in
\iffull
\Cref{def:def-ni-exp} and \Cref{def:def-ni-decl}.
\fi

Then, our main soundness theorem states that a well-typed program in our information-flow control type system will be non-interfering.
\begin{theorem}[Main Soundness Theorem]\label{thm:def-ni-stmt}
If $~\stmtTyping{pc}{\Gamma}{\Delta}{stmt}{\Gamma'}$, then  $\NIstmt{pc}{\Gamma}{\Delta}{stmt}{\Gamma'}$.
\end{theorem}
We present similar non-interference theorems for expressions and declarations in
\iffull
\Cref{ni-exp} and \Cref{ni-decl}.
\fi
\begin{proofsk}
  We prove non-interference theorems for statements, expressions and
  declarations together as a mutual induction on the typing derivation. The
  detailed proof of \Cref{thm:def-ni-stmt} is given in
  \iffull
    \Cref{app:proof-ni}.
  \fi
  The most involved case is the rule for function calls (\textsc{T-FnCall}),
  where we must slightly strengthen the non-interference definition for
  expressions, statements, and declarations.
\end{proofsk}

\section{Implementation and Case Studies}
\label{sec:eval}

To evaluate our type system, we implemented a type-checker for annotated P4
programs and used it to analyze a range of example programs exhibiting
different kinds of errors.
We call our tool \name.
Our information-flow control type system is implemented as an extension of the type checker in the
\textsc{p4c} compiler~\citep{p4c}, the reference compiler for
$\textsc{P4}_{16}$~\citep{p416}.
The target of our type checker is the \textsf{simple_switch} based on the \textsc{BMv2} behavioral model.
Our implementation adds about $700$ LOC to \textsc{p4c} and supports the
$\mathcal{L} = \{\textsf{high}, \textsf{low}\}$ lattice, and a simple diamond
lattice from \Cref{fig:lattice}, $\mathcal{L} = \{\textsf{high}, \textsf{alice},
\textsf{bob}, \textsf{low}\}$ for modeling isolation specifications.
Standard P4 types can be annotated with a security label from the lattice;
unannotated types default to \textsf{low}.

We evaluate our implementation by comparing the typechecking time of the secure
programs presented in the case studies using the \name typechecker with the
typechecking time of their uninstrumented insecure counterparts using the
original \textsc{p4c} compiler. \Cref{tab:eval-time} shows that our
implementation incurs an overhead of 5\% (or 30ms) on average in comparison to
the reference \textsc{p4c} compiler when evaluated on the instrumented and
uninstrumented versions of the same program. We believe this overhead is
reasonable for an unoptimized implementation that builds on the stock p4c
compiler; developing a more optimized implementation is a direction for future
work.

\begin{table}[htbp]
\centering
\caption{Typechecking time in milliseconds.}
\begin{tabular}{lccc}
\hline
Program               & Unannotated, p4c            & Annotated, \name \\
\hline
D2R &  534 & 599 \\
App & 593 & 600 \\
Lattice & 495 & 527\\
Topology & 554 & 591\\
Cache & 538 & 550 \\
\hline
Average & 543 & 573 \\
\hline
\end{tabular}
\label{tab:eval-time}
\end{table}

In the rest of the section, we present our case studies.

\subsection{Dataplane Routing with Priorities}
In traditional networks, the control plane is responsible for \emph{routing}, determining how to send a packet from source to destination, while the data plane is responsible for \emph{forwarding}, sending a packet to its next hop.
\citet{d2r} have shown that using programmable switches, one can handle routing in the data plane, avoiding the control plane entirely.
In their scheme, called D2R, when a switch receives a packet, it uses pre-loaded information about the network topology and local knowledge about link failures to perform a breadth-first search (BFS) and find a path to the target destination address.
D2R uses P4 mechanisms (e.g., stacks) to perform the BFS computation entirely on the switch, without needing to communicate with the control plane.

\begin{lstlisting}[caption={D2R: Dataplane Routing},label={lst:d2r}, numbers=left,  breaklines=true, escapechar=|, style=topper]
header bfs_t {
  <bit<32>, low> curr;
  <bit<32>, low> tried_links;
  <bit<32>, high> num_hops;
  // ...
}
header ipv4_t {
  <bit<3>, low> priority;
  // ...
}
struct headers {
  bfs_t bfs;
  ipv4_t ipv4;
  // ...
}

control D2R_Ingress(headers hdr) {
  <bit<32>, high> failures
  = num_bits_set(hdr.bfs.tried_links) - hdr.bfs.num_hops; |\label{ln:d2r:num-fails}|

  table bfs_step { ... }
  table forward {
    key = { hdr.bfs.next_node: exact; }
    actions = { forwarding(failures); NoAction; }
  }
  action forwarding(in <bit<32>, high> failures) {
    if (failures >= THRESHOLD) {
      hdr.ipv4.priority = PRIO_1; // Leak |\label{ln:d2r:leak1}|
    }
    else {
      hdr.ipv4.priority = PRIO_2; // Leak |\label{ln:d2r:leak2}|
    }
    // ... normal forwarding logic ...
  }
  apply {	|\label{ln:d2r:apply}|
    if (hdr.bfs.curr != hdr.ipv4.dstAddr) {
      bfs_step.apply(); |\label{ln:d2r:not-completed}|
    } else {
      forward.apply(); |\label{ln:d2r:completed}|
    }
    // repeat applications of bfs
  }
}
\end{lstlisting}%

We consider an extension of D2R where packets that encounter a higher number of link failures will receive higher priority.
\Cref{lst:d2r} gives schematic code for the main headers and control block implementing this variant of data plane routing.
The \textsf{bfs_t} headers describe the auxiliary information carried in the packets to perform the BFS, e.g., which links have been tried, while the \textsf{ipv4_t} headers contain information for standard packet forwarding.
In the control block \textsf{D2R_Ingress}, the number of failures count (\Cref{ln:d2r:num-fails}) can be computed from the vector of links that have been tried, \textsf{hdr.bfs.tried_links}, and the number of traversed links, \textsf{hdr.bfs.num_hops}.
The table \textsf{bfs_step} performs one step of BFS; the details are not important for our purposes.
Since P4 does not support loops, an iterative search algorithm like BFS is modeled in the apply block on \Cref{ln:d2r:apply} by unrolling the loop.
If the BFS search has not completed, \ie the current node in the BFS search is not the destination node (\Cref{ln:d2r:not-completed}), the BFS table is applied again (we elide the details of this BFS search algorithm which can be found in \cite{d2r}).
When the BFS search has successfully completed (\Cref{ln:d2r:completed}), the forwarding table is applied and packet priorities are assigned based on the number of failures encountered by the packet.

Using failure information to prioritize packets may leak information.
For instance, there are several potential reasons why \textsf{hdr.bfs.num_hops} could be secret---e.g., the packet could be transiting a private network and one might not want to reveal whether the network has reliable or unreliable links.
If \textsf{hdr.bfs.num_hops} is annotated as high security, the program is rejected by our typechecker because the \textsf{forwarding} action writes data to the low-security priority after branching on the number of the failures, which is high security (\Cref{ln:d2r:leak1,ln:d2r:leak2}).
This is an example of an indirect leak: the program branches on the secret, and then writes to public fields.

To remedy this information leak, we can modify the scheme so that the priority is computed based on non-sensitive information.
For instance, we can assign priority based on the total number of links that a packet tried to cross.
This count is an approximate proxy for the number of failures: as the number of failures rises, the packet tries more links.
This change can be implemented by removing \textsf{hdr.bfs.num_hops} in \Cref{ln:d2r:num-fails}, giving a program that is accepted by our typechecker.

A similar kind of leak can manifest in the implementation of NetChain \citep{netchain}, an in-network implementation of chain replication on top of a key-value store. The implementation assigns \textsf{roles} to the various switches in the network to determine the head, tail, or internal nodes  of the chain, which among various actions determines if the node sends out a reply or not. If the \textsf{roles} header field is labeled as a secret field, this can give away private topological information. When instrumented with a \textsf{high} label on \textsf{role}, the typechecker flagged implicit leaks in the implementation.

\subsection{Modeling Timing for In-Network Caching}
Like other IFC systems, our type system can model different notions of
adversary-observable data. For an example, we can consider a key-value store
with an in-network cache~\citep{netcache}. These systems are a prominent
application of data plane computing: switches can quickly retrieve hot items,
keep track of which items are frequently requested, and notify the controller
about which items should be stored on the switch. While the result of a query
should be the same no matter where the item is stored, an observer may be able
to detect variations in timing: data that is stored on the switch is returned
faster, while data that is stored on the controller takes longer to access. In
some cases, this \emph{timing side-channel} may allow an adversary to learn
about the state of the system.

While Core P4 does not model timing aspects of program behavior, we can still model timing information leaks by augmenting the program with new variables holding data that a timing-sensitive adversary may be able to observe.
For example, \Cref{lst:cache} gives a schematic P4 program implementing a simple cache.
The switch first tries to fetch data locally (\Cref{ln:fetchlocally}).
If the request hits then the table runs action \textsf{cache_hit}, while if the request misses then the table runs action \textsf{cache_miss}.
Both actions record the hit or miss in \textsf{hdr.resp.hit}.
We mark this field as a low-security (publicly visible) variable, to model an adversary who can distinguish whether a request was serviced by the cache or the controller.
If the query is sensitive information, \textsf{hdr.req.query} is declared as high security.
Our typechecker rejects this program because of an information leak: the actions \textsf{cache_hit} and \textsf{cache_miss} write to the low-security field \textsf{hdr.response.hit} (\Cref{ln:cache:leak1,ln:cache:leak2}), but they are invoked in a table with a high-security key \textsf{hdr.req.query} (\Cref{ln:query}).
This is again an indirect leak, modeling a simple timing side-channel.

\begin{lstlisting}[caption={In-network cache},label={lst:cache}, numbers=left, keywords={low,high},  breaklines=true, escapechar=|, style=topper]
header request_t { <bit<8>, high> query; }
header response_t { <bool, low> hit; <bit<32>, low> value; }
struct headers { request_t req; response_t resp; eth_t eth; }

control Cache_Ingress(headers hdr) {
  action cache_hit(<bit<32>, low> value) {
    hdr.resp.value = value;
    hdr.resp.hit = true; |\label{ln:cache:leak1}|
  }
  action cache_miss() { hdr.resp.hit = false; |\label{ln:cache:leak2}| }
  table fetch_from_cache {
    key = { hdr.req.query: exact; } |\label{ln:query}|
    actions = { cache_hit; cache_miss; }
  }
  apply {
    fetch_from_cache.apply(); |\label{ln:fetchlocally}|
    // ... if miss, try to fetch from controller ...
  }
}
\end{lstlisting}%

\subsection{Preventing Manipulation in Resource Allocation}
The examples we have seen so far use IFC to guarantee confidentiality: secret information (\textsf{high}) should not leak into publicly visible outputs (\textsf{low}).
As is well-known, if we interpret high-security data as ``untrusted'' and low-security data as ``trusted'', IFC systems can also ensure \emph{integrity}: untrusted inputs should not affect trusted outputs.
To demonstrate, suppose several applications are running on separate subnetworks behind a single gateway switch, which is responsible for forwarding packets to their destination subnetwork and allocate resources to the application flows.
We consider a very simple form of resource allocation, where a switch caters to the needs of latency-sensitive applications by increasing the priority of packets belonging to such applications.
%
%
The P4 program in \Cref{lst:app} gives the main logic for a gateway switch that accomplishes this task.
In addition to ordinary IP headers, packet headers in this setting also include an application ID \textsf{hdr.app.appID} indicating which application the packet belongs to.
In the control block, the table \mbox{\textsf{app_resources}} matches on the application ID, and then calls \mbox{\textsf{set_priority}} with the desired priority level.
This action then sets the priority level of the packet by writing to \textsf{hdr.ipv4.priority} (\Cref{ln:setpriority}).
Finally, the switch forwards the packet to the destination address \textsf{hdr.ipv4.dstAddr}.

While this program behaves well when clients are honest, a malicious client may manipulate the switch to increase the priority of their packets.
Specifically, since \textsf{hdr.app.appID} is used to determine priority but not used to forward the packets, a client may report a false application ID.
This issue can be detected by our IFC system if we label \textsf{hdr.app.appID} as untrusted (\textsf{high}) and \textsf{hdr.ipv4.priority} as trusted (\textsf{low}): setting priority based on application ID is an information-flow violation.

\begin{lstlisting}[caption={Resource Allocation},label={lst:app}, numbers=left,  breaklines=true, escapechar=|, style=topper]
header app_t { <bit<8>, high> appID; }
header ipv4_t {
  <bit<32>, low> dstAddr;
  <bit<32>, low> priority;
  // ...
}
struct headers {
  app_t app;
  ipv4_t ipv4;
  // ...
}

control App_Ingress(headers hdr) {
  action set_priority(<bit<3>, low> priority) {
    hdr.ipv4.priority = priority; |\label{ln:setpriority}|
  }
  table app_resources {
    key = { hdr.app.appID: exact; } |\label{ln:priority:key}|
    actions = { set_priority; }
  }
  apply {
    set_priority.apply();
    // ... forward the packet to hdr.ipv4.dstAddr ...
  }
}
\end{lstlisting}%

To address this problem, we can set the priority based on the destination address instead, by matching on \textsf{hdr.ipv4.dstAddr} instead of \textsf{hdr.app.appID} on \Cref{ln:priority:key}.
It is reasonable to model this header as trusted (\textsf{low}) because if a client were to manipulate this data, the packet would be delivered to the wrong destination.
In the modified program, the priority is now only computed based on trusted data in \textsf{hdr.ipv4.dstAddr} and the typechecker accepts this program because there is no integrity violation.

\subsection{Ensuring Network Isolation}
The previous example changes the interpretation of security labels in order to
establish different properties with IFC. For our final case study, we show how
our type system can use a richer lattice to enforce network isolation
properties.

Suppose we have a private network used by two clients, Alice and Bob, who run
dataplane programs on two separate nodes (the precise topology is not important,
but a sketch can be see in \Cref{fig:switch}).  Nodes pass around a shared
packet header with separate fields for Alice and for Bob, and we want to ensure
that Alice does not touch Bob's fields, and vice versa.  Furthermore, the
network operator wants to carry telemetry data alongside the packets
(\emph{in-band network telemetry}~\citep{int-whitepaper}) this data may depend
on Alice or Bob's data, but neither Alice nor Bob should be able to use
telemetry data.

\begin{figure}[!t]
  \begin{subfigure}[b]{.45\linewidth}
    \center
    \begin{tikzpicture}[scale=.35]
      \node (one) {Bob Switch};
      \node[below=of one ] (a) {Alice switch};
      \draw (a) -- (one);
    \end{tikzpicture}
    \caption{Network Topology}
    \label{fig:switch}
  \end{subfigure}%
  \begin{subfigure}[b]{.45\linewidth}
    \center
    \begin{tikzpicture}[scale=.35]
      \node (one) at (0,2) {$\top$};
      \node (a) at (-2,0) {$B$};
      \node (b) at (2,0) {$A$};
      \node (zero) at (0,-2) {$\bot$};
      \draw (zero) -- (a) -- (one) -- (b) -- (zero) -- (a) -- (zero);
    \end{tikzpicture}
    \caption{Security Lattice}
    \label{fig:lattice}
  \end{subfigure}
  \caption{Security lattice for a network topology\label{fig:lattice-ex}}
\end{figure}
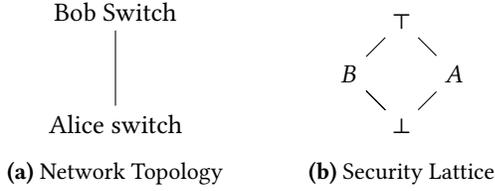

We can model this isolation property as non-interference with a four-point
diamond lattice with labels $\{ A, B, \top, \bot \}$ (\Cref{fig:lattice}).
Non-interference ensures that data from level $\chi$ can flow to variables
labeled $\chi'$ if and only if $\chi \sqsubseteq \chi'$. Thus, if we label
Alice's fields $A$ and label Bob's fields $B$, then Alice's data cannot influence
Bob's fields, and vice versa. Similarly, $\top$-labeled fields can depend on all
data, but cannot influence data below $\top$. For instance, telemetry data can
be labeled $\top$: both Alice and Bob can accumulate data into $\top$-labeled
fields (e.g., increment a counter), but neither Alice nor Bob are able to leak
information from $\top$-labeled data into their own fields. Finally, fields
labeled $\bot$ contain globally visible data that cannot depend on other fields
above $\bot$. For example, we can pre-configure a packet's route through the
private network in $\bot$-labeled fields: this ensures that information from
Alice or Bob does not influence routing, potentially leading to an indirect leak
or isolation failure.

Labeling data from the four-point lattice can already rule out many kinds of
leaks. However, it still allows some leaks involving $\bot$-labeled data. For
instance, Alice may write Bob's fields with $\bot$-labeled data, while Bob may
use $\bot$-labeled data to modify $\bot$-labeled data. While potentially
undesirable, neither of these actions violates IFC since high data is allowed to
depend on low data. To prevent these behaviors, we can additionally typecheck
Alice's code with $pc$ label $A$, and typecheck Bob's code with $pc$ label $B$.
Then, non-interference guarantees that Alice can only write to fields labeled
$A$ or $\top$, and Bob can only write to fields labeled $B$ or $\top$.

\begin{lstlisting}[caption={Network Isolation and Telemetry},label={lst:isolation}, numbers=left,  breaklines=true, escapechar=|, style=topper]
struct headers {
  <alice_t, A> alice_data;
  <bob_t, B> bob_data;
  <telem_t, top> telem;
  <eth_t, bot> eth;
}

// typed at pc = |\color{pred}{A}|
control Alice_Ingress(headers hdr) {
  action set_by_alice(<bob_t, A> value) {
    // Error: should not have written to Bob's field
    hdr.bob = value; |\label{ln:isolation:explicit}|
  }
  table update_by_alice {
    // Error: should not have used telemetry field
    key = { hdr.telem: exact; } |\label{ln:isolation:implicit}|
    actions  = { set_by_alice; }
  }
  apply { update_by_alice.apply(); }
}

// typed at pc = |\color{pred}{B}|
control Bob_Ingress(headers hdr) {
  action set_by_bob() {
    // Allowed: modify telemetry using telemetry information
    hdr.telem = hdr.telem + 1;
  }
  table update {
    key = { hdr.eth.dstAddr: exact; }
    actions = { set_by_bob; NoAction; }
  }
  apply { update_by_bob.apply(); }
}
\end{lstlisting}%

\Cref{lst:isolation} shows schematic versions of programs implementing the Alice
and Bob switches.  Both the switches have a single action. The packet header
carries one of the four security labels.  In this example, we consider that
\textsf{hdr.alice_data} and \textsf{hdr.bob_data} are Alice's and Bob's data,
respectively; \textsf{hdr.eth} cannot be updated by either switch, but it can be
used by both the switches; and \textsf{hdr.telem} can be updated by any switch
but it should not be visible to Alice or Bob. Then, isolation can be established
by checking two judgements:
\begin{align*}
\stmtTyping{A}{\Gamma}{\Delta}{\textsf{update_by_alice}()}{\Gamma'} \\
\stmtTyping{B}{\Gamma}{\Delta}{\textsf{update_by_bob}()}{\Gamma'}
\end{align*}

Programs that incorrectly access packet headers will be flagged by the
typechecker. For instance, in \textsf{Alice_Ingress}, the switch tries to write
to Bob's field, \Cref{ln:isolation:explicit} and on \Cref{ln:isolation:implicit}
it attempts to use the telemetry field \textsf{hdr.telem}, which can only be
written to, not read. Our typechecker flags both leaks. A safe version of
Alice's switch program is shown in \Cref{lst:isolation-fixed}. In contrast,
\textsf{Bob_Ingress} is accepted by the typechecker: it applies a table that
branches on the $\bot$-labeled header $\textsf{hdr.eth}$, and the action
\textsf{set_by_bob} only modifies the $\top$-level header $\textsf{hdr.telem}$, 
incrementing a counter.

\begin{lstlisting}[caption={Isolation Respecting Switch
Program},label={lst:isolation-fixed}, numbers=left,  breaklines=true, escapechar=|,style=topper]
// typed at pc = |\color{pred}{A}|
control Alice_Ingress(headers hdr) {
  action set_by_alice(<alice_data, A> value) {
    hdr.alice_data = value;
  }
  table update_by_alice {
    key = { hdr.alice_data: exact; }
    action  = { set_by_alice; }
  }
  apply { update_by_alice.apply(); }
}
\end{lstlisting}%

While our concrete example only involves two switches and two parties, the same
idea can be directly generalized to more parties by adding additional labels at
the level of $A$ and $B$. Then, our typechecker can ensure that programs written
by different parties act on only their own packet headers. Richer dataflow
policies could potentially be enforced by using more complex lattices; this is
an interesting direction for future work.


\section{Related Work}
\label{sec:rw}

\paragraph*{Security in programmable networks.}
Recent works explore the security and privacy implications of programmable
networks. For instance, in-network systems can be used to defend against
denial-of-service attacks~\citep{ripple,fastflex}, obfuscate network
topology~\citep{nethide}, mitigate covert channels~\citep{netwarden}, and
enforce custom security policies~\citep{poise,d2r}. Tools have also been
developed for helping operators test their dataplane programs against
adversarial inputs (e.g., \citep{p4wn}). Our work complements these systems by
detecting security and privacy bugs in programs running on programmable
switches.

\paragraph*{Network verification.}
The network verification literature is too vast to summarize here; methods have
have targeted many aspects of networked systems, including routing
protocols~(e.g., \citep{bagpipe,minesweeper,bonsai,shapeshifter}), network
configurations~(e.g., \citep{config2spec,netdice}), and network
controllers~(e.g.,~\citep{avenir,nv}). Techniques have also been developed for
verifying dataplane programs~(e.g., \citep{netkat,coalg-netkat}).  Some works
also allow one to automatically repair faulty configurations~\cite{qarc} or to
automatically synthesize policy-compliant ones~\cite{genesis,zeppelin}.

Our work focuses on dataplane programs written in the P4 language~\citep{p4},
building on the core version of P4 developed by \citet{petr4}. Perhaps the most
closely related work is \textsf{p4v}~\citep{p4v}, a verification system for P4
programs. Using \textsf{p4v}, a P4 program is verified against a logical
specification by extracting a logical formula, which can be dispatched to
solvers like Z3. \citet{p4v} use \textsf{p4v} to verify basic correctness
properties, e.g., a program does not read or write invalid headers, or a program
implements the desired functionality correctly. While our system cannot verify
the general properties established by \textsf{p4v}, our target non-interference
property cannot be established in \textsf{p4v} since it relates a program's
behavior on pairs of inputs~\citep{hyperp}. Furthermore, our type-based analysis
is lightweight and does not require automated solvers.

Two closely related type-system based works that explore properties orthogonal
to non-interference properties are \textsf{SafeP4} \citep{safep4} and
\textsf{$\Pi$4} \citep{pi4}. \textsf{SafeP4} aims at catching invalid header
access bugs, while \textsf{$\Pi$4} presents a dependently-typed extension of P4
for verifying richer properties that SafeP4 could not cover. Unlike
\textsf{$\Pi$4}, \name has a light-weight typechecking algorithm that does not
involve constraint solving.  Furthermore, our system builds on Core P4, a more
realistic formal model of P4. For example, Core P4 models different calling
conventions of P4 functions (e.g., pass by value and pass by reference) and
control flow signals. These features introduced new opportunities for implicit
leaks, which our type system rules out.

\paragraph*{Information-flow control.}
Our approach belongs to a line of research on information-flow control (IFC), a
type-based method of expressing and verifying a wide variety of security
properties. Starting from work by \citet{denning} and \citet{volpano}, there are
now many information-flow control systems ensuring different variants of
non-interference against different kinds of adversaries; the survey
by~\citet{ifc-survey} is a good introduction to this area. Existing systems
target general-purpose programming languages (e.g.,~\citep{flowcaml,jif}). Our
work brings this idea to languages for programmable networks.

\section{Conclusion and Future Directions}
\label{sec:conc}

We have designed an information-flow control type system for P4 and demonstrated
how it can verify networking properties for programs running on programmable
switches.

We see several possibilities for further investigation. First, our
non-interference theorems treat P4 programs as mapping a single input packet to
a single output packet, but,P4 allows programming switches that can maintain
internal state and recirculate packets for additional processing. These features
could lead to security leaks if an adversary can observe sequences of input and
output packets, and it would be interesting to establish non-interference in
this richer setting. Second, it could be interesting to refine our analysis with
information or assumptions about the control plane~\citep{p4v}.

\begin{acks}
  This work benefited substantially from discussions about P4 and Core P4 with
  Eric Campbell, Ryan Doenges, and Nate Foster. We thank the reviewers and our
  shepherd, Jedidiah McClurg, for their close reading and constructive feedback.
  This work is partially supported by NSF grants \#2152831 and \#1943130.
\end{acks}

\balance
\bibliographystyle{ACM-Reference-Format}
\bibliography{header,bibfile}

\iffull
\onecolumn
\appendix

\section{Grammar}

\paragraph*{Expressions}

$\begin{array}{rclr}
exp & ::= & b & \text{Boolean}\\
                 &\OR & n_w  &\text{integers or bits of width w}\\
                 &\OR & x &\text{variable} \\
                 &\OR & exp_1[exp_2] &\text{array indexing} \\
                 &\OR & exp_1 \oplus exp_2 &\text{binary operation}\\
                 &\OR & \{\overline{f_i = exp_i}\} &\text{record} \\
                 &\OR & exp.f_i &\text{field projection}\\
                 &\OR & exp_1(\overline{exp_2}) &\text{function call}
\end{array}$
\paragraph*{Statements}

$\begin{array}{rclr}
{stmt} & ::= & exp_1(\overline{exp_2}) &\text{function call}\\
               & \OR & exp_1 := exp_2 &\text{assignment} \\
              &  \OR & \terminal{if}~ (exp_1)~ stmt_1~ \terminal{else}~ stmt_2 &\text{conditional}\\
              &  \OR & \{\overline{stmt}\} &\text{sequencing}\\
               & \OR & \terminal{exit} &\text{exit}\\
                &\OR & \terminal{return}~ exp &\text{return}\\
               & \OR & var\_decl &\text{variable declaration}
\end{array}$

\paragraph*{Declaration}

$\begin{array}{rclr}
{prg} & ::=&  \overline{typ\_decl}~ ctrl\_body  \\
{ctrl\_body} &::= & \overline{decl} ~stmt \\
{decl} & ::=&  var\_decl \OR obj\_decl \OR typ\_decl \\
{var\_decl} & ::=&  \tau~x:= exp
        \OR \tau~x
        \\
{typ\_decl} & ::=& \terminal{match\_kind}~\{\overline{f}\} \OR \terminal{typedef}~\tau~X\\
{obj\_decl} & ::=& \terminal{table}~x~\{\overline{key}~\overline{act}\} \\
                &\OR & \terminal{function}~ \tau_{ret}~x~(\overline{d~y:\tau}) \{stmt\}
\end{array}$

$\begin{array}{rclr}
{d}  & ::= &  in \OR inout\\
{lval} & ::=&  x \\
        &\OR & lval.f \\
        &\OR & lval[n]\\
{key} & ::= & exp: x \\
{act} & ::=&  x(\overline{exp}, \overline{x: \tau})\\
\end{array}$
\section{Typing Rules}
$\Delta \vdash \tau \rightsquigarrow \tau'$ are judgements that resolve the base types for typedefs. We use the same definition as presented in Petr4's sections A.7 and A.8~\citep{petr4}. Note that the grammar that we consider doesn't support $bit\langle exp \rangle$ as we have discounted slice operations, instead we have $bit \langle n \rangle$, where $n$ is some constant.

\paragraph*{Expression Typing Rules}
\begin{mathpar}
\inferrule*[right=T-SubType-PC]
{
\ordinaryTyping[pc']{\Gamma}{\Delta}{exp}{\type{\tau}{\chi}} \\
pc \sqsubseteq pc'
}
{
\ordinaryTyping[pc]{\Gamma}{\Delta}{exp}{\type{\tau}{\chi}}
}

\inferrule*[right=T-SubType-In]
{
\ordinaryTyping[pc]{\Gamma}{\Delta}{exp}{\type{\tau}{\chi}~goes~ in} \\
\chi \sqsubseteq \chi'
}
{
\ordinaryTyping[pc]{\Gamma}{\Delta}{exp}{\type{\tau}{\chi'}~goes~ in}
}

\inferrule*[right=T-Int]
{
}
{
\ordinaryTyping[pc]{\Gamma}{\Delta}{n_{\infty}}{\type{int}{\bot}~goes~ in}
}

\inferrule*[right=T-Var]
{   x \in \dom{\Gamma} \qquad
\Gamma(x) = \type{\tau}{\chi}
}
{
\ordinaryTyping[pc]{\Gamma}{\Delta}{x}{\type{\tau}{\chi} \text{~goes inout}}
}

\inferrule*[right=T-BinOP]
{
\ordinaryTyping[pc]{\Gamma}{\Delta} {exp_{1}}{\type{\rho_{1}}{\chi_{1}}} \\
\ordinaryTyping[pc]{\Gamma}{\Delta}{exp_{2}}{\type{\rho_{2}}{\chi_{2}}} \\\\
\mathcal{T}(\Delta; \oplus; \rho_{1}; \rho_{2}) = \rho_{3} \\
\chi_{1} \sqsubseteq \chi' \\
\chi_{2} \sqsubseteq \chi'
}
{
\ordinaryTyping[pc]{\Gamma}{\Delta}{exp_{1} \oplus exp_{2}}{\type{\rho_3}{\chi'} ~goes~ in}
}

\inferrule*[right=T-Rec]
{
\listOrdinaryTyping[pc]{\Gamma}{\Delta}{\{\overline{exp: \type{\tau_i}{\chi_i}}\}}
}
{
\ordinaryTyping[pc]{\Gamma}{\Delta}{\{ \overline{f: exp} \}}{\type{\{ \overline{f: \langle \tau_i, \chi_i \rangle} \}}{\bot}~goes~ in}
}

\inferrule*[right=T-MemRec]
{
\ordinaryTyping[pc]{\Gamma}{\Delta}{exp}{\type{\{ \overline{f_i: \langle \tau_i, \chi_i \rangle} \}}{\bot}}~goes~d
}
{
\ordinaryTyping[pc]{\Gamma}{\Delta}{exp.f_{i}}{\type{\tau_{i}}{\chi_{i}}~ goes~ d}
}

\inferrule*[right=T-Index]
{
\ordinaryTyping[pc]{\Gamma}{\Delta}{exp_{1}}{\type{\type{\tau}{\chi_1}[n]}{\bot}~goes~ d} \\\\
\ordinaryTyping[pc]{\Gamma}{\Delta}{exp_{2}}{\type{bit \langle 32 \rangle}{\chi_2}} \\\\
\chi_2 \sqsubseteq \chi_1
}
{
\ordinaryTyping[pc]{\Gamma}{\Delta} {exp_{1}[exp_{2}]}{\type{\tau}{\chi_1}~ goes~ d }
}

\inferrule*[right=T-MemHdr]
{
\ordinaryTyping[pc]{\Gamma}{\Delta}{exp} {\type{header \{ \overline{f_i: \langle \tau_i, \chi_i \rangle} \}}{\bot}~ goes~ d}
}
{
\ordinaryTyping[pc]{\Gamma}{\Delta}{exp.f_{i}}{\type{\tau_{i}}{\chi_{i}}~goes~ d}
}

\inferrule*[right=T-Call]
{
\Gamma, \Delta \vdash_{pc} exp_{1}: \langle \overline{d~\type{\tau_i}{\chi_i}}
\xrightarrow{pc_{fn}} \langle \tau_{ret}, \chi_{ret}\rangle, \bot \rangle \\\\
\Gamma, \Delta \vdash_{pc} \overline{exp_{2}:\type{\tau_i}{\chi_i}~ goes~d} \\
pc \sqsubseteq pc_{fn}
}
{
\Gamma, \Delta \vdash_{pc} exp_{1} (\overline{exp_{2}}): \type{\tau_{ret}}{\chi_{ret}} ~\text{goes in}
}

\inferrule*[right=T-FuncDecl]
{
\Gamma_1 = \Gamma[\overline{x_{i}: \type{\tau_{i}'}{\chi_{i}}}, \terminal{return}: \type{\tau_{ret}'}{\chi_{ret}}]\\
\declTyping{pc_{fn}}{\Gamma_{1}}{\Delta}{stmt}{\Gamma_{2}} \\\\
\Delta \vdash \tau_i \rightsquigarrow \tau_i'~\text{for each}~\tau_i \\
\Delta \vdash \tau_{ret} \rightsquigarrow \tau_{ret}' \\
\Gamma' = \Gamma[x: \type{\overline{d~\type{\tau_{i}'}{\chi_{i} }} \xrightarrow[]{pc_{fn}} \type{\tau_{ret}'}{\chi_{ret}}}{\bot}]
}
{
\declTyping{pc}{\Gamma}{\Delta}{\terminal{function}~\type{\tau_{ret}}{\chi_{ret}}~x~(\overline{d~ x_{i}: \type{\tau_{i}}{\chi_{i}}}) \{stmt\}}{\Gamma'}{\Delta}
}

\end{mathpar}

\paragraph*{Statement Typing Rules}
\begin{mathpar}
\inferrule*[right=T-Empty]
{
~
}
{
\stmtTyping{pc}{\Gamma}{\Delta}{\{ \}}{\Gamma}
}

\inferrule*[right=T-Exit]
{
~
}
{
\stmtTyping{\bot}{\Gamma}{\Delta}{\terminal{exit}}{\Gamma}
}

\inferrule*[right=T-conditional]
{
\ordinaryTyping[pc]{\Gamma}{\Delta}{exp}{\type{bool}{\chi_1}}\\\\
\stmtTyping{\chi_2}{\Gamma}{\Delta}{stmt_{1}}{\Gamma_{1}}\\
\stmtTyping{\chi_2}{\Gamma}{\Delta}{stmt_{2}}{\Gamma_{2}} \\
\chi_1 \sqsubseteq \chi_2 \\
pc \sqsubseteq \chi_2
}
{
\stmtTyping{pc}{\Gamma}{\Delta}{\terminal{if}~(exp)~~stmt_{1}~\terminal{else}~stmt_{2}}{\Gamma}
}

\inferrule*[right=T-Seq]
{
\Gamma, \Delta \vdash_{pc} stmt_{1} \dashv \Gamma_{1} \qquad
\Gamma_{1}, \Delta \vdash_{pc} \{ \overline{stmt_{2}} \} \dashv \Gamma_{2}
}
{
\Gamma, \Delta \vdash_{pc} \{ stmt_{1}; \overline{stmt_{2}} \} \dashv \Gamma_2
}

\inferrule*[right=T-Return]
{
\Gamma,\Delta \vdash_{pc}  exp: \langle \tau, \chi_{ret} \rangle \\
\Gamma(\terminal{return}) = \type{\tau_{ret}}{\chi_{ret}} \\
\Delta    \vdash \tau_{ret} \rightsquigarrow \tau
}
{
\Gamma, \Delta \vdash_{\bot} \terminal{return}~ exp \dashv \Gamma
}

\inferrule*[right=T-Assign]
{
\Gamma,\Delta \vdash_{pc} exp_{1}:\type{\tau}{\chi_{1}}~goes~ inout \\
\Gamma,\Delta \vdash_{pc} exp_{2} :\type{\tau}{\chi_{2}}\\
\chi_{2} \sqsubseteq \chi_{1}\\
pc \sqsubseteq \chi_{1}
}
{
\Gamma, \Delta \vdash_{pc} exp_{1} := exp_{2} \dashv \Gamma
}

\inferrule*[right=T-Decl]
{
\Gamma, \Delta \vdash_{pc} var\_decl \dashv \Gamma_1, \Delta
}
{
\Gamma, \Delta \vdash_{pc} var\_decl \dashv \Gamma_1
}

\inferrule*[right=T-FnCallStmt]
{
\Gamma, \Delta \vdash_{pc} exp_1(\overline{exp_2}): \type{\tau_{ret}}{\chi_{ret}}
}
{
\Gamma, \Delta \vdash_{pc} exp_1(\overline{exp_2}) \dashv \Gamma
}

\inferrule*[right=T-TblCall]
{
\Gamma, \Delta \vdash_{pc} exp: \langle table(pc_{tbl}), \bot \rangle \\
 pc \sqsubseteq pc_{tbl}
}
{
\Gamma, \Delta \vdash_{pc} exp() \dashv \Gamma
}
\end{mathpar}

\paragraph*{Declaration Typing Rules}

\begin{mathpar}
    \inferrule*[right=T-VarDecl]
    {~}
    {\Gamma, \Delta \vdash_{pc} \type{\tau}{\chi}~ x \dashv \Gamma [x: \langle \tau, \chi \rangle], \Delta}

    \inferrule*[]
{
\ordinaryTyping[pc]{\Gamma}{\Delta}{exp}{\type{\tau'}{\chi}} \\
\Delta \vdash \tau \rightsquigarrow \tau'
}
{
\Gamma; \Delta \vdash_{pc} \type{\tau}{\chi}~ x:= exp \dashv \Gamma [x: \langle \tau', \chi\rangle]; \Delta
}

      \inferrule*[right=T-TblDecl]
    {
    \newordinaryTyping[pc_{tbl}]{\Gamma}{\Delta}{\overline{exp_k:\type{\tau_k}{\chi_k}}}{} \\
    \newordinaryTyping[pc_{tbl}]{\Gamma}{\Delta}{\overline{x_k:\type{match\_kind}{\bot}}}{} \\\\
    \newordinaryTyping[pc_{tbl}]{\Gamma}{\Delta}{act_{a_j}: \type{\overline{d\type{\tau_{a_{ji}}}{\chi_{a_{ji}}}}~;\overline{\type{\tau_{c_{ji}}}{\chi_{c_{ji}}}} \xrightarrow{pc_{fn_j}} \type{unit}{\bot}}{\bot}}{},~\text{for all}~ j\\\\
    \newordinaryTyping[pc_{tbl}]{\Gamma}{\Delta}{\overline{exp_{a_{ji}}:\type{\tau_{a_{ji}}}{\chi_{a_{ji}}}goes~d}}{}  \\
    {\chi_k} \sqsubseteq {pc_{fn_j}}~\text{for all}~j,k \\
    pc_a \sqsubseteq pc_{fn_j}, \text{for all}~j \\
    {\chi_k} \sqsubseteq {pc_{tbl}}~\text{for all}~k\\
    pc_{tbl} \sqsubseteq pc_a
    }
    {
    \declTyping{pc}{\Gamma}{\Delta}{\text{table}~x~ \{\overline{exp_k: x_k}~ \overline{act_{a_j}(\overline{exp_{a_{ji}}})}\}}{\Gamma[x: \type{table(pc_{tbl})}{\bot}]}{\Delta}
    }

\inferrule*[right=T-FuncDecl]
{
\Gamma_1 = \Gamma[\overline{x_{i}: \type{\tau_{i}'}{\chi_{i}}}, \terminal{return}: \type{\tau_{ret}'}{\chi_{ret}}]\\
\declTyping{pc_{fn}}{\Gamma_{1}}{\Delta}{stmt}{\Gamma_{2}} \\\\
\Delta \vdash \tau_i \rightsquigarrow \tau_i'~\text{for each}~\tau_i \\
\Delta \vdash \tau_{ret} \rightsquigarrow \tau_{ret}' \\
\Gamma' = \Gamma[x: \type{\overline{d~\type{\tau_{i}'}{\chi_{i} }} \xrightarrow[]{pc_{fn}} \type{\tau_{ret}'}{\chi_{ret}}}{\bot}]
}
{
\declTyping{pc}{\Gamma}{\Delta}{\terminal{function}~\type{\tau_{ret}}{\chi_{ret}}~x~(\overline{d~ x_{i}: \type{\tau_{i}}{\chi_{i}}}) \{stmt\}}{\Gamma'}{\Delta}
}

    \end{mathpar}
 \section{Definition}
Let $\tau_{fn} =\langle  \overline{ d~\rho} \xrightarrow{pc_{fn}} \rho_{ret}, \bot \rangle$ and $\tau_{tbl}=\type{table(pc_{tbl})}{\bot}$

\begin{definition}[Store typing]\label{def:store-typing}
Store typing context $\Xi$ is a partial map from the locations to types, $\Xi: \mathbb{L} \to \type{\tau}{\chi}$.
A memory-store $\mu$ is well-typed in a store-typing context $\Xi$, which can be represented as $\semanticStore{\Xi}{\Delta}{\mu}$, if for every location, $l \in \dom{\mu}$ there exists a type, $\type{\tau}{\chi} = \Xi(l)$ and
$\ordinaryTyping{\Xi}{\Delta}{\mu(l)}{\type{\tau}{\chi}}$ (value typing is defined in \Cref{vtyping}).
\end{definition}

\begin{definition}[Typing of environment] \label{def:env-typing}
	$\Xi \vdash \epsilon: \Gamma$ is defined as
	\begin{mathpar}
	  \inferrule{~}{\Xi \vdash []: []}

	  \inferrule{\Xi \vdash \epsilon: \Gamma \\ \Xi(l)= \type{\tau}{\chi}}{\Xi \vdash (\epsilon, x \mapsto l): (\Gamma, x: \type{\tau}{\chi})}

	  \inferrule{\Xi \vdash \epsilon: \Gamma}{\Xi \vdash \epsilon: \Gamma, \textsf{return} \mapsto l}
	\end{mathpar}
\end{definition}
\begin{definition}[Semantic typing of store and environment] \label{def:store-env-semantic typing}
A pair of store and environment $\langle \mu, \epsilon \rangle$ is semantically well-typed $\semanticStoreEnv{\Xi}{\Delta}{\mu}{\epsilon}{\Gamma}$ if the following conditions hold:
	\begin{enumerate}
		\item $\semanticStore{\Xi}{\Delta}{\mu}$
		\item $\Xi \vdash \epsilon: \Gamma$
		\item For any $x \in \dom{\epsilon}$, $\epsilon(x) \in \dom{\mu}$,
		\item For all $x$ in $\dom{\epsilon}$ and some $\Gamma_{fn} \subseteq \Gamma$ and any $pc$,  if $\Gamma, \Delta \vdash_{pc} x: \tau_{fn}$, $\mu(\epsilon(x)) = clos(\epsilon_{c},...)$, and $\Xi \models \epsilon_c: \Gamma_{fn}$, then $\dom{\epsilon_{c}} \subseteq \dom{\epsilon}$ and $\semanticStoreEnv{\Xi}{\Delta}{\mu}{\epsilon_{c}}{\Gamma_{fn}}$. Here $\tau_{fn}$ is the function type. We elide the full view of the closures in this definition.
\item For all $x$ in $\dom{\epsilon}$ and some $\Gamma_{tbl} \subseteq \Gamma$ and any $pc$,  if $\Gamma, \Delta \vdash_{pc} x: \tau_{tbl}$,  $\mu(\epsilon(x)) = table~l~(\epsilon_{c},...)$, and $\Xi \models \epsilon_c: \Gamma_{tbl}$, then $\dom{\epsilon_{c}} \subseteq \dom{\epsilon}$ and $\semanticStoreEnv{\Xi}{\Delta}{\mu}{\epsilon_c}{\Gamma_{tbl}}$. Here $\tau_{tbl}$ is the table type.
	\end{enumerate}
\end{definition}

\begin{definition}[Semantic typing for a pair of memory stores and environments] \label{def:mem-store-pair-semantic}
	$\semanticBelowPCState{l}{\Xi_a}{\Xi_b}{\Delta}{\mu_{a}}{\epsilon_{a}}{\mu_{b}}{\epsilon_{b}}{\Gamma}$ holds when
		\begin{enumerate}
			\item $\semanticStoreEnv{\Xi_a}{\Delta}{\mu_{a}}{\epsilon_{a}}{\Gamma}$ and $\semanticStoreEnv{\Xi_b}{\Delta}{\mu_{b}}{\epsilon_{b}}{\Gamma}$
			\item $\dom{\epsilon_a} = \dom{\epsilon_b}$
			\item For any $x \in \dom{\epsilon_a} = \dom{\epsilon_b}$, $\NIval{l}{\Xi_a}{\Xi_b}{\Delta}{\mu_a(\epsilon_{a}(x))}{\mu_b(\epsilon_{b}(x))}{\Gamma(x)}$ (defined in \Cref{def:def-ni-val}),
			\item For all $x$ in $\dom{\epsilon_a} = \dom{\epsilon_b}$ and some $\Gamma_{fn} \subseteq \Gamma$ and any $pc$,  if $\Gamma, \Delta \vdash_{pc} x: \tau_{fn}$, $\mu_a(\epsilon_a(x)) = clos(\epsilon_{c_a},...)$, $\mu_b(\epsilon_b(x)) = clos(\epsilon_{c_b},...)$, $\Xi_a \models \epsilon_{c_a}: \Gamma_{fn}$, and $\Xi_b \models \epsilon_{c_b}: \Gamma_{fn}$ , then $\semanticBelowPCState{l}{\Xi_a}{\Xi_b}{\Delta}{\mu_{a}}{\epsilon_{c_a}}{\mu_{b}}{\epsilon_{c_b}}{\Gamma_{fn}}$,
			\item For all $x$ in $\dom{\epsilon_a} = \dom{\epsilon_b}$ and some $\Gamma_{tbl} \subseteq \Gamma$ and any $pc$,  if $\Gamma, \Delta \vdash_{pc} x: \tau_{tbl}$, $\mu_a(\epsilon_a(x)) = table~l_a~(\epsilon_{c_a},...)$, $\mu_b(\epsilon_b(x)) = table~l_b~(\epsilon_{c_b},...)$, $\Xi_a \models \epsilon_{c_a}: \Gamma_{tbl}$, and $\Xi_b \models \epsilon_{c_b}: \Gamma_{tbl}$ , then $\semanticBelowPCState{l}{\Xi_a}{\Xi_b}{\Delta}{\mu_{a}}{\epsilon_{c_a}}{\mu_{b}}{\epsilon_{c_b}}{\Gamma_{tbl}}$.
		\end{enumerate}
\end{definition}

\begin{definition}[Non-interference for Expressions] \label{def:def-ni-exp}
	$\NIexp{pc}{\Gamma}{\Delta}{exp}{\type{\tau}{\chi}}$ holds
     if for any $\Xi_a$, $\Xi_b$, $\mu_{a}$, $\mu_{b}$, $\epsilon_{a}$, $\epsilon_{b}$,  $\mu_{a}'$, $\mu_{b}'$, and any security level $l$,
	 \begin{enumerate}

	   \item \emph{Variable at level lower than $l$ are indistinguishable at the beginning.}$\semanticBelowPCState{l}{\Xi_a}{ \Xi_b} {\Delta}{\mu_{a}}{\epsilon_{a}}{\mu_b}{\epsilon_b}{\Gamma}$,
	   \item $\evalsto
    {\config[exp]{\mathcal{C}; \Delta}{\mu_{a}}{\epsilon_{a}}}
    {\configval{\mu_{a}'}{val_a}}$,
    \item $\evalsto
    {\config[exp]{\mathcal{C}; \Delta}{\mu_{b}}{\epsilon_{b}}}
    {\configval{\mu_{b}'}{val_b}}$
	 \end{enumerate}
	 implies there exists a $\Xi_a'$, $\Xi_b'$ such that
	 \begin{enumerate}

	   \item \emph{Effects on any variable at level lower than $l$ should be indistinguishable.} $\semanticBelowPCState{l}{\Xi_{a}'}{\Xi_{b}'}{\Delta}{\mu_{a}'}{\epsilon_{a}}{\mu_{b}'}{\epsilon_{b}}{\Gamma}$,
	   \item \emph{PC is used to bound writes.} For any $l_a \in \dom{\mu_a}$ and $l_b \in \dom{\mu_b}$ such that $\ordinaryTyping[]{\Xi_a}{\Delta}{\mu_a(l_a)}{\type{\tau}{\chi}}$ and $\ordinaryTyping[]{\Xi_b}{\Delta}{\mu_b(l_b)}{\type{\tau}{\chi}}$ and $pc \nsqsubseteq \chi$, we have $\mu_{a}'(l_a) = \mu_{a}(l_a)$ and $\mu_{b}'(l_b) = \mu_{b}(l_b)$,
   \item For any $l_a \in \dom{\mu_a}$ such that $\ordinaryTyping[]{\Xi_a}{\Delta}{\mu_a(l_a)}{\tau_{clos}}$, where $\tau_{clos} \in \{\tau_{fn}, \tau_{tbl}\}$, then $\mu_a'(l_a) = \mu_a(l_a)$,
   \item For any $l_b \in \dom{\mu_b}$ such that $\ordinaryTyping[]{\Xi_b}{\Delta}{\mu_b(l_b)}{\tau_{clos}}$, where $\tau_{clos} \in \{\tau_{fn}, \tau_{tbl}\}$, then $\mu_b'(l_b) = \mu_b(l_b)$,
	   \item $\Xi_a \subseteq \Xi_a'$, $\Xi_b \subseteq \Xi_b'$, $\dom{\mu_a} \subseteq \dom{\mu_a'}$, and $\dom{\mu_b} \subseteq \dom{\mu_b'}$,
	   \item $\NIval{l}{\Xi_a'}{\Xi_b'}{\Delta}{val_{a}}{val_b}{\type{\tau}{\chi}}$.
	 \end{enumerate}

\end{definition}

\begin{definition}[Non-interference for values]\label{def:def-ni-val}
$\NIval{l}{\Xi_a}{\Xi_b}{\Delta}{val_{a}}{val_b}{\type{\tau}{\chi}}$ holds when:
\begin{enumerate}
\item $\ordinaryTyping[]{\Xi_a}{\Delta}{val_a}{\langle \tau, \chi \rangle}$ and $\ordinaryTyping[]{\Xi_b}{\Delta}{val_b}{\langle \tau, \chi \rangle}$ (value typing is defined in \Cref{vtyping}),
\item If $\tau \notin \{ \rho[n], ~\{\overline{f: \rho }\},~ header~\{\overline{f: \rho }\}, ~\overline{d~\rho} \xrightarrow{pc} \rho_{ret}, ~table(pc_{tbl})\}$ and  $\chi \sqsubseteq l$, then $val_{a} = val_{b}$,
\item If $\tau =\rho[n]$, $\tau = \{\overline{f: \rho }\}$ or $\tau = header~\{\overline{f: \rho }\}$, then $\NIval{l}{\Xi_a}{\Xi_b}{\Delta}{val_{ai}}{val_{bi}}{\type{\tau_{i}}{\chi_{i}}}$, for all $val_{ai} \in \overline{val_{ai}} = val_a$ and $val_{bi} \in \overline{val_{bi}} = val_b$, $\rho = \type{\tau_i}{\chi_i}$.
\item  If $\type{\tau}{\chi} = \tau_{fn}$, then $\NIclos{}{\Xi_a}{\Xi_b}{\Delta}{val_{a}}{val_b}{\tau_{fn}}$ (\Cref{clos-def}),
\item If $\type{\tau}{\chi} = \tau_{tbl}$, then $\NItbl{}{\Xi_a}{\Xi_b}{\Delta}{val_{a}}{val_b}{\tau_{tbl}}$ (\Cref{ni-tbl}).
%
\end{enumerate}
\end{definition}

\begin{definition} \label{clos-def}
	$\NIclos{}{\Xi_a}{\Xi_b}{\Delta}{val_{a}}{val_b}{\tau_{fn}}$, where $val_a$ and $val_b$ are of the form $clos(\epsilon_{c_a}, \overline{d x: \type{\tau}{\chi}}, \type{\tau_{ret}}{\chi_{ret}},\text{stmt})$ and	 $clos(\epsilon_{c_b}, \overline{d x: \type{\tau}{\chi}}, \type{\tau_{ret}}{\chi_{ret}},\text{stmt})$  holds when there exists a $\Gamma$ such that the following are satisfied:
	\begin{enumerate}
		\item $\Xi_a \vdash \epsilon_{c_a}: \Gamma$ and $\Xi_b \vdash \epsilon_{c_b}: \Gamma$
		\item for any $pc$, $\Gamma, \Delta \vdash_{pc} val_a :\langle  \overline{d~\langle\tau, \chi \rangle} \xrightarrow{pc_{fn}} \langle \tau_{ret}, \chi_{ret}\rangle, \bot \rangle$ and $\Gamma, \Delta \vdash_{pc} val_b :\langle \overline{d~\langle\tau, \chi\rangle} \xrightarrow{pc_{fn}} \langle \tau_{ret}, \chi_{ret}\rangle, \bot \rangle$.
		\item $\Gamma[\overline{x:\type{\tau}{\chi}}, return: \type{\tau_{ret}}{\chi_{ret}}], \Delta \vdash_{pc_{fn}} stmt \dashv \Gamma'$
		\item $val_{a} =_{clos} val_{b}$.
	\end{enumerate}
	Here, $val_{a} =_{clos} val_{b}$ is defined as two closures with $\dom{\epsilon_{c_a}} = \dom{\epsilon_{c_b}}$.
\end{definition}

\begin{definition}\label{ni-tbl}
 $\NItbl{}{\Xi_a}{\Xi_b}{\Delta}{val_{a}}{val_b}{\tau_{tbl}}$, where
 \[val_a = table~l_a~(\epsilon_{a}, \overline{exp_k: x_k}, \overline{act_{a_j}(\overline{exp_{a_{ji}}}, \overline{y_{c_{ji}}: \type{\tau_{c_{ji}}}{\chi_{c_{ji}}}})})\] and
\[val_b = table~l_b~(\epsilon_{b}, \overline{exp_k: x_k}, \overline{act_{a_j}(\overline{exp_{a_{ji}}}, \overline{y_{c_{ji}}: \type{\tau_{c_{ji}}}{\chi_{c_{ji}}}})})\] holds when there exists a $\Gamma$ and $pc_a$ such that the following are satisfied:

\begin{enumerate}
\item $\Xi_a \models \epsilon_{a}: \Gamma$ and $\Xi_b \models \epsilon_{b}: \Gamma$
\item for any $pc$, $\Gamma; \Delta \vdash_{pc} val_a : \type{table(pc_{tbl})}{\bot}$, $\Gamma; \Delta \vdash_{pc} val_b: \type{table(pc_{tbl})}{\bot}$.
\item $\ordinaryTyping[pc_{tbl}]{\Gamma}{\Delta}{x_k}{\type{match\_kind}{\bot}}$ for each $x_k \in \overline{x_k}$
\item $\ordinaryTyping[pc_{tbl}]{\Gamma}{\Delta}{exp_k}{\type{\tau_k}{\chi_{k}}}$ for each $exp_k \in \overline{exp_k}$
\item $\ordinaryTyping[pc_{tbl}]{\Gamma}{\Delta}{act_{aj}}{\type{\overline{d~\type{\tau_{a_{ji}}}{\chi_{a_{ji}}}}~;\overline{\type{\tau_{c_{ji}}}{\chi_{c_{ji}}}} \xrightarrow{pc_{fn_j}} \type{unit}{\bot}}{\bot}}$ for each $act_{a_j} \in \overline{act_{a_j}}$
\item $\ordinaryTyping[pc_{tbl}]{\Gamma}{\Delta}{exp_{a_{ji}}}{\type{\tau_{a_{ji}}}{\chi_{a_{ji}}}~goes~d}$ for each $exp_{a_{ji}} \in \overline{exp_{a_{ji}}}$
\item $val_a =_{tbl} val_b$.
\item $\chi_k \sqsubseteq {pc_{fn_j}}$, for all $j,k$
\item $pc_a \sqsubseteq pc_{fn_j}$, for all $j$
\item ${\chi_k} \sqsubseteq {pc_{tbl}}~\text{for all}~k$
\item $pc_{tbl} \sqsubseteq pc_a$.
\end{enumerate}
	Here, $table~l_a =_{tbl} table~l_b$ is defined as two table values with $\dom{\epsilon_{a}} = \dom{\epsilon_{b}}$. Their control plane entries will be the same.
\end{definition}

\begin{definition}[Non-interference for statements]
	For any security lable $l$, $\NIstmt{pc}{\Gamma}{\Delta}{stmt}{\Gamma'}$ holds	 for any 	 $\Xi_a$, $\Xi_b$, $\mu_{a}$, $\mu_{b}$, $\epsilon_{a}$, $\epsilon_{b}$, $\mu_{a}'$, $\mu_{b}'$, $\epsilon_{a}'$, $\epsilon_{b}'$ if
\begin{enumerate}
\item  $\semanticBelowPCState{l}{\Xi_a}{\Xi_b}{\Delta}{\mu_{a}}{\epsilon_{a}}{\mu_{b}}{\epsilon_{b}}{\Gamma}$,
\item $\evalsto{\config[stmt]{\mathcal{C}; \Delta}{\mu_{a}}{\epsilon_{a}}}{\config{\mu_{a}'}{\epsilon_{a}'}{sig_{1}}}$,
\item $\evalsto{\config[stmt]{\mathcal{C}; \Delta}{\mu_{b}}{\epsilon_{b}}}{\config{\mu_{b}'}{\epsilon_{b}'}{sig_{2}}}$
  \end{enumerate}
  then there exists $\Xi_a'$, $\Xi_b'$, such that
  \begin{enumerate}
    \item $\stmtTyping{pc}{\Gamma}{\Delta}{stmt}{\Gamma'}$,
    \item $\semanticBelowPCState{l}{\Xi_a'}{\Xi_b'}{\Delta} {\mu_{a}'}{\epsilon_{a}'}{\mu_{b}'}{\epsilon_{b}'}{\Gamma'}$,
    \item $\semanticBelowPCState{l}{\Xi_a'}{\Xi_b'}{\Delta} {\mu_{a}'}{\epsilon_{a}}{\mu_{b}'}{\epsilon_{b}}{\Gamma}$,
    \item \emph{PC is used to bound writes.} For any $l_a \in \dom{\mu_a}$ and $l_b \in \dom{\mu_b}$ such that $\ordinaryTyping[]{\Xi_a}{\Delta}{\mu_a(l_a)}{\type{\tau}{\chi}}$ and $\ordinaryTyping[]{\Xi_b}{\Delta}{\mu_b(l_b)}{\type{\tau}{\chi}}$ and $pc \nsqsubseteq \chi$, we have $\mu_{a}'(l_a) = \mu_{a}(l_a)$ and $\mu_{b}'(l_b) = \mu_{b}(l_b)$,
   \item For any $l_a \in \dom{\mu_a}$ such that $\ordinaryTyping[]{\Xi_a}{\Delta}{\mu_a(l_a)}{\tau_{clos}}$, where $\tau_{clos} \in \{\tau_{fn}, \tau_{tbl}\}$, then $\mu_a'(l_a) = \mu_a(l_a)$,
   \item For any $l_b \in \dom{\mu_b}$ such that $\ordinaryTyping[]{\Xi_b}{\Delta}{\mu_b(l_b)}{\tau_{clos}}$, where $\tau_{clos} \in \{\tau_{fn}, \tau_{tbl}\}$, then $\mu_b'(l_b) = \mu_b(l_b)$,
   \item $sig_1=sig_2= cont$ or $sig_1=sig_2= exit$ or $sig_1=\textsf{return}~val_1;~ sig_2=\textsf{return}~val_2$,
	 \item If $sig_{1}=\textsf{return}~ val_{1}$ and $sig_{2}=\textsf{return}~ val_{2}$, then $\Xi_a', \Xi_b', \Delta \models_{l} NI(val_{1}, val_2): \type{\tau_{ret}'}{\chi_{ret}}$, where $\Delta \vdash \tau_{ret} \rightsquigarrow \tau_{ret}'$ and $\Gamma[\terminal{return}] = \type{\tau_{ret}}{\chi_{ret}}$,
	 \item $\Xi_a \subseteq \Xi_a'$, $\Xi_b \subseteq \Xi_b'$, $\dom{\mu_a} \subseteq \dom{\mu_a'}$, $\dom{\mu_b} \subseteq \dom{\mu_b'}$, $\dom{\epsilon_a} \subseteq \dom{\epsilon_a'}$, and $\dom{\epsilon_b} \subseteq \dom{\epsilon_b'}$.
 \end{enumerate}
\end{definition}

\begin{definition}[Non-interference for declaration statements]\label{def:def-ni-decl}
For any security lable $l$, $\NIdecl{pc}{\Gamma}{\Delta}{decl}{\Gamma'}{\Delta_1}$ holds for any
$\Xi_a$, $\Xi_b$, $\mu_{a}$, $\mu_{b}$, $\epsilon_{a}$, $\epsilon_{b}$, $\mu_{a}'$, $\mu_{b}'$, $\epsilon_{a}'$, $\epsilon_{b}'$ if
\begin{enumerate}
\item  $\semanticBelowPCState{l}{\Xi_a}{\Xi_b}{\Delta}{\mu_{a}}{\epsilon_{a}}{\mu_{b}}{\epsilon_{b}}{\Gamma}$,
\item $\evalsto
{\config[decl]{\mathcal{C}; \Delta}{\mu_{a}}{\epsilon_{a}}}
{\config[cont]{\Delta_1}{\mu_{a}'}{\epsilon_{a}'}}$,
\item $\evalsto
{\config[decl]{\mathcal{C}; \Delta}{\mu_{b}}{\epsilon_{b}}}
{\config[cont]{\Delta_1}{\mu_{b}'}{\epsilon_{b}'}}$,
\end{enumerate}
then there exists $\Xi_a'$, $\Xi_b'$ such that
\begin{enumerate}
\item $\declTyping{pc}{\Gamma}{\Delta}{decl}{\Gamma'}{\Delta_1}$,
\item $\semanticBelowPCState{l}{\Xi_a'}{\Xi_b'}{\Delta_1}{\mu_{a}'}{\epsilon_{a}'}{\mu_{b}'}{\epsilon_{b}'}{\Gamma'}$ and $\semanticBelowPCState{l}{\Xi_a'}{\Xi_b'}{\Delta_1}{\mu_{a}'}{\epsilon_{a}}{\mu_{b}'}{\epsilon_{b}}{\Gamma}$
\item \emph{PC is used to bound writes.} For any $l_a \in \dom{\mu_a}$ and $l_b \in \dom{\mu_b}$ such that $\ordinaryTyping[]{\Xi_a}{\Delta}{\mu_a(l_a)}{\type{\tau}{\chi}}$ and $\ordinaryTyping[]{\Xi_b}{\Delta}{\mu_b(l_b)}{\type{\tau}{\chi}}$ and $pc \nsqsubseteq \chi$, we have $\mu_{a}'(l_a) = \mu_{a}(l_a)$ and $\mu_{b}'(l_b) = \mu_{b}(l_b)$,
\item For any $l_a \in \dom{\mu_a}$ such that $\ordinaryTyping[]{\Xi_a}{\Delta}{\mu_a(l_a)}{\tau_{clos}}$, where $\tau_{clos} \in \{\tau_{fn}, \tau_{tbl}\}$, then $\mu_a'(l_a) = \mu_a(l_a)$,
\item For any $l_b \in \dom{\mu_b}$ such that $\ordinaryTyping[]{\Xi_b}{\Delta}{\mu_b(l_b)}{\tau_{clos}}$, where $\tau_{clos} \in \{\tau_{fn}, \tau_{tbl}\}$, then $\mu_b'(l_b) = \mu_b(l_b)$,
\item $\Xi_a \subseteq \Xi_a'$, $\Xi_b \subseteq \Xi_b'$, $\dom{\mu_a} \subseteq \dom{\mu_a'}$, $\dom{\mu_b} \subseteq \dom{\mu_b'}$, $\dom{\epsilon_a} \subseteq \dom{\epsilon_a'}$, and $\dom{\epsilon_b} \subseteq \dom{\epsilon_b'}$, $\Delta \subseteq \Delta_1$.
\end{enumerate}
\end{definition}

\section{Theorems}

\begin{theorem} \label{ni-exp}
If $\ordinaryTyping[pc]{\Gamma}{\Delta}{exp}{\type{\tau}{\chi}}$, then $\NIexp{pc}{\Gamma}{\Delta}{exp}{\type{\tau}{\chi}}$.
\end{theorem}

\begin{theorem}\label{ni-stmt}
If $\stmtTyping{pc}{\Gamma}{\Delta}{stmt}{\Gamma'}$, then  $\NIstmt{pc}{\Gamma}{\Delta}{stmt}{\Gamma'}$.
\end{theorem}

\begin{theorem}\label{ni-decl}
If $\declTyping{pc}{\Gamma}{\Delta}{decl}{\Gamma'}{\Delta'}$, then $\NIdecl{pc}{\Gamma}{\Delta}{decl}{\Gamma'}{\Delta'}$.
\end{theorem}

\section{Lemmas}
\begin{lemma}
Suppose $\semanticStore{\Xi}{\Delta}{\mu}$. For any $l$, if $l \notin \dom{\Xi}$, then $l \notin \dom{\mu}$.
\end{lemma}
\begin{proof}
Direct proof by expanding the definition of the $\semanticStore{\Xi}{\Delta}{\mu}$.
\end{proof}
\begin{lemma} \label{aux}
If $\semanticTyping{\Xi}{\Delta}{val}{\type{\tau}{\chi}}$ and $\Xi \subseteq \Xi'$, then $\semanticTyping{\Xi'}{\Delta}{val}{\type{\tau}{\chi}}$.
\end{lemma}
\begin{proof}
By induction on the typing derivations of value typing judgement.
\end{proof}
\begin{lemma} \label{lem:store-typing-extended}
Suppose $\semanticStore{\Xi}{\Delta}{\mu}$ and for any $l_a \notin \dom{\Xi}$, let $~\Xi' = \Xi[l_a \mapsto \type{\tau'}{\chi'}]$, $\mu' = \mu[l_a \mapsto val]$, and $\ordinaryTyping{\Xi'}{\Delta}{\mu'(l_a)}{\type{\tau'}{\chi'}}$. Then $\semanticStore{\Xi'}{\Delta}{\mu'}$.
\end{lemma}
\begin{proof}
By \Cref{def:store-typing}, $l_a \notin \dom{\Xi}$ implies $l_a \notin \dom{\mu}$. For all $l \in \dom{\mu'}$, there are two cases:
\begin{itemize}
\item $l \in \dom{\mu}$. By the definitions of $\mu'$ and $\Xi'$ we know that for the locations in this case ($l \in \dom{\mu'} \cap \dom{\mu}$), $\mu'(l) = \mu(l)$ and $\Xi'(l) = \Xi(l)$.
Using $\semanticStore{\Xi}{\Delta}{\mu}$, we can conclude that for the locations in this case, there exists a type, $\type{\tau}{\chi} = \Xi(l) = \Xi'(l)$ and
$\ordinaryTyping{\Xi}{\Delta}{\mu(l)}{\type{\tau}{\chi}}$.
Using $\mu'(l) = \mu(l)$, we can say $\ordinaryTyping{\Xi}{\Delta}{\mu'(l)}{\type{\tau}{\chi}}$.
Applying \Cref{aux} with $\Xi \subseteq \Xi'$, we conclude $\ordinaryTyping{\Xi'}{\Delta}{\mu'(l)}{\type{\tau}{\chi}}$.
\item $l = l_a$. For the last case where $l=l_a$, we can see that $\Xi'(l_a)=\type{\tau'}{\chi'}$  and  $\ordinaryTyping{\Xi'}{\Delta}{\mu'(l_a)}{\type{\tau'}{\chi'}}$.
\end{itemize}
Therefore, we have shown that for every location, $l \in \dom{\mu'}$ there exists a type, $\type{\tau}{\chi} = \Xi'(l)$ and $\ordinaryTyping{\Xi'}{\Delta}{\mu'(l)}{\type{\tau}{\chi}}$.
\end{proof}
\begin{lemma} \label{lem-env-subtyping}
If $\Xi \vdash \epsilon: \Gamma$, then for any $\Xi'$ such that $\Xi \subseteq \Xi'$, we have $\Xi' \vdash \epsilon: \Gamma$.
\end{lemma}
\begin{proof}
	Direct proof using the definition of  $\Xi' \vdash \epsilon: \Gamma$
\end{proof}

\begin{lemma}\label{lem:only-mem-superset}
Suppose $~\semanticStoreEnv{\Xi_a}{\Delta}{\mu_{a}}{\epsilon_{a}}{\Gamma}$. Let $\Xi_a \subseteq \Xi_a'$, $\mu_a \subseteq \mu_a'$. If $\semanticStore{\Xi_a'}{\Delta}{\mu_a'}$, then $~\semanticStoreEnv{\Xi_a'}{\Delta}{\mu_{a}'}{\epsilon_{a}}{\Gamma}$.
\end{lemma}
\begin{proof}
By definition, $\Xi_a' = \Xi_a \cup \{\overline{l_a \mapsto \type{\tau}{\chi}}\}$, $\mu_a' = \mu_a \cup \{\overline{l_a \mapsto val_a}\}$ and $l_a \notin \Xi_a$. This is followed by:
$\mu_a'(\epsilon_a(x)) = \mu_a(\epsilon_a(x))$ and $\Xi_a'(\epsilon_a(x)) = \Xi_a(\epsilon_a(x))$, for any $x \in \epsilon_a$.
We prove this lemma by induction on the $\dom{\epsilon_a}$
\begin{enumerate}
  \item Base case. $\dom{\epsilon_a}=\emptyset$. Trivial.
  \item To prove $~\semanticStoreEnv{\Xi_a'}{\Delta}{\mu_{a}'}{\epsilon_{a}}{\Gamma}$, we need to show that
  \begin{enumerate}
    \item $\semanticStore{\Xi_a'}{\Delta}{\mu_a'}$ (already given),
    \item $\Xi_a' \vdash \epsilon_a: \Gamma$ (follows from \Cref{lem-env-subtyping}),
    \item for any $x \in \dom{\epsilon_a}$, we have $\epsilon_a(x) \in \dom{\mu_a} \subseteq \dom{\mu_a'}$,
    \item For all $x$ in $\dom{\epsilon_a}$ and any $\Gamma_{clos} \subseteq \Gamma$ and any $pc$,  if $\Gamma, \Delta \vdash_{pc} x: \tau_{clos}$, $\mu_a'(\epsilon_a(x)) = val_{clos}'$, $val_{clos} = \{clos(\epsilon_c,...), table~l_a (\epsilon_c,...)\}$, and $\Xi_a' \models \epsilon_c: \Gamma_{clos}$, then $\dom{\epsilon_{c}} \subseteq \dom{\epsilon_a}$ and $\semanticStoreEnv{\Xi_a'}{\Delta}{\mu_a'}{\epsilon_{c}}{\Gamma_{clos}}$. Note that $val_{clos}' = \mu_a(\epsilon_a(x))$. This implies the $\epsilon_c$ will still satisfy $\dom{\epsilon_{c}} \subseteq \dom{\epsilon_a}$. By applying the induction hypothesis on $~\semanticStoreEnv{\Xi_a}{\Delta}{\mu_{a}}{\epsilon_{c}}{\Gamma}$ with $\dom{\epsilon_{c}} \subseteq \dom{\epsilon_a}$, we can conclude that $\semanticStoreEnv{\Xi_a'}{\Delta}{\mu_a'}{\epsilon_{c}}{\Gamma_{clos}}$ holds.
  \end{enumerate}
\end{enumerate}
\end{proof}
\begin{lemma}[Non-interference with Subtyping]\label{lem:ni-subtyping}
If $\NIval{l}{\Xi_a}{\Xi_b}{\Delta}{val_{a}}{val_b}{\type{\tau}{\chi}}$ and $\chi \sqsubseteq \chi'$, then $\NIval{l}{\Xi_a}{\Xi_b}{\Delta}{val_{a}}{val_b}{\type{\tau}{\chi'}}$.
\end{lemma}
\begin{proof}
To show $\NIval{l}{\Xi_a}{\Xi_b}{\Delta}{val_{a}}{val_b}{\type{\tau}{\chi'}}$, we need to show the following:
\begin{enumerate}
\item
$\Xi_a, \Delta \vdash_{} val_a: \langle \tau, \chi' \rangle$. The above holds true using the TV-SubType rule in  \Cref{vtyping}, since we have the premise $\Xi_a, \Delta \vdash_{} val_a: \langle \tau, \chi \rangle$ and $\chi \sqsubseteq \chi'$
\item
$\Xi_b, \Delta \vdash_{} val_b: \langle \tau, \chi' \rangle$. Similar to the above case.
\item if $\tau \notin \{ \{\overline{f: \rho }\},~ header~\{\overline{f: \rho }\}, ~\overline{d~\rho} \xrightarrow{pc} \rho_{ret}, ~table(pc_{tbl})\}$ and $\chi' \sqsubseteq l$, then $val_{a} = val_{b}$. Since we know $\chi \sqsubseteq \chi'$, if $\chi' \sqsubseteq l$, then $\chi \sqsubseteq l$, and $val_{a} = val_{b}$ (according to $\NIval{l}{\Xi_a}{\Xi_b}{\Delta}{val_{a}}{val_b}{\type{\tau}{\chi}}$).
\end{enumerate}
\end{proof}
\begin{lemma} \label{lem-ni-val-subtype}
If $\NIval{l}{\Xi_a}{\Xi_b}{\Delta}{val_{a}}{val_b}{\type{\tau}{\chi}}$ and $\Xi_a \subseteq \Xi_a'$ and $\Xi_b' \subseteq \Xi_b'$, then $\NIval{l}{\Xi_a'}{\Xi_b'}{\Delta}{val_{a}}{val_b}{\type{\tau}{\chi}}$.
\end{lemma}
\begin{proof}
Direct proof using the definition of $\NIval{l}{\Xi_a}{\Xi_b}{\Delta}{val_{a}}{val_b}{\type{\tau}{\chi}}$ and applying  \Cref{aux}, and \Cref{lem-env-subtyping}
\end{proof}
\begin{definition}\label{def-unused}
$\textsf{unused}(\langle\mu, \epsilon\rangle, x, l_a)$ holds when the following are satisfied:
\begin{itemize}
\item if $x \in \dom{\epsilon}$, then $\epsilon(x) \ne l_a$,
\item for all $y \in \dom{\epsilon}$, if $\mu(\epsilon(y)) = val_{clos}$, where $val_{clos} \in \{clos(\epsilon_c,...), table~l(\epsilon_c,...)\}$, then $\textsf{unused}(\langle\mu, \epsilon_{clos}\rangle, x, l_a)$ holds.
\end{itemize}
\end{definition}
\begin{lemma} \label{unused-stmt}
Suppose $\semanticStoreEnv{\Xi}{\Delta}{\mu}{\epsilon}{\Gamma}$ and $\langle \Delta, \mu, \epsilon, stmt \rangle \Downarrow \langle \mu', \epsilon', sig \rangle$. For some $\epsilon_a$ with $\semanticStoreEnv{\Xi_a}{\Delta}{\mu}{\epsilon_a}{\Gamma'}$ and some variable $x \in \dom{\epsilon_a}$, if $\textsf{unused}(\langle\mu, \epsilon\rangle, x, \epsilon_a(x))$, then $\mu(\epsilon_a(x)) = \mu'(\epsilon_a(x))$.
\end{lemma}
\begin{proof}
By induction on the evaluation derivation of $stmt$. Involves mutual induction of \Cref{unused-exp}, \Cref{unused-stmt}, and \Cref{unused-decl}.
Some of the interesting bits include concluding $\semanticStoreEnv{\Xi}{\Delta'}{\mu'}{\epsilon_a}{\Gamma'}$ (using the definition), $\textsc{unused}(\langle \mu', \epsilon' \rangle, x, \epsilon_a(x))$, and the fact that declaration introduces new locations into $\mu'$.
\end{proof}
\begin{lemma}\label{unused-exp}
Suppose $\semanticStoreEnv{\Xi}{\Delta}{\mu}{\epsilon}{\Gamma}$ and $\langle \Delta, \mu, \epsilon, exp \rangle \Downarrow \langle \mu', val \rangle$. For some $\epsilon_a$ with $\semanticStoreEnv{\Xi_a}{\Delta}{\mu}{\epsilon_a}{\Gamma'}$ and some variable $x \in \dom{\epsilon_a}$, if $\textsf{unused}(\langle\mu, \epsilon\rangle, x, \epsilon_a(x))$, then $\mu(\epsilon_a(x)) = \mu'(\epsilon_a(x))$.
\end{lemma}
\begin{proof}
By induction on the evaluation derivation of $exp$. Involves mutual induction of \Cref{unused-exp}, \Cref{unused-stmt}, and \Cref{unused-decl}.
Some of the interesting bits include concluding $\semanticStoreEnv{\Xi}{\Delta}{\mu'}{\epsilon_a}{\Gamma'}$ (using the definition), $\textsc{unused}(\langle \mu', \epsilon \rangle, x, \epsilon_a(x))$.
\end{proof}
\begin{lemma}\label{unused-decl}
Suppose $\semanticStoreEnv{\Xi}{\Delta}{\mu}{\epsilon}{\Gamma}$ and $\langle \Delta, \mu, \epsilon, decl \rangle \Downarrow \langle \Delta', \mu', \epsilon', cont \rangle$. For some $\epsilon_a$ with $\semanticStoreEnv{\Xi_a}{\Delta}{\mu}{\epsilon_a}{\Gamma'}$ and some variable $x \in \dom{\epsilon_a}$, if $\textsf{unused}(\langle\mu, \epsilon\rangle, x, \epsilon_a(x))$, then $\mu(\epsilon_a(x)) = \mu'(\epsilon_a(x))$.
\end{lemma}
\begin{proof}
By induction on the evaluation derivation of $decl$. Involves mutual induction of \Cref{unused-exp}, \Cref{unused-stmt}, and \Cref{unused-decl}.
Some of the interesting bits include concluding $\semanticStoreEnv{\Xi}{\Delta'}{\mu'}{\epsilon_a}{\Gamma'}$ (using the definition), $\textsc{unused}(\langle \mu', \epsilon' \rangle, x, \epsilon_a(x))$.
\end{proof}
\begin{lemma} \label{lem:pair-env-extend}
	Suppose $~\semanticBelowPCState{pc}{\Xi_a}{\Xi_b}{\Delta}{\mu_{a}}{\epsilon_{a1}}{\mu_{b}}{\epsilon_{b1}}{\Gamma_1}$, $~\semanticBelowPCState{pc}{\Xi_a}{\Xi_b}{\Delta}{\mu_{a}}{\epsilon_{a2}}{\mu_{b}}{\epsilon_{b2}}{\Gamma_2}$, where $\epsilon_{a2} = \{\overline{x \mapsto l_a}\}$ and $\epsilon_{b2} = \{\overline{x \mapsto l_b}\}$ and $\Gamma_2 = [\{\overline{x \mapsto \type{\tau}{\chi}}\}]$.
Then $~\semanticBelowPCState{pc}{\Xi_a}{\Xi_b}{\Delta}{\mu_{a}}{\epsilon_{a1}[\overline{x \mapsto l_a}]}{\mu_{b}}{\epsilon_{b1}[\overline{x \mapsto l_a}]}{\Gamma_1[\overline{x \mapsto \type{\tau}{\chi}}]}$.
\end{lemma}
\begin{proof}
To prove all the requirements of the \Cref{def:mem-store-pair-semantic}, we use the fact that independently all $y \in \dom{\epsilon_{a1}} = \dom{\epsilon_{b1}}$ and $y \in \dom{\epsilon_{a2}} = \dom{\epsilon_{b2}}$ satisfy the required properties. Now, in the extended environment all $y \in \dom{\epsilon_{a1}[\overline{x \mapsto l_a}]} = \dom{\epsilon_{b1}[\overline{x \mapsto l_b}]}$, will also satisfy the properties by reducing to either an element in $\dom{\epsilon_{a1}}$ or $\dom{\epsilon_{a2}}$
\end{proof}

\section{L-value Evaluation Rules}
For a term to be a well-formed l-value, the directionality of the term should be inout. Therefore, only the following typing judgements can be used in the derivation of a well-formed l-value:
\paragraph*{Valid L-Value Expression Typing Rules}
\begin{list}{•}{}
\item T-Var
\item T-Index
\item T-MemHdr
\item T-MemRec
\end{list}
Therefore, l-value is given by the following grammar:
\begin{align*}
base &::= x  \\
lval &::= base
            \OR lval.f_i
            \OR lval[n]
\end{align*}
\paragraph*{\textsf{lval_base}}
Note that only $base \in \epsilon$, where $\epsilon$ is the environment in which the l-value is evaluated.
Other l-values like the ones corresponding to a header field or an array index do not map to a location in the environment.
Instead, it is the header variable or the array variable that has an entry in the environment.
For instance, to write to a header field $lval.f_i$ where $lval$ is the l-value of the header that needs to be updated, the value of the header variable given by $lval$ is updated and there is no variable $lval.f_i$ in $\epsilon$.
The value at location pointed by $\epsilon(lval)$ is then overwritten with the new header value.
Therefore, we define a function $\textsc{lval_base}(lval)$ to return the l-value of the base variable that will be touched while writing to the $lval$.
This is inductively defined using:
\begin{align*}
  \textsc{lval_base}(base) &= base \\
  \textsc{lval_base}(lval.f_i) &= \textsc{lval_base}(lval) \\
  \textsc{lval_base}(lval[n]) &= \textsc{lval_base}(lval)
\end{align*}

\subsection{L-value Equality Relation} \label{def:lval-eq}
We inductively define an equality relation on l-value expressions as follows:
\[
x =_{lval} x
\]
\[
\inferrule*[]
{
lval_{a} =_{lval} lval_{b}
}
{
lval_{a}.f =_{lval} lval_{b}.f
}
\]
\[
\inferrule*[]
{
lval_{a} =_{lval} lval_{b} \qquad n_{a} = n_{b} \qquad n_{a}: int
}
{
lval_{a}[n_{a}] =_{lval} lval_{b}[n_{b}]
}
\]

\begin{definition} \label{lem:lval-eval}
For any security label $l$, $\NIlvalEval[pc]{\Gamma}{\Delta}{lval\_exp}{\type{\tau}{\chi}}$ holds for any $\Xi_a$, $\Xi_b$, $\mu_{a}$, $\mu_{b}$, $\epsilon_{a}$, $\epsilon_{b}$, $\mu_a'$, $\mu_b'$, if
\begin{enumerate}
\item $\semanticBelowPCState{l}{\Xi_a}{\Xi_b}{\Delta}{\mu_a}{\epsilon_a}{\mu_b}{\epsilon_b}{\Gamma}$,
\item $\evalstolval{\config[lval\_exp]{\mathcal{C};\Delta}{\mu_{a}}{\epsilon_{a}}}{\configval{\mu_{a}'}{lval_a}}$ and $\evalstolval{\config[lval\_exp]{\mathcal{C};\Delta}{\mu_{b}}{\epsilon_{b}}}{\configval{\mu_{b}'}{lval_b}}$
\end{enumerate}
then there exists some $\Xi_a'$, $\Xi_b'$ such that
\begin{enumerate}
\item $\Xi_a \subseteq \Xi_a'$, $\Xi_b \subseteq \Xi_b'$, $\dom{\mu_a} \subseteq \dom{\mu_a'}$ and $\dom{\mu_b}\subseteq \dom{\mu_b'}$,
\item $\semanticBelowPCState{l}{\Xi_{a}'}{\Xi_{b}'}{\Delta}{\mu_{a}'}{\epsilon_{a}}{\mu_{b}'}{\epsilon_{b}}{\Gamma}$,
\item For any $l_a \in \dom{\mu_a}$ such that $\ordinaryTyping[]{\Xi_a}{\Delta}{\mu_a(l_a)}{\tau_{clos}}$, where $\tau_{clos} \in \{\tau_{fn}, \tau_{tbl}\}$, then $\mu_a'(l_a) = \mu_a(l_a)$. Similarly for any $l_b \in \dom{\mu_b}$ such that $\ordinaryTyping[]{\Xi_b}{\Delta}{\mu_b(l_b)}{\tau_{clos}}$, where $\tau_{clos} \in \{\tau_{fn}, \tau_{tbl}\}$, then $\mu_b'(l_b) = \mu_b(l_b)$,
\item if $\chi \sqsubseteq l$, then $\lvalEquality{lval_a}{lval_b}$,
\item $\textsc{lval_base}(lval_a) \in \dom{\epsilon_a}$, $\textsc{lval_base}(lval_b) \in \dom{\epsilon_b}$ and $\textsc{lval_base}(lval_a) = \textsc{lval_base}(lval_b)$,
\item $\ordinaryTyping[pc]{\Gamma}{\Delta}{lval_a}{\type{\tau}{\chi}}$ and $\ordinaryTyping[pc]{\Gamma}{\Delta}{lval_b}{\type{\tau}{\chi}}$,
\item for any $l_a' \in \dom{\mu_{a}}$ and $l_b' \in \dom{\mu_{b}}$ such that $\ordinaryTyping[]{\Xi_{a}}{\Delta}{\mu_{a}(l_a')}{\type{\tau}{\chi}}$ and $\ordinaryTyping[]{\Xi_{b}}{\Delta}{\mu_{b}(l_b')}{\type{\tau}{\chi}}$ and $pc \nsqsubseteq \chi$, we have $\mu_{a}(l_a') = \mu_{a1}(l_a')$ and $\mu_{b}(l_b') = \mu_{b1}(l_b')$
\end{enumerate}
\end{definition}

\begin{lemma} \label{lem-lval-eval}
For any security label $l$, if $\ordinaryTyping[pc]{\Gamma}{\Delta}{lval\_exp}{\type{\tau}{\chi}}$, then $\NIlvalEval[pc]{\Gamma}{\Delta}{lval\_exp}{\type{\tau}{\chi}}$.
\end{lemma}

\begin{proof}
We prove this by induction on the typing derivation of $\ordinaryTyping[pc]{\Gamma}{\Delta}{lval\_exp}{\type{\tau}{\chi}}$, where we choose the last typing rule to be the different cases.
\begin{enumerate}
\item \textbf{T-Var}

Consider the case where the last typing rule in the derivation of l-value expression is T-Var
\[   \inferrule*[right=T-Var]
{   x \in \dom{\Gamma} \qquad
\Gamma(x) = \type{\tau}{\chi}
}
{
\ordinaryTyping[pc]{\Gamma}{\Delta}{x}{\type{\tau}{\chi} \text{~goes inout}}
}
\]
then we need to show that for any $\Xi_a$, $\Xi_b$, $\mu_{a}$, $\mu_{b}$, $\epsilon_{a}$, $\epsilon_{b}$, $\mu_a'$, $\mu_b'$, if
\begin{equation} \label{leval-var-hyp}
\semanticBelowPCState{l}{\Xi_a}{\Xi_b}{\Delta}{\mu_a}{\epsilon_a}{\mu_b}{\epsilon_b}{\Gamma}
\end{equation}
and $x$ is evaluated to get its l-value in two initial configuration $\langle \mu_a, \epsilon_a \rangle$ and $\langle \mu_b, \epsilon_b \rangle$ as follows:
\begin{mathpar}
\inferrule
{~}
{ \langle \mathcal{C}, \Delta_{}, \mu_{a}, \epsilon_{a}, x \rangle \ \Downarrow_{lval}  \langle \mu_{a}, x\rangle
}

\inferrule
{~}
{ \langle \mathcal{C}, \Delta_{}, \mu_{b}, \epsilon_{b}, x \rangle \ \Downarrow_{lval}  \langle \mu_{b}, x \rangle
}
\end{mathpar}
then there exists some $\Xi_a'$ and $\Xi_b'$ satisfying the following properties:
\begin{enumerate}
\item \label{leval-var-tp-1} $\Xi_a \subseteq \Xi_a'$, $\Xi_b \subseteq \Xi_b'$, $\dom{\mu_a} \subseteq \dom{\mu_a'}$, $\dom{\mu_b}\subseteq \dom{\mu_b'}$, and $\semanticBelowPCState{l}{\Xi_{a}'}{\Xi_{b}'}{\Delta}{\mu_{a}'}{\epsilon_{a}}{\mu_{b}'}{\epsilon_{b}}{\Gamma}$, where in this case $\mu_a' = \mu_a$ and $\mu_b' = \mu_b$,
\item \label{leval-var-tp-3} For any $l_a \in \dom{\mu_a}$ such that $\ordinaryTyping[]{\Xi_a}{\Delta}{\mu_a(l_a)}{\tau_{clos}}$, where $\tau_{clos} \in \{\tau_{fn}, \tau_{tbl}\}$, then $\mu_a'(l_a) = \mu_a(l_a)$. Similarly for any $l_b \in \dom{\mu_b}$ such that $\ordinaryTyping[]{\Xi_b}{\Delta}{\mu_b(l_b)}{\tau_{clos}}$, where $\tau_{clos} \in \{\tau_{fn}, \tau_{tbl}\}$, then $\mu_b'(l_b) = \mu_b(l_b)$,
\item  \label{leval-var-tp-2} $\textsc{lval_base}(lval_a) \in \dom{\epsilon_a}$, $\textsc{lval_base}(lval_b) \in \dom{\epsilon_b}$ and $\textsc{lval_base}(lval_a) = \textsc{lval_base}(lval_b)$.
\item \label{leval-var-tp-4} $\ordinaryTyping[pc]{\Gamma}{\Delta}{lval_a}{\type{\tau}{\chi}}$ and $\ordinaryTyping[pc]{\Gamma}{\Delta}{lval_b}{\type{\tau}{\chi}}$,
\item \label{leval-var-tp-5} for any $l_a' \in \dom{\mu_{a}}$ and $l_b' \in \dom{\mu_{b}}$ such that $\ordinaryTyping[]{\Xi_{a}}{\Delta}{\mu_{a}(l_a')}{\type{\tau}{\chi}}$ and $\ordinaryTyping[]{\Xi_{b}}{\Delta}{\mu_{b}(l_b')}{\type{\tau}{\chi}}$ and $pc \nsqsubseteq \chi$, we have $\mu_{a}(l_a') = \mu_{a}'(l_a')$ and $\mu_{b}(l_b') = \mu_{b}'(l_b')$.
\end{enumerate}
Showing that \Cref{leval-var-tp-1} holds for $\Xi_a' = \Xi_a$, $\Xi_b' = \Xi_b$, $\mu_a' = \mu_a$, and $\mu_b' = \mu_b$ is same as proving \Cref{leval-var-hyp}, which is already given.
\Cref{leval-var-tp-3} is trivial as the memory stores do not change.
So is \Cref{leval-var-tp-5}.
\Cref{leval-var-tp-2} is immediate since $lval_a =_{lval} x =_{lval} lval_b$.
Additionally, $x \in \dom{\epsilon_a} = \dom{\epsilon_b}$ since $x \in \dom{\Gamma}$, $\Xi_a \vdash \epsilon_a: \Gamma$, and $\Xi_b \vdash \epsilon_b: \Gamma$.
\Cref{leval-var-tp-4} follows from $\ordinaryTyping[pc]{\Gamma}{\Delta}{x}{\type{\tau}{\chi}}$.

\item \textbf{T-MemRec}

Consider the case where the last typing rule in the derivation of l-value expression is T-MemRec
\[  \inferrule*[right=T-MemRec]
{
\ordinaryTyping[pc]{\Gamma}{\Delta}{exp}{\type{\{ \overline{f_i: \langle \tau_i, \chi_i \rangle} \}}{\bot}}~goes~d
}
{
\ordinaryTyping[pc]{\Gamma}{\Delta}{exp.f_{i}}{\type{\tau_{i}}{\chi_{i}}~ goes~ d}
}
\]
then we need to show that for any $\Xi_a$, $\Xi_b$, $\mu_{a}$, $\mu_{b}$, $\epsilon_{a}$, $\epsilon_{b}$, $\mu_a'$, $\mu_b'$, if
\[
\semanticBelowPCState{l}{\Xi_a}{\Xi_b}{\Delta}{\mu_a}{\epsilon_a}{\mu_b}{\epsilon_b}{\Gamma}
\]
and $exp.f_i$ is evaluated to get its l-value in two initial configuration $\langle \mu_a, \epsilon_a \rangle$ and $\langle \mu_b, \epsilon_b \rangle$ as follows:
\begin{mathpar}
\inferrule*[]
{
 \langle \mathcal{C}, \Delta_{}, \mu_{a}, \epsilon_{a}, exp \rangle \Downarrow_{lval}
\langle \mu_{a}', lval_{a}\rangle
}
{
 \langle \mathcal{C}, \Delta_{}, \mu_{a}, \epsilon_{a}, exp.f_i \rangle \Downarrow_{lval}
\langle \mu_{a}', lval_{a}.f_i\rangle
}

\inferrule*[]
{
 \langle \mathcal{C}, \Delta_{}, \mu_{b}, \epsilon_{b}, exp \rangle \Downarrow_{lval}
\langle \mu_{b}', lval_{b} \rangle
}
{
 \langle \mathcal{C}, \Delta_{}, \mu_{b}, \epsilon_{b}, exp.f_i \rangle \Downarrow_{lval}   \langle \mu_{b}', lval_{b}.f_i \rangle
}
\end{mathpar}
then there exists some $\Xi_a'$ and $\Xi_b'$ satisfying the following properties:
\begin{enumerate}
\item \label{leval-rec-tp-1} $\Xi_a \subseteq \Xi_a'$, $\Xi_b \subseteq \Xi_b'$, $\dom{\mu_a} \subseteq \dom{\mu_a'}$ and $\dom{\mu_b}\subseteq \dom{\mu_b'}$.
Also, $\semanticBelowPCState{l}{\Xi_{a}'}{\Xi_{b}'}{\Delta}{\mu_{a}'}{\epsilon_{a}}{\mu_{b}'}{\epsilon_{b}}{\Gamma}$.
\item \label{leval-rec-tp-3} For any $l_a \in \dom{\mu_a}$ such that $\ordinaryTyping[]{\Xi_a}{\Delta}{\mu_a(l_a)}{\tau_{clos}}$, where $\tau_{clos} \in \{\tau_{fn}, \tau_{tbl}\}$, then $\mu_a'(l_a) = \mu_a(l_a)$. Similarly for any $l_b \in \dom{\mu_b}$ such that $\ordinaryTyping[]{\Xi_b}{\Delta}{\mu_b(l_b)}{\tau_{clos}}$, where $\tau_{clos} \in \{\tau_{fn}, \tau_{tbl}\}$, then $\mu_b'(l_b) = \mu_b(l_b)$,
\item  \label{leval-rec-tp-2}$\textsc{lval_base}(lval_a.f_i) = \textsc{lval_base}(lval_b.f_i)$,  $\textsc{lval_base}(lval_a.f_i) = \textsc{lval_base}(lval_a) \in \dom{\epsilon_a}$ and $\textsc{lval_base}(lval_b.f_i) = \textsc{lval_base}(lval_b) \in \dom{\epsilon_b}$,
\item \label{leval-rec-tp-4} $\ordinaryTyping[pc]{\Gamma}{\Delta}{lval_a.f_i}{\type{\tau_i}{\chi_i}}$ and $\ordinaryTyping[pc]{\Gamma}{\Delta}{lval_b.f_i}{\type{\tau_i}{\chi_i}}$,
\item \label{leval-rec-tp-5} if $\chi \sqsubseteq l$, then $\lvalEquality{lval_a.f_i}{lval_b.f_i}$.
\item \label{leval-rec-tp-6} for any $l_a' \in \dom{\mu_{a}}$ and $l_b' \in \dom{\mu_{b}}$ such that $\ordinaryTyping[]{\Xi_{a}}{\Delta}{\mu_{a}(l_a')}{\type{\tau}{\chi}}$ and $\ordinaryTyping[]{\Xi_{b}}{\Delta}{\mu_{b}(l_b')}{\type{\tau}{\chi}}$ and $pc \nsqsubseteq \chi$, we have $\mu_{a}(l_a') = \mu_{a}'(l_a')$ and $\mu_{b}(l_b') = \mu_{b}'(l_b')$
\end{enumerate}

By applying induction hypothesis on the typing derivation of $exp$, we conclude $\NIlvalEval[pc]{\Gamma}{\Delta}{exp}{\type{\{ \overline{f_i: \langle \tau_i, \chi_i \rangle} \}}{\bot}}$. Since $exp$ is evaluated to get its l-value in two initial configuration $\langle \mu_a, \epsilon_a \rangle$ and $\langle \mu_b, \epsilon_b \rangle$, where $\semanticBelowPCState{l}{\Xi_a}{\Xi_b}{\Delta}{\mu_a}{\epsilon_a}{\mu_b}{\epsilon_b}{\Gamma}$, there exists some $\Xi_{a1}'$ and $\Xi_{b1}'$ satisfying $\Xi_a \subseteq \Xi_{a1}'$ and $\Xi_b \subseteq \Xi_{b1}'$, $\dom{\mu_{a}} \subseteq \dom{\mu_a'}$, $\dom{\mu_{b}} \subseteq \dom{\mu_b'}$  and the following:
\begin{equation} \label{leval-rec-in-1-1}
\semanticBelowPCState{l}{\Xi_{a1}'}{\Xi_{b1}'}{\Delta}{\mu_{a1}}{\epsilon_{a}}{\mu_{b1}}{\epsilon_{b}}{\Gamma},
\end{equation}
\begin{equation} \label{leval-rec-in-1-2}
\textsc{lval_base}(lval_a) = \textsc{lval_base}(lval_b) \text{,}~ \textsc{lval_base}(lval_a) \in \dom{\epsilon_a}~\text{and}~ \textsc{lval_base}(lval_b) \in \dom{\epsilon_b}
\end{equation}
\begin{equation} \label{leval-rec-in-1-3}
\ordinaryTyping[pc]{\Gamma}{\Delta}{lval_a}{\type{\{ \overline{f_i: \langle \tau_i, \chi_i \rangle} \}}{\bot}}~ \text{and}~ \ordinaryTyping[pc]{\Gamma}{\Delta}{lval_b}{\type{\{ \overline{f_i: \langle \tau_i, \chi_i \rangle} \}}{\bot}}
\end{equation}
\begin{equation} \label{leval-rec-in-1-4}
\text{if}~ \bot \sqsubseteq l, ~\text{then}~ \lvalEquality{lval_a}{lval_b}.
\end{equation}
and For any $l_a \in \dom{\mu_a}$ such that $\ordinaryTyping[]{\Xi_a}{\Delta}{\mu_a(l_a)}{\tau_{clos}}$, where $\tau_{clos} \in \{\tau_{fn}, \tau_{tbl}\}$, then $\mu_a'(l_a) = \mu_a(l_a)$. Similarly for any $l_b \in \dom{\mu_b}$ such that $\ordinaryTyping[]{\Xi_b}{\Delta}{\mu_b(l_b)}{\tau_{clos}}$, where $\tau_{clos} \in \{\tau_{fn}, \tau_{tbl}\}$, then $\mu_b'(l_b) = \mu_b(l_b)$. With this we have shown \Cref{leval-rec-tp-3}.

Also, for any $l_a' \in \dom{\mu_{a}}$ and $l_b' \in \dom{\mu_{b}}$ such that $\ordinaryTyping[]{\Xi_{a}}{\Delta}{\mu_{a}(l_a')}{\type{\tau}{\chi}}$ and $\ordinaryTyping[]{\Xi_{b}}{\Delta}{\mu_{b}(l_b')}{\type{\tau}{\chi}}$ and $pc \nsqsubseteq \chi$, we have $\mu_{a}(l_a') = \mu_{a1}(l_a')$ and $\mu_{b}(l_b') = \mu_{b1}(l_b')$. This proves \Cref{leval-rec-tp-6}.
\Cref{leval-rec-in-1-1} proves \Cref{leval-rec-tp-1}.
Proof of \Cref{leval-rec-tp-2} follows from the definition of \textsc{lval_base} and \Cref{leval-rec-in-1-2}.
Using \Cref{leval-rec-in-1-3} and T-MemRec we conclude \Cref{leval-rec-tp-4}.
\Cref{leval-rec-in-1-4} with the definition \Cref{def:lval-eq} proves \Cref{leval-rec-tp-5}. Note that we have proved that $\lvalEquality{lval_a.f_i}{lval_b.f_i}$.
\item \textbf{T-MemHdr}

Consider the case where the last typing rule in the derivation of l-value expression is T-MemHdr
\[  \inferrule*[right=T-MemHdr]
{
\ordinaryTyping[pc]{\Gamma}{\Delta}{exp} {\type{header \{ \overline{f_i: \langle \tau_i, \chi_i \rangle} \}}{\bot}~ goes~ d}
}
{
\ordinaryTyping[pc]{\Gamma}{\Delta}{exp.f_{i}}{\type{\tau_{i}}{\chi_{i}}~goes~ d}
}
\]
then showing all the required properties for an evaluation rule as follows is similar to the T-MemRec case.
\begin{mathpar}
\inferrule*[]
{
 \langle \mathcal{C}, \Delta_{}, \mu_{a}, \epsilon_{a}, exp \rangle \Downarrow_{lval}
\langle \mu_{a1}, lval_{a} \rangle
}
{
 \langle \mathcal{C}, \Delta_{}, \mu_{a}, \epsilon_{a}, exp.f \rangle \Downarrow_{lval}
\langle \mu_{a1}, lval_{a}.f \rangle
}

\inferrule*[]
{
 \langle \mathcal{C}, \Delta_{}, \mu_{b}, \epsilon_{b}, exp \rangle \Downarrow_{lval}
\langle \mu_{b1}, lval_{b} \rangle
}
{
 \langle \mathcal{C}, \Delta_{}, \mu_{b}, \epsilon_{b}, exp.f \rangle \Downarrow_{lval}   \langle \mu_{b1}, lval_{b}.f \rangle
}
\end{mathpar}

\item \textbf{T-Index}

Consider the case where the last typing rule in the derivation of l-value expression is T-Index
\[  \inferrule*[right=T-Index]
{
\ordinaryTyping[pc]{\Gamma}{\Delta}{exp_{1}}{\type{\type{\tau}{\chi_1}[n]}{\bot}~goes~ d} \\\\
\ordinaryTyping[pc]{\Gamma}{\Delta}{exp_{2}}{\type{bit \langle 32 \rangle}{\chi_2}} \\\\
\chi_2 \sqsubseteq \chi_1
}
{
\ordinaryTyping[pc]{\Gamma}{\Delta} {exp_{1}[exp_{2}]}{\type{\tau}{\chi_1}~ goes~ d }
}
\]
then we need to show that for any $\Xi_a$, $\Xi_b$, $\mu_{a}$, $\mu_{b}$, $\epsilon_{a}$, $\epsilon_{b}$, $\mu_a'$, $\mu_b'$, if
\[
\semanticBelowPCState{l}{\Xi_a}{\Xi_b}{\Delta}{\mu_a}{\epsilon_a}{\mu_b}{\epsilon_b}{\Gamma}
\]
and $exp_{1}[exp_{2}]$ is evaluated to get its l-value in two initial configuration $\langle \mu_a, \epsilon_a \rangle$ and $\langle \mu_b, \epsilon_b \rangle$ as follows:

\begin{mathpar}
\inferrule*[]
{
 \langle \mathcal{C}, \Delta_{}, \mu_{a}, \epsilon_{a}, exp_{1} \rangle \Downarrow_{lval}  \langle \mu_{a1}, lval_{a} \rangle \\
 \langle \mathcal{C}, \Delta_{}, \mu_{a1}, \epsilon_{a}, exp_{2} \rangle \Downarrow  \langle \mu_{a2}, n_{a} \rangle
}
{
 \langle \mathcal{C}, \Delta_{}, \mu_{a}, \epsilon_{a}, exp_{1}[exp_{2}] \rangle  \Downarrow_{lval}  \langle \mu_{a2}, lval_{a}[n_{a}] \rangle
}

\inferrule*[]
{
 \langle \mathcal{C}, \Delta_{}, \mu_{b}, \epsilon_{b}, exp_{1} \rangle \Downarrow_{lval}  \langle \mu_{b1}, lval_{b} \rangle \\
 \langle \mathcal{C}, \Delta_{}, \mu_{b1}, \epsilon_{b}, exp_{2} \rangle \Downarrow  \langle \mu_{b2}, n_{b} \rangle
}
{
 \langle \mathcal{C}, \Delta_{}, \mu_{b}, \epsilon_{b}, exp_{1}[exp_{2}] \rangle  \Downarrow_{lval}  \langle \mu_{b2}, lval_{b}[n_{b}]\rangle
}
\end{mathpar}
then there exists some $\Xi_a'$ and $\Xi_b'$ satisfying the following properties:
\begin{enumerate}
\item \label{leval-index-tp-1} $\Xi_a \subseteq \Xi_a'$, $\Xi_b \subseteq \Xi_b'$, $\dom{\mu_a} \subseteq \dom{\mu_a'}$ and $\dom{\mu_b}\subseteq \dom{\mu_b'}$, where $\mu_a' = \mu_{a2}$ and $\mu_b' = \mu_{b2}$.
Also, $\semanticBelowPCState{l}{\Xi_{a}'}{\Xi_{b}'}{\Delta}{\mu_{a}'}{\epsilon_{a}}{\mu_{b}'}{\epsilon_{b}}{\Gamma}$.
\item \label{leval-index-tp-3} For any $l_a \in \dom{\mu_a}$ such that $\ordinaryTyping[]{\Xi_a}{\Delta}{\mu_a(l_a)}{\tau_{clos}}$, where $\tau_{clos} \in \{\tau_{fn}, \tau_{tbl}\}$, then $\mu_a'(l_a) = \mu_a(l_a)$. Similarly for any $l_b \in \dom{\mu_b}$ such that $\ordinaryTyping[]{\Xi_b}{\Delta}{\mu_b(l_b)}{\tau_{clos}}$, where $\tau_{clos} \in \{\tau_{fn}, \tau_{tbl}\}$, then $\mu_b'(l_b) = \mu_b(l_b)$,
\item  \label{leval-index-tp-2}$\textsc{lval_base}(lval_a[n_a]) = \textsc{lval_base}(lval_b[n_b])$, $\textsc{lval_base}(lval_{a}[n_{a}]) \in \dom{\epsilon_a}$ and $\textsc{lval_base}(lval_{b}[n_{b}]) \in \dom{\epsilon_b}$,
\item \label{leval-index-tp-4} $\ordinaryTyping[pc]{\Gamma}{\Delta}{lval_a[n_a]}{\type{\tau}{\chi_1}}$ and $\ordinaryTyping[pc]{\Gamma}{\Delta}{lval_b[n_b]}{\type{\tau}{\chi_1}}$,
\item \label{leval-index-tp-5} if $\chi_1 \sqsubseteq l$, then $\lvalEquality{lval_a[n_a]}{lval_b[n_b]}$.
\item \label{leval-index-tp-6} for any $l_a' \in \dom{\mu_{a}}$ and $l_b' \in \dom{\mu_{b}}$ such that $\ordinaryTyping[]{\Xi_{a}}{\Delta}{\mu_{a}(l_a')}{\type{\tau}{\chi}}$ and $\ordinaryTyping[]{\Xi_{b}}{\Delta}{\mu_{b}(l_b')}{\type{\tau}{\chi}}$ and $pc \nsqsubseteq \chi$, we have $\mu_{a}(l_a') = \mu_{a2}(l_a')$ and $\mu_{b}(l_b') = \mu_{b2}(l_b')$
\end{enumerate}

By applying induction hypothesis on the typing derivation of $exp_1$, we conclude $\NIlvalEval[pc]{\Gamma}{\Delta}{exp_1}{\type{\type{\tau}{\chi_1}[n]}{\bot}}$. Since $exp_1$ is evaluated to get its l-value in two initial configuration $\langle \mu_a, \epsilon_a \rangle$ and $\langle \mu_b, \epsilon_b \rangle$, where $\semanticBelowPCState{l}{\Xi_a}{\Xi_b}{\Delta}{\mu_a}{\epsilon_a}{\mu_b}{\epsilon_b}{\Gamma}$, there exists some $\Xi_{a1}'$ and $\Xi_{b1}'$ satisfying $\Xi_a \subseteq \Xi_{a1}'$ and $\Xi_b \subseteq \Xi_{b1}'$, $\dom{\mu_{a}} \subseteq \dom{\mu_{a1}}$, $\dom{\mu_{b}} \subseteq \dom{\mu_{b1}}$ and the following:
\begin{equation} \label{leval-index-in-1-1}
\semanticBelowPCState{l}{\Xi_{a1}'}{\Xi_{b1}'}{\Delta}{\mu_{a1}}{\epsilon_{a}}{\mu_{b1}}{\epsilon_{b}}{\Gamma},
\end{equation}
\begin{equation} \label{leval-index-in-1-2}
\textsc{lval_base}(lval_a) = \textsc{lval_base}(lval_b) \text{,}~ \textsc{lval_base}(lval_a) \in \dom{\epsilon_a}~\text{and}~\textsc{lval_base}(lval_b) \in \dom{\epsilon_b}
\end{equation}
\begin{equation} \label{leval-index-in-1-3}
\ordinaryTyping[pc]{\Gamma}{\Delta}{lval_a}{\type{\type{\tau}{\chi_1}[n]}{\bot}}~ \text{and}~ \ordinaryTyping[pc]{\Gamma}{\Delta}{lval_b}{\type{\type{\tau}{\chi_1}[n]}{\bot}}
\end{equation}
\begin{equation} \label{leval-index-in-1-4}
\text{if}~ \bot \sqsubseteq l, ~\text{then}~ \lvalEquality{lval_a}{lval_b}.
\end{equation}
and For any $l_a \in \dom{\mu_a}$ such that $\ordinaryTyping[]{\Xi_a}{\Delta}{\mu_a(l_a)}{\tau_{clos}}$, where $\tau_{clos} \in \{\tau_{fn}, \tau_{tbl}\}$, then $\mu_{a1}(l_a) = \mu_a(l_a)$. Similarly for any $l_b \in \dom{\mu_b}$ such that $\ordinaryTyping[]{\Xi_b}{\Delta}{\mu_b(l_b)}{\tau_{clos}}$, where $\tau_{clos} \in \{\tau_{fn}, \tau_{tbl}\}$, then $\mu_{b1}(l_b) = \mu_b(l_b)$.
Also, for any $l_a' \in \dom{\mu_{a}}$ and $l_b' \in \dom{\mu_{b}}$ such that $\ordinaryTyping[]{\Xi_{a}}{\Delta}{\mu_{a}(l_a')}{\type{\tau}{\chi}}$ and $\ordinaryTyping[]{\Xi_{b}}{\Delta}{\mu_{b}(l_b')}{\type{\tau}{\chi}}$ and $pc \nsqsubseteq \chi$, we have $\mu_{a}(l_a') = \mu_{a1}(l_a')$ and $\mu_{b}(l_b') = \mu_{b1}(l_b')$.

By applying induction hypothesis of \Cref{ni-exp} on the $exp_2$, we get $\NIexp{pc}{\Gamma}{\Delta}{exp_2}{\type{bit \langle 32 \rangle}{\chi_2}}$. Since $exp_2$ is evaluated in an initial configuration satisfying \Cref{leval-rec-in-1-1}, we can conclude that there exists some $\Xi_{a2}'$ and $\Xi_{b2}'$ satisfying $\Xi_{a1}' \subseteq \Xi_{a2}'$, $\Xi_{b1}' \subseteq \Xi_{b2}'$, and $\dom{\mu_{a1}} \subseteq \dom{\mu_{a2}}$, $\dom{\mu_{b1}} \subseteq \dom{\mu_{b2}}$ and the following:
\[
\semanticBelowPCState{l}{\Xi_{a2}'}{\Xi_{b2}'}{\Delta}{\mu_{a2}}{\epsilon_{a}}{\mu_{b2}}{\epsilon_{b}}{\Gamma}
\]
\[
\NIval{l}{\Xi_{a2}'}{\Xi_{b2}'}{\Delta}{n_{a}}{n_{b}}{\type{bit \langle 32 \rangle}{\chi_{2}}}
\]
And finally, For any $l_a \in \dom{\mu_a}$ such that $\ordinaryTyping[]{\Xi_a}{\Delta}{\mu_a(l_a)}{\tau_{clos}}$, where $\tau_{clos} \in \{\tau_{fn}, \tau_{tbl}\}$, then $\mu_{a2}(l_a) = \mu_{a1}(l_a)$. Similarly for any $l_b \in \dom{\mu_b}$ such that $\ordinaryTyping[]{\Xi_b}{\Delta}{\mu_b(l_b)}{\tau_{clos}}$, where $\tau_{clos} \in \{\tau_{fn}, \tau_{tbl}\}$, then $\mu_{b2}(l_b) = \mu_{b1}(l_b)$.
This equation proves \Cref{leval-index-tp-3}, since $\dom{\mu_a} \subseteq \dom{\mu_{a1}} \subseteq \dom{\mu_{a2}}$ and $\dom{\mu_b} \subseteq \dom{\mu_{b1}} \subseteq \dom{\mu_{b2}}$. Similarly, we prove \Cref{leval-index-tp-6}

Since the type of $n_{a}$ is $\tau= {bit \langle 32 \rangle}$, we can say that $n_{a} = n_{b}$ if the $\chi_1 \sqsubseteq l$. Therefore, \Cref{leval-index-in-1-4} proves \Cref{leval-index-tp-5}.
Using the definition of \textsc{lval_base} along with \Cref{leval-index-in-1-2} we can conclude \Cref{leval-index-tp-2}.
\Cref{leval-index-in-1-3} with T-Index proves \Cref{leval-index-tp-4}.
\end{enumerate}
\end{proof}

\section{L-value Writing}
\begin{lemma}\label{no-side-lval}
Let $\ordinaryTyping[pc]{\Gamma}{\Delta}{lval\_exp}{\type{\tau}{\chi}}$, $\evalstolval{\config[lval\_exp]{\mathcal{C};\Delta}{\mu}{\epsilon}}{\configval{\mu'}{lval}}$. Suppose $ \langle \mathcal{C}, \Delta_{}, \mu_{1}, \epsilon_{1}, lval \rangle \Downarrow  \langle \mu_{2}, val \rangle$, $\dom{\mu'} \subseteq \dom{\mu_1}$ and $\dom{\epsilon} \subseteq \dom{\epsilon_1}$. Then $\mu_2 = \mu_1$.
\end{lemma}
\begin{proof}
By induction on the typing derivation of $lval\_exp$. Intuitively, $lval$ has no unevaluated expression, so evaluating a normalized value will not have side-effects.
\end{proof}

\begin{lemma} \label{lw-val-ni}
Let $\ordinaryTyping[pc]{\Gamma}{\Delta}{lval\_exp}{\type{\tau}{\chi}}$, $\evalstolval{\config[lval\_exp]{\mathcal{C};\Delta}{\mu_{a}}{\epsilon_{a}}}{\configval{\mu_{a}'}{lval_a}}$ and $\evalstolval{\config[lval\_exp]{\mathcal{C};\Delta}{\mu_{b}}{\epsilon_{b}}}{\configval{\mu_{b}'}{lval_b}}$.

Suppose $\semanticBelowPCState{l}{\Xi_{a}}{\Xi_{b}}{\Delta}{\mu_a}{\epsilon_{a}}{\mu_{b}}{\epsilon_{b}}{\Gamma}$, $\semanticBelowPCState{l}{\Xi_{a1}}{\Xi_{b1}}{\Delta}{\mu_{a1}}{\epsilon_{a}}{\mu_{b1}}{\epsilon_{b}}{\Gamma}$, and $\NIval{l}{\Xi_{a1}}{\Xi_{b1}}{\Delta}{val_a}{val_b}{\type{\tau}{\chi}}$, where $\Xi_a \subseteq \Xi_{a1}$, $\Xi_b \subseteq \Xi_{b1}$, $\dom{\mu_a} \subseteq \dom{\mu_{a1}}$, $\dom{\mu_b} \subseteq \dom{\mu_{b1}}$.

If $ \langle \mathcal{C}, \Delta,\mu_{a1}, \epsilon_{a}, lval_a := val_a \rangle \Downarrow_{write} \mu_{a1}' $, $ \langle \mathcal{C}, \Delta,\mu_{b1}, \epsilon_{b}, lval_b := val_b \rangle \Downarrow_{write} \mu_{b1}'$, then
\begin{enumerate}
\item $\NIval{l}{\Xi_{a1}}{\Xi_{b1}}{\Delta}{\mu_{a1}'(\epsilon_{a}(\textsc{lval_base}(lval_a))}{\mu_{b1}'(\epsilon_{b}(\textsc{lval_base}(lval_b))}{\Gamma(\textsc{lval_base}(lval_b))}$,

\item For any $l_a \in \dom{\mu_a}$ such that $\ordinaryTyping[]{\Xi_a}{\Delta}{\mu_a(l_a)}{\tau_{clos}}$, where $\tau_{clos} \in \{\tau_{fn}, \tau_{tbl}\}$, then $\mu_{a1}'(l_a) = \mu_{a1}(l_a)$. Similarly for any $l_b \in \dom{\mu_b}$ such that $\ordinaryTyping[]{\Xi_b}{\Delta}{\mu_b(l_b)}{\tau_{clos}}$, where $\tau_{clos} \in \{\tau_{fn}, \tau_{tbl}\}$, then $\mu_{b1}'(l_b) = \mu_{b1}(l_b)$.
\item for any $l_a \in \dom{\mu_{a1}}$ and $l_b \in \dom{\mu_{b1}}$ such that $l_b \ne \epsilon_{b}(\textsc{lval_base}(lval_b))$ and $l_a \ne \epsilon_{a}(\textsc{lval_base}(lval_a))$, we have $\mu_{a1}'(l_a) = \mu_{a1}(l_a)$ and $\mu_{b1}'(l_b) = \mu_{b1}(l_b)$,
\end{enumerate}

\end{lemma}
\begin{proof}
By induction hypothesis on the typing derivation of $lval\_exp$.
\begin{enumerate}
\item \textbf{T-Var}

If the l-value expression's typing derivation ends with a variable typing rule, then write to the l-value follows the following  evaluation, where $\mu_{a1}' = \mu_{a1}[l_a\coloneq val_a]$, and $\mu_{b1}' = \mu_{b1}[l_b\coloneq val_b]$.
\[
\inferrule*[]
{
\epsilon_{a}(x) = l_a
}
{ \langle \mathcal{C}, \Delta_{}, \mu_{a1}, \epsilon_{a}, x \coloneq val_a \rangle \Downarrow_{write}  \mu_{a1}[l_a\coloneq val_a]		}
\]
\[\inferrule*[]
{
\epsilon_{b}(x) = l_b
}
{ \langle \mathcal{C}, \Delta_{}, \mu_{b1}, \epsilon_{b}, x \coloneq val_b \rangle \Downarrow_{write}  \mu_{b1}[l_b \coloneq val_b]	}
\]
According to the evaluation rule, $\mu_{a1}'(\epsilon_{a}(\textsc{lval_base}(x)) = \mu_{a1}'(\epsilon_{a}(x)) = val_a$ and $\mu_{b1}'(\epsilon_{b}(\textsc{lval_base}(x)) = \mu_{b}'(\epsilon_{b}(x)) = val_b$. Since we already know that $\NIval{l}{\Xi_{a1}}{\Xi_{b1}}{\Delta}{val_a}{val_b}{\type{\tau}{\chi}}$, we have proved the requirement. Since the memory store doesn't change for other location's besides that of $x$, showing the other two requirements are direct.
\item \textbf{T-Mem}

If the l-value expression's $\ordinaryTyping[pc]{\Gamma}{\Delta}{exp.f_i}{\type{\tau_i}{\chi_i}}$ typing derivation ends with a T-MemRec rule, then write to the l-value follows the following evaluation, where $\mu_{a1}' = \mu_{a3}$ and $\mu_{b1}' = \mu_{b3}$.
\[
\inferrule*[]
{
 \langle \mathcal{C}, \Delta_{}, \mu_{a1}, \epsilon_{a}, lval_a \rangle \Downarrow  \langle \mu_{a2}, \{\overline{f_j = val_{f_a}}\} \rangle \\
 \langle \mathcal{C}, \Delta_{}, \mu_{a2}, \epsilon_{a}, lval_a \coloneq \{f_{i} = val_a, ~ \overline{f_{j\ne i}=val_{f_a}} \} \rangle \Downarrow_{write} \mu_{a3}
}
{
 \langle \mathcal{C}, \Delta_{}, \mu_{a1}, \epsilon_{a}, lval_a.f_{i} \coloneq val_a \rangle \Downarrow_{write}  \mu_{a3}
}
\]

\[
\inferrule*[]
{
 \langle \mathcal{C}, \Delta_{}, \mu_{b1}, \epsilon_{b}, lval_b \rangle \Downarrow  \langle \mu_{b2}, \{\overline{f_j = val_{f_b}}\} \rangle \\
 \langle \mathcal{C}, \Delta_{}, \mu_{b2}, \epsilon_{b}, lval_b \coloneq \{f_{i} = val_b, ~ \overline{f_{j\ne i}=val_{f_b}} \} \rangle \Downarrow_{write} \mu_{b3}
}
{
 \langle \mathcal{C}, \Delta_{}, \mu_{b1}, \epsilon_{b}, lval_b.f_{i} \coloneq val_b \rangle \Downarrow_{write}  \mu_{b3}
}
\]
We know that $\textsc{lval_base}(lval_a.f_i) = \textsc{lval_base}(lval_b.f_i)$ using \Cref{lem-lval-eval}.
We can have two cases:
\begin{itemize}
\item $\chi_i \sqsubseteq l$.
According to \Cref{lem-lval-eval}, this implies that $lval_a.f_i =_{lval} lval_b.f_i$. Therefore, by inversion of the equality defined in \Cref{def:lval-eq}, $lval_a =_{lval} lval_b$. By using \Cref{no-side-lval} to get the value of the respective l-value, we get $\mu_{a2} = \mu_{a1}$ and $\mu_{b2} = \mu_{b1}$. This implies $\semanticBelowPCState{l}{\Xi_{a1}}{\Xi_{b1}}{\Delta}{\mu_{a2}}{\epsilon_{a}}{\mu_{b2}}{\epsilon_{b}}{\Gamma_1}$.
$lval_a$ and $lval_b$ are returned by sub-expression $exp$ of $lval\_exp = exp.f_i$, therefore, by \Cref{lem-lval-eval}, we have $\ordinaryTyping[pc]{\Gamma}{\Delta}{lval_a}{\type{\{f: \type{\tau_i}{\chi_i}\}}{\bot}}$ and $\ordinaryTyping[pc]{\Gamma}{\Delta}{lval_b}{\type{\{f: \type{\tau_i}{\chi_i}\}}{\bot}}$. Therefore, we can apply induction hypothesis of \Cref{ni-exp} to evaluate the value of a well-typed expression under two different configurations. By applying induction hypothesis of \Cref{ni-exp} on evaluating $lval_a =_{lval} lval_b$, we get $\NIval{l}{\Xi_a}{\Xi_b}{\Delta}{\{\overline{f_j = val_{f_a}}\}}{\{\overline{f_j = val_{f_b}}\}}{\type{\{f: \type{\tau_i}{\chi_i}\}}{\bot}}$. Since $val_a$ and $val_b$ are given to be non-interfering, we have
\[\NIval{l}{\Xi_a}{\Xi_b}{\Delta}{\{f_{i} = val_a, ~ \overline{f_{j\ne i}=val_{f_a}} \}}{\{f_{i} = val_b, ~ \overline{f_{j\ne i}=val_{f_b}} \}}{\type{\{f: \type{\tau_i}{\chi_i}\}}{\bot}}\]
We can apply the induction hypothesis of this lemma to write to two l-value expressions generated from a well-typed $exp$, and conclude that
\[\NIval{l}{\Xi_{a1}}{\Xi_{b1}}{\Delta}{\mu_{a3}(\epsilon_{a}(\textsc{lval_base}(lval_a))}{\mu_{b3}(\epsilon_{b}(\textsc{lval_base}(lval_b))}{\Gamma(\textsc{lval_base}(lval_b))}\]
Since $\textsc{lval_base}(lval_a.f_i) = \textsc{lval_base}(lval_a)$ and $\textsc{lval_base}(lval_b.f_i) = \textsc{lval_base}(lval_b)$, we have proved the necessary. Also, the two other requirements follow from the results of this induction hypothesis.
\item $\chi_i \nsqsubseteq l$.
Since $lval_a$ and $lval_b$ are evaluated from $exp$, by using \Cref{lem-lval-eval}, we can conclude that $lval_a =_{lval} lval_b$ because $\bot \sqsubseteq l$. Also, the type of both $lval_a$ and $lval_b$ is ${\type{\{f: \type{\tau_i}{\chi_i}\}}{\bot}}$. Now, by applying induction hypothesis of \Cref{ni-exp} to evaluate the value of a well-typed expression $lval_a =_{lval} lval_b$ under two different configurations, we conclude $\NIval{l}{\Xi_a}{\Xi_b}{\Delta}{\{\overline{f_j = val_{f_a}}\}}{\{\overline{f_j = val_{f_b}}\}}{\type{\{f: \type{\tau_i}{\chi_i}\}}{\bot}}$.
Similar to the previous case, we can apply the induction hypothesis of this lemma on the lval-write to $lval_a$ and $lval_b$ because they are both evaluated from sub-expression $exp$ of $exp.f_i$ (this can be checked from the lval-evaluation derivation). This can conclude that $\NIval{l}{\Xi_{a1}}{\Xi_{b1}}{\Delta}{\mu_{a3}(\epsilon_{a}(\textsc{lval_base}(lval_a))}{\mu_{b3}(\epsilon_{b}(\textsc{lval_base}(lval_b))}{\Gamma(\textsc{lval_base}(lval_b))}$. Also, the two other requirements follow from the results of this induction hypothesis.
\end{itemize}
\item \textbf{T-Hdr}

Follows similarly.
\[
\inferrule*[]
{
 \langle \mathcal{C}, \Delta_{}, \mu_{}, \epsilon_{}, lval \rangle \Downarrow  \langle \mu_{1}, \text{header}\{valid=~true, ~\overline{f = val_{f}}\} \rangle \\
 \langle \mathcal{C}, \Delta_{}, \mu_{1}, \epsilon_{}, lval \coloneq header\{valid = true, f_{i} = val, ~ \overline{f_{\ne i}=val_{f}} \} \rangle \Downarrow_{write} \mu_{2}
}
{
 \langle \mathcal{C}, \Delta_{}, \mu_{}, \epsilon_{}, lval.f_{i} \coloneq val \rangle \Downarrow_{write}  \mu_{2}
}
\]
\item \textbf{T-Index}

If the l-value expression's typing derivation ends with a T-Index rule
\[  \inferrule*[right=T-Index]
{
\ordinaryTyping[pc]{\Gamma}{\Delta}{exp_{1}}{\type{\type{\tau}{\chi_1}[n]}{\bot}~goes~ d} \\\\
\ordinaryTyping[pc]{\Gamma}{\Delta}{exp_{2}}{\type{bit \langle 32 \rangle}{\chi_2}} \\\\
\chi_2 \sqsubseteq \chi_1
}
{
\ordinaryTyping[pc]{\Gamma}{\Delta} {exp_{1}[exp_{2}]}{\type{\tau}{\chi_1}~ goes~ d }
}
\]
then write to the l-value follows the following evaluation, where $\mu_{a1}' = \mu_{a3}$ and $\mu_{b1}' = \mu_{b3}$.
\[
\inferrule*[]
{
 \langle \mathcal{C}, \Delta_{}, \mu_{a1}, \epsilon_{a}, lval_a \rangle \Downarrow  \langle \mu_{a2}, stack~ \tau~\{\overline{val_a}\} \rangle \\
 \langle \mathcal{C}, \Delta_{}, \mu_{a2}, \epsilon_{a}, lval_a \coloneq stack~\tau~\{...,val_{a_{n_a-1}}, val_a, val_{a_{n_a+1}} ,...\} \rangle \Downarrow_{write} \mu_{a3}
}
{
 \langle \mathcal{C}, \Delta_{}, \mu_{a1}, \epsilon_{a}, lval_a[n_a] \coloneq val_a \rangle \Downarrow_{write}  \mu_{a3}
}
\]

\[
\inferrule*[]
{
 \langle \mathcal{C}, \Delta_{}, \mu_{b1}, \epsilon_{b}, lval_b \rangle \Downarrow  \langle \mu_{b2}, stack~ \tau~\{\overline{val_b}\} \rangle \\
 \langle \mathcal{C}, \Delta_{}, \mu_{b2}, \epsilon_{b}, lval_b \coloneq stack~\tau~\{...,val_{b_{n_b-1}}, val_b, val_{b_{n_b+1}} ,...\} \rangle \Downarrow_{write} \mu_{b3}
}
{
 \langle \mathcal{C}, \Delta_{}, \mu_{b1}, \epsilon_{b}, lval_b[n_b] \coloneq val_b \rangle \Downarrow_{write}  \mu_{b3}
}
\]

We know that $\textsc{lval_base}(lval_a[n_a]) = \textsc{lval_base}(lval_b[n_b])$ using \Cref{lem-lval-eval} on $exp_1$.
We can have two cases:
\begin{itemize}
\item $\chi_1 \sqsubseteq l$. \Cref{lem-lval-eval} implies that $lval_a[n_a] =_{lval} lval_b[n_b]$. Therefore, $lval_a =_{lval} lval_b$ and $n_a = n_b$. Using similar argument to the "record" case, we can show that $\NIval{l}{\Xi_{a1}}{\Xi_{b1}}{\Delta}{\mu_{a3}(\epsilon_{a}(\textsc{lval_base}(lval_a))}{\mu_{b3}(\epsilon_{b}(\textsc{lval_base}(lval_b))}{\Gamma(\textsc{lval_base}(lval_b))}$. And by using the definition of $\textsc{lval_base}$, we conclude \[\NIval{l}{\Xi_{a1}}{\Xi_{b1}}{\Delta}{\mu_{a3}(\epsilon_{a}(\textsc{lval_base}(lval_a[n_a]))}{\mu_{b3}(\epsilon_{b}(\textsc{lval_base}(lval_b[n_b]))}{\Gamma(\textsc{lval_base}(lval_b))}\]
\item $\chi_1 \nsqsubseteq l$.
We can have the following cases:
\begin{itemize}
\item Case $\chi_2 \sqsubseteq l$. In this case $n_a = n_b$. (using induction hypothesis of  \Cref{ni-exp} on $exp_2$ evaluation in $\Downarrow_{lval}$ of $exp_1[exp_2]$). Observe that $\ordinaryTyping[pc]{\Gamma}{\Delta}{exp_1}{\type{\type{\tau}{\chi_1}[n]}{\bot}}$, and because $\bot \sqsubseteq l$, we have $lval_a =_{lval} lval_b$ (by using \Cref{lem-lval-eval}, $lval_a$ is the lvalue generated from $exp_1$). By applying \Cref{ni-exp} on $lval_a$'s evaluation, we get $\NIval{l}{\Xi_{a1}}{\Xi_{b1}}{\Delta}{stack~\tau~\{\overline{val_a}\}}{stack~\tau~\{\overline{val_b}\}}{\type{\type{\tau}{\chi_1}[n]}{\bot}}$. Similar to the previous cases, by applying induction hypothesis of this lemma on the lval-write to $lval_a$ and $lval_b$ that are generated from the same $lval\_exp$, we can conclude that \[\NIval{l}{\Xi_{a1}}{\Xi_{b1}}{\Delta}{\mu_{a3}(\epsilon_{a}(\textsc{lval_base}(lval_a))}{\mu_{b3}(\epsilon_{b}(\textsc{lval_base}(lval_b))}{\Gamma(\textsc{lval_base}(lval_b))}\]
\item Case $\chi_2 \nsqsubseteq l$. In this case $n_a$ and $n_b$ can be $n_a \ne n_b$ (using induction hypothesis of  \Cref{ni-exp} on $exp_2$ evaluation in $\Downarrow_{lval}$ of $exp_1[exp_2]$).
As $\bot \sqsubseteq pc$, $lval_a =_{lval} lval_b$. By applying \Cref{ni-exp} on $lval_a$'s evaluation, we get $\NIval{l}{\Xi_{a1}}{\Xi_{b1}}{\Delta}{stack~\tau~\{\overline{val_a}\}}{stack~\tau~\{\overline{val_b}\}}{\type{\type{\tau}{\chi_1}[n]}{\bot}}$. However, with $\chi_2 \nsqsubseteq l$ and $\chi_1 \nsqsubseteq l$ and according to \Cref{def:def-ni-val}, we have $\NIval{l}{\Xi_{a1}}{\Xi_{b1}}{\Delta}{stack~\tau~\{...,val_{a_{n_a-1}}, val_a, val_{a_{n_a+1}} ,...\}}{stack~\tau~\{...,val_{b_{n_b-1}}, val_b, val_{b_{n_b+1}} ,...\}}{\type{\tau}{\chi_1}[n]}$. Similar to the previous case, by applying induction hypothesis of this lemma on the lval-write to $lval_a$ and $lval_b$ that are generated from the same $lval\_exp$, we can conclude that \[\NIval{l}{\Xi_{a1}}{\Xi_{b1}}{\Delta}{\mu_{a3}(\epsilon_{a}(\textsc{lval_base}(lval_a))}{\mu_{b3}(\epsilon_{b}(\textsc{lval_base}(lval_b))}{\Gamma(\textsc{lval_base}(lval_b))}\]
\end{itemize}
\end{itemize}
\end{enumerate}
\end{proof}

\begin{lemma} \label{lem:lvalue-write-final}
Let $\ordinaryTyping[pc]{\Gamma}{\Delta}{lval\_exp}{\type{\tau}{\chi}}$, $\evalstolval{\config[lval\_exp]{\mathcal{C};\Delta}{\mu_{a}}{\epsilon_{a}}}{\configval{\mu_{a}'}{lval_a}}$ and $\evalstolval{\config[lval\_exp]{\mathcal{C};\Delta}{\mu_{b}}{\epsilon_{b}}}{\configval{\mu_{b}'}{lval_b}}$.

Suppose $\semanticBelowPCState{l}{\Xi_{a}}{\Xi_{b}}{\Delta}{\mu_a}{\epsilon_{a}}{\mu_{b}}{\epsilon_{b}}{\Gamma}$, $\semanticBelowPCState{l}{\Xi_{a1}}{\Xi_{b1}}{\Delta}{\mu_{a1}}{\epsilon_{a}}{\mu_{b1}}{\epsilon_{b}}{\Gamma}$, and $\NIval{l}{\Xi_{a1}}{\Xi_{b1}}{\Delta}{val_a}{val_b}{\type{\tau}{\chi}}$, where $\Xi_a \subseteq \Xi_{a1}$, $\Xi_b \subseteq \Xi_{b1}$, $\dom{\mu_a} \subseteq \dom{\mu_{a1}}$, $\dom{\mu_b} \subseteq \dom{\mu_{b1}}$.
If $ \langle \mathcal{C}, \Delta,\mu_{a1}, \epsilon_{a}, lval_a := val_a \rangle \Downarrow_{write} \mu_{a1}' $ and $ \langle \mathcal{C}, \Delta,\mu_{b1}, \epsilon_{b}, lval_b := val_b \rangle \Downarrow_{write} \mu_{b1}' $, then $\semanticBelowPCState{l}{\Xi_{a1}}{\Xi_{b1}}{\Delta}{\mu_{a1}'}{\epsilon_{a}}{\mu_{b1}'}{\epsilon_{b}}{\Gamma}$.
\end{lemma}
\begin{proof}
Follows from \Cref{lw-val-ni}.
\end{proof}

\section{Function Evaluation Strategy}
\begin{lemma} \label{lem:copy-in-out-one}
	Consider the following well-typed expressions $\ordinaryTyping[pc]{\Gamma}{\Delta}{exp}{\type{\tau}{\chi'}}$, and $d~x:\langle \tau, \chi\rangle := exp$, where $\chi' \sqsubseteq \chi$ is evaluated in two different initial configurations $\langle \mu_{a}, \epsilon_{a} \rangle$ and $\langle \mu_{b}, \epsilon_{b} \rangle$ satisfying $\semanticBelowPCState{l}{\Xi_a}{\Xi_b}{\Delta}{\mu_{a}}{\epsilon_{a}}{\mu_{b}}{\epsilon_{b}}{\Gamma}$ as follows:
\[\evalstocopy
    {\config[d~x:\langle \tau, \chi\rangle := exp]{\mathcal{C}; \Delta}{\mu_{a}}{\epsilon_{a}}}
    {\configcopy{\mu_{a}'}{x \mapsto l_a}{lval_a \mapsto l_a}}\]
and
\[\evalstocopy
    {\config[d~x:\langle \tau, \chi\rangle := exp]{\mathcal{C}; \Delta}{\mu_{b}}{\epsilon_{b}}}
    {\configcopy{\mu_{b}'}{x \mapsto l_b}{lval_b \mapsto l_b}}\]
then
\begin{enumerate}
	\item $\semanticBelowPCState{l}{\Xi_a'}{\Xi_b'}{\Delta}{\mu_{a}'}{x \mapsto l_a}{\mu_{b}'}{x \mapsto l_b}{\Gamma'}$, for some $\Xi_a'$, $\Xi_b'$, $\mu_{a}'$, $\mu_{b}'$ such that $\Xi_a \subseteq \Xi_a'$ and $\Xi_b \subseteq \Xi_b'$, $\dom{\mu_{a}} \subseteq \dom{\mu_{a}'}$, $\dom{\mu_{b}} \subseteq \dom{\mu_{b}'}$ and $\Gamma' = \{x \mapsto \type{\tau}{\chi}\}$.
	\item $\semanticBelowPCState{l}{\Xi_a'}{\Xi_b'}{\Delta}{\mu_{a}'}{\epsilon_{a}}{\mu_{b}'}{\epsilon_{b}}{\Gamma}$,
	\item $\textsc{lval_base}(lval_a) \in \dom{\epsilon_a}$, $\textsc{lval_base}(lval_b) \in \dom{\epsilon_b}$, and $\textsc{lval_base}(lval_a) = \textsc{lval_base}(lval_b)$,
	\item $l_a \in \dom{\mu_{a}'}$ and $l_b \in \dom{\mu_{b}'}$,
	\item $l_a$ and $l_b$ are fresh locations, $l_a \notin \Xi_a$ and $l_b \notin \Xi_b$,
	\item For any $l_a \in \dom{\mu_a}$ such that $\ordinaryTyping[]{\Xi_a}{\Delta}{\mu_a(l_a)}{\tau_{clos}}$, where $\tau_{clos} \in \{\tau_{fn}, \tau_{tbl}\}$, then $\mu_a'(l_a) = \mu_a(l_a)$,
  \item For any $l_b \in \dom{\mu_b}$ such that $\ordinaryTyping[]{\Xi_b}{\Delta}{\mu_b(l_b)}{\tau_{clos}}$, where $\tau_{clos} \in \{\tau_{fn}, \tau_{tbl}\}$, then $\mu_b'(l_b) = \mu_b(l_b)$,
  \item \emph{PC is used to bound writes.} For any $l_a \in \dom{\mu_a}$ and $l_b \in \dom{\mu_b}$ such that $\ordinaryTyping[]{\Xi_a}{\Delta}{\mu_a(l_a)}{\type{\tau}{\chi}}$ and $\ordinaryTyping[]{\Xi_b}{\Delta}{\mu_b(l_b)}{\type{\tau}{\chi}}$ and $pc \nsqsubseteq \chi$, we have $\mu_{a}'(l_a) = \mu_{a}(l_a)$ and $\mu_{b}'(l_b) = \mu_{b}(l_b)$.
\end{enumerate}

\textbf{Note.} By \Cref{def:mem-store-pair-semantic}, $\dom{\epsilon_a} = \dom{\epsilon_b}$.
\end{lemma}
\begin{proof}
Case analysis on the possible directionalities d for the arguments.
\begin{enumerate}
\item \textbf{Copy In}
If the statement $in~x:\langle \tau, \chi\rangle := exp$ is evaluated in two different initial configurations $\langle \mu_a, \epsilon_a \rangle$ and $\langle \mu_b, \epsilon_b \rangle$ satisfying
		\[
			\semanticBelowPCState{l}{\Xi_a}{\Xi_b}{\Delta}{\mu_{a}}{\epsilon_{a}}{\mu_{b}}{\epsilon_{b}}{\Gamma}
    \]

				as follows:
				\[
						\inferrule*[]
						{
							 \langle \mathcal{C}, \Delta, \mu_{a}, \epsilon_{a}, exp\rangle \Downarrow  \langle \mu_{a1}, val_{a}\rangle \qquad l_{a}~ fresh
						}
						{ \langle \mathcal{C}, \Delta, \mu_{a}, \epsilon_{a}, in~ x: \langle \tau, \chi \rangle= exp\rangle \Downarrow_{copy}  \langle \mu_{a1}[l_{a} \mapsto val_{a}], x\mapsto l_{a}, []\rangle
						}
				\]

				\[
						\inferrule*[]
						{
							 \langle \mathcal{C}, \Delta, \mu_{b}, \epsilon_{b}, exp\rangle \Downarrow  \langle \mu_{b1}, val_{b}\rangle \qquad l_{b}~ fresh
						}
						{ \langle \mathcal{C}, \Delta, \mu_{b}, \epsilon_{b}, in~ x: \langle \tau, \chi \rangle= exp\rangle \Downarrow_{copy}  \langle \mu_{b1}[l_{b} \mapsto val_{b}], x\mapsto l_{b}, []\rangle
						}
				\]

    then we need to show each of the following, where $\mu_{a}' = \mu_{a1}[l_{a} \mapsto val_{a}]$ and $\mu_{b}' = \mu_{b1}[l_{b} \mapsto val_{b}]$, $\Gamma' = \{x \mapsto \type{\tau}{\chi}\}$.
\begin{enumerate}
	\item \label{copy-part-a} $\semanticBelowPCState{l}{\Xi_a'}{\Xi_b'}{\Delta}{\mu_{a}'}{x \mapsto l_a}{\mu_{b}'}{x \mapsto l_b}{\Gamma'}$, for some $\Xi_a'$, $\Xi_b'$, $\mu_a'$, $\mu_b'$ such that $\Xi_a \subseteq \Xi_a'$ and $\Xi_b \subseteq \Xi_b'$, $\dom{\mu_a} \subseteq \dom{\mu_a'}$, $\dom{\mu_b} \subseteq \dom{\mu_b'}$.
	\item \label{copy-part-b} $\semanticBelowPCState{l}{\Xi_a'}{\Xi_b'}{\Delta}{\mu_{a}'}{\epsilon_a}{\mu_{b}'}{\epsilon_b}{\Gamma}$
	\item \label{copy-part-c} Since the set of l-values, \ie the third element of the final tuple is empty, vacuously we have $\textsc{lval_base}(lval_a) \in \dom{\epsilon_a}$, $\textsc{lval_base}(lval_b) \in \dom{\epsilon_b}$, and $\textsc{lval_base}(lval_a) = \textsc{lval_base}(lval_b)$.
	\item \label{copy-part-d} $l_a \in \dom{\mu_{a}'}$, $l_b \in \dom{\mu_{b}'}$, $l_a$ and $l_b$ are fresh locations, $l_a \notin \Xi_a$ and $l_b \notin \Xi_b$,
	\item \label{copy-part-e}For any $l_a \in \dom{\mu_a}$ such that $\ordinaryTyping[]{\Xi_a}{\Delta}{\mu_a(l_a)}{\tau_{clos}}$, where $\tau_{clos} \in \{\tau_{fn}, \tau_{tbl}\}$, then $\mu_a'(l_a) = \mu_a(l_a)$. Similarly for any $l_b \in \dom{\mu_b}$ such that $\ordinaryTyping[]{\Xi_b}{\Delta}{\mu_b(l_b)}{\tau_{clos}}$, where $\tau_{clos} \in \{\tau_{fn}, \tau_{tbl}\}$, then $\mu_b'(l_b) = \mu_b(l_b)$,
	\item For any $l_a \in \dom{\mu_a}$ and $l_b \in \dom{\mu_b}$ such that $\ordinaryTyping[]{\Xi_a}{\Delta}{\mu_a(l_a)}{\type{\tau}{\chi}}$ and $\ordinaryTyping[]{\Xi_b}{\Delta}{\mu_b(l_b)}{\type{\tau}{\chi}}$ and $pc \nsqsubseteq \chi$, we have $\mu_{a}'(l_a) = \mu_{a}(l_a)$ and $\mu_{b}'(l_b) = \mu_{b}(l_b)$.
\end{enumerate}

By applying the induction hypothesis of \Cref{ni-exp} on $\ordinaryTyping[pc]{\Gamma}{\Delta}{exp}{\type{\tau}{\chi'}}$, we conclude that $\NIexp{pc}{\Gamma}{\Delta}{exp}{\type{\tau}{\chi'}}$. This can be expanded to show that there exist some $\Xi_{a1}$, $\Xi_{b1}$, $\mu_{a1}$, $\mu_{b1}$ satisfying $\Xi_a \subseteq \Xi_{a1}$, $\Xi_b \subseteq \Xi_{b1}$, $\dom{\mu_a} \subseteq \dom{\mu_{a1}}$,  $\dom{\mu_b} \subseteq \dom{\mu_{b1}}$ and the following:
	\begin{equation}	\label{copy-in-ih-1-1}
		\semanticBelowPCState{l}{\Xi_{a1}}{\Xi_{b1}}{\Delta}{\mu_{a1}}{\epsilon_{a}}{\mu_{b1}}{\epsilon_{b}}{\Gamma}
	\end{equation}
	\begin{equation}\label{copy-in-ih-1-red}
		\NIval{l}{\Xi_{a1}}{\Xi_{b1}}{\Delta}{val_{a}}{val_b}{\type{\tau}{\chi'}}
	\end{equation}
and for any $l_a \in \dom{\mu_a}$ such that $\ordinaryTyping[]{\Xi_a}{\Delta}{\mu_a(l_a)}{\tau_{clos}}$, where $\tau_{clos} \in \{\tau_{fn}, \tau_{tbl}\}$, then $\mu_{a1}(l_a) = \mu_a(l_a)$. Similarly for any $l_b \in \dom{\mu_b}$ such that $\ordinaryTyping[]{\Xi_b}{\Delta}{\mu_b(l_b)}{\tau_{clos}}$, where $\tau_{clos} \in \{\tau_{fn}, \tau_{tbl}\}$, then $\mu_{b1}(l_b) = \mu_b(l_b)$,
Also, for any $l_a \in \dom{\mu_a}$ and $l_b \in \dom{\mu_b}$ such that $\ordinaryTyping[]{\Xi_a}{\Delta}{\mu_a(l_a)}{\type{\tau}{\chi}}$ and $\ordinaryTyping[]{\Xi_b}{\Delta}{\mu_b(l_b)}{\type{\tau}{\chi}}$ and $pc \nsqsubseteq \chi$, we have $\mu_{a1}(l_a) = \mu_{a}(l_a)$ and $\mu_{b1}(l_b) = \mu_{b}(l_b)$,

Using \Cref{lem:ni-subtyping}, we can reduce the \Cref{copy-in-ih-1-red} as follows since $\chi' \sqsubseteq \chi$:
	\begin{equation}	\label{copy-in-ih-1-2}
		\NIval{l}{\Xi_{a1}}{\Xi_{b1}}{\Delta}{val_{a}}{val_b}{\type{\tau}{\chi}}
	\end{equation}

To prove \Cref{copy-part-a}, we take $\Xi_{a}' = \Xi_{a1}[l_a \mapsto \type{\tau}{\chi}]$, $\Xi_{b}' = \Xi_{b1}[l_b \mapsto \type{\tau}{\chi}]$, and $\mu_{a}' = \mu_{a1}[l_a \mapsto val_a]$ and $\mu_{b}' = \mu_{b1}[l_b \mapsto val_b]$. Now to prove $\semanticBelowPCState{l}{\Xi_a'}{\Xi_b'}{\Delta}{\mu_{a}'}{x \mapsto l_a}{\mu_{b}'}{x \mapsto l_a}{\Gamma'}$, we need to show:
		\begin{enumerate}
			\item $\semanticStoreEnv{\Xi_a'}{\Delta}{\mu_{a}'}{\{x \mapsto l_a\}}{\{x \mapsto \type{\tau}{\chi}\}}$ and $\semanticStoreEnv{\Xi_b'}{\Delta}{\mu_{b}'}{\{x \mapsto l_b\}}{\{x \mapsto \type{\tau}{\chi}\}}$.	$\semanticStoreEnv{\Xi_a'}{\Delta}{\mu_{a}'}{\{x \mapsto l_a\}}{\{x \mapsto \type{\tau}{\chi}\}}$ holds as $\semanticStore{\Xi_a'}{\Delta}{\mu_a'}$ (using \Cref{lem:store-typing-extended}) and $\Xi_a' \vdash \{x \mapsto l_a\}: \{x \mapsto \type{\tau}{\chi}\}$ (by definition). Since $x$ is not of function type (as we do not support higher-order function), we do not need to prove the third/ fourth property of \Cref{def:store-env-semantic typing}. Similarly, $\semanticStoreEnv{\Xi_b'}{\Delta}{\mu_{b}'}{\{x \mapsto l_b\}}{\{x \mapsto \type{\tau}{\chi}\}}$ also holds.

			\item $\dom{\{x \mapsto l_a\}} = \dom{\{x \mapsto l_b\}}$. Trivial.
			\item $\NIval{l}{\Xi_a'}{\Xi_b'}{\Delta}{\mu_a'(l_a)}{\mu_b'(l_b)}{\type{\tau}{\chi})}$

			Applying \Cref{lem-ni-val-subtype} on \Cref{copy-in-ih-1-2} with $\Xi_{a1} \subseteq \Xi_a'$ and $\Xi_{b1} \subseteq \Xi_{b}'$, we conclude $\NIval{l}{\Xi_{a}'}{\Xi_{b}'}{\Delta}{val_{a}}{val_b}{\type{\tau}{\chi}}$.
			As $\mu_a'(l_a) = val_a$ and $\mu_b'(l_b) = val_b$, we have shown the necessary.

		\end{enumerate}
    With this we have shown \Cref{copy-part-a}. Observe that we do not need to show properties related to closure variables because $x$ is not a closure variable in our setting.

    Can't this be proved by saying that old locations have same value?
		To prove \Cref{copy-part-b}, we apply \Cref{lem:only-mem-superset} on \Cref{copy-in-ih-1-1} with $\Xi_{a1} \subseteq \Xi_{a}'$, $\Xi_{b1} \subseteq \Xi_{b}'$, $\mu_{a1} \subseteq \mu_a'$, and $\mu_{b1} \subseteq \mu_b'$, to conclude $~\semanticBelowPCState{l}{\Xi_a'}{\Xi_b'}{\Delta}{\mu_{a}'}{\epsilon_{a}}{\mu_{b}'}{\epsilon_{b}}{\Gamma}$.

\Cref{copy-part-d} can be seen in the final configuration of the evaluation rule.
\Cref{copy-part-e} is satisfied using the result of applying induction hypothesis of non-interference for expression.
				\item \textbf{Copy out}

If the statement, $out~x:\langle \tau, \chi\rangle := exp$, is evaluated in two different initial configurations $\langle \mu_a, \epsilon_a \rangle$ and $\langle \mu_b, \epsilon_b \rangle$ satisfying $\semanticBelowPCState{l}{\Xi_a}{\Xi_b}{\Delta}{\mu_{a}}{\epsilon_{a}}{\mu_{b}}{\epsilon_{b}}{\Gamma}$ as follows:
				\[
						\inferrule*[]
						{
							 \langle \mathcal{C}, \Delta, \mu_{a}, \epsilon_{a}, exp\rangle \Downarrow_{lval}  \langle \mu_{a1}, lval_{a}\rangle \qquad l_{a}~ fresh
						}
						{
						 \langle \mathcal{C}, \Delta, \mu_{a}, \epsilon_{a}, out~ x: \langle \tau, \chi \rangle= exp\rangle \Downarrow_{copy}  \langle \mu_{a1}[l_{a} \mapsto init_{\Delta} \tau], x\mapsto l_{a}, [lval_{a}:=l_{a}]\rangle
						}
				\]

				\[
						\inferrule*[]
						{
							 \langle \mathcal{C}, \Delta, \mu_{b}, \epsilon_{b}, exp\rangle \Downarrow_{lval}  \langle \mu_{b1}, lval_{b}\rangle \qquad l_{b}~ fresh
						}
						{
						 \langle \mathcal{C}, \Delta, \mu_{b}, \epsilon_{b}, out~ x: \langle \tau, \chi \rangle= exp\rangle \Downarrow_{copy}  \langle \mu_{b1}[l_{b} \mapsto init_{\Delta} \tau], x\mapsto l_{b}, [lval_{b}:=l_{b}]\rangle
						}
				\]

    Then we need to show each of the following, where $\mu_a' =  \mu_{a1}[l_{a} \mapsto init_{\Delta} \tau]$ and $\mu_b' = \mu_{b1}[l_{b} \mapsto init_{\Delta} \tau]$:
\begin{enumerate}
	\item \label{copy-out-part-a} $\semanticBelowPCState{l}{\Xi_a'}{\Xi_b'}{\Delta}{\mu_{a}'}{x \mapsto l_a}{\mu_{b}'}{x \mapsto l_b}{\Gamma'}$, for some $\Xi_a'$, $\Xi_b'$, $\mu_a'$, $\mu_b'$ such that $\Xi_a \subseteq \Xi_a'$ and $\Xi_b \subseteq \Xi_b'$, $\dom{\mu_a} \subseteq \dom{\mu_a'}$, $\dom{\mu_b'} \subseteq \dom{\mu_b'}$ and $\Gamma' = \{x \mapsto \type{\tau}{\chi}\}$.
	\item \label{copy-out-part-b} $\semanticBelowPCState{l}{\Xi_a'}{\Xi_b'}{\Delta}{\mu_{a}'}{\epsilon_a}{\mu_{b}'}{\epsilon_b}{\Gamma}$
	\item \label{copy-out-part-c} $\textsc{lval_base}(lval_a) \in \dom{\epsilon_a}$, $\textsc{lval_base}(lval_b) \in \dom{\epsilon_b}$, and $\textsc{lval_base}(lval_a) = \textsc{lval_base}(lval_b)$.
	\item \label{copy-out-part-d} $l_a \in \dom{\mu_{a}'}$ and $l_b \in \dom{\mu_{b}'}$. $l_a$ and $l_b$ are fresh locations, $l_a \notin \Xi_a$ and $l_b \notin \Xi_b$,
	\item \label{copy-out-part-e} for any $l_a \in \dom{\mu_a}$ such that $\ordinaryTyping[]{\Xi_a}{\Delta}{\mu_a(l_a)}{\tau_{clos}}$, where $\tau_{clos} \in \{\tau_{fn}, \tau_{tbl}\}$, then $\mu_a'(l_a) = \mu_a(l_a)$. Similarly for any $l_b \in \dom{\mu_b}$ such that $\ordinaryTyping[]{\Xi_b}{\Delta}{\mu_b(l_b)}{\tau_{clos}}$, where $\tau_{clos} \in \{\tau_{fn}, \tau_{tbl}\}$, then $\mu_b'(l_b) = \mu_b(l_b)$,
	\item for any $l_a \in \dom{\mu_a}$ and $l_b \in \dom{\mu_b}$ such that $\ordinaryTyping[]{\Xi_a}{\Delta}{\mu_a(l_a)}{\type{\tau}{\chi}}$ and $\ordinaryTyping[]{\Xi_b}{\Delta}{\mu_b(l_b)}{\type{\tau}{\chi}}$ and $pc \nsqsubseteq \chi$, we have $\mu_{a}'(l_a) = \mu_{a}(l_a)$ and $\mu_{b}'(l_b) = \mu_{b}(l_b)$.
\end{enumerate}

The proof for this case is similar to Copy-in, with the difference that here we use \Cref{lem-lval-eval} to conclude $\semanticBelowPCState{l}{\Xi_{a1}}{\Xi_{b1}}{\Delta}{\mu_{a1}}{\epsilon_{a}}{\mu_{b1}}{\epsilon_{b}}{\Gamma}$, for some $\Xi_{a1}$, $\Xi_{b1}$, $\mu_{a1}$ and $\mu_{b1}$ satisfying $\Xi_a \subseteq \Xi_{a1}$, $\Xi_b \subseteq \Xi_{b1}$, $\dom{\mu_a} \subseteq \dom{\mu_{a1}}$ and $\dom{\mu_b} \subseteq \dom{\mu_{b1}}$, \Cref{copy-out-part-c}, and \Cref{copy-out-part-d}.
    From here proving all the cases is similar to the Copy-in.
				\item \textbf{Copy inout}

				If the statement, $inout~x:\langle \tau, \chi\rangle := exp$, is evaluated in two different initial configurations $\langle \mu_a, \epsilon_a \rangle$ and $\langle \mu_b, \epsilon_b \rangle$ satisfying $\semanticBelowPCState{l}{\Xi_a}{\Xi_b}{\Delta}{\mu_{a}}{\epsilon_{a}}{\mu_{b}}{\epsilon_{b}}{\Gamma}$ as follows:

				\[
						\inferrule*[]
						{
							 \langle \mathcal{C}, \Delta, \mu_{a}, \epsilon_{a}, exp\rangle \Downarrow_{lval}  \langle \mu_{a1}, lval_{a}\rangle \\
							 \langle \mathcal{C}, \Delta, \mu_{a1}, \epsilon_{a}, lval_{a}\rangle \Downarrow  \langle \mu_{a2}, val_{a}\rangle \\
							 l_{a}~ fresh
						}
						{
						 \langle \mathcal{C}, \Delta, \mu_{a}, \epsilon_{a}, inout~ x: \langle \tau, \chi \rangle= exp\rangle \Downarrow_{copy}  \langle \mu_{a2}[l_{a} \mapsto val_{a}], x\mapsto l_{a}, [lval_{a}:=l_a]\rangle
						}
				\]

				\[
						\inferrule*[]
						{
							 \langle \mathcal{C}, \Delta, \mu_{b}, \epsilon_{b}, exp\rangle \Downarrow_{lval}  \langle \mu_{b1}, lval_{b}\rangle \\
							 \langle \mathcal{C}, \Delta, \mu_{b1}, \epsilon_{b}, lval_{b}\rangle \Downarrow  \langle \mu_{b2}, val_{b}\rangle \\
							 l_{b}~ fresh
						}
						{
						 \langle \mathcal{C}, \Delta, \mu_{b}, \epsilon_{b}, inout~ x: \langle \tau, \chi \rangle= exp\rangle \Downarrow_{copy}  \langle \mu_{b2}[l_{b} \mapsto val_{b}], x\mapsto l_{b}, [lval_{b}:=l_b]\rangle
						}
					\]
	Then we need to show each of the following, where $\mu_a' = \mu_{a2}[l_{a} \mapsto val_{a}]$, $\mu_b' = \mu_{b2}[l_{b} \mapsto val_{b}]$:
\begin{enumerate}
	\item \label{copy-io-part-a} $\semanticBelowPCState{l}{\Xi_a'}{\Xi_b'}{\Delta}{\mu_{a}'}{x \mapsto l_a}{\mu_{b}'}{x \mapsto l_b}{\Gamma'}$, for some $\Xi_a'$, $\Xi_b'$, $\mu_a'$, $\mu_b'$ such that $\Xi_a \subseteq \Xi_a'$ and $\Xi_b \subseteq \Xi_b'$, $\dom{\mu_a} \subseteq \dom{\mu_a'}$, $\dom{\mu_b'} \subseteq \dom{\mu_b'}$ and $\Gamma' = \{x \mapsto \type{\tau}{\chi}\}$.
	\item \label{copy-io-part-b} $\semanticBelowPCState{l}{\Xi_a'}{\Xi_b'}{\Delta}{\mu_{a}'}{\epsilon_a}{\mu_{b}'}{\epsilon_b}{\Gamma}$,
	\item \label{copy-io-part-c} $\textsc{lval_base}(lval_a) \in \dom{\epsilon_a}$, $\textsc{lval_base}(lval_b) \in \dom{\epsilon_b}$, and $\textsc{lval_base}(lval_a) = \textsc{lval_base}(lval_b)$.
	\item \label{copy-io-part-d} $l_a \in \dom{\mu_{a}'}$ and $l_b \in \dom{\mu_{b}'}$. $l_a$ and $l_b$ are fresh locations, $l_a \notin \Xi_a$ and $l_b \notin \Xi_b$,
	\item \label{copy-io-part-e}For any $l_a \in \dom{\mu_a}$ such that $\ordinaryTyping[]{\Xi_a}{\Delta}{\mu_a(l_a)}{\tau_{clos}}$, where $\tau_{clos} \in \{\tau_{fn}, \tau_{tbl}\}$, then $\mu_a'(l_a) = \mu_a(l_a)$. Similarly for any $l_b \in \dom{\mu_b}$ such that $\ordinaryTyping[]{\Xi_b}{\Delta}{\mu_b(l_b)}{\tau_{clos}}$, where $\tau_{clos} \in \{\tau_{fn}, \tau_{tbl}\}$, then $\mu_b'(l_b) = \mu_b(l_b)$.
 \item For any $l_a \in \dom{\mu_a}$ and $l_b \in \dom{\mu_b}$ such that $\ordinaryTyping[]{\Xi_a}{\Delta}{\mu_a(l_a)}{\type{\tau}{\chi}}$ and $\ordinaryTyping[]{\Xi_b}{\Delta}{\mu_b(l_b)}{\type{\tau}{\chi}}$ and $pc \nsqsubseteq \chi$, we have $\mu_{a}'(l_a) = \mu_{a}(l_a)$ and $\mu_{b}'(l_b) = \mu_{b}(l_b)$.
\end{enumerate}
By applying the induction hypothesis of \Cref{lem:lval-eval}, we conclude that
\begin{enumerate}
\item $\Xi_a \subseteq \Xi_{a1}$, $\Xi_b \subseteq \Xi_{b1}$, $\dom{\mu_a} \subseteq \dom{\mu_{a1}}$ and $\dom{\mu_b}\subseteq \dom{\mu_{b1}}$,
\item $\semanticBelowPCState{l}{\Xi_{a1}}{\Xi_{b1}}{\Delta}{\mu_{a1}}{\epsilon_{a}}{\mu_{b1}}{\epsilon_{b}}{\Gamma}$
\item For any $l_a \in \dom{\mu_a}$ such that $\ordinaryTyping[]{\Xi_a}{\Delta}{\mu_a(l_a)}{\tau_{clos}}$, where $\tau_{clos} \in \{\tau_{fn}, \tau_{tbl}\}$, then $\mu_{a1}(l_a) = \mu_a(l_a)$. Similarly for any $l_b \in \dom{\mu_b}$ such that $\ordinaryTyping[]{\Xi_b}{\Delta}{\mu_b(l_b)}{\tau_{clos}}$, where $\tau_{clos} \in \{\tau_{fn}, \tau_{tbl}\}$, then $\mu_{b1}(l_b) = \mu_b(l_b)$,
\item if $\chi' \sqsubseteq pc$, then $\lvalEquality{lval_a}{lval_b}$,
\item $\textsc{lval_base}(lval_a) \in \dom{\epsilon_a}$, $\textsc{lval_base}(lval_b) \in \dom{\epsilon_b}$ and $\textsc{lval_base}(lval_a) = \textsc{lval_base}(lval_b)$.
\item $\ordinaryTyping[pc]{\Gamma}{\Delta}{lval_a}{\type{\tau}{\chi}}$ and $\ordinaryTyping[pc]{\Gamma}{\Delta}{lval_b}{\type{\tau}{\chi}}$
\item For any $l_a \in \dom{\mu_{a1}}$ and $l_b \in \dom{\mu_{b1}}$ such that $\ordinaryTyping[]{\Xi_{a1}}{\Delta}{\mu_{a1}(l_a)}{\type{\tau}{\chi}}$ and $\ordinaryTyping[]{\Xi_{b1}}{\Delta}{\mu_{b1}(l_b)}{\type{\tau}{\chi}}$ and $pc \nsqsubseteq \chi$, we have $\mu_{a1}(l_a) = \mu_{a2}(l_a)$ and $\mu_{b1}(l_b) = \mu_{b2}(l_b)$.
\end{enumerate}

By applying \Cref{ni-exp} on the expressions, $lval_a$ and $lval_b$ (which satisfy $lval_a  lval_b$), where $\ordinaryTyping[pc]{\Gamma}{\Delta}{lval_a}{\type{\tau}{\chi}}$ and $\ordinaryTyping[pc]{\Gamma}{\Delta}{lval_b}{\type{\tau}{\chi}}$ we get $\NIval{l}{\Xi_{a1}'}{\Xi_{b1}'}{\Delta}{val_a}{val_b}{\type{\tau}{\chi'}}$. Here, $\Xi_{a1} \subseteq \Xi_{a1}'$ and $\Xi_{b1} \subseteq \Xi_{b1}'$. By applying \Cref{no-side-lval}, we conclude that $\mu_{a2} = \mu_{a1}$ and $\mu_{b2} = \mu_{b1}$. Now similar to the proof for copy-in, we can prove that \Cref{copy-io-part-a} and \Cref{copy-io-part-b}. The other parts directly follow from the above induction results.
\end{enumerate}
\end{proof}

Lifting the copy-in-out rules to a list of statements, we arrive at the following lemma:
\begin{lemma} \label{lem:copy-in-out}
	Consider well-typed expressions$\overline{\ordinaryTyping[pc]{\Gamma}{\Delta}{exp}{\type{\tau}{\chi'}}}$ and the statement, $\overline{d~x:\langle \tau, \chi\rangle := exp}$, , where $\chi' \sqsubseteq \chi$ that is evaluated in two different initial configurations $\langle \mu_a, \epsilon_a \rangle$ and $\langle \mu_b, \epsilon_b \rangle$ satisfying $\semanticBelowPCState{l}{\Xi_a}{\Xi_b}{\Delta}{\mu_{a}}{\epsilon_{a}}{\mu_{b}}{\epsilon_{b}}{\Gamma}$ as follows:

	$\evalstocopy
    {\config[\overline{d~x:\langle \tau, \chi\rangle := exp}]{\mathcal{C};\Delta}{\mu_{a1}}{\epsilon_{a}}}
    {\configcopy{\mu_{a2}}{\overline{x \mapsto l_a}}{\overline{lval_a \mapsto l_a}}}$
    and

    $\evalstocopy
    {\config[\overline{d~x:\langle \tau, \chi\rangle := exp}]{\mathcal{C};\Delta}{\mu_{b1}}{\epsilon_{b}}}
    {\configcopy{\mu_{b2}}{\overline{x \mapsto l_b}}{\overline{lval_b \mapsto l_b}}}$

    Then:
\begin{enumerate}
	\item $\semanticBelowPCState{l}{\Xi_a'}{\Xi_b'}{\Delta}{\mu_{a2}}{\overline{x \mapsto l_a}}{\mu_{b2}}{\overline{x \mapsto l_b}}{\Gamma'}$, for some $\Xi_a'$, $\Xi_b'$, $\mu_a'$, $\mu_b'$ such that $\Xi_a \subseteq \Xi_a'$ and $\Xi_b \subseteq \Xi_b'$, $\dom{\mu_a} \subseteq \dom{\mu_a'}$, $\dom{\mu_b'} \subseteq \dom{\mu_b'}$ and $\Gamma' = \{\overline{x \mapsto \type{\tau}{\chi}}\}$,
	\item $\semanticBelowPCState{l}{\Xi_a'}{\Xi_b'}{\Delta}{\mu_{a2}}{\epsilon_a}{\mu_{b2}}{\epsilon_b}{\Gamma}$,
	\item $\textsc{lval_base}(lval_a) \in \dom{\epsilon_a}$, $\textsc{lval_base}(lval_b) \in \dom{\epsilon_b}$, and $\textsc{lval_base}(lval_a) = \textsc{lval_base}(lval_b)$ for each $lval_a$ and $lval_b$,
	\item $\overline{l_a} \in \dom{\mu_{a2}}$ and $\overline{l_b} \in \dom{\mu_{b2}}$
	\item $\overline{l_a}$ and $\overline{l_b}$ are fresh locations, $\overline{l_a} \notin \Xi_a$ and $\overline{l_b} \notin \Xi_b$,
	\item For any $l_a \in \dom{\mu_a}$ such that $\ordinaryTyping[]{\Xi_a}{\Delta}{\mu_a(l_a)}{\tau_{clos}}$, where $\tau_{clos} \in \{\tau_{fn}, \tau_{tbl}\}$, then $\mu_a'(l_a) = \mu_a(l_a)$. Similarly for any $l_b \in \dom{\mu_b}$ such that $\ordinaryTyping[]{\Xi_b}{\Delta}{\mu_b(l_b)}{\tau_{clos}}$, where $\tau_{clos} \in \{\tau_{fn}, \tau_{tbl}\}$, then $\mu_b'(l_b) = \mu_b(l_b)$,
	\item For any $l_a \in \dom{\mu_a}$ and $l_b \in \dom{\mu_b}$ such that $\ordinaryTyping[]{\Xi_a}{\Delta}{\mu_a(l_a)}{\type{\tau}{\chi}}$ and $\ordinaryTyping[]{\Xi_b}{\Delta}{\mu_b(l_b)}{\type{\tau}{\chi}}$ and $pc \nsqsubseteq \chi$, we have $\mu_{a}'(l_a) = \mu_{a}(l_a)$ and $\mu_{b}'(l_b) = \mu_{b}(l_b)$.
\end{enumerate}

\textbf{Note. } By \Cref{def:mem-store-pair-semantic}, $\dom{\epsilon_a} = \dom{\epsilon_b}$.
\end{lemma}

\section{Proof of Non-Interference} \label{app:proof-ni}
\paragraph*{Proof of \Cref{ni-exp}}
The proof is given by induction on the typing derivation of the expression and the cases are given by the last typing rule in the expression's typing derivation.
\begin{enumerate}
\item \textbf{T-Int}
If the typing derivation ends with the following last rule
\[
\inferrule*[right=T-Int]
{
}
{
\ordinaryTyping[pc]{\Gamma}{\Delta}{n_{w}}{\type{int}{\bot}~goes~ in}
}
\]
then we need to show that for any $\Xi_a$, $\Xi_b$, $\mu_{a}$, $\mu_{b}$, $\epsilon_{a}$, $\epsilon_{b}$, $\mu_a'$, $\mu_b'$ satisfying
\begin{equation}\label{int-initial-hypo}
\semanticBelowPCState{l}{\Xi_a}{\Xi_b}{\Delta}{\mu_a}{\epsilon_a}{\mu_b}{\epsilon_b}{\Gamma}
\end{equation}
if the expression $n_w$ is evaluated under two different initial configurations $\langle \mu_a, \epsilon_a \rangle$ and $\langle \mu_b, \epsilon_b \rangle$ as follows:
\[
\inferrule*[right=Eval 1]
{~}
{ \langle \mathcal{C}, \Delta; \mu_{a}; \epsilon_{a}; n_{w} \rangle \Downarrow \langle \mu_{a}, n_{w}\rangle}
\]
\[
\inferrule*[right=Eval 2]
{~}
{ \langle \mathcal{C}, \Delta; \mu_{b}; \epsilon_{b}; n_{w} \rangle \Downarrow \langle \mu_{b}, n_{w}\rangle}
\]
then there exists some $\Xi_a'$ and $\Xi_b'$, such that the following hold:
\begin{enumerate}
\item $\ordinaryTyping[pc]{\Gamma}{\Delta}{n_w}{\type{int}{\bot}}$. Already given in the hypothesis of this theorem,
\item \label{int-tp-1} $\Xi_a \subseteq \Xi_a'$, $\Xi_b \subseteq \Xi_b'$, $\dom{\mu_a} \subseteq \dom{\mu_a'}$, and $\dom{\mu_b} \subseteq \dom{\mu_b'}$ and $\semanticBelowPCState{l}{\Xi_{a}'}{\Xi_{b}'}{\Delta}{\mu_{a}'}{\epsilon_{a}}{\mu_{b}'}{\epsilon_{b}}{\Gamma}$. Here, $\mu_a' = \mu_{a}$ and $\mu_b' = \mu_{b}$,
\item \label{int-tp-2} $\NIval{l}{\Xi_a'}{\Xi_b'}{\Delta}{val_{a}}{val_b}{\type{int}{\bot}}$,
\item For any $l_a \in \dom{\mu_a}$ such that $\ordinaryTyping[]{\Xi_a}{\Delta}{\mu_a(l_a)}{\tau_{clos}}$, where $\tau_{clos} \in \{\tau_{fn}, \tau_{tbl}\}$, then $\mu_a'(l_a) = \mu_a(l_a)$. Similarly for any $l_b \in \dom{\mu_b}$ such that $\ordinaryTyping[]{\Xi_b}{\Delta}{\mu_b(l_b)}{\tau_{clos}}$, where $\tau_{clos} \in \{\tau_{fn}, \tau_{tbl}\}$, then $\mu_b'(l_b) = \mu_b(l_b)$. This is trivial, since memory store doesn't change.
\item For any $l_a \in \dom{\mu_a}$ and $l_b \in \dom{\mu_b}$ such that $\ordinaryTyping[]{\Xi_a}{\Delta}{\mu_a(l_a)}{\type{\tau}{\chi}}$ and $\ordinaryTyping[]{\Xi_b}{\Delta}{\mu_b(l_b)}{\type{\tau}{\chi}}$ and $pc \nsqsubseteq \chi$, we have $\mu_{a}'(l_a) = \mu_{a}(l_a)$ and $\mu_{b}'(l_b) = \mu_{b}(l_b)$. This is trivial, since memory store doesn't change.
\end{enumerate}
First we will prove \Cref{int-tp-1}. Let $\Xi_a' = \Xi_a$ and $\Xi_b' = \Xi_b$, now showing \Cref{int-tp-1} is same as showing
\begin{equation}\label{int-tp-1-red}
\semanticBelowPCState{l}{\Xi_{a}}{\Xi_{b}}{\Delta}{\mu_{a}'}{\epsilon_{a}}{\mu_{b}'}{\epsilon_{b}}{\Gamma}
\end{equation}
From the evaluation rule, we know $\mu_a' = \mu_a$ and $\mu_b' = \mu_b$. Therefore, showing \Cref{int-tp-1-red} is same as showing
\begin{equation*}
\semanticBelowPCState{l}{\Xi_{a}}{\Xi_{b}}{\Delta}{\mu_{a}}{\epsilon_{a}}{\mu_{b}}{\epsilon_{b}}{\Gamma}
\end{equation*}

This is what we had started out with in \Cref{int-initial-hypo}. Therefore we
have shown \Cref{int-tp-1}.

Next to show \Cref{int-tp-2}, we first expand the definition for  \NIforval and prove each of its requirement.
Since $\tau = \type{int}{\bot}$, using the syntactic typing, $\ordinaryTyping[pc]{\Gamma}{\Delta}{n_w}{\type{int}{\bot}~goes~ in}$ we can show that $\semanticTyping[pc]{\Xi_a'}{\Delta}{n_w}{ \type{int}{\bot}}$ and $\semanticTyping[pc]{\Xi_b'}{\Delta}{n_w}{ \type{int}{\bot}}$. Also, since both integers have equal value $val_a = n_w = val_b$, we have shown \NIforval.
\item \textbf{T-Bool} Similar to \textbf{E-Int}.
\item \textbf{T-Var}
If the typing derivation ends with the following last rule
\[      \inferrule*[right=T-Var]
{   x \in \dom{\Gamma} \qquad
\Gamma(x) = \type{\tau}{\chi}
}
{
\ordinaryTyping[pc]{\Gamma}{\Delta}{x}{\type{\tau}{\chi} \text{~goes inout}}
}
\]
then we need to show that for any $\Xi_a$, $\Xi_b$, $\mu_{a}$, $\mu_{b}$, $\epsilon_{a}$, $\epsilon_{b}$, $\mu_a'$, $\mu_b'$ satisfying
\begin{equation}\label{var-initial-hypo}
\semanticBelowPCState{l}{\Xi_a}{\Xi_b}{\Delta}{\mu_a}{\epsilon_a}{\mu_b}{\epsilon_b}{\Gamma},
\end{equation}
if the expression $x$ is evaluated under two different initial configurations $\langle \mu_a, \epsilon_a \rangle$ and $\langle \mu_b, \epsilon_b \rangle$ as follows:
\[
\inferrule*[right=Eval 1]
{
\epsilon_a(x) = l_{a} \qquad
\mu_{a}(l_{a}) = val_{a}
}
{ \langle \mathcal{C}, \Delta; \mu_{a}; \epsilon_{a}; x \rangle \Downarrow \langle \mu_{a}, val_{a}\rangle}
\]

\[
\inferrule*[right=Eval 2]
{
\epsilon_b(x) = l_{b} \qquad
\mu_{b}(l_{b}) = val_{b}
}
{ \langle \mathcal{C}, \Delta; \mu_{b}; \epsilon_{b}; x \rangle \Downarrow \langle \mu_{b}, val_{b}\rangle}
\]
then there exists some $\Xi_a'$ and $\Xi_b'$, such that the following hold:
\begin{enumerate}
\item $\ordinaryTyping[pc]{\Gamma}{\Delta}{x}{\type{\tau}{\chi}}$. Already given in the hypothesis of this theorem,
\item \label{var-tp-1}$\Xi_a \subseteq \Xi_a'$, $\Xi_b \subseteq \Xi_b'$, $\dom{\mu_a} \subseteq \dom{\mu_a'}$, and $\dom{\mu_b} \subseteq \dom{\mu_b'}$ and $\semanticBelowPCState{l}{\Xi_{a}'}{\Xi_{b}'}{\Delta}{\mu_{a}'}{\epsilon_{a}}{\mu_{b}'}{\epsilon_{b}}{\Gamma}$. Here, $\mu_a' = \mu_{a}$ and $\mu_b' = \mu_{b}$,
\item \label{var-tp-2} $\NIval{l}{\Xi_a'}{\Xi_b'}{\Delta}{val_{a}}{val_b}{\type{\tau}{\chi}}$,
\item For any $l_a \in \dom{\mu_a}$ such that $\ordinaryTyping[]{\Xi_a}{\Delta}{\mu_a(l_a)}{\tau_{clos}}$, where $\tau_{clos} \in \{\tau_{fn}, \tau_{tbl}\}$, then $\mu_a'(l_a) = \mu_a(l_a)$. Similarly for any $l_b \in \dom{\mu_b}$ such that $\ordinaryTyping[]{\Xi_b}{\Delta}{\mu_b(l_b)}{\tau_{clos}}$, where $\tau_{clos} \in \{\tau_{fn}, \tau_{tbl}\}$, then $\mu_b'(l_b) = \mu_b(l_b)$.
\item For any $l_a \in \dom{\mu_a}$ and $l_b \in \dom{\mu_b}$ such that $\ordinaryTyping[]{\Xi_a}{\Delta}{\mu_a(l_a)}{\type{\tau}{\chi}}$ and $\ordinaryTyping[]{\Xi_b}{\Delta}{\mu_b(l_b)}{\type{\tau}{\chi}}$ and $pc \nsqsubseteq \chi$, we have $\mu_{a}'(l_a) = \mu_{a}(l_a)$ and $\mu_{b}'(l_b) = \mu_{b}(l_b)$.
\end{enumerate}
First we will prove \Cref{var-tp-1}. Let $\Xi_a' = \Xi_a$ and $\Xi_b' = \Xi_b$, now showing \Cref{var-tp-1} is same as showing
\begin{equation}\label{var-tp-1-red}
\semanticBelowPCState{l}{\Xi_{a}}{\Xi_{b}}{\Delta}{\mu_{a}'}{\epsilon_{a}}{\mu_{b}'}{\epsilon_{b}}{\Gamma}
\end{equation}
From the evaluation rule, we know $\mu_a' = \mu_a$ and $\mu_b' = \mu_b$.
Therefore, showing \Cref{var-tp-1-red} is same as showing
\begin{equation}\label{int-tp-1-final}
\semanticBelowPCState{l}{\Xi_{a}}{\Xi_{b}}{\Delta}{\mu_{a}}{\epsilon_{a}}{\mu_{b}}{\epsilon_{b}}{\Gamma}
\end{equation}

This is what we had started out with in \Cref{var-initial-hypo}. Therefore we
have shown \Cref{int-tp-1}.


Next,  we use the following property from the definition of \Cref{int-tp-1-final}
\begin{equation}
\text{for any} ~x,~ \NIval{l}{\Xi_a}{\Xi_b}{\Delta}{\mu_a(\epsilon_{a}(x))}{\mu_b(\epsilon_{b}(x))}{\Gamma(x)}
\end{equation}
to conclude that $\NIval{l}{\Xi_a}{\Xi_b}{\Delta}{val_a}{val_b}{\type{\tau}{\chi}}$.

\item \textbf{T-SubType-In}
In case the last typing rule is the following and we need to prove that $\NIexp{pc}{\Gamma}{\Delta}{exp}{\type{\tau}{\chi'}}$.
\[
\inferrule*[right=T-SubType-In]
{
\ordinaryTyping[pc]{\Gamma}{\Delta}{exp}{\type{\tau}{\chi}~goes~ in} \\
\chi \sqsubseteq \chi'
}
{
\ordinaryTyping[pc]{\Gamma}{\Delta}{exp}{\type{\tau}{\chi'}~goes~ in}
}
\]

By applying the induction hypothesis of this theorem, we get $\NIexp{pc}{\Gamma}{\Delta}{exp}{\type{\tau}{\chi}}$. Now we need to show that if $\chi \sqsubseteq \chi'$, then $\NIexp{pc}{\Gamma}{\Delta}{exp}{\type{\tau}{\chi'}}$. To show NI of expression, we need to first show that the final memory stores are below-pc equivalent. This is already available from the expansion of $\NIexp{pc}{\Gamma}{\Delta}{exp}{\type{\tau}{\chi}}$. In addition, we need to show that the value that this expression evaluates to is still respecting non-interference of values with the security label $\chi'$ as defined in \Cref{def:def-ni-val}. To do this we use \Cref{lem:ni-subtyping}.
\item \textbf{T-BinOp}
If the typing derivation ends with the following last rule
\[  \inferrule*[right=T-BinOP]
{
\ordinaryTyping[pc]{\Gamma}{\Delta} {exp_{1}}{\type{\rho_{1}}{\chi_{1}}} \\
\ordinaryTyping[pc]{\Gamma}{\Delta}{exp_{2}}{\type{\rho_{2}}{\chi_{2}}} \\\\
\mathcal{T}(\Delta; \oplus; \rho_{1}; \rho_{2}) = \rho_{3} \\
\chi_{1} \sqsubseteq \chi' \\
\chi_{2} \sqsubseteq \chi'
}
{
\ordinaryTyping[pc]{\Gamma}{\Delta}{exp_{1} \oplus exp_{2}}{\type{\rho_3}{\chi'} ~goes~ in}
}
\]
then we need to show that for any $\Xi_a$, $\Xi_b$, $\mu_{a}$, $\mu_{b}$, $\epsilon_{a}$, $\epsilon_{b}$, $\mu_a'$, $\mu_b'$ satisfying
\begin{equation}\label{binop-initial-hypo}
\semanticBelowPCState{l}{\Xi_a}{\Xi_b}{\Delta}{\mu_a}{\epsilon_a}{\mu_b}{\epsilon_b}{\Gamma},
\end{equation}
if the expression $exp_{1} \oplus exp_{2}$ is evaluated under two different initial configurations $\langle \mu_a, \epsilon_a \rangle$ and $\langle \mu_b, \epsilon_b \rangle$ as follows:
\[
\inferrule*[right=Eval 1]
{
 \langle \mathcal{C}, \Delta; \mu_{a}; \epsilon_{a}; exp_{1} \rangle \Downarrow \langle \mu_{a1}, val_{a1} \rangle \\
 \langle \mathcal{C}, \Delta; \mu_{a1}; \epsilon_{a}; exp_{2} \rangle \Downarrow \langle \mu_{a2}, val_{a2} \rangle
}
{ \langle \mathcal{C}, \Delta; \mu_{a}; \epsilon_{a}; exp_{1} \oplus exp_{2} \rangle \Downarrow \langle \mu_{a2}, \mathbb{E} (\oplus, val_{a1}, val_{a2})\rangle}
\]

\[
\inferrule*[right=Eval 2]
{
 \langle \mathcal{C}, \Delta; \mu_{b}; \epsilon_{b}; exp_{1} \rangle \Downarrow \langle \mu_{b1}, val_{b1} \rangle \\
 \langle \mathcal{C}, \Delta; \mu_{b1}; \epsilon_{b}; exp_{2} \rangle \Downarrow \langle \mu_{b2}, val_{b2} \rangle
}
{ \langle \mathcal{C}, \Delta; \mu_{b}; \epsilon_{b}; exp_{1} \oplus exp_{2} \rangle \Downarrow \langle \mu_{b2}, \mathbb{E} (\oplus, val_{b1}, val_{b2})\rangle}
\]
then there exists some $\Xi_a'$ and $\Xi_b'$, such that the following hold:
\begin{enumerate}
\item $\ordinaryTyping[pc]{\Gamma}{\Delta}{exp_{1} \oplus exp_{2}}{\type{\rho_3}{\chi'}}$. Already given in the hypothesis of this theorem,
\item \label{binop-tp-1}$\Xi_a \subseteq \Xi_a'$, $\Xi_b \subseteq \Xi_b'$, $\dom{\mu_a} \subseteq \dom{\mu_a'}$, and $\dom{\mu_b} \subseteq \dom{\mu_b'}$ and $\semanticBelowPCState{l}{\Xi_{a}'}{\Xi_{b}'}{\Delta}{\mu_{a}'}{\epsilon_{a}}{\mu_{b}'}{\epsilon_{b}}{\Gamma}$. Here, $\mu_a' = \mu_{a}$ and $\mu_b' = \mu_{b}$,
\item \label{binop-tp-2} $\NIval{l}{\Xi_a'}{\Xi_b'}{\Delta}{\mathbb{E} (\oplus, val_{a1}, val_{a2})}{\mathbb{E} (\oplus, val_{b1}, val_{b2})}{\type{\rho_3}{\chi'}}$,
\item \label{binop-tp-3} For any $l_a \in \dom{\mu_a}$ such that $\ordinaryTyping[]{\Xi_a}{\Delta}{\mu_a(l_a)}{\tau_{clos}}$, where $\tau_{clos} \in \{\tau_{fn}, \tau_{tbl}\}$, then $\mu_a'(l_a) = \mu_a(l_a)$. Similarly for any $l_b \in \dom{\mu_b}$ such that $\ordinaryTyping[]{\Xi_b}{\Delta}{\mu_b(l_b)}{\tau_{clos}}$, where $\tau_{clos} \in \{\tau_{fn}, \tau_{tbl}\}$, then $\mu_b'(l_b) = \mu_b(l_b)$.
\item \label{binop-tp-4} For any $l_a \in \dom{\mu_a}$ and $l_b \in \dom{\mu_b}$ such that $\ordinaryTyping[]{\Xi_a}{\Delta}{\mu_a(l_a)}{\type{\tau}{\chi}}$ and $\ordinaryTyping[]{\Xi_b}{\Delta}{\mu_b(l_b)}{\type{\tau}{\chi}}$ and $pc \nsqsubseteq \chi$, we have $\mu_{a}'(l_a) = \mu_{a}(l_a)$ and $\mu_{b}'(l_b) = \mu_{b}(l_b)$.
\end{enumerate}
We repeatedly apply induction hypothesis on the typing derivation of $exp_1$ and $exp_2$ to get:
\begin{equation}\label{binop-ih-1}
\NIexp{pc}{\Gamma}{\Delta}{exp_1}{\type{\rho_1}{\chi_1}}
\end{equation}
and
\begin{equation}\label{binop-ih-2}
\NIexp{pc}{\Gamma}{\Delta}{exp_2}{\type{\rho_2}{\chi_2}}
\end{equation}

Expanding \Cref{binop-ih-1} we get:
\begin{equation}\label{binop-ih-1-1}
\semanticBelowPCState{l}{\Xi_{a1}'}{\Xi_{b1}'}{\Delta}{\mu_{a1}}{\epsilon_{a}}{\mu_{b1}}{\epsilon_{b}}{\Gamma}~,
\end{equation}
where $\Xi_{a} \subseteq \Xi_{a1}'$ and $\Xi_{b} \subseteq \Xi_{b1}'$
\begin{equation}\label{binop-ih-1-2}
\NIval{l}{\Xi_{a1}'}{\Xi_{b1}'}{\Delta}{val_{a1}}{val_{b1}}{\type{\rho_1}{\chi_{1}}}
\end{equation}

Similarly expanding \Cref{binop-ih-2} we get:
\begin{equation}\label{binop-ih-2-1}
\semanticBelowPCState{l}{\Xi_{a2}'}{\Xi_{b2}'}{\Delta}{\mu_{a2}}{\epsilon_{a}}{\mu_{b2}}{\epsilon_{b}}{\Gamma}~,
\end{equation}
where $\Xi_{a1}' \subseteq \Xi_{a2}'$ and $\Xi_{b1}' \subseteq \Xi_{b2}'$
\begin{equation}\label{binop-ih-2-2}
\NIval{l}{\Xi_{a2}'}{\Xi_{b2}'}{\Delta}{val_{a2}}{val_{b2}}{\type{\rho_2}{\chi_{2}}}
\end{equation}

Using \Cref{binop-ih-1-1} and \Cref{binop-ih-2-1}, we conclude
$\semanticBelowPCState{l}{\Xi_{a}'}{\Xi_{b}'}{\Delta}{\mu_{a2}}{\epsilon_{a}}{\mu_{b2}}{\epsilon_{b}}{\Gamma}$, where $\Xi_{a}' = \Xi_{a2}'$ and $\Xi_{b}' = \Xi_{b2}'$.
This proves \Cref{binop-tp-1}.

We assume the following about $\mathbb{E}$:
\begin{equation}
if x_1 = x_2~ and~ y_1 = y_2,~ then~\mathbb{E}(\oplus, x_1, y_1) = \mathbb{E}(\oplus, x_2, y_2)
\end{equation}
Thus, if the parameters to the evaluation function $\mathbb{E}$ are non-interfering,
$\NIval{l}{\Xi_a}{\Xi_b}{\Delta}{val_{a1}}{val_{b1}}{\type{\rho_1}{\chi_1}}$ and $\NIval{l}{\Xi_a}{\Xi_b}{\Delta}{val_{a2}}{val_{b2}}{\type{\rho_2}{\chi_2}}$, then the resultant value will also be non-interfering \[\NIval{l}{\Xi_a}{\Xi_b}{\Delta}{\mathbb{E}(\oplus, val_{a1}, val_{a2})}{\mathbb{E}(\oplus, val_{b1}, val_{b2})}{\type{\rho_3}{\chi'}},\] where $\chi_1 \sqsubseteq \chi'$  and $\chi_2 \sqsubseteq \chi'$ and $\rho_3 = \mathcal{T}(\oplus, \rho_1, \rho_2)$.

We consider only binary operations returning integers, bit vectors and booleans.


Using \Cref{binop-ih-1-2}, $\Xi_{a1}' \subseteq \Xi_{a2}'$ and the \Cref{lem-ni-val-subtype} we have:
\begin{equation}\label{binop-ih-1-2-ext}
\NIval{l}{\Xi_{a2}'}{\Xi_{b2}'}{\Delta}{val_{a1}}{val_{b1}}{\type{\rho_1}{\chi_{1}}}
\end{equation}
Using the above equation with \Cref{binop-ih-2-2} and the above assumption about the $\mathbb{E}$ function, we conclude:
\[\NIval{l}{\Xi_{a2}'}{\Xi_{b2}'}{\Delta}{\mathbb{E} (\oplus, val_{a1}, val_{a2})}{\mathbb{E} (\oplus, val_{b1}, val_{b2})}{\type{\rho_3}{\chi'}}\]

Now, we prove \Cref{binop-tp-3}. We know from \Cref{binop-ih-1}  that  for any $l_a \in \dom{\mu_a}$ such that $\ordinaryTyping[]{\Xi_a}{\Delta}{\mu_a(l_a)}{\tau_{clos}}$, where $\tau_{clos} \in \{\tau_{fn}, \tau_{tbl}\}$, then $\mu_{a1}(l_a) = \mu_a(l_a)$. Similarly for any $l_b \in \dom{\mu_b}$ such that $\ordinaryTyping[]{\Xi_b}{\Delta}{\mu_b(l_b)}{\tau_{clos}}$, where $\tau_{clos} \in \{\tau_{fn}, \tau_{tbl}\}$, then $\mu_{b1}(l_b) = \mu_b(l_b)$.
\Cref{binop-ih-2} also implies that for any $l_a \in \dom{\mu_{a1}}$ such that $\ordinaryTyping[]{\Xi_{a2}}{\Delta}{\mu_{a1}(l_a)}{\tau_{clos}}$, where $\tau_{clos} \in \{\tau_{fn}, \tau_{tbl}\}$, then $\mu_{a2}(l_a) = \mu_{a1}(l_a)$. Similarly for any $l_b \in \dom{\mu_{b1}}$ such that $\ordinaryTyping[]{\Xi_{b2}}{\Delta}{\mu_{b1}(l_b)}{\tau_{clos}}$, where $\tau_{clos} \in \{\tau_{fn}, \tau_{tbl}\}$, then $\mu_{b2}(l_b) = \mu_{b1}(l_b)$.
We also know that $\Xi_a \subseteq \Xi_{a1}$, this implies that $l_a \in \dom{\mu_a}$ will also be present in $\dom{\mu_{a1}}$. Therefore, we can show \Cref{binop-tp-3}. \Cref{binop-tp-4} can be similarly shown.

\item \textbf{T-Rec}
If the typing derivation ends with the following last rule
\[	\inferrule*[right=T-Rec]
{
\listOrdinaryTyping[pc]{\Gamma}{\Delta}{\{\overline{exp: \type{\tau_i}{\chi_i}}\}}
}
{
\ordinaryTyping[pc]{\Gamma}{\Delta}{\{ \overline{f: exp} \}}{\type{\{ \overline{f: \langle \tau_i, \chi_i \rangle} \}}{\bot}~goes~ in}
}\]
then we need to show that for any $\Xi_a$, $\Xi_b$, $\mu_{a}$, $\mu_{b}$, $\epsilon_{a}$, $\epsilon_{b}$, $\mu_a'$, $\mu_b'$ satisfying
\begin{equation}\label{rec-initial-hypo}
\semanticBelowPCState{l}{\Xi_a}{\Xi_b}{\Delta}{\mu_a}{\epsilon_a}{\mu_b}{\epsilon_b}{\Gamma},
\end{equation}
if the expression $\{ \overline{f: exp} \}$ is evaluated under two different initial configurations $\langle \mu_a, \epsilon_a \rangle$ and $\langle \mu_b, \epsilon_b \rangle$ as follows:
\[
\inferrule*[right=Eval 1]
{
 \langle \mathcal{C}, \Delta; \mu_{a}; \epsilon_{a}; \overline{exp} \rangle \Downarrow \langle \mu_{a}', \overline{val_{a}}\rangle
}
{ \langle \mathcal{C}, \Delta; \mu_{a}; \epsilon_{a}; \{\overline{f = exp}\} \rangle \Downarrow \langle \mu_{a}', \{\overline{f = val_{a}}\}\rangle}
\]
\[
\inferrule*[right=Eval 2]
{
 \langle \mathcal{C}, \Delta; \mu_{b}; \epsilon_{b}; \overline{exp} \rangle \Downarrow \langle \mu_{b}', \overline{val_{b}}\rangle
}
{ \langle \mathcal{C}, \Delta; \mu_{b}; \epsilon_{b}; \{\overline{f = exp}\} \rangle \Downarrow \langle \mu_{b}', \{\overline{f = val_{b}}\}\rangle}
\]
then there exists some $\Xi_a'$ and $\Xi_b'$, such that the following hold:
\begin{enumerate}
\item $\ordinaryTyping[pc]{\Gamma}{\Delta}{\{\overline{f = exp}\}}{\type{\{ \overline{f: \langle \tau_i, \chi_i \rangle} \}}{\bot}}$. Already given in the hypothesis of this theorem,
\item \label{rec-tp-1}$\Xi_a \subseteq \Xi_a'$, $\Xi_b \subseteq \Xi_b'$, $\dom{\mu_a} \subseteq \dom{\mu_a'}$, and $\dom{\mu_b} \subseteq \dom{\mu_b'}$ and $\semanticBelowPCState{l}{\Xi_{a}'}{\Xi_{b}'}{\Delta}{\mu_{a}'}{\epsilon_{a}}{\mu_{b}'}{\epsilon_{b}}{\Gamma}$,
\item \label{rec-tp-2} $\NIval{l}{\Xi_a'}{\Xi_b'}{\Delta}{\{\overline{f = val_{a}}\}}{\{\overline{f = val_{b}}\}}{\type{\{ \overline{f: \langle \tau_i, \chi_i \rangle} \}}{\bot}}$,
\item \label{rec-tp-3} For any $l_a \in \dom{\mu_a}$ such that $\ordinaryTyping[]{\Xi_a}{\Delta}{\mu_a(l_a)}{\tau_{clos}}$, where $\tau_{clos} \in \{\tau_{fn}, \tau_{tbl}\}$, then $\mu_a'(l_a) = \mu_a(l_a)$. Similarly for any $l_b \in \dom{\mu_b}$ such that $\ordinaryTyping[]{\Xi_b}{\Delta}{\mu_b(l_b)}{\tau_{clos}}$, where $\tau_{clos} \in \{\tau_{fn}, \tau_{tbl}\}$, then $\mu_b'(l_b) = \mu_b(l_b)$.
\item \label{rec-tp-4} For any $l_a \in \dom{\mu_a}$ and $l_b \in \dom{\mu_b}$ such that $\ordinaryTyping[]{\Xi_a}{\Delta}{\mu_a(l_a)}{\type{\tau}{\chi}}$ and $\ordinaryTyping[]{\Xi_b}{\Delta}{\mu_b(l_b)}{\type{\tau}{\chi}}$ and $pc \nsqsubseteq \chi$, we have $\mu_{a}'(l_a) = \mu_{a}(l_a)$ and $\mu_{b}'(l_b) = \mu_{b}(l_b)$,
\end{enumerate}

We repeatedly apply induction hypothesis on each $\ordinaryTyping[pc]{\Gamma}{\Delta}{exp}{\type{\tau_i}{\chi_i}}$ in the sequence $\seqordinaryTyping[pc]{\Gamma}{\Delta}{exp}{\type{\tau}{\chi}}$. The last memory store we arrive at is given by $\mu_a'$ and $\mu_b'$ in the two evaluations. Therefore, after repeated application of induction hypothesis we get,
\begin{equation}\label{rec-final}
\seqNIexp{pc}{\Gamma}{\Delta}{exp}{\type{\tau_i}{\chi_i}}
\end{equation}
Since we evaluate $\overline{exp}$ in initial configurations satisfying \Cref{rec-initial-hypo}, this can be expanded to conclude that there exists some $\Xi_{a}'$ and $\Xi_{b}'$ satisfying  $\Xi_{a} \subseteq \Xi_{a}'$,	$\Xi_{b} \subseteq \Xi_{b}'$, $\dom{\mu_{a}} \subseteq \dom{\mu_{a}'}$, $\dom{\mu_{b}} \subseteq \dom{\mu_{b}'}$  and all of the following:
\begin{equation}\label{rec-ih-1}
\semanticBelowPCState{l}{\Xi_{a}'}{\Xi_{b}'}{\Delta}{\mu_{a}'}{\epsilon_{a}}{\mu_{b}'}{\epsilon_{b}}{\Gamma}~,
\end{equation}
\begin{equation}\label{rec-ih-2}
\NIval{l}{\Xi_{a}'}{\Xi_{b}'}{\Delta}{\overline{ val_{a}}}{\overline{val_{b}}}{\overline{\langle \tau_i, \chi_i \rangle}}
\end{equation}
This is to be interpreted as a sequence of non-interfering values.
\Cref{rec-ih-1} proves the goal in \Cref{rec-tp-1}.

\Cref{rec-ih-2} can be interpreted as satisfying $\NIval{l}{\Xi_{a}'}{\Xi_{b}'}{\Delta}{val_{a}}{val_{b}}{\type{\tau_i}{\chi_i}}$ for each $val_{a}$, $val_{b}$.

We use the TV-rec rule with \Cref{rec-ih-2} to conclude that
$\semanticTyping[]{\Xi_{a}'}{\Delta}{\{\overline{f = val_{a}}\}}{\type{\{ \overline{f: \langle \tau_i, \chi_i \rangle} \}}{\bot}}$
and
$\semanticTyping[]{\Xi_{b}'}{\Delta}{\{\overline{f = val_{b}}\}}{\type{\{ \overline{f: \langle \tau_i, \chi_i \rangle} \}}{\bot}}$. Therefore, we have shown \Cref{rec-tp-2}.
\Cref{rec-tp-3}  and \Cref{rec-tp-4} follows from \Cref{rec-final}.

\item \textbf{T-MemRec} \label{recmem-exp}
If the typing derivation ends with the following last rule
\[	  \inferrule*[right=T-MemRec]
{
\ordinaryTyping[pc]{\Gamma}{\Delta}{exp}{\type{\{ \overline{f_i: \langle \tau_i, \chi_i \rangle} \}}{\bot}}~goes~d
}
{
\ordinaryTyping[pc]{\Gamma}{\Delta}{exp.f_{i}}{\type{\tau_{i}}{\chi_{i}}~ goes~ d}
}\]
then we need to show that for any $\Xi_a$, $\Xi_b$, $\mu_{a}$, $\mu_{b}$, $\epsilon_{a}$, $\epsilon_{b}$, $\mu_a'$, $\mu_b'$ satisfying
\begin{equation}\label{recmem-initial-hypo}
\semanticBelowPCState{l}{\Xi_a}{\Xi_b}{\Delta}{\mu_a}{\epsilon_a}{\mu_b}{\epsilon_b}{\Gamma},
\end{equation}
if the expression $exp.f_{i}$ is evaluated under two different initial configurations $\langle \mu_a, \epsilon_a \rangle$ and $\langle \mu_b, \epsilon_b \rangle$ as follows:
\[
\inferrule*[right=Eval 1]
{
 \langle \mathcal{C}, \Delta; \mu_{a}; \epsilon_{a}; exp \rangle \Downarrow \langle \mu_{a}', \{\overline{f_i: \type{\tau}{\chi} = val_{ai}}\}\rangle
}
{ \langle \mathcal{C}, \Delta; \mu_{a}; \epsilon_{a}; exp.f_{i} \rangle \Downarrow \langle \mu_{a}', val_{ai}\rangle}
\]

\[
\inferrule*[right=Eval 2]
{
 \langle \mathcal{C}, \Delta; \mu_{b}; \epsilon_{b}; exp \rangle \Downarrow \langle \mu_{b}', \{\overline{f_i: \type{\tau}{\chi} = val_{bi}}\}\rangle
}
{ \langle \mathcal{C}, \Delta; \mu_{b}; \epsilon_{b}; exp.f_{i} \rangle \Downarrow \langle \mu_{b}', val_{bi}\rangle}
\]
then there exists some $\Xi_a'$ and $\Xi_b'$, such that the following hold:
\begin{enumerate}
\item $\ordinaryTyping[pc]{\Gamma}{\Delta}{exp.f_{i}}{\type{\tau_{i}}{\chi_{i}}}$. Already given in the hypothesis of this theorem,
\item \label{recmem-tp-1}$\Xi_a \subseteq \Xi_a'$, $\Xi_b \subseteq \Xi_b'$, $\dom{\mu_a} \subseteq \dom{\mu_a'}$, and $\dom{\mu_b} \subseteq \dom{\mu_b'}$ and $\semanticBelowPCState{l}{\Xi_{a}'}{\Xi_{b}'}{\Delta}{\mu_{a}'}{\epsilon_{a}}{\mu_{b}'}{\epsilon_{b}}{\Gamma}$.
\item \label{recmem-tp-2} $\NIval{l}{\Xi_a'}{\Xi_b'}{\Delta}{val_{ai}}{val_{bi}}{\type{\tau_{i}}{\chi_{i}}}$,
\item \label{recmem-tp-3} For any $l_a \in \dom{\mu_a}$ such that $\ordinaryTyping[]{\Xi_a}{\Delta}{\mu_a(l_a)}{\tau_{clos}}$, where $\tau_{clos} \in \{\tau_{fn}, \tau_{tbl}\}$, then $\mu_a'(l_a) = \mu_a(l_a)$. Similarly for any $l_b \in \dom{\mu_b}$ such that $\ordinaryTyping[]{\Xi_b}{\Delta}{\mu_b(l_b)}{\tau_{clos}}$, where $\tau_{clos} \in \{\tau_{fn}, \tau_{tbl}\}$, then $\mu_b'(l_b) = \mu_b(l_b)$.
\item \label{recmem-tp-4} For any $l_a \in \dom{\mu_a}$ and $l_b \in \dom{\mu_b}$ such that $\ordinaryTyping[]{\Xi_a}{\Delta}{\mu_a(l_a)}{\type{\tau}{\chi}}$ and $\ordinaryTyping[]{\Xi_b}{\Delta}{\mu_b(l_b)}{\type{\tau}{\chi}}$ and $pc \nsqsubseteq \chi$, we have $\mu_{a}'(l_a) = \mu_{a}(l_a)$ and $\mu_{b}'(l_b) = \mu_{b}(l_b)$,
\end{enumerate}

By applying induction hypothesis on the typing derivation of $exp$, which is evaluated in an initial configuration satisfying \Cref{recmem-initial-hypo}, we get:
\[
\NIexp{pc}{\Gamma}{\Delta}{exp}{\type{\{ \overline{f_i: \langle \tau_i, \chi_i \rangle} \}}{\bot}}
\]
This implies that there exists a $\Xi_{a}'$, $\Xi_{b}'$, such that $\Xi_a \subseteq \Xi_a'$, $\Xi_b \subseteq \Xi_b'$, $\dom{\mu_a} \subseteq \dom{\mu_a'}$, $\dom{\mu_b} \subseteq \dom{\mu_b'}$ and the following:
\begin{equation}\label{memrec-ih-1}
\semanticBelowPCState{l}{\Xi_{a}'}{\Xi_{b}'}{\Delta}{\mu_{a}'}{\epsilon_{a}}{\mu_{b}'}{\epsilon_{b}}{\Gamma}~,
\end{equation}
This proves \Cref{recmem-tp-1}.
\begin{equation}\label{memrec-ih-2}
\NIval{l}{\Xi_{a}'}{\Xi_{b}'}{\Delta}{\{\overline{f_i: \type{\tau}{\chi} = val_{ai}}\}}{\{\overline{f_i: \type{\tau}{\chi} = val_{bi}}\}}{\type{\{ \overline{f_i: \langle \tau_i, \chi_i \rangle} \}}{\bot}}
\end{equation}
Using the \Cref{def:def-ni-val}, we can observe that for each $val_{ai}$ and $val_{bi}$ the following holds:
\begin{equation}\label{menmrec-ih-2-1}
\NIval{l}{\Xi_{a}'}{\Xi_{b}'}{\Delta}{val_{ai}}{val_{bi}}{\type{\tau_i}{\chi_i}}
\end{equation}
This proves \Cref{recmem-tp-2}.
\Cref{recmem-tp-3} and \Cref{recmem-tp-4} is also a conclusion of applying the induction hypothesis.

\item \textbf{T-Index}
If the typing derivation ends with the following last rule
\[\inferrule*[right=T-Index]
{
\ordinaryTyping[pc]{\Gamma}{\Delta}{exp_{1}}{\type{\type{\tau}{\chi_1}[n]}{\bot}~goes~ d} \\\\
\ordinaryTyping[pc]{\Gamma}{\Delta}{exp_{2}}{\type{bit \langle 32 \rangle}{\chi_2}} \\\\
\chi_2 \sqsubseteq \chi_1
}
{
\ordinaryTyping[pc]{\Gamma}{\Delta} {exp_{1}[exp_{2}]}{\type{\tau}{\chi_1}~ goes~ d }
}\]
then we need to show that for any $\Xi_a$, $\Xi_b$, $\mu_{a}$, $\mu_{b}$, $\epsilon_{a}$, $\epsilon_{b}$, $\mu_a'$, $\mu_b'$ satisfying
\begin{equation}\label{index-initial-hypo}
\semanticBelowPCState{l}{\Xi_a}{\Xi_b}{\Delta}{\mu_a}{\epsilon_a}{\mu_b}{\epsilon_b}{\Gamma},
\end{equation}
if the expression $exp_{1}[exp_{2}]$ is evaluated under two different initial configurations $\langle \mu_a, \epsilon_a \rangle$ and $\langle \mu_b, \epsilon_b \rangle$ as follows.
Observe that if $exp_{2}$ evaluates to a value within the array bounds the following rule will be used; otherwise $\textsc{Eval 1 error}$.
\[
\inferrule*[right=Eval 1]
{
 \langle \mathcal{C}, \Delta; \mu_{a}; \epsilon_{a}; exp_{1} \rangle \Downarrow \langle \mu_{a1}, stack~\tau \{\overline{val_a}\}\rangle \\
 \langle \mathcal{C}, \Delta; \mu_{a1}; \epsilon_{a}; exp_{2} \rangle \Downarrow \langle \mu_{a2}, n_{a}\rangle  \\
0 \leq n_{a} < len(\overline{val_{a}})
}
{ \langle \mathcal{C}, \Delta; \mu_{a}; \epsilon_{a};exp_{1}[exp_{2}] \rangle \Downarrow \langle \mu_{a2}, val_{a~n_{a}}\rangle}
\]

\[
\inferrule*[right=Eval 1 error]
{
 \langle \mathcal{C}, \Delta; \mu_{a}; \epsilon_{a}; exp_{1} \rangle \Downarrow \langle \mu_{a1}, stack~\tau \{\overline{val_a}\}\rangle \\
 \langle \mathcal{C}, \Delta; \mu_{a1}; \epsilon_{a}; exp_{2} \rangle \Downarrow \langle \mu_{a2}, n_{a}\rangle  \\
n_{a} \geq len(\overline{val_a})
}
{ \langle \mathcal{C}, \Delta; \mu_{a}; \epsilon_{a};exp_{1}[exp_{2}] \rangle \Downarrow \langle \mu_{a2}, havoc(\tau)\rangle}
\]

If $exp_{2}$ evaluates to a value within the array bounds the following rule will be used; otherwise $Eval_{2error}$
\[
\inferrule*[right=Eval 2]
{
 \langle \mathcal{C}, \Delta; \mu_{b}; \epsilon_{b}; exp_{1} \rangle \Downarrow \langle \mu_{b1}, stack~\tau \{\overline{val_b}\}\rangle \\
 \langle \mathcal{C}, \Delta; \mu_{b1}; \epsilon_{b}; exp_{2} \rangle \Downarrow \langle \mu_{b2}, n_{b}\rangle  \\
0 \leq n_{b} < len(\overline{val_b})
}
{ \langle \mathcal{C}, \Delta; \mu_{b}; \epsilon_{b};exp_{1}[exp_{2}] \rangle \Downarrow \langle \mu_{b2}, val_{b~n_{b}}\rangle}
\]
\[
\inferrule*[right=Eval 2 error]
{
 \langle \mathcal{C}, \Delta; \mu_{b}; \epsilon_{b}; exp_{1} \rangle \Downarrow \langle \mu_{b1}, stack~\tau \{\overline{val_b}\}\rangle \\
 \langle \mathcal{C}, \Delta; \mu_{b1}; \epsilon_{b}; exp_{2} \rangle \Downarrow \langle \mu_{b2}, n_{b}\rangle  \\
n_{b} \geq len(\overline{val_b})
}
{ \langle \mathcal{C}, \Delta; \mu_{b}; \epsilon_{b};exp_{1}[exp_{2}] \rangle \Downarrow \langle \mu_{b2}, havoc(\tau)\rangle}
\]
then there exists some $\Xi_a'$ and $\Xi_b'$, such that the following hold:
\begin{enumerate}
\item $\ordinaryTyping[pc]{\Gamma}{\Delta}{exp_{1}[exp_{2}]}{\type{\tau}{\chi_1}}$. Already given in the hypothesis of this theorem,
\item \label{index-tp-1}$\Xi_a \subseteq \Xi_a'$, $\Xi_b \subseteq \Xi_b'$, $\dom{\mu_a} \subseteq \dom{\mu_a'}$, and $\dom{\mu_b} \subseteq \dom{\mu_b'}$ and $\semanticBelowPCState{l}{\Xi_{a}'}{\Xi_{b}'}{\Delta}{\mu_{a}'}{\epsilon_{a}}{\mu_{b}'}{\epsilon_{b}}{\Gamma}$. Here, $\mu_a' = \mu_{a2}$ and $\mu_b' = \mu_{b2}$,
\item \label{index-tp-2} $\NIval{l}{\Xi_a'}{\Xi_b'}{\Delta}{val_{a}'}{val_{b}'}{\type{\tau}{\chi_1}}$, where $val_a' \in \{val_{a~n_a}, havoc(\tau)\}$ and $val_b' \in \{val_{b~n_b}, havoc(\tau)\}$,
\item \label{index-tp-3} For any $l_a \in \dom{\mu_a}$ such that $\ordinaryTyping[]{\Xi_a}{\Delta}{\mu_a(l_a)}{\tau_{clos}}$, where $\tau_{clos} \in \{\tau_{fn}, \tau_{tbl}\}$, then $\mu_a'(l_a) = \mu_a(l_a)$. Similarly for any $l_b \in \dom{\mu_b}$ such that $\ordinaryTyping[]{\Xi_b}{\Delta}{\mu_b(l_b)}{\tau_{clos}}$, where $\tau_{clos} \in \{\tau_{fn}, \tau_{tbl}\}$, then $\mu_b'(l_b) = \mu_b(l_b)$.
\item \label{index-tp-4} For any $l_a \in \dom{\mu_a}$ and $l_b \in \dom{\mu_b}$ such that $\ordinaryTyping[]{\Xi_a}{\Delta}{\mu_a(l_a)}{\type{\tau}{\chi}}$ and $\ordinaryTyping[]{\Xi_b}{\Delta}{\mu_b(l_b)}{\type{\tau}{\chi}}$ and $pc \nsqsubseteq \chi$, we have $\mu_{a}'(l_a) = \mu_{a}(l_a)$ and $\mu_{b}'(l_b) = \mu_{b}(l_b)$,
\end{enumerate}

By applying induction hypothesis on the typing derivation of $exp_1$ that is evaluated in configuration satisfying \Cref{index-initial-hypo}, we conclude that there exist some $\Xi_{a1}'$ and $\Xi_{b1}'$ satisfying $\Xi_{a} \subseteq \Xi_{a1}'$,	$\Xi_{b} \subseteq \Xi_{b1}'$ and all of the following:
\begin{equation}\label{index-h-1-1}
\semanticBelowPCState{l}{\Xi_{a1}'}{\Xi_{b1}'}{\Delta}{\mu_{a1}}{\epsilon_{a}}{\mu_{b1}}{\epsilon_{b}}{\Gamma}~,
\end{equation}
\begin{equation}\label{index-h-1-2}
\NIval{l}{\Xi_{a1}'}{\Xi_{b1}'}{\Delta}{stack~\tau\{\overline{val_{a}}\}}{stack~\tau\{\overline{val_{b}}\}}{\type{\type{\tau}{\chi_1}[n]}{\bot}}
\end{equation}
Using the \Cref{def:def-ni-val}, we can observe that for each $val_{a}$ and $val_{b}$ the following holds:
\begin{equation}\label{index-ih-2-1}
\NIval{l}{\Xi_{a1}'}{\Xi_{b1}'}{\Delta}{val_{a}}{val_{b}}{\type{\tau}{\chi_1}}
\end{equation}

By applying induction hypothesis on the typing derivation of $exp_2$ that is evaluated in configuration satisfying \Cref{index-h-1-1}, we conclude that there exist some $\Xi_{a2}'$ and $\Xi_{b2}'$ satisfying $\Xi_{a1}' \subseteq \Xi_{a2}'$,	$\Xi_{b1}' \subseteq \Xi_{b2}'$ and all of the following:
\begin{equation}\label{index-h-2-1}
\semanticBelowPCState{l}{\Xi_{a2}'}{\Xi_{b2}'}{\Delta}{\mu_{a2}}{\epsilon_{a}}{\mu_{b2}}{\epsilon_{b}}{\Gamma}~,
\end{equation}
\begin{equation}\label{index-h-2-2}
\NIval{l}{\Xi_{a2}'}{\Xi_{b2}'}{\Delta}{n_{a}}{n_{b}}{\type{bit \langle 32 \rangle}{\chi_2}}
\end{equation}
\Cref{index-h-2-1} proves the requirement of \Cref{index-tp-1}.
To prove \Cref{index-tp-2} we consider the following cases for the final values $val_a'$ and $val_b'$:
\begin{itemize}
\item Index within bound. In this case both the evaluations use the same evaluation rules.

If $\chi_1 \sqsubseteq l$, then $\chi_2 \sqsubseteq l$, which implies that $n_a = n_b = n_{32}$. We can observe that in this case we will have $val_a' = val_{a~n_{32}}$ and $val_b' = val_{b~n_{32}}$. Using \Cref{index-ih-2-1}, we conclude that $\NIval{l}{\Xi_{a1}'}{\Xi_{b1}'}{\Delta}{val_{a}'}{val_{b}'}{\type{\tau}{\chi_1}}$.  By applying \Cref{lem-ni-val-subtype}, we will get $\NIval{l}{\Xi_{a2}'}{\Xi_{b2}'}{\Delta}{val_{a}'}{val_{b}'}{\type{\tau}{\chi_1}}$. We have shown \Cref{index-tp-2}.

If $\chi_1 \nsqsubseteq l$, according to the \Cref{def:def-ni-val},  $\NIval{l}{\Xi_{a1}'}{\Xi_{b1}'}{\Delta}{val_{a}'}{val_{b}'}{\type{\tau}{\chi_1}}$ will hold even if $val_a' \ne val_b'$. Therefore, even if $n_a \ne n_b$, $val_a' = val_{a~n_{a}}$ and $val_b' = val_{b~n_{b}}$, we will have $\NIval{l}{\Xi_{a1}'}{\Xi_{b1}'}{\Delta}{val_{a}'}{val_{b}'}{\type{\tau}{\chi_1}}$.
\item One index is out-of-bound. In this case one of the evaluation will yield the $havoc(\tau)$ and $n_a$ and $n_b$ should have differed. This implies $\chi_2 \nsqsubseteq l$, which implies $\chi_1 \nsqsubseteq l$. As described in the previous case, $\NIval{l}{\Xi_{a1}'}{\Xi_{b1}'}{\Delta}{val_{a~n_{a}}}{havoc(\tau)}{\type{\tau}{\chi_1}}$ is true according to the \Cref{def:def-ni-val}.
\item Both indices are out-of-bound. In this case the values will be of the form $val_a' = havoc(\tau) = val_b'$. According to the \Cref{def:def-ni-val}, $\NIval{l}{\Xi_{a1}'}{\Xi_{b1}'}{\Delta}{havoc(\tau)}{havoc(\tau)}{\type{\tau}{\chi_1}}$ is satisfied.
\end{itemize}
\item \textbf{T-HdrMem}
If the typing derivation ends with the following last rule
\[ \inferrule*[right=T-MemHdr]
{
\ordinaryTyping[pc]{\Gamma}{\Delta}{exp} {\type{header \{ \overline{f_i: \langle \tau_i, \chi_i \rangle} \}}{\bot}~ goes~ d}
}
{
\ordinaryTyping[pc]{\Gamma}{\Delta}{exp.f_{i}}{\type{\tau_{i}}{\chi_{i}}~goes~ d}
}
\]
then we need to show that for any $\Xi_a$, $\Xi_b$, $\mu_{a}$, $\mu_{b}$, $\epsilon_{a}$, $\epsilon_{b}$, $\mu_a'$, $\mu_b'$ satisfying
\begin{equation}\label{hdr-initial-hypo}
\semanticBelowPCState{l}{\Xi_a}{\Xi_b}{\Delta}{\mu_a}{\epsilon_a}{\mu_b}{\epsilon_b}{\Gamma},
\end{equation}
if the expression $exp.f_{i}$ is evaluated under two different initial configurations $\langle \mu_a, \epsilon_a \rangle$ and $\langle \mu_b, \epsilon_b \rangle$ as follows:
\[
\inferrule*[right=Eval 1]
{
 \langle \mathcal{C}, \Delta; \mu_{a}; \epsilon_{a}; exp \rangle \Downarrow \langle \mu_{a}', header\{valid, \overline{f: \tau = val_{a}}\}\rangle
}
{ \langle \mathcal{C}, \Delta; \mu_{a}; \epsilon_{a}; exp.f_{i} \rangle \Downarrow \langle \mu_{a}', val_{ai}\rangle}
\]

\[
\inferrule*[right=Eval 2]
{
 \langle \mathcal{C}, \Delta; \mu_{b}; \epsilon_{b}; exp \rangle \Downarrow \langle \mu_{b}', header\{valid, \overline{f: \tau = val_{b}}\}\rangle
}
{ \langle \mathcal{C}, \Delta; \mu_{b}; \epsilon_{b}; exp.f_{i} \rangle \Downarrow \langle \mu_{b}', val_{bi}\rangle}
\]
then there exists some $\Xi_a'$ and $\Xi_b'$, such that the following hold:
\begin{enumerate}
\item $\ordinaryTyping[pc]{\Gamma}{\Delta}{exp.f_{i}}{\type{\tau_{i}}{\chi_{i}}}$. Already given in the hypothesis of this theorem.
\item $\Xi_a \subseteq \Xi_a'$, $\Xi_b \subseteq \Xi_b'$, $\dom{\mu_a} \subseteq \dom{\mu_a'}$, and $\dom{\mu_b} \subseteq \dom{\mu_b'}$ and $\semanticBelowPCState{l}{\Xi_{a}'}{\Xi_{b}'}{\Delta}{\mu_{a}'}{\epsilon_{a}}{\mu_{b}'}{\epsilon_{b}}{\Gamma}$,
\item $\NIval{l}{\Xi_a'}{\Xi_b'}{\Delta}{val_{ai}}{val_{bi}}{\type{\tau_{i}}{\chi_{i}}}$,
\item For any $l_a \in \dom{\mu_a}$ such that $\ordinaryTyping[]{\Xi_a}{\Delta}{\mu_a(l_a)}{\tau_{clos}}$, where $\tau_{clos} \in \{\tau_{fn}, \tau_{tbl}\}$, then $\mu_a'(l_a) = \mu_a(l_a)$. Similarly for any $l_b \in \dom{\mu_b}$ such that $\ordinaryTyping[]{\Xi_b}{\Delta}{\mu_b(l_b)}{\tau_{clos}}$, where $\tau_{clos} \in \{\tau_{fn}, \tau_{tbl}\}$, then $\mu_b'(l_b) = \mu_b(l_b)$.
\end{enumerate}
We consider only valid headers in this information-flow control system.
Similar to case \ref{recmem-exp}, we apply induction hypothesis on typing derivation of $exp$ followed by inverting the value typing for headers.

%
%
\item \textbf{T-FuncCall}
If the typing derivation ends with the following last rule
\[
\inferrule*[right=T-Call]
{
\Gamma, \Delta \vdash_{pc} exp_{1}: \langle \overline{d~\type{\tau_i}{\chi_i}} \xrightarrow{pc_{fn}} \langle \tau_{ret}, \chi_{ret}\rangle, \bot \rangle\\
\Gamma, \Delta \vdash_{pc} \overline{exp_{2}:\type{\tau_i}{\chi_i}~ goes~d} \\
pc \sqsubseteq pc_{fn}
}
{
\Gamma, \Delta \vdash_{pc} exp_{1} (\overline{exp_{2}}): \type{\tau_{ret}}{\chi_{ret}} ~\text{goes in}
}
\]
then we need to show that for any $\Xi_a$, $\Xi_b$, $\mu_{a}$, $\mu_{b}$, $\epsilon_{a}$, $\epsilon_{b}$, $\mu_a'$, $\mu_b'$ satisfying
\begin{equation} \label{fncall-initial-hypo}
\semanticBelowPCState{l}{\Xi_a}{\Xi_b}{\Delta}{\mu_{a}}{\epsilon_{a}}{\mu_{b}}{\epsilon_{b}}{\Gamma}
\end{equation}
if the function call expression $exp_{1} (\overline{exp_{2}})$ is evaluated under two different initial configurations $\langle \mu_a, \epsilon_a \rangle$ and $\langle \mu_b, \epsilon_b \rangle$ as follows:
\[
\inferrule*[]
{
 \langle \mathcal{C}, \Delta, \mu_{a}, \epsilon_{a}, exp_{1} \rangle \Downarrow \langle \mu_{a1},clos(\epsilon_{c_a}, \overline{d~x: \langle \tau, \chi \rangle}, \langle \tau_{ret}, \chi_{ret} \rangle, \text{stmt}) \rangle \\
\langle  \Delta, \mu_{a1}, \epsilon_{a}, \overline{d x: \langle \tau, \chi \rangle :=exp_{2} }\rangle \Downarrow_{copy} \langle \mu_{a2}, \overline{x \mapsto l_a}, \overline{lval_a:= l_a} \rangle \\
 \langle \mathcal{C}, \Delta, \mu_{a2}, \epsilon_{c_a}[\overline{x \mapsto l_a}], stmt \rangle \Downarrow \langle \mu_{a3}, \epsilon_{a2}, \text{return~}val_{a}\rangle \\
 \langle \mathcal{C}, \Delta, \mu_{a3}, \epsilon_{a},  \overline{lval:=\mu_{a3}(l)}\rangle \Downarrow_{write} \mu_{4}
}
{
 \langle \mathcal{C}, \Delta, \mu_{a}, \epsilon_{a}, exp_{1} (\overline{exp_{2}})\rangle \Downarrow \langle \mu_{a4}, val_{a} \rangle
}
\]
\[
\inferrule*[]
{
 \langle \mathcal{C}, \Delta, \mu_{b}, \epsilon_{b}, exp_{1} \rangle \Downarrow \langle \mu_{b1},clos(\epsilon_{c_b},\overline{d~x: \langle \tau, \chi \rangle}, \langle \tau_{ret}, \chi_{ret} \rangle, \text{stmt}) \rangle \\
\langle  \Delta, \mu_{b1}, \epsilon_{b}, \overline{d x: \langle \tau, \chi \rangle :=exp_{2} }\rangle \Downarrow_{copy} \langle \mu_{b2}, \overline{x \mapsto l_b}, \overline{lval_b:= l_b} \rangle \\
 \langle \mathcal{C}, \Delta, \mu_{b2}, \epsilon_{c_b}[\overline{x \mapsto l_b}], stmt \rangle \Downarrow \langle\mu_{b3}, \epsilon_{b2}, \text{return~}val_{b} \rangle \\
 \langle \mathcal{C}, \Delta, \mu_{b3}, \epsilon_{b},  \overline{lval_b:=\mu_{b3}(l_b)}\rangle \Downarrow_{write} \mu_{b4}
}
{
 \langle \mathcal{C}, \Delta, \mu_{b}, \epsilon_{b}, exp_{1} (\overline{exp_{2}})\rangle \Downarrow \langle \mu_{b4}, val_{b}\rangle
}
\]
then there exists some $\Xi_a'$ and $\Xi_b'$, such that the following hold:
\begin{enumerate}
\item $\ordinaryTyping[pc]{\Gamma}{\Delta}{exp_{1} (\overline{exp_{2}})}{\type{\tau_{ret}}{\chi_{ret}}}$. Already given in the hypothesis of this theorem.
\item \label{fn-tp-1}$\Xi_a \subseteq \Xi_a'$, $\Xi_b \subseteq \Xi_b'$, $\dom{\mu_a} \subseteq \dom{\mu_a'}$, $\dom{\mu_b} \subseteq \dom{\mu_b'}$, and $\semanticBelowPCState{l}{\Xi_{a}'}{\Xi_{b}'}{\Delta}{\mu_{a}'}{\epsilon_{a}}{\mu_{b}'}{\epsilon_{b}}{\Gamma}$. Here, $\mu_a' = \mu_{a4}$ and $\mu_b' = \mu_{b4}$,
\item \label{fn-tp-2}For any $l_a \in \dom{\mu_a}$ and $l_b \in \dom{\mu_b}$ such that $\ordinaryTyping[]{\Xi_a}{\Delta}{\mu_a(l_a)}{\type{\tau}{\chi}}$ and $\ordinaryTyping[]{\Xi_b}{\Delta}{\mu_b(l_b)}{\type{\tau}{\chi}}$ and $pc \nsqsubseteq \chi$, we have $\mu_{a}'(l_a) = \mu_{a}(l_a)$ and $\mu_{b}'(l_b) = \mu_{b}(l_b)$,
\item $\NIval{l}{\Xi_a'}{\Xi_b'}{\Delta}{val_{a}}{val_b}{\type{\tau_{ret}}{\chi_{ret}}}$,
\item \label{fn-tp-3}For any $l_a \in \dom{\mu_a}$ such that $\ordinaryTyping[]{\Xi_a}{\Delta}{\mu_a(l_a)}{\tau_{clos}}$, where $\tau_{clos} \in \{\tau_{fn}, \tau_{tbl}\}$, then $\mu_a'(l_a) = \mu_a(l_a)$. Similarly for any $l_b \in \dom{\mu_b}$ such that $\ordinaryTyping[]{\Xi_b}{\Delta}{\mu_b(l_b)}{\tau_{clos}}$, where $\tau_{clos} \in \{\tau_{fn}, \tau_{tbl}\}$, then $\mu_b'(l_b) = \mu_b(l_b)$.
\end{enumerate}

By applying induction hypothesis of \Cref{ni-exp} on $exp_1$, which is evaluated in an initial configuration satisfying \Cref{fncall-initial-hypo},
we get: $\NIexp{l}{\Gamma}{\Delta}{exp_1}{\langle \overline{d~\type{\tau_i}{\chi_i}} \xrightarrow{pc_{fn}} \langle \tau_{ret}, \chi_{ret}\rangle, \bot \rangle}$.
This implies that there exists some $\Xi_{a1}$, $\Xi_{b1}$, $\mu_{a1}$, $\mu_{b1}$
satisfying $\Xi_a \subseteq \Xi_{a1}$ and $\Xi_b \subseteq \Xi_{b1}$, $\dom{\mu_a} \subseteq \dom{\mu_{a1}}$,
and $\dom{\mu_b} \subseteq \dom{\mu_{b1}}$ and the following:
\begin{equation} \label{fn-call-ih-1-1}
\semanticBelowPCState{l}{\Xi_{a1}}{\Xi_{b1}}{\Delta}{\mu_{a1}}{\epsilon_{a}}{\mu_{b1}}{\epsilon_{b}}{\Gamma}
\end{equation}
\begin{equation}\label{fn-call-ih-1-2}
\NIval{l}{\Xi_{a1}}{\Xi_{b1}}{\Delta}{val_{a1}}{val_{b1}}{\tau_{fn}}
\end{equation}
Here $val_{a1} = clos(\epsilon_{c_a}, \overline{d x: \type{\tau}{\chi}}, \type{\tau_{ret}}{\chi_{ret}}, stmt)$ and
$val_{b1} = clos(\epsilon_{c_b}, \overline{d x: \type{\tau}{\chi}}, \type{\tau_{ret}}{\chi_{ret}}, stmt)$.

Since $\type{\tau}{\chi} = \tau_{fn}$, by using \Cref{fn-call-ih-1-2} we conclude that $\NIclos{}{\Xi_{a1}}{\Xi_{b1}}{\Delta}{val_{a1}}{val_{b1}}{{\langle \overline{d\type{\tau_i}{\chi_i}} \xrightarrow{pc_{fn}} \langle \tau_{ret}, \chi_{ret}\rangle, \bot \rangle}}$.
Expanding the non-interference definition for closure (\Cref{clos-def}), we conclude that there exists some $\Gamma_{fn}$,
such that the following properties are satisfied:
\[\Xi_{a1}, \Delta \models \epsilon_{c_a}: \Gamma_{fn}\]
\[\Xi_{b1}, \Delta \models \epsilon_{c_b}: \Gamma_{fn}\]
\[\Gamma_{fn}; \Delta \vdash_{pc} clos(\epsilon_{c_a}, \overline{d x: \type{\tau}{\chi}}, \type{\tau_{ret}}{\chi_{ret}}, \text{stmt}):\langle \overline{d\langle \tau, \chi\rangle} \xrightarrow{pc_{fn}} \langle \tau_{ret}, \chi_{ret}\rangle, \bot \rangle\]
\[\Gamma_{fn}; \Delta \vdash_{pc} clos(\epsilon_{c_b}, \overline{d x: \type{\tau}{\chi}}, \type{\tau_{ret}}{\chi_{ret}}, \text{stmt}):\langle \overline{d \langle \tau, \chi\rangle} \xrightarrow{pc_{fn}} \langle \tau_{ret}, \chi_{ret}\rangle, \bot \rangle\]
\begin{equation}
\Gamma_{fn}[\overline{x:\type{\tau}{\chi}}, \terminal{return}: \type{\tau_{ret}}{\chi_{ret}}], \Delta \vdash_{pc_{fn}}stmt \dashv \Gamma_{fn2}
\end{equation}

Application of the induction hypothesis on $exp_1$ also grantees that the closure values do not change in the transition from $\mu_a$ to $\mu_{a1}$ and $\mu_b$ to $\mu_{b1}$. Therefore, we can apply the property of closure values in the state given by \Cref{fn-call-ih-1-1} to the closure values returned after the evaluation of  $exp_1$.
\Cref{fn-call-ih-1-1} concludes that for any $x \in \dom{\epsilon_a}$, satisfying $\Gamma \vdash x: \tau_{fn}$, $\mu_{a1}(\epsilon_a(x)) = clos(\epsilon_{c},...)$, and $\Xi_{a1} \vdash \epsilon_c: \Gamma_c $, we will have $\dom{\epsilon_{c}} \subseteq \dom{\epsilon_a}$ and $\semanticStoreEnv{\Xi_{a1}}{\Delta}{\mu_{a1}}{\epsilon_{c}}{\Gamma_{c}}$ . Here $\tau_{f} =\langle  \overline{d \langle\tau, \chi \rangle} \xrightarrow{pc_{fn}} \langle \tau_{ret}, \chi_{ret}\rangle, \bot \rangle$, for any $\tau, \chi, \tau_{ret}, \chi_{ret}$. This implies that $\dom{\epsilon_{c_a}} \subseteq \dom{\epsilon_{a}}$ and $\semanticStoreEnv{\Xi_{a1}}{\Delta}{\mu_{a1}}{\epsilon_{c_a}}{\Gamma_{fn}}$. Similarly $\dom{\epsilon_{c_b}} \subseteq \dom{\epsilon_{b}}$ and $\semanticStoreEnv{\Xi_{b1}}{\Delta}{\mu_{b1}}{\epsilon_{c_b}}{\Gamma_{fn}}$.

Using \Cref{lem:copy-in-out} for the evaluation of $d x: \langle \tau, \chi \rangle :=\exp_{2} $ in the initial configuration satisfying \Cref{fn-call-ih-1-1}, we conclude the following:
\begin{enumerate}
\item \label{fndecl-ih-2-1} $\semanticBelowPCState{l}{\Xi_{a2}}{\Xi_{b2}}{\Delta}{\mu_{a2}}{\overline{x \mapsto l_a}}{\mu_{b2}}{\overline{x \mapsto l_b}}{\Gamma'}$, for some $\Xi_{a2}$, $\Xi_{b2}$, $\mu_{a2}$, $\mu_{b2}$ such that $\Xi_{a1} \subseteq \Xi_{a2}$ and $\Xi_{b1} \subseteq \Xi_{b2}$, $\dom{\mu_{a1}} \subseteq \dom{\mu_{a2}}$, $\dom{\mu_{b1}} \subseteq \dom{\mu_{b2}}$ and $\Gamma' = \{\overline{x \mapsto \type{\tau}{\chi}}\}$.
\item \label{fndecl-ih-2-2} $\semanticBelowPCState{l}{\Xi_{a2}}{\Xi_{b2}}{\Delta}{\mu_{a2}}{\epsilon_a}{\mu_{b2}}{\epsilon_b}{\Gamma}$
\item For any $l_a' \in \dom{\mu_{a1}}$ and $l_b' \in \dom{\mu_{b1}}$ such that $\ordinaryTyping[]{\Xi_{a1}}{\Delta}{\mu_{a1}(l_a')}{\type{\tau}{\chi}}$ and $\ordinaryTyping[]{\Xi_{b1}}{\Delta}{\mu_{b1}(l_b')}{\type{\tau}{\chi}}$ and $pc \nsqsubseteq \chi$, we have $\mu_{a1}(l_a') = \mu_{a2}(l_a')$ and $\mu_{b2}(l_b') = \mu_{b1}(l_b')$,
\item  $\textsc{lval_base}(lval_a) \in \dom{\epsilon_a}$, $\textsc{lval_base}(lval_b) \in \dom{\epsilon_b}$, and $\textsc{lval_base}(lval_a) = \textsc{lval_base}(lval_b)$ for each $lval_a$ and $lval_b$.
\item $\overline{l_a} \in \dom{\mu_{a2}}$ and $\overline{l_b} \in \dom{\mu_{b2}}$
\item \label{loc-in-store} $l_a$ and $l_b$ are fresh locations, $l_a \notin \Xi_{a1}$ and $l_b \notin \Xi_{b1}$
\item \label{unchanged-clos} For any $l_a \in \dom{\mu_{a1}}$ such that $\ordinaryTyping[]{\Xi_{a1}}{\Delta}{\mu_{a1}(l_a)}{\tau_{clos}}$, where $\tau_{clos} \in \{\tau_{fn}, \tau_{tbl}\}$, then $\mu_{a1}(l_a) = \mu_{a2}(l_a)$. Similarly for any $l_b \in \dom{\mu_{b1}}$ such that $\ordinaryTyping[]{\Xi_b}{\Delta}{\mu_{b1}(l_b)}{\tau_{clos}}$, where $\tau_{clos} \in \{\tau_{fn}, \tau_{tbl}\}$, then $\mu_{b1}(l_b) = \mu_{b2}(l_b)$.
\end{enumerate}
Given \Cref{unchanged-clos}, we can observe that some closure variable $x$ that evaluated to the closures returned on evaluating $exp_1$ will have the same value in $\mu_{a2}$.  Therefore, by expanding \Cref{fndecl-ih-2-2} we conclude
\begin{equation} \label{fn-decl-ih-2-3}
\semanticBelowPCState{l}{\Xi_{a2}}{\Xi_{b2}}{\Delta}{\mu_{a2}}{\epsilon_{c_a}}{\mu_{b2}}{\epsilon_{c_b}}{\Gamma_{fn}}
\end{equation}
Combining \Cref{fndecl-ih-2-1} and \Cref{fn-decl-ih-2-3} using \Cref{lem:pair-env-extend} we get:
\begin{equation} \label{fn-decl-ih-2-4}
\semanticBelowPCState{l}{\Xi_{a2}}{\Xi_{b2}}{\Delta}{\mu_{a2}}{\epsilon_{c_a}[\overline{x \mapsto l_a}]}{\mu_{b2}}{\epsilon_{c_b}[\overline{x \mapsto l_b}]}{\Gamma_{fn}[\overline{x \mapsto \type{\tau}{\chi}}]}
\end{equation}
Note that \Cref{loc-in-store} enforces that $l_a$ and $l_b$ are present in $\mu_{a2}$ and $\mu_{b2}$.

By using the induction hypothesis of \Cref{ni-stmt} on $stmt$ that is evaluated in the initial configuration satisfying \Cref{fn-decl-ih-2-4}, we conclude $\NIstmt{pc}{\Gamma_{fn}[\overline{x:\type{\tau}{\chi}}, \terminal{return}: \type{\tau_{ret}}{\chi_{ret}}]}{\Delta}{stmt}{\Gamma_{fn1}}$ or
there exist some $\Xi_{a3}$, $\Xi_{b3}$, $\mu_{a3}$, $\mu_{b3}$, $\epsilon_{a2}$, and $\epsilon_{b2}$ such that $\Xi_{a2} \subseteq \Xi_{a3}$ and $\Xi_{b2} \subseteq \Xi_{b3}$, $\dom{\mu_{a2}} \subseteq \dom{\mu_{a3}}$, $\dom{\mu_{b2}} \subseteq \dom{\mu_{b3}}$, $\dom{\epsilon_{c_a}[\overline{x \mapsto l_a}]}\subseteq \dom{\epsilon_{a2}}$, and $\dom{\epsilon_{c_b}[\overline{x \mapsto l_b}]} \subseteq \dom{\epsilon_{b2}}$ satisfying:
\begin{equation}
\semanticBelowPCState{l}{\Xi_{a3}}{\Xi_{b3}}{\Delta}{\mu_{a3}}{\epsilon_{a2}}{\mu_{b3}}{\epsilon_{b2}}{\Gamma_{fn1}}
\end{equation}
\begin{equation} \label{clos-eq}
\semanticBelowPCState{l}{\Xi_{a3}}{\Xi_{b3}}{\Delta}{\mu_{a3}}{\epsilon_{c_a}[\overline{x \mapsto l_a}]}{\mu_{b3}}{\epsilon_{c_b}[\overline{x \mapsto l_b}]}{\Gamma_{fn}[\overline{x:\type{\tau}{\chi}}, \terminal{return}: \type{\tau_{ret}}{\chi_{ret}}]}
\end{equation}
and none of the locations with security label $pc \sqsubseteq \chi$ will be updated between $\mu_{a2}$ and $\mu_{a3}$, $\mu_{b2}$ and $\mu_{b3}$.

We know that $\semanticBelowPCState{l}{\Xi_{a2}}{\Xi_{b2}}{\Delta}{\mu_{a2}}{\epsilon_a}{\mu_{b2}}{\epsilon_b}{\Gamma}$.
Any $y \in \dom{\epsilon_a}=\dom{\epsilon_b}$ can satisfy one of the following:
\begin{enumerate}
\item $\epsilon_a(y) = \epsilon_{c_a}[\overline{x \mapsto l_a}](y)$ and $\epsilon_b(y) = \epsilon_{c_b}[\overline{x \mapsto l_b}](y)$, then $\mu_{a3}(\epsilon_a(y)) = \mu_{a3}(\epsilon_{c_a}[\overline{x \mapsto l_a}](y))$ and $\mu_{b3}(\epsilon_b(y)) = \mu_{b3}(\epsilon_{c_b}[\overline{x \mapsto l_b}](y))$. This variable has non-interfering value (\Cref{clos-eq}).
\item \label{unused-vars-fn-call}  $\textsc{unused}(\mu_{a2}, \epsilon_{c_a}, y, \epsilon_a(y))$ (\Cref{def-unused}) and $\textsc{unused}(\mu_{b2}, \epsilon_{c_b}, y, \epsilon_b(y))$, then $\mu_{a3}(y) = \mu_{a2}(y)$ and $\mu_{b3}(y) = \mu_{b2}(y)$. $\semanticBelowPCState{l}{\Xi_{a2}}{\Xi_{b2}}{\Delta}{\mu_{a2}}{\epsilon_a}{\mu_{b2}}{\epsilon_b}{\Gamma}$ guarantees this value to be non-interfering.
\item \label{used-vars-fn-call}$\neg \textsc{unused}(\mu_{a2}, \epsilon_{c_a}, y, \epsilon_a(y))$ and $\neg \textsc{unused}(\mu_{b2}, \epsilon_{c_b}, y, \epsilon_b(y))$, then there exists some closure value with $\epsilon_{c_a}'$ and $\dom{\epsilon_{c_a}'} \subseteq \dom{\epsilon_{c_a}}$, and $\epsilon_{c_b}'$ and $\dom{\epsilon_{c_b}'} \subseteq \dom{\epsilon_{c_b}}$  where $\epsilon_a(y) = \epsilon_{c_a}'(y)$ and $\epsilon_b(y) = \epsilon_{c_b}'(y)$. From \Cref{clos-eq}, we know that $\semanticBelowPCState{l}{\Xi_{a3}}{\Xi_{b3}}{\Delta}{\mu_{a3}}{\epsilon_{c_a}'}{\mu_{b3}}{\epsilon_{c_b}'}{\Gamma_{clos}}$.
\end{enumerate}
To conclude that
\begin{equation} \label{eq-4}
\semanticBelowPCState{l}{\Xi_{a3}}{\Xi_{b3}}{\Delta}{\mu_{a3}}{\epsilon_{a}}{\mu_{b3}}{\epsilon_{b}}{\Gamma}
\end{equation}
we also need to ensure that for all $x$ in $\dom{\epsilon_a} = \dom{\epsilon_b}$
and some $\Gamma_{clos} \subseteq \Gamma$ and any $pc$,  if $\Gamma, \Delta
\vdash_{pc} x: \tau_{clos}$ with closure environments $\epsilon_{c_a}'$ and
$\epsilon_{c_b}'$ in the two states, then
$\semanticBelowPCState{l}{\Xi_{a3}}{\Xi_{b3}}{\Delta}{\mu_{a3}}{\epsilon_{c_a}'}{\mu_{b3}}{\epsilon_{c_b}'}{\Gamma_{clos}}$.
For closure variables satisfying \Cref{used-vars-fn-call}, this will follow from closure properties in \Cref{clos-eq}.
For variables satisfying   \Cref{unused-vars-fn-call}, this will follow from the fact that the variables in their closure environments can again be \textsc{unused} (implies unchanged between $\mu_{a2}$ and $\mu_{a3}$, $\mu_{b2}$ and $\mu_{b3}$) or \textsc{used} (in this case we already know from \Cref{clos-eq} that such variables satisfy non-interference of values).

Using \Cref{lem:lvalue-write-final} on $\mu_{a3}$, $\mu_{b3}$ to assign non-interfering values (\Cref{eq-4} implies that the store has non-interfering values) to l-values, we conclude
\begin{equation}
\semanticBelowPCState{l}{\Xi_{a4}}{\Xi_{b4}}{\Delta} {\mu_{a5}}{\epsilon_{a}}{\mu_{b5}}{\epsilon_{b}}{\Gamma}
\end{equation}
Since $\Xi_a \subseteq \Xi_{a1} \subseteq \Xi_{a2} \subseteq \Xi_{a3} \subseteq \Xi_{a4}$ and $\Xi_b \subseteq \Xi_{b1} \subseteq \Xi_{b2} \subseteq \Xi_{b3} \subseteq \Xi_{b4}$, showing the above equation is same as showing \Cref{fn-tp-1}.
Proof of \Cref{fn-tp-2} and \Cref{fn-tp-3} follows from the results of the application of the theorem for NI for expression, statements above and the fact that domain of memory stores have increasing domains.

\item \textbf{T-MatchKind}
Trivial
\[ \inferrule*[right=T-MemHdr]
{
  match\_kind \{\overline{f}\} \in \Delta(match\_kind) \\
  f_i \in \overline{f}
}
{
\ordinaryTyping[pc]{\Gamma}{\Delta}{match\_kind.f_i}{\type{match\_kind \{\overline{f}\}}{\bot}~goes~in}
}
\]
Evaluation rule
\[
\inferrule*[right=Eval 1]
{
  match\_kind \{\overline{f}\} \in \Delta(match\_kind)
}
{ \langle \mathcal{C}, \Delta; \mu_{a}; \epsilon_{a}; match\_kind.f_i \rangle \Downarrow \langle \mu_{a}, f_i\rangle}
\]

\[
\inferrule*[right=Eval 2]
{
 match\_kind \{\overline{f}\} \in \Delta(match\_kind)
}
{ \langle \mathcal{C}, \Delta; \mu_{b}; \epsilon_{b}; exp.f_{i} \rangle \Downarrow \langle \mu_{b}, f_i\rangle}
\]

\end{enumerate}
\paragraph*{Proof on \Cref{ni-stmt}}
The non-interference theorem for statements is given in \Cref{ni-stmt}.
\begin{enumerate}
\item \textbf{\textsc{T-Empty}}
The last typing rule in the derivation of an empty statement will be:

\[ \inferrule*[]
{
}
{
\stmtTyping{pc}{\Gamma}{\Delta}{\{ \}}{\Gamma}
}
\]

Given the above typing judgement holds for, $\{\}$, statement, we need to show that for any $\Xi_a$, $\Xi_b$, $\mu_{a}$, $\mu_{b}$, $\epsilon_{a}$, $\epsilon_{b}$, $\mu_a'$, $\mu_b'$, $\epsilon_a'$, $\epsilon_b'$ satisfying
\begin{equation} \label{stmt-empty-hyp}
\semanticBelowPCState{l}{\Xi_a}{\Xi_b}{\Delta}{\mu_{a}}{\epsilon_{a}}{\mu_{b}}{\epsilon_{b}}{\Gamma}
\end{equation}
if the statement, $\{\}$ is evaluated under two different initial configurations $\langle \mu_a, \epsilon_a \rangle$ and $\langle \mu_b, \epsilon_b \rangle$ as follows:

\begin{mathpar}
\inferrule*[]
{
~
}
{
 \langle \mathcal{C}, \Delta, \mu_{a}, \epsilon_{a}, \{\} \rangle \Downarrow \langle \mu_{a}, \epsilon_{a}, cont \rangle
}

 \inferrule*[]
{
~
}
{
 \langle \mathcal{C}, \Delta, \mu_{b}, \epsilon_{b}, \{\} \rangle \Downarrow \langle \mu_{b}, \epsilon_{b}, cont \rangle
}
\end{mathpar}

Then there exists some $\Xi_a'$ and $\Xi_b'$, such that the following hold:
\begin{enumerate}
\item $\stmtTyping{pc}{\Gamma}{\Delta}{\{\}}{\Gamma'}{}$. This is already the theorem's hypothesis.
\item We have $\Xi_a \subseteq \Xi_a'$, $\Xi_b \subseteq \Xi_b'$, $\dom{\mu_a}
  \subseteq \dom{\mu_a'}$, $\dom{\mu_b} \subseteq \dom{\mu_b'}$,
  $\dom{\epsilon_a} \subseteq \dom{\epsilon_a'}$, and $\dom{\epsilon_b}
  \subseteq \dom{\epsilon_b'}$,
and $\semanticBelowPCState{l}{\Xi_a'}{\Xi_b'}{\Delta} {\mu_{a}'}{\epsilon_{a}'}{\mu_{b}'}{\epsilon_{b}'}{\Gamma'}$. In this case $\mu_a' = \mu_a$, $\mu_b' = \mu_b$, $\epsilon_a' = \epsilon_a$, $\epsilon_b' = \epsilon_b$.

With $\Xi_a' = \Xi_a$, $\Xi_b' = \Xi_b$, the above equation reduces to showing \Cref{stmt-empty-hyp}.
\item For any $l_a \in \dom{\mu_a}$ such that $\ordinaryTyping[]{\Xi_a}{\Delta}{\mu_a(l_a)}{\tau_{clos}}$, where $\tau_{clos} \in \{\tau_{fn}, \tau_{tbl}\}$, then $\mu_a'(l_a) = \mu_a(l_a)$. Similarly for any $l_b \in \dom{\mu_b}$ such that $\ordinaryTyping[]{\Xi_b}{\Delta}{\mu_b(l_b)}{\tau_{clos}}$, where $\tau_{clos} \in \{\tau_{fn}, \tau_{tbl}\}$, then $\mu_b'(l_b) = \mu_b(l_b)$. This is evident as the memory store remains unchanged.
\item $sig$ in any two evaluations are of the same form. In this case $sig_1 = cont = sig_2$.
\item For any $l_a \in \dom{\mu_a}$ and $l_b \in \dom{\mu_b}$ such that $\ordinaryTyping[]{\Xi_a}{\Delta}{\mu_a(l_a)}{\type{\tau}{\chi}}$ and $\ordinaryTyping[]{\Xi_b}{\Delta}{\mu_b(l_b)}{\type{\tau}{\chi}}$ and $pc \nsqsubseteq \chi$, we have $\mu_{a}'(l_a) = \mu_{a}(l_a)$ and $\mu_{b}'(l_b) = \mu_{b}(l_b)$. The stores remain unchanged.
\end{enumerate}

\item \textbf{\textsc{T-Exit}}
\[  \inferrule*[]
{
    ~
}
{
     \langle \mathcal{C}, \Delta, \sigma, \epsilon, exit \rangle \Downarrow \langle \sigma, \epsilon, exit \rangle
}
\]

Similar to the empty statement case. This time the $sig_{1}=sig_{2}=exit$

\item \textbf{\textsc{T-Cond}}
The last rule in the typing derivation of a conditional statement will be:
\[
\inferrule*[right=T-Cond]
{
\ordinaryTyping[pc]{\Gamma}{\Delta}{exp}{\type{bool}{\chi_1}}\\\\
\stmtTyping{\chi_2}{\Gamma}{\Delta}{stmt_{1}}{\Gamma_{1}}\\
\stmtTyping{\chi_2}{\Gamma}{\Delta}{stmt_{2}}{\Gamma_{2}} \\
\chi_1 \sqsubseteq \chi_2 \\
pc \sqsubseteq \chi_2
}
{
\stmtTyping{pc}{\Gamma}{\Delta}{\terminal{if}~(exp)~~stmt_{1}~\terminal{else}~stmt_{2}}{\Gamma}
}
\]

Given the above typing judgement holds for, $\textsf{if}~(exp)~ stmt_{1}~\textsf{else}~stmt_{2}$, statement, we need to show that for any $\Xi_a$, $\Xi_b$, $\mu_{a}$, $\mu_{b}$, $\epsilon_{a}$, $\epsilon_{b}$, $\mu_a'$, $\mu_b'$, $\epsilon_a'$, $\epsilon_b'$ satisfying
\begin{equation} \label{stmt-ite-hyp}
\semanticBelowPCState{l}{\Xi_a}{\Xi_b}{\Delta}{\mu_{a}}{\epsilon_{a}}{\mu_{b}}{\epsilon_{b}}{\Gamma}
\end{equation}
if the statement, $\text{if}~(exp)~stmt_{1}~\text{else}~stmt_{2}$ is evaluated under two different initial configurations $\langle \mu_a, \epsilon_a \rangle$ and $\langle \mu_b, \epsilon_b \rangle$ as follows (in a given evaluation, a conditional statement can have the $exp$ evaluate to true or false):

\paragraph*{Boolean guard evaluates to false}
\begin{mathpar}
\inferrule*[]
{
 \langle \mathcal{C}, \Delta, \mu_{a}, \epsilon_{a}, exp \rangle \Downarrow  \langle \mu_{a1}, false \rangle \\
 \langle \mathcal{C}, \Delta, \mu_{a1}, \epsilon_{a}, stmt_{2} \rangle \Downarrow  \langle \mu_{a2}, \epsilon_{a1}, sig_{a1} \rangle
}
{
 \langle \mathcal{C}, \Delta, \mu_{a}, \epsilon_{a}, \terminal{if}~ (exp)~ stmt_{1}~  \terminal{else}~ stmt_{2}  \rangle  \Downarrow  \langle \mu_{a2}, \epsilon_{a}, sig_{a1} \rangle
}

\inferrule*[]
{
 \langle \mathcal{C}, \Delta, \mu_{b}, \epsilon_{b}, exp \rangle \Downarrow  \langle \mu_{b1}, false \rangle \\
 \langle \mathcal{C}, \Delta, \mu_{b1}, \epsilon_{b}, stmt_{2} \rangle \Downarrow  \langle \mu_{b2}, \epsilon_{b1}, sig_{b1} \rangle
}
{
 \langle \mathcal{C}, \Delta, \mu_{b}, \epsilon_{b}, \terminal{if}~ (exp)~ stmt_{1}~  \terminal{else}~ stmt_{2} \rangle  \Downarrow  \langle \mu_{b2}, \epsilon_{b}, sig_{b1} \rangle
}
\end{mathpar}

\paragraph*{Boolean guard evaluates to true}
\begin{mathpar}
\inferrule*[]
{
 \langle \mathcal{C}, \Delta, \mu_{a}, \epsilon_{a}, exp \rangle \Downarrow  \langle \mu_{a1}, true \rangle \\
 \langle \mathcal{C}, \Delta, \mu_{a1}, \epsilon_{a}, stmt_{1} \rangle \Downarrow  \langle \mu_{a2}, \epsilon_{a1}, sig_{a2} \rangle
}
{
 \langle \mathcal{C}, \Delta, \mu_{a}, \epsilon_{a}, \terminal{if}~ (exp)~ stmt_{1}~  \terminal{else}~ stmt_{2} \rangle  \Downarrow  \langle \mu_{a2}, \epsilon_{a}, sig_{a2} \rangle
}

\inferrule*[]
{
 \langle \mathcal{C}, \Delta, \mu_{b}, \epsilon_{b}, exp \rangle \Downarrow  \langle \mu_{b1}, true \rangle \\
 \langle \mathcal{C}, \Delta, \mu_{b1}, \epsilon_{b}, stmt_{1} \rangle \Downarrow  \langle \mu_{b2}, \epsilon_{b1}, sig_{b2} \rangle
}
{
 \langle \mathcal{C}, \Delta, \mu_{b}, \epsilon_{b}, \terminal{if}~ (exp)~ stmt_{1}~  \terminal{else}~ stmt_{2}  \rangle  \Downarrow  \langle \mu_{b2}, \epsilon_{b}, sig_{b2} \rangle
}
\end{mathpar}
Then there exists some $\Xi_a'$ and $\Xi_b'$, such that the following hold:
\begin{enumerate}
\item $\stmtTyping{pc}{\Gamma}{\Delta}{\terminal{if}~ (exp)~ stmt_{1}~  \terminal{else}~ stmt_{2}}{\Gamma'}{}$. This is already the theorem's hypothesis.
\item \label{stmt-ite-tp-2} We have $\Xi_a \subseteq \Xi_a'$, $\Xi_b \subseteq
  \Xi_b'$, $\dom{\mu_a} \subseteq \dom{\mu_a'}$, $\dom{\mu_b} \subseteq
  \dom{\mu_b'}$, $\dom{\epsilon_a} \subseteq \dom{\epsilon_a'}$, and
  $\dom{\epsilon_b} \subseteq \dom{\epsilon_b'}$,
and $\semanticBelowPCState{l}{\Xi_a'}{\Xi_b'}{\Delta} {\mu_{a}'}{\epsilon_{a}'}{\mu_{b}'}{\epsilon_{b}'}{\Gamma'}$. In this case $\mu_a' = \mu_{a2}$, $\mu_b' = \mu_{b2}$, $\epsilon_a' = \epsilon_a$, $\epsilon_b' = \epsilon_b$.
\item \label{stmt-ite-tp-4} For any $l_a \in \dom{\mu_a}$ such that $\ordinaryTyping[]{\Xi_a}{\Delta}{\mu_a(l_a)}{\tau_{clos}}$, where $\tau_{clos} \in \{\tau_{fn}, \tau_{tbl}\}$, then $\mu_a'(l_a) = \mu_a(l_a)$. Similarly for any $l_b \in \dom{\mu_b}$ such that $\ordinaryTyping[]{\Xi_b}{\Delta}{\mu_b(l_b)}{\tau_{clos}}$, where $\tau_{clos} \in \{\tau_{fn}, \tau_{tbl}\}$, then $\mu_b'(l_b) = \mu_b(l_b)$.
\item \label{stmt-ite-tp-5} For any $l_a \in \dom{\mu_a}$ and $l_b \in \dom{\mu_b}$ such that $\ordinaryTyping[]{\Xi_a}{\Delta}{\mu_a(l_a)}{\type{\tau}{\chi}}$ and $\ordinaryTyping[]{\Xi_b}{\Delta}{\mu_b(l_b)}{\type{\tau}{\chi}}$ and $pc \nsqsubseteq \chi$, we have $\mu_{a}'(l_a) = \mu_{a}(l_a)$ and $\mu_{b}'(l_b) = \mu_{b}(l_b)$,
\item \label{stmt-ite-tp-3} Final $sig$ in any two evaluations are of the same form. We will show this by proving that despite both the branches yielding independent $sig_{a1}$, $sig_{a2}$ (similarly for $b$), the typing rule will ensure that the final $sig$ will be of the same form.
\end{enumerate}
In the following part, we prove the last four requirements.
By applying induction hypothesis of \Cref{ni-exp} on the well-typed $exp$ that is evaluated in an initial state satisfying \Cref{stmt-ite-hyp}, we conclude that  there exists some $\Xi_{a1}'$ and $\Xi_{b1}'$ such that $\Xi_{a} \subseteq \Xi_{a1}'$ and $\Xi_{b} \subseteq \Xi_{b1}'$, $\dom{\mu_{a1}} \supseteq \dom{\mu_{a}}$, $\dom{\mu_{b1}} \supseteq \dom{\mu_{b}}$ and the following hold:
\begin{equation}\label{stmt-ite-ih-1}
\semanticBelowPCState{l}{\Xi_{a1}'}{\Xi_{b1}'}{\Delta}{\mu_{a1}}{\epsilon_{a}}{\mu_{b1}}{\epsilon_{b}}{\Gamma}~,
\end{equation}
for any $l_a \in \dom{\mu_a}$ such that $\ordinaryTyping[]{\Xi_a}{\Delta}{\mu_a(l_a)}{\tau_{clos}}$, where $\tau_{clos} \in \{\tau_{fn}, \tau_{tbl}\}$, then $\mu_{a1}(l_a) = \mu_a(l_a)$. Similarly for any $l_b \in \dom{\mu_b}$ such that $\ordinaryTyping[]{\Xi_b}{\Delta}{\mu_b(l_b)}{\tau_{clos}}$, where $\tau_{clos} \in \{\tau_{fn}, \tau_{tbl}\}$, then $\mu_{b1}(l_b) = \mu_b(l_b)$,

for any $l_a \in \dom{\mu_a}$ and $l_b \in \dom{\mu_b}$ such that $\ordinaryTyping[]{\Xi_a}{\Delta}{\mu_a(l_a)}{\type{\tau}{\chi}}$ and $\ordinaryTyping[]{\Xi_b}{\Delta}{\mu_b(l_b)}{\type{\tau}{\chi}}$ and $pc \nsqsubseteq \chi$, we have $\mu_{a1}(l_a) = \mu_{a}(l_a)$ and $\mu_{b1}(l_b) = \mu_{b}(l_b)$,
\begin{equation}\label{stmt-ite-ih-2}
\NIval{l}{\Xi_{a1}'}{\Xi_{b1}'}{\Delta}{val_{a1}}{val_{b1}}{\type{bool}{\chi_1}}
\end{equation}
To interpret this judgement, we consider two cases for $\chi_1$:
\begin{itemize}
\item $\chi_1 \sqsubseteq l$. This implies $val_{a1} = val_{b1}$. Therefore, both the evaluations will either take \emph{true} branch or both take \emph{false} branch.
We prove the required results for the \emph{true} case; proof for the other case follows similarly.
By applying the current theorem's induction hypothesis on the well-typed $stmt_1$ that is evaluated in an initial configuration satisfying \Cref{stmt-ite-ih-1}, we conclude that given $\evalsto{\config[stmt_1]{\mathcal{C};\Delta}{\mu_{a1}}{\epsilon_{a}}}{\config{\mu_{a2}}{\epsilon_{a1}}{sig_{a}}}$ and $\evalsto{\config[stmt_1]{\mathcal{C};\Delta}{\mu_{b1}}{\epsilon_{b}}}{\config{\mu_{b2}}{\epsilon_{b1}}{sig_{b}}}$ there exists some $\Xi_{a2}'$ and $\Xi_{b2}'$, such that $\Xi_{a1}' \subseteq \Xi_{a2}'$, $\Xi_{b1}' \subseteq \Xi_{b2}'$, $\dom{\mu_{a2}} \supseteq \dom{\mu_{a1}}$, $\dom{\mu_{b2}} \supseteq \dom{\mu_{b1}}$, $\dom{\epsilon_a} \subseteq \dom{\epsilon_{a1}}$, $\dom{\epsilon_b} \subseteq \dom{\epsilon_{b1}}$, the signals satisfy the property of being of the same form (this proves the requirement in \Cref{stmt-ite-tp-3}) and
\begin{equation} \label{stmt-ite-ih-3}
\semanticBelowPCState{l}{\Xi_{a2}'}{\Xi_{b2}'}{\Delta} {\mu_{a2}}{\epsilon_{a1}}{\mu_{b2}}{\epsilon_{b1}}{\Gamma'},
\end{equation}
\begin{equation} \label{stmt-ite-ih-3-1}
\semanticBelowPCState{l}{\Xi_{a2}'}{\Xi_{b2}'}{\Delta} {\mu_{a2}}{\epsilon_{a}}{\mu_{b2}}{\epsilon_{b}}{\Gamma},
\end{equation}
We already know from above that $\Xi_a \subseteq \Xi_{a1}' \subseteq \Xi_{a2}'$, $\Xi_b \subseteq \Xi_{b1}' \subseteq \Xi_{b2}'$, $\dom{\mu_{a2}} \supseteq \dom{\mu_{a1}} \supseteq \dom{\mu_{a}}$, $\dom{\mu_{b2}} \supseteq \dom{\mu_{b1}} \supseteq \dom{\mu_{b}}$.
Therefore, the \Cref{stmt-ite-ih-3-1} proves the results needed to show \Cref{stmt-ite-tp-2}.
Applying the induction hypothesis also concludes that for any $l_a \in \dom{\mu_a}$ such that $\ordinaryTyping[]{\Xi_a}{\Delta}{\mu_a(l_a)}{\tau_{clos}}$, where $\tau_{clos} \in \{\tau_{fn}, \tau_{tbl}\}$, then $\mu_{a2}(l_a) = \mu_{a1}(l_a)= \mu_a(l_a)$. Similarly for any $l_b \in \dom{\mu_b}$ such that $\ordinaryTyping[]{\Xi_b}{\Delta}{\mu_b(l_b)}{\tau_{clos}}$, where $\tau_{clos} \in \{\tau_{fn}, \tau_{tbl}\}$, then $\mu_{b2}(l_b) = \mu_{b1}(l_b) = \mu_b(l_b)$.
This proves the result needed to show \Cref{stmt-ite-tp-4}.

Applying the induction hypothesis also gives us that for any $l_a \in \dom{\mu_{a1}}$ and $l_b \in \dom{\mu_{b1}}$ such that $\ordinaryTyping[]{\Xi_{a1}'}{\Delta}{\mu_{a1}(l_a)}{\type{\tau}{\chi}}$ and $\ordinaryTyping[]{\Xi_{b1}'}{\Delta}{\mu_{b1}(l_b)}{\type{\tau}{\chi}}$ and $pc \nsqsubseteq \chi$, we have $\mu_{a1}(l_a) = \mu_{a2}(l_a)$ and $\mu_{b1}(l_b) = \mu_{b2}(l_b)$.  As $\dom{\mu_{a}} \subseteq \dom{\mu_{a1}} \subseteq \dom{\mu_{a2}}$, this proves the result needed to show \Cref{stmt-ite-tp-5}.

\item $\chi_1 \nsqsubseteq l$. In this case the conditional guards might differ causing different branches to be taken. However, $\chi_1 \nsqsubseteq l$ implies $\chi' \nsqsubseteq l$. Since we know that $stmt_1$ and $stmt_2$ are well-typed at $\chi'$, which means store locations at $\chi' \nsqsubseteq \chi$ remain unchanged across $\mu_{a1}$ and $\mu_{a2}$, and $\mu_{b1}$ and $\mu_{b2}$. This implies locations at $\chi \sqsubseteq l$ remain unchanged. Therefore, we can conclude from \Cref{stmt-ite-ih-1} that
\begin{equation}
\semanticBelowPCState{l}{\Xi_{a2}'}{\Xi_{b2}'}{\Delta}{\mu_{a2}}{\epsilon_{a}}{\mu_{b2}}{\epsilon_{b}}{\Gamma}
\end{equation}
$stmt_1$ and $stmt_2$ are well-typed at $pc=\chi'$. Since $\chi' \nsqsubseteq l$ and $\bot \sqsubseteq l$, we know that $\chi' \nsqsubseteq \bot$. This implies that $\terminal{return}$ and $\terminal{exit}$ statements cannot be in these statement block because these two statements are well typed at the $pc=\bot$ only. Therefore, only $sig$ that can be returned in these statement blocks are $cont$. With this we prove that the final $sig$ are of the same kind.
\end{itemize}

\item \textbf{\textsc{T-Seq-1}} \label{seq}
The last rule in the typing derivation of a block of statements will be:
\[ \inferrule*[right=T-Seq]
{
\Gamma, \Delta \vdash_{pc} stmt_{1} \dashv \Gamma_{1} \qquad
\Gamma_{1}, \Delta \vdash_{pc} \{ \overline{stmt_{2}} \} \dashv \Gamma_{2}
}
{
\Gamma, \Delta \vdash_{pc} \{ stmt_{1}; \overline{stmt_{2}} \} \dashv \Gamma_2
}
\]

Given the above typing judgement holds for the statement, $ \{ stmt_{1}; \overline{stmt_{2}} \}$, we need to show that for any $\Xi_a$, $\Xi_b$, $\mu_{a}$, $\mu_{b}$, $\epsilon_{a}$, $\epsilon_{b}$, $\mu_a'$, $\mu_b'$, $\epsilon_a'$, $\epsilon_b'$ satisfying
\begin{equation} \label{stmt-blk-hyp}
\semanticBelowPCState{l}{\Xi_a}{\Xi_b}{\Delta}{\mu_{a}}{\epsilon_{a}}{\mu_{b}}{\epsilon_{b}}{\Gamma}
\end{equation}
If the statement, $ \{ stmt_{1}; \overline{stmt_{2}} \}$ is evaluated under two different initial configurations $\langle \mu_a, \epsilon_a \rangle$ and $\langle \mu_b, \epsilon_b \rangle$, then there exists some $\Xi_a'$ and $\Xi_b'$, such that the following hold:
\begin{enumerate}
\item $\declTyping{pc}{\Gamma}{\Delta}{\{ stmt_{1}, \overline{stmt_{2}} \}}{\Gamma'}{}$. This is already the theorem's hypothesis.
\item \label{stmt-blk-tp-2} We have $\Xi_a \subseteq \Xi_a'$, $\Xi_b \subseteq
  \Xi_b'$, $\dom{\mu_a} \subseteq \dom{\mu_a'}$, $\dom{\mu_b} \subseteq
  \dom{\mu_b'}$, $\dom{\epsilon_a} \subseteq \dom{\epsilon_a'}$, and
  $\dom{\epsilon_b} \subseteq \dom{\epsilon_b'}$,
and $\semanticBelowPCState{l}{\Xi_a'}{\Xi_b'}{\Delta} {\mu_{a}'}{\epsilon_{a}'}{\mu_{b}'}{\epsilon_{b}'}{\Gamma'}$. In this case $\mu_a' = \mu_{a2}$, $\mu_b' = \mu_{b2}$, $\epsilon_a' = \epsilon_{a2}$, $\epsilon_b' = \epsilon_{b2}$.
We also need to show that $\semanticBelowPCState{l}{\Xi_a'}{\Xi_b'}{\Delta} {\mu_{a}'}{\epsilon_{a}}{\mu_{b}'}{\epsilon_{b}}{\Gamma}$.
\item \label{stmt-blk-tp-4} For any $l_a \in \dom{\mu_a}$ such that $\ordinaryTyping[]{\Xi_a}{\Delta}{\mu_a(l_a)}{\tau_{clos}}$, where $\tau_{clos} \in \{\tau_{fn}, \tau_{tbl}\}$, then $\mu_a'(l_a) = \mu_a(l_a)$. Similarly for any $l_b \in \dom{\mu_b}$ such that $\ordinaryTyping[]{\Xi_b}{\Delta}{\mu_b(l_b)}{\tau_{clos}}$, where $\tau_{clos} \in \{\tau_{fn}, \tau_{tbl}\}$, then $\mu_b'(l_b) = \mu_b(l_b)$.
\item For any $l_a' \in \dom{\mu_{a}}$ and $l_b' \in \dom{\mu_{b}}$ such that $\ordinaryTyping[]{\Xi_{a}}{\Delta}{\mu_{a}(l_a')}{\type{\tau}{\chi}}$ and $\ordinaryTyping[]{\Xi_{b}}{\Delta}{\mu_{b}(l_b')}{\type{\tau}{\chi}}$ and $pc \nsqsubseteq \chi$, we have $\mu_{a}(l_a') = \mu_{a}'(l_a')$ and $\mu_{b}(l_b') = \mu_{b}'(l_b')$,
\item \label{stmt-blk-tp-3} $sig$ in any two evaluations are of the same form.
\end{enumerate}
There are three cases for this evaluation: involving return statement, exit statement, or ordinary statements. We explain the ordinary statements case in detail, and the other two follow similarly.

\begin{mathpar}
\inferrule*[]
{
         \langle \mathcal{C}, \Delta, \mu_{a}, \epsilon_{a}, stmt_{1} \rangle \Downarrow  \langle \mu_{a1}, \epsilon_{a1}, cont \rangle \\
         \langle \mathcal{C}, \Delta, \mu_{a1}, \epsilon_{a1}, \{\overline{stmt_{2}}\} \rangle \Downarrow   \langle \mu_{a2}, \epsilon_{a2}, sig_a \rangle
}
{
 \langle \mathcal{C}, \Delta, \mu_{a}, \epsilon_{a}, \{ stmt_{1}, \overline{stmt_{2}} \} \rangle  \Downarrow  \langle \mu_{a2}, \epsilon_{a2}, sig_a \rangle
}

\inferrule*[]
{
         \langle \mathcal{C}, \Delta, \mu_{b}, \epsilon_{b}, stmt_{1} \rangle \Downarrow  \langle \mu_{b1}, \epsilon_{b1}, cont \rangle \\
         \langle \mathcal{C}, \Delta, \mu_{b1}, \epsilon_{b1}, \{\overline{stmt_{2}}\} \rangle \Downarrow   \langle \mu_{b2}, \epsilon_{b2}, sig_b \rangle
}
{
 \langle \mathcal{C}, \Delta, \mu_{b}, \epsilon_{b}, \{ stmt_{1}, \overline{stmt_{2}} \} \rangle  \Downarrow  \langle \mu_{b2}, \epsilon_{b2}, sig_b \rangle
}
\end{mathpar}

In the following part, we prove the last three requirements.
Since $stmt_1$ is evaluated in an initial configuration satisfying \Cref{stmt-blk-hyp}, by applying induction hypothesis on the typing derivation of $stmt_1$, we conclude that given $\evalsto{\config[stmt_1]{\mathcal{C};\Delta}{\mu_{a}}{\epsilon_{a}}}{\config{\mu_{a1}}{\epsilon_{a1}}{cont}}$ and $\evalsto{\config[stmt_1]{\mathcal{C};\Delta}{\mu_{b}}{\epsilon_{b}}}{\config{\mu_{b1}}{\epsilon_{b1}}{cont}}$ there exists some $\Xi_{a1}'$ and $\Xi_{b1}'$, such that $\Xi_{a}' \subseteq \Xi_{a1}'$, $\Xi_{b}' \subseteq \Xi_{b1}'$, $\dom{\mu_{a1}} \supseteq \dom{\mu_{a}}$, $\dom{\mu_{b1}} \supseteq \dom{\mu_{b}}$, $\dom{\epsilon_a} \subseteq \dom{\epsilon_{a1}}$, $\dom{\epsilon_b} \subseteq \dom{\epsilon_{b1}}$, the signals satisfy the property of being of the same form (in both case it is $cont$) and
\begin{equation} \label{stmt-blk-ih-1-3}
\semanticBelowPCState{l}{\Xi_{a1}'}{\Xi_{b1}'}{\Delta} {\mu_{a1}}{\epsilon_{a1}}{\mu_{b1}}{\epsilon_{b1}}{\Gamma_1},
\end{equation}
\begin{equation} \label{stmt-blk-ih-1-4}
\semanticBelowPCState{l}{\Xi_{a1}'}{\Xi_{b1}'}{\Delta} {\mu_{a1}}{\epsilon_{a}}{\mu_{b1}}{\epsilon_{b}}{\Gamma},
\end{equation}
and for any $l_a \in \dom{\mu_a}$ such that $\ordinaryTyping[]{\Xi_a}{\Delta}{\mu_a(l_a)}{\tau_{clos}}$, where $\tau_{clos} \in \{\tau_{fn}, \tau_{tbl}\}$, then $\mu_{a1}(l_a) = \mu_a(l_a)$. Similarly for any $l_b \in \dom{\mu_b}$ such that $\ordinaryTyping[]{\Xi_b}{\Delta}{\mu_b(l_b)}{\tau_{clos}}$, where $\tau_{clos} \in \{\tau_{fn}, \tau_{tbl}\}$, then $\mu_{b1}(l_b) = \mu_b(l_b)$.

$\overline{stmt_2}$ is a sequence of statements, so we apply induction hypothesis repeatedly on each statement and conclude that the final states after evaluation of the sequence of statements $\evalsto{\config[\overline{stmt_2}]{\mathcal{C};\Delta}{\mu_{a1}}{\epsilon_{a1}}}{\config{\mu_{a2}}{\epsilon_{a2}}{sig_{a}}}$ and $\evalsto{\config[\overline{stmt_2}]{\mathcal{C};\Delta}{\mu_{b1}}{\epsilon_{b1}}}{\config{\mu_{b2}}{\epsilon_{b2}}{sig_{b}}}$ there exists some $\Xi_{a2}'$ and $\Xi_{b2}'$, such that $\Xi_{a1}' \subseteq \Xi_{a2}'$, $\Xi_{b1}' \subseteq \Xi_{b2}'$, $\dom{\mu_{a2}} \supseteq \dom{\mu_{a1}}$, $\dom{\mu_{b2}} \supseteq \dom{\mu_{b1}}$, $\dom{\epsilon_{a1}} \subseteq \dom{\epsilon_{a2}}$, $\dom{\epsilon_{b1}} \subseteq \dom{\epsilon_{b2}}$, the signals satisfy the property of being of the same form (this proves the requirement in \Cref{stmt-blk-tp-3}) and
\begin{equation}
\semanticBelowPCState{l}{\Xi_{a2}'}{\Xi_{b2}'}{\Delta} {\mu_{a2}}{\epsilon_{a2}}{\mu_{b2}}{\epsilon_{b2}}{\Gamma_2}
\end{equation}
\begin{equation} \label{stmt-blk-ih-2-4}
\semanticBelowPCState{l}{\Xi_{a2}'}{\Xi_{b2}'}{\Delta} {\mu_{a2}}{\epsilon_{a1}}{\mu_{b2}}{\epsilon_{b1}}{\Gamma_1}
\end{equation}
and for any $l_a \in \dom{\mu_a}$ such that $\ordinaryTyping[]{\Xi_a}{\Delta}{\mu_a(l_a)}{\tau_{clos}}$, where $\tau_{clos} \in \{\tau_{fn}, \tau_{tbl}\}$, then $\mu_{a2}(l_a) = \mu_{a1}(l_a)= \mu_a(l_a)$. Similarly for any $l_b \in \dom{\mu_b}$ such that $\ordinaryTyping[]{\Xi_b}{\Delta}{\mu_b(l_b)}{\tau_{clos}}$, where $\tau_{clos} \in \{\tau_{fn}, \tau_{tbl}\}$, then $\mu_{b2}(l_b) = \mu_{b1}(l_b)= \mu_b(l_b)$.
This proves the result needed to show \Cref{stmt-ite-tp-4}.
Since we know that $\dom{\epsilon_a} \subseteq \dom{\epsilon_{a1}}$, any $x \in \dom{\epsilon_a}$ will also be in $\dom{\epsilon_{a1}}$. Similarly for $\epsilon_b$. There can be two cases due to shadowing a variable name:
\begin{enumerate}
\item $\epsilon_{a}(x) = \epsilon_{a1}(x)$, $\epsilon_{b}(x) = \epsilon_{b1}(x)$. In this case, $\mu_{a2}(\epsilon_{a}(x)) = \mu_{a2}(\epsilon_{a1}(x))$ and $\mu_{b2}(\epsilon_{b}(x)) = \mu_{b2}(\epsilon_{b1}(x))$. We know that these variables satisfy non-interference in $\mu_{a2}$ and $\mu_{b2}$ from \Cref{stmt-blk-ih-2-4}.
\item $\epsilon_{a}(x) \ne \epsilon_{a1}(x)$ and $\epsilon_{b}(x) \ne \epsilon_{b1}(x)$.
\begin{enumerate}
  \item \label{unused-vars-seq} If $\textsc{unused}(\langle\mu_{a1}, \epsilon_{a1} \rangle, x, \epsilon_{a}(x))$ and $\textsc{unused}(\langle\mu_{b1}, \epsilon_{b1} \rangle, x, \epsilon_{b}(x))$, then $\mu_{a2}(\epsilon_a(x)) = \mu_{a1}(\epsilon_{a}(x))$ and $\mu_{b2}(\epsilon_b(x)) = \mu_{b1}(\epsilon_{b}(x))$, which we know are non-interfering from \Cref{stmt-blk-ih-1-4}.
  \item \label{used-vars-seq} If $\neg \textsc{unused}(\langle\mu_{a1}, \epsilon_{a1} \rangle, x, \epsilon_{a}(x))$ and $\neg \textsc{unused}(\langle\mu_{b1}, \epsilon_{b1} \rangle, x, \epsilon_{b}(x))$, then there exists some closure value with environment $\epsilon_{c_a}'$ and $\dom{\epsilon_{c_a}'} \subseteq \dom{\epsilon_{a1}}$, and $\epsilon_{c_b}'$ and $\dom{\epsilon_{c_b}'} \subseteq \dom{\epsilon_{b1}}$  where $x \in \dom{\epsilon_{c_a}'}$and $\epsilon_a(x) = \epsilon_{c_a}'(x)$. Also, $\epsilon_b(y) = \epsilon_{c_b}'(y)$. From \Cref{stmt-blk-ih-2-4}, we know that $\semanticBelowPCState{l}{\Xi_{a2}'}{\Xi_{b2}'}{\Delta}{\mu_{a2}}{\epsilon_{c_a}'}{\mu_{b2}}{\epsilon_{c_b}'}{\Gamma_{clos}}$. This implies that this variable $x$ will have non-interfering values in $\mu_{a2}$ and $\mu_{b2}$.
\end{enumerate}
\end{enumerate}
We also need to ensure that for all $x$ in $\dom{\epsilon_a} = \dom{\epsilon_b}$
and some $\Gamma_{clos} \subseteq \Gamma$ and any $pc$,  if $\Gamma, \Delta
\vdash_{pc} x: \tau_{clos}$ with closure environments $\epsilon_{c_a}'$ and
$\epsilon_{c_b}'$ in the two states, then
$\semanticBelowPCState{l}{\Xi_{a2}'}{\Xi_{b2}'}{\Delta}{\mu_{a2}}{\epsilon_{c_a}'}{\mu_{b2}}{\epsilon_{c_b}'}{\Gamma_{clos}}$.
For closure variables satisfying \Cref{used-vars-seq}, this will follow from closure properties in \Cref{stmt-blk-ih-2-4}.
For variables ratifying   \Cref{unused-vars-seq}, this will follow from the fact that the variables in their closure environments can again be \textsc{unused} (implies unchanged between $\mu_{a1}$ and $\mu_{a2}$, $\mu_{b1}$ and $\mu_{b2}$) or \textsc{used} (in this case we already know from \Cref{stmt-blk-ih-2-4} that such variables satisfy non-interference of values).
By combining the observation that all variables in $\epsilon_a$ and $\epsilon_b$ are non-interfering, we can conclude
\[
\semanticBelowPCState{l}{\Xi_{a2}'}{\Xi_{b2}'}{\Delta}{\mu_{a2}}{\epsilon_{a}}{\mu_{b2}}{\epsilon_{b}}{\Gamma}
\]
This proves \Cref{stmt-blk-tp-2}.

For reference, the evaluation rules for the other two cases are as follows:
\begin{mathpar}
\inferrule*[]
{
         \langle \mathcal{C}, \Delta, \mu_{a}, \epsilon_{a}, stmt_{1} \rangle \Downarrow  \langle \mu_{a1}, \epsilon_{a1}, return~val_{a} \rangle
}
{
 \langle \mathcal{C}, \Delta, \mu_{a}, \epsilon_{a}, \{ stmt_{1}, \overline{stmt_{2}} \} \rangle  \Downarrow  \langle \mu_{a1}, \epsilon_{a1}, return~val_{a} \rangle
}

\inferrule*[]
{
         \langle \mathcal{C}, \Delta, \mu_{b}, \epsilon_{b}, stmt_{1} \rangle \Downarrow  \langle \mu_{b1}, \epsilon_{b1}, return~val_{b} \rangle
}
{
 \langle \mathcal{C}, \Delta, \mu_{b}, \epsilon_{b}, \{ stmt_{1}, \overline{stmt_{2}} \} \rangle  \Downarrow  \langle \mu_{b1}, \epsilon_{b1}, return~val_{b} \rangle
}
\end{mathpar}

\begin{mathpar}
\inferrule*[]
{
         \langle \mathcal{C}, \Delta, \mu_{a}, \epsilon_{a}, stmt_{1} \rangle \Downarrow  \langle \mu_{a1}, \epsilon_{a1}, return~val_{a} \rangle
}
{
 \langle \mathcal{C}, \Delta, \mu_{a}, \epsilon_{a}, \{ stmt_{1}, \overline{stmt_{2}} \} \rangle  \Downarrow  \langle \mu_{a1}, \epsilon_{a1}, return~val_{a} \rangle
}

\inferrule*[]
{
         \langle \mathcal{C}, \Delta, \mu_{b}, \epsilon_{b}, stmt_{1} \rangle \Downarrow  \langle \mu_{b1}, \epsilon_{b1}, exit \rangle
}
{
 \langle \mathcal{C}, \Delta, \mu_{b}, \epsilon_{b}, \{ stmt_{1}, \overline{stmt_{2}} \} \rangle  \Downarrow  \langle \mu_{b1}, \epsilon_{b1}, exit \rangle
}
\end{mathpar}

\item \textbf{\textsc{T-Return}}
The last rule in the typing derivation of a return will be:
\[
\inferrule*[right=T-Return]
{
\Gamma,\Delta \vdash_{pc}  exp: \langle \tau, \chi_{ret} \rangle \\
\Gamma(\terminal{return}) = \type{\tau_{ret}}{\chi_{ret}} \\
\Delta    \vdash \tau_{ret} \rightsquigarrow \tau
}
{
\Gamma, \Delta \vdash_{pc} \terminal{return}~ exp \dashv \Gamma
}
\]
Given the above typing judgement holds for, $return~exp$, we need to show that for any $\Xi_a$, $\Xi_b$, $\mu_{a}$, $\mu_{b}$, $\epsilon_{a}$, $\epsilon_{b}$, $\mu_a'$, $\mu_b'$, $\epsilon_a'$, $\epsilon_b'$ satisfying
\begin{equation} \label{stmt-ret-hyp}
\semanticBelowPCState{l}{\Xi_a}{\Xi_b}{\Delta}{\mu_{a}}{\epsilon_{a}}{\mu_{b}}{\epsilon_{b}}{\Gamma}
\end{equation}
if the statement, $return~exp$ is evaluated under two different initial configurations $\langle \mu_a, \epsilon_a \rangle$ and $\langle \mu_b, \epsilon_b \rangle$ as follows:

\begin{mathpar}
\inferrule*[]
{
     \langle \mathcal{C}, \Delta, \mu_{a}, \epsilon_{a}, exp \rangle  \Downarrow  \langle \mu_{a1}, val_{a} \rangle
}
{
 \langle \mathcal{C}, \Delta, \mu_{a}, \epsilon_{a}, \texttt{return exp} \rangle  \Downarrow  \langle \mu_{a1}, \epsilon_{a}, \texttt{return } val_{a}\rangle
}

\inferrule*[]
{
     \langle \mathcal{C}, \Delta, \mu_{b}, \epsilon_{b}, exp \rangle  \Downarrow  \langle \mu_{b1}, val_{b} \rangle
}
{
 \langle \mathcal{C}, \Delta, \mu_{b}, \epsilon_{b}, \texttt{return exp} \rangle  \Downarrow  \langle \mu_{b1}, \epsilon_{b}, \texttt{return } val_{b}\rangle
}
\end{mathpar}
Then there exists some $\Xi_a'$ and $\Xi_b'$, such that the following hold:
\begin{enumerate}
\item $\stmtTyping{pc}{\Gamma}{\Delta}{\terminal{return}~exp}{\Gamma'}{}$, where $\Gamma' = \Gamma$. This is already the theorem's hypothesis.
\item \label{stmt-ret-tp-2} $\Xi_a \subseteq \Xi_a'$, $\Xi_b \subseteq \Xi_b'$, $\dom{\mu_a} \subseteq \dom{\mu_a'}$, $\dom{\mu_b} \subseteq \dom{\mu_b'}$, $\dom{\epsilon_a} \subseteq \dom{\epsilon_a'}$, $\dom{\epsilon_b} \subseteq \dom{\epsilon_b'}$, and $\semanticBelowPCState{l}{\Xi_a'}{\Xi_b'}{\Delta} {\mu_{a}'}{\epsilon_{a}'}{\mu_{b}'}{\epsilon_{b}'}{\Gamma'}$. In this case $\mu_a' = \mu_{a1}$, $\mu_b' = \mu_{b1}$, $\epsilon_a' = \epsilon_a$, $\epsilon_b' = \epsilon_b$.
\item \label{stmt-ret-tp-4} For any $l_a \in \dom{\mu_a}$ such that $\ordinaryTyping[]{\Xi_a}{\Delta}{\mu_a(l_a)}{\tau_{clos}}$, where $\tau_{clos} \in \{\tau_{fn}, \tau_{tbl}\}$, then $\mu_a'(l_a) = \mu_a(l_a)$. Similarly for any $l_b \in \dom{\mu_b}$ such that $\ordinaryTyping[]{\Xi_b}{\Delta}{\mu_b(l_b)}{\tau_{clos}}$, where $\tau_{clos} \in \{\tau_{fn}, \tau_{tbl}\}$, then $\mu_b'(l_b) = \mu_b(l_b)$.
\item \label{stmt-ret-tp-3} $sig$ in any two evaluations are of the same form.
\item \label{stmt-ret-tp-5} For any $l_a' \in \dom{\mu_{a}}$ and $l_b' \in \dom{\mu_{b}}$ such that $\ordinaryTyping[]{\Xi_{a}}{\Delta}{\mu_{a}(l_a')}{\type{\tau}{\chi}}$ and $\ordinaryTyping[]{\Xi_{b}}{\Delta}{\mu_{b}(l_b')}{\type{\tau}{\chi}}$ and $pc \nsqsubseteq \chi$, we have $\mu_{a}(l_a') = \mu_{a}'(l_a')$ and $\mu_{b}(l_b') = \mu_{b}'(l_b')$,
\end{enumerate}
Since $exp$ is evaluated in an initial configuration satisfying
\Cref{stmt-ret-hyp}, by applying induction hypothesis of \Cref{ni-exp} on the
typing derivation of $exp$, we conclude that there exists some $\Xi_{a}'$,
$\Xi_{b}'$, $\mu_{a1}$, and $\mu_{b1}$ such that $\Xi_a \subseteq \Xi_{a}'$,
$\Xi_b \subseteq \Xi_{b}'$, $\dom{\mu_a} \subseteq \dom{\mu_{a1}}$, $\dom{\mu_b}
\subseteq \dom{\mu_{b1}}$ and the following holds:
\begin{equation} \label{stmt-ret-ih-1-1}
\semanticBelowPCState{l}{\Xi_{a}'}{\Xi_{b}'}{\Delta} {\mu_{a1}}{\epsilon_{a}}{\mu_{b1}}{\epsilon_{b}}{\Gamma},
\end{equation}
For any $l_a \in \dom{\mu_a}$ such that $\ordinaryTyping[]{\Xi_a}{\Delta}{\mu_a(l_a)}{\tau_{clos}}$, where $\tau_{clos} \in \{\tau_{fn}, \tau_{tbl}\}$, then $\mu_{a1}(l_a) = \mu_a(l_a)$. Similarly for any $l_b \in \dom{\mu_b}$ such that $\ordinaryTyping[]{\Xi_b}{\Delta}{\mu_b(l_b)}{\tau_{clos}}$, where $\tau_{clos} \in \{\tau_{fn}, \tau_{tbl}\}$, then $\mu_{b1}(l_b) = \mu_b(l_b)$.
This proves \Cref{stmt-ret-tp-2} and \Cref{stmt-ret-tp-4}.
Also, for any $l_a' \in \dom{\mu_{a}}$ and $l_b' \in \dom{\mu_{b}}$ such that $\ordinaryTyping[]{\Xi_{a}}{\Delta}{\mu_{a}(l_a')}{\type{\tau}{\chi}}$ and $\ordinaryTyping[]{\Xi_{b}}{\Delta}{\mu_{b}(l_b')}{\type{\tau}{\chi}}$ and $pc \nsqsubseteq \chi$, we have $\mu_{a}(l_a') = \mu_{a1}(l_a')$ and $\mu_{b}(l_b') = \mu_{b1}(l_b')$. This proves \Cref{stmt-ret-tp-5}.
The above applying of the induction hypothesis also shows
\begin{equation}\label{stmt-ret-ih-1-2}
\NIval{l}{\Xi_a'}{\Xi_b'}{\Delta}{val_{a}}{val_b}{\type{\tau_{ret}}{\chi_{ret}}}.
\end{equation}
Since the signal in this case is of the form $ret~val$, we need to show that
\begin{equation*} \label{stmt-ret-sig}
\Xi_a', \Xi_b', \Delta \models_{pc} NI(val_{a}, val_b): \type{\tau_{ret}}{\chi_{ret}}
\end{equation*}
This is already given by \Cref{stmt-ret-ih-1-2}.

\item \textbf{\textsc{T-Assign}}
The last rule in the typing derivation of an assignment statement will be:
\[
\inferrule*[right=T-Assign]
{
\Gamma,\Delta \vdash_{pc} exp_{1}:\type{\tau}{\chi_{1}}~goes~ inout \\
\Gamma,\Delta \vdash_{pc} exp_{2} :\type{\tau}{\chi_{2}}\\
\chi_{2} \sqsubseteq \chi_{1}\\
pc \sqsubseteq \chi_{1}
}
{
\Gamma, \Delta \vdash_{pc} exp_{1} := exp_{2} \dashv \Gamma
}
\]

Given the above typing judgement holds for, $ exp_{1} := exp_{2}$, we need to show that for any $\Xi_a$, $\Xi_b$, $\mu_{a}$, $\mu_{b}$, $\epsilon_{a}$, $\epsilon_{b}$, $\mu_a'$, $\mu_b'$, $\epsilon_a'$, $\epsilon_b'$ satisfying
\begin{equation} \label{stmt-assign-hyp}
\semanticBelowPCState{l}{\Xi_a}{\Xi_b}{\Delta}{\mu_{a}}{\epsilon_{a}}{\mu_{b}}{\epsilon_{b}}{\Gamma}
\end{equation}
if the statement, $ exp_{1} := exp_{2}$ is evaluated under two different initial configurations $\langle \mu_a, \epsilon_a \rangle$ and $\langle \mu_b, \epsilon_b \rangle$ as follows:

\begin{mathpar}
\inferrule*[]
{
 \langle \mathcal{C}, \Delta, \mu_{a}, \epsilon_{a}, exp_{1} \rangle \Downarrow_{lval}  \langle \mu_{a1}, lval_{a} \rangle \\
 \langle \mathcal{C}, \Delta, \mu_{a1}, \epsilon_{a}, exp_{2} \rangle \Downarrow  \langle \mu_{a2}, val_{a} \rangle \\
 \langle \mathcal{C}, \Delta, \mu_{a2}, \epsilon_{a}, lval_{a} := val_{a} \rangle  \Downarrow_{write} \mu_{a3}
}
{
 \langle \mathcal{C}, \Delta, \mu_{a}, \epsilon_{a}, \{exp_{1} := exp_{2}\} \rangle  \Downarrow  \langle \mu_{a3}, \epsilon_{a}, cont \rangle
}

\inferrule*[]
{
 \langle \mathcal{C}, \Delta, \mu_{b}, \epsilon_{b}, exp_{1} \rangle \Downarrow_{lval}  \langle \mu_{b1}, lval_{b} \rangle \\
 \langle \mathcal{C}, \Delta, \mu_{b1}, \epsilon_{b}, exp_{2} \rangle \Downarrow  \langle \mu_{b2}, val_{b} \rangle \\
 \langle \mathcal{C}, \Delta, \mu_{b2}, \epsilon_{b}, lval_{b} := val_{b} \rangle  \Downarrow_{write} \mu_{b3}
}
{
 \langle \mathcal{C}, \Delta, \mu_{a}, \epsilon_{a}, \{exp_{1} := exp_{2}\} \rangle  \Downarrow  \langle \mu_{a3}, \epsilon_{a}, cont \rangle
}
\end{mathpar}

Then there exists some $\Xi_a'$ and $\Xi_b'$, such that the following hold:
\begin{enumerate}
\item $\declTyping{pc}{\Gamma}{\Delta}{exp_{1} := exp_{2}}{\Gamma'}{}$. This is already the theorem's hypothesis.
\item \label{stmt-assign-tp-2} We have $\Xi_a \subseteq \Xi_a'$, $\Xi_b
  \subseteq \Xi_b'$, $\dom{\mu_a} \subseteq \dom{\mu_a'}$, $\dom{\mu_b}
  \subseteq \dom{\mu_b'}$, $\dom{\epsilon_a} \subseteq \dom{\epsilon_a'}$, and
  $\dom{\epsilon_b} \subseteq \dom{\epsilon_b'}$,
and $\semanticBelowPCState{l}{\Xi_a'}{\Xi_b'}{\Delta} {\mu_{a}'}{\epsilon_{a}'}{\mu_{b}'}{\epsilon_{b}'}{\Gamma'}$. In this case $\mu_a' = \mu_{a2}$, $\mu_b' = \mu_{b2}$, $\epsilon_a' = \epsilon_a$, $\epsilon_b' = \epsilon_b$.
\item \label{stmt-assign-tp-4} For any $l_a \in \dom{\mu_a}$ such that $\ordinaryTyping[]{\Xi_a}{\Delta}{\mu_a(l_a)}{\tau_{clos}}$, where $\tau_{clos} \in \{\tau_{fn}, \tau_{tbl}\}$, then $\mu_a'(l_a) = \mu_a(l_a)$. Similarly for any $l_b \in \dom{\mu_b}$ such that $\ordinaryTyping[]{\Xi_b}{\Delta}{\mu_b(l_b)}{\tau_{clos}}$, where $\tau_{clos} \in \{\tau_{fn}, \tau_{tbl}\}$, then $\mu_b'(l_b) = \mu_b(l_b)$.
\item \label{stmt-assign-tp-5}  For any $l_a' \in \dom{\mu_{a}}$ and $l_b' \in \dom{\mu_{b}}$ such that $\ordinaryTyping[]{\Xi_{a}}{\Delta}{\mu_{a}(l_a')}{\type{\tau}{\chi}}$ and $\ordinaryTyping[]{\Xi_{b}}{\Delta}{\mu_{b}(l_b')}{\type{\tau}{\chi}}$ and $pc \nsqsubseteq \chi$, we have $\mu_{a}(l_a') = \mu_{a}'(l_a')$ and $\mu_{b}(l_b') = \mu_{b}'(l_b')$,
\item \label{stmt-assign-tp-3} $sig$ in any two evaluations are of the same form. In this case $sig_1 = cont = sig_2$.
\end{enumerate}

By applying \Cref{lem-lval-eval} on $exp_1$, which is evaluated in an initial configuration satisfying \Cref{stmt-assign-hyp}, we conclude:

There exists some $\Xi_{a1}$, $\Xi_{b1}$, $\mu_{a1}$ and $\mu_{b1}$ satisfying $\Xi_a \subseteq \Xi_{a1}$, $\Xi_b \subseteq \Xi_{b1}$, $\dom{\mu_a} \subseteq \dom{\mu_{a1}}$ and $\dom{\mu_b} \subseteq \dom{\mu_{b1}}$ and the following:
\begin{equation} \label{stmt-assign-ih-1-1}
\semanticBelowPCState{l}{\Xi_{a1}}{\Xi_{b1}}{\Delta}{\mu_{a1}}{\epsilon_{a}}{\mu_{b1}}{\epsilon_{b}}{\Gamma}
\end{equation}
For any $l_a \in \dom{\mu_a}$ such that $\ordinaryTyping[]{\Xi_a}{\Delta}{\mu_a(l_a)}{\tau_{clos}}$, where $\tau_{clos} \in \{\tau_{fn}, \tau_{tbl}\}$, then $\mu_{a1}(l_a) = \mu_a(l_a)$. Similarly for any $l_b \in \dom{\mu_b}$ such that $\ordinaryTyping[]{\Xi_b}{\Delta}{\mu_b(l_b)}{\tau_{clos}}$, where $\tau_{clos} \in \{\tau_{fn}, \tau_{tbl}\}$, then $\mu_{b1}(l_b) = \mu_b(l_b)$.

Also, for any $l_a' \in \dom{\mu_{a}}$ and $l_b' \in \dom{\mu_{b}}$ such that $\ordinaryTyping[]{\Xi_{a}}{\Delta}{\mu_{a}(l_a')}{\type{\tau}{\chi}}$ and $\ordinaryTyping[]{\Xi_{b}}{\Delta}{\mu_{b}(l_b')}{\type{\tau}{\chi}}$ and $pc \nsqsubseteq \chi$, we have $\mu_{a}(l_a') = \mu_{a1}(l_a')$ and $\mu_{b}(l_b') = \mu_{b1}(l_b')$,

Also, if $\chi_1 \sqsubseteq l$, then $\lvalEquality{lval_a}{lval_b}$. Also, $\textsc{lval_base}(lval_a) \in \dom{\epsilon_a}$ and $\textsc{lval_base}(lval_b) \in \dom{\epsilon_b}$.

By applying induction \Cref{ni-exp} on $exp_2$, which is evaluated in an initial configuration satisfying \Cref{stmt-assign-ih-1-1}, we can conclude:

There exists some $\Xi_{a2}$, $\Xi_{b2}$, $\mu_{a2}$, and $\mu_{b2}$ such that $\Xi_{a1} \subseteq \Xi_{a2}$, $\Xi_{b1} \subseteq \Xi_{b2}$, $\dom{\mu_{a1}} \subseteq \dom{\mu_{a2}}$, $\dom{\mu_{b1}} \subseteq \dom{\mu_{b2}}$ and the following hold:
\begin{equation}
\semanticBelowPCState{l}{\Xi_{a2}}{\Xi_{b2}}{\Delta} {\mu_{a2}}{\epsilon_{a}}{\mu_{b2}}{\epsilon_{b}}{\Gamma},
\end{equation}
\begin{equation}\label{stmt-assign-in-2-2}
\NIval{l}{\Xi_{a2}}{\Xi_{b2}}{\Delta}{val_{a}}{val_b}{\type{\tau}{\chi_{2}}}
\end{equation}

Using \Cref{lem:lvalue-write-final} on l-value write in expressions $lval_a:= val_a$ and $lval_b  := val_b$, we get that
\begin{equation}
\semanticBelowPCState{l}{\Xi_{a2}}{\Xi_{b2}}{\Delta} {\mu_{a3}}{\epsilon_{a}}{\mu_{b3}}{\epsilon_{b}}{\Gamma},
\end{equation}
Since $\Xi_a \subseteq \Xi_{a1} \subseteq \Xi_{a2}$ and $\Xi_b \subseteq \Xi_{b1} \subseteq \Xi_{b2}$, showing the above equation is same as showing \Cref{stmt-assign-tp-2}.
Observe that the $lval_a$ and $lval_b$ have security level $pc \sqsubseteq \chi_1 $, and \Cref{lem:lvalue-write-final} states that only the location given by $\epsilon_a(\textsc{lval_base}(lval_a))$ is updated in the $\mu_{a3}$ and similarly $\mu_{b3}$. Therefore, we have proved \Cref{stmt-assign-tp-5}. Proof of \Cref{stmt-assign-tp-4} follows similarly from the results of applying the above induction hypothesis.
\item \textbf{\textsc{T-VarDecl}}
A well-formed declaration statement will satisfy the following typing rule:\\
			    \[ \inferrule*[]
			        {
			            \Gamma; \Delta \vdash_{pc} \text{var_decl} \dashv \Gamma'; \Delta_1
			        }
			        {
			            \Gamma; \Delta \vdash_{pc} \text{var_decl} \dashv \Gamma'
			        }
			    \]

			    \[  \inferrule*[]
			        {
			                  \langle \mathcal{C}, \Delta, \mu_{a}, \epsilon_{a}, \text{var_decl} \rangle \Downarrow   \langle \Delta_1, \mu_{a}', \epsilon_{a}', cont \rangle
			        }
			        {
			             \langle \mathcal{C}, \Delta, \mu_{a}, \epsilon_{a}, \text{var_decl} \rangle  \Downarrow  \langle \mu_{a}', \epsilon_{a}', cont \rangle
			        }
			    \]

			    \[  \inferrule*[]
			        {
			                  \langle \mathcal{C}, \Delta, \mu_{b}, \epsilon_{b}, \text{var_decl} \rangle \Downarrow   \langle \Delta_1, \mu_{b}', \epsilon_{b}', cont \rangle
			        }
			        {
			             \langle \mathcal{C}, \Delta, \mu_{b}, \epsilon_{b}, \text{var_decl} \rangle  \Downarrow  \langle \mu_{b}', \epsilon_{b}', cont \rangle
			        }
			    \]

			The proof of this case follows from applying the induction hypothesis for NI for declarations. In case of \text{var_decl} $\Delta_1 = \Delta$.
\item \textbf{\textsc{T-TblCall}}
\[
	\inferrule*[right=T-TblCall]
{
\Gamma, \Delta \vdash_{pc} exp: \langle table(pc_{tbl}), \bot \rangle \\
 pc \sqsubseteq pc_{tbl}
}
{
\Gamma, \Delta \vdash_{pc} exp() \dashv \Gamma
}
\]

Given the above typing judgement holds for, $exp()$ statement, we need to show that for any $\Xi_a$, $\Xi_b$, $\mu_{a}$, $\mu_{b}$, $\epsilon_{a}$, $\epsilon_{b}$, $\mu_a'$, $\mu_b'$, $\epsilon_a'$, $\epsilon_b'$ satisfying

\begin{equation} \label{tbl-apply-hyp}
\semanticBelowPCState{l}{\Xi_a}{\Xi_b}{\Delta}{\mu_{a}}{\epsilon_{a}}{\mu_{b}}{\epsilon_{b}}{\Gamma}
\end{equation}
If the statement, $exp()$ is evaluated under two different initial configurations $\langle \mu_a, \epsilon_a \rangle$ and $\langle \mu_b, \epsilon_b \rangle$ as follows,
\[
	\inferrule*[]
{
\langle \mathcal{C}, \Delta, \mu_a, \epsilon_a, exp \rangle \Downarrow \langle \mu_{a1}, table~l_{a}(\epsilon_{c_a}, \overline{exp_{k}: x_k}, \overline{act_{a_j}(\overline{exp_{a_j i}}, \overline{y_{c}: \type{\tau_c}{\chi_c}})})\rangle \\
\langle	\mathcal{C}, \Delta, \mu_{a1}, \epsilon_{c}, \overline{exp_k}\rangle \Downarrow \langle \mu_{a2}, \overline{val_{ka}}\rangle \\
\langle \mathcal{C}, l_{a}, \overline{val_{ka}: x_k}, act_{a_j}(\overline{y_{c_{ji}}: \type{\tau_{c_{ji}}}{\chi_{c_{ji}}}})\rangle \Downarrow_{match} \langle act_{a_j}(\overline{exp_{c_{ji}}})\rangle		           \\
\langle \mathcal{C}, \Delta, \mu_{a2}, \epsilon_{c_a}, act_{aj}(\overline{exp_{a_{ji}}}, \overline{exp_{c_{ji}}})\rangle \Downarrow \langle \mu_{a3}, \epsilon_{c_a}', cont\rangle
}
{
\langle \mathcal{C}, \Delta, \mu_{a}, \epsilon_{a}, exp()\rangle \Downarrow \langle \mu_{a3}, \epsilon_{a}, cont \rangle
}
\]

\[
	\inferrule*[]
{
\langle \mathcal{C}, \Delta, \mu_b, \epsilon_b, exp \rangle \Downarrow \langle \mu_{b1}, table~l_{b}(\epsilon_{c_b}, \overline{exp_{k}: x_k}, \overline{act_{a_j}(\overline{exp_{a_{ji}}}, \overline{y_{c_j}: \type{\tau_{c_j}}{\chi_{c_j}}})})\rangle \\
\langle \mathcal{C}, \Delta, \mu_{b1}, \epsilon_{c_b}, \overline{exp_k}\rangle \Downarrow \langle \mu_{b2}, \overline{val_{kb}}\rangle \\
\langle \mathcal{C}, l_{b}, \overline{val_{kb}: x_k}, act_{a_j}(\overline{y_{c_{ji}}: \type{\tau_{c_{ji}}}{\chi_{c_{ji}}}})\rangle \Downarrow_{match} \langle act_{a_j}(\overline{exp_{c_{ji}}})\rangle		           \\
\langle \mathcal{C}, \Delta, \mu_{b2}, \epsilon_{c_b}, act_{aj'}(\overline{exp_{a_{ji'}}}, \overline{exp_{c_{ji'}}})\rangle \Downarrow \langle \mu_{b3}, \epsilon_{c_b}', cont\rangle
}
{
\langle \mathcal{C}, \Delta, \mu_{b}, \epsilon_{b}, exp()\rangle \Downarrow \langle \mu_{b3}, \epsilon_{b}, cont \rangle
}
\]
Then there exists some $\Xi_a'$ and $\Xi_b'$, such that the following hold:
\begin{enumerate}
\item $\stmtTyping{pc}{\Gamma}{\Delta}{exp()}{\Gamma'}{}$, where $\Gamma' = \Gamma$. This is already the theorem's hypothesis.
\item \label{tbl-apply-tp-2} We have $\Xi_a \subseteq \Xi_a'$, $\Xi_b \subseteq \Xi_b'$, $\dom{\mu_a} \subseteq \dom{\mu_a'}$, $\dom{\mu_b} \subseteq \dom{\mu_b'}$, $\dom{\epsilon_a} \subseteq \dom{\epsilon_a'}$, and $\dom{\epsilon_b} \subseteq \dom{\epsilon_b'}$
and $\semanticBelowPCState{l}{\Xi_a'}{\Xi_b'}{\Delta} {\mu_{a}'}{\epsilon_{a}'}{\mu_{b}'}{\epsilon_{b}'}{\Gamma'}$. In this case $\epsilon_a' = \epsilon_a$, $\epsilon_b' = \epsilon_b$, $\mu_a' = \mu_{a3}$ and $\mu_b' = \mu_{b3}$.
\item For any $l_a \in \dom{\mu_a}$ and $l_b \in \dom{\mu_b}$ such that $\ordinaryTyping[]{\Xi_a}{\Delta}{\mu_a(l_a)}{\type{\tau}{\chi}}$ and $\ordinaryTyping[]{\Xi_b}{\Delta}{\mu_b(l_b)}{\type{\tau}{\chi}}$ and $pc \nsqsubseteq \chi$, we have $\mu_{a}'(l_a) = \mu_{a}(l_a)$ and $\mu_{b}'(l_b) = \mu_{b}(l_b)$,
\item \label{tbl-apply-tp-4} For any $l_a \in \dom{\mu_a}$ such that $\ordinaryTyping[]{\Xi_a}{\Delta}{\mu_a(l_a)}{\tau_{clos}}$, where $\tau_{clos} \in \{\tau_{fn}, \tau_{tbl}\}$, then $\mu_a'(l_a) = \mu_a(l_a)$. Similarly for any $l_b \in \dom{\mu_b}$ such that $\ordinaryTyping[]{\Xi_b}{\Delta}{\mu_b(l_b)}{\tau_{clos}}$, where $\tau_{clos} \in \{\tau_{fn}, \tau_{tbl}\}$, then $\mu_b'(l_b) = \mu_b(l_b)$.
\item \label{tbl-apply-tp-3} $sig$ in any two evaluations are of the same form. In this case $sig_1 = cont = sig_2$.
\end{enumerate}

To show that the final state satisfies \Cref{tbl-apply-tp-2} we start by showing that final state after evaluating all the sub-step in the table evaluation satisfies \Cref{tbl-apply-tp-2}.
\paragraph*{Evaluating table expression}
By applying induction hypothesis of \Cref{ni-exp} on the well-typed $exp$, we get \[\NIexp{pc}{\Gamma}{\Delta}{exp}{\type{table(pc_{tbl})}{\bot}}.\] Since $exp$ is evaluated in an initial configuration satisfying \Cref{tbl-apply-hyp}, we can expand the NI for expression definition to conclude that
there exists some $\Xi_{a1}$, $\Xi_{b1}$, satisfying $\Xi_a \subseteq \Xi_{a1}$, $\Xi_b \subseteq \Xi_{b1}$, $\dom{\mu_a} \subseteq \dom{\mu_{a1}}$ and $\dom{\mu_b} \subseteq \dom{\mu_{b1}}$ and the following:
\begin{equation}\label{tbl-apply-ih-1-1}
\semanticBelowPCState{l}{\Xi_{a1}}{\Xi_{b1}}{\Delta}{\mu_{a1}}{\epsilon_{a1}}{\mu_{b1}}{\epsilon_{b1}}{\Gamma}~,
\end{equation}
\begin{equation}\label{tbl-apply-ih-1-2}
\NIval{l}{\Xi_{a1}}{\Xi_{b1}}{\Delta}{val_{a}}{val_{b}}{\type{table(pc_{tbl})}{\bot}}
\end{equation}
where
\[
  val_a = table~l_a~(\epsilon_{a}, \overline{exp_k: x_k}, \overline{act_{a_j}(\overline{exp_{a_{ji}}}, \overline{y_{c_{ji}}: \type{\tau_{c_{ji}}}{\chi_{c_{ji}}}})})
\]
and
\[
  val_b = table~l_b~(\epsilon_{b}, \overline{exp_k: x_k}, \overline{act_{b_j}(\overline{exp_{a_{ji}}}, \overline{y_{c_{ji}}: \type{\tau_{c_{ji}}}{\chi_{c_{ji}}}})}) .
\]

\Cref{tbl-apply-ih-1-2} expands to give $\NItbl{pc}{\Xi_{a1}}{\Xi_{b1}}{\Delta}{val_{a}}{val_b}{\type{table(pc_{tbl})}{\bot}}$, which implies that there exists a $\Gamma_{tbl}$ and $pc_a$ such that
\begin{enumerate}
\item $\Xi_{a1} \models \epsilon_{c_a}: \Gamma_{tbl}$ and $\Xi_{b1} \models \epsilon_{c_b}: \Gamma_{tbl}$
\item Well-typed. $\Gamma_{tbl}; \Delta \vdash_{pc} table~l_a~(\epsilon_{a}, \overline{exp_k: x_k}, \overline{act_{a_j}(\overline{exp_{a_{ji}}}, \overline{y_{c_{ji}}: \type{\tau_{c_{ji}}}{\chi_{c_{ji}}}})}): \type{table({pc_a},{pc_{tbl}})}{\bot}$.
Similarly, we have $\Gamma_{tbl}; \Delta \vdash_{pc} table~l_b~(\epsilon_{b}, \overline{exp_k: x_k}, \overline{act_{b_j}(\overline{exp_{a_{ji}}}, \overline{y_{c_{ji}}: \type{\tau_{c_{ji}}}{\chi_{c_{ji}}}})}): \type{table({pc_a},{pc_{tbl}})}{\bot}$
\item $\ordinaryTyping[pc_{tbl}]{\Gamma_{tbl}}{\Delta}{x_k}{\type{match\_kind}{\bot}}$ for each $x_k \in \overline{x_k}$
\item $\ordinaryTyping[pc_{tbl}]{\Gamma_{tbl}}{\Delta}{exp_k}{\type{\tau_k}{\chi_k}}$ for each $exp_k \in \overline{exp_k}$
\item $\ordinaryTyping[pc_{tbl}]{\Gamma_{tbl}}{\Delta}{act_{aj}}{\type{\overline{d\type{\tau_{a_{ji}}}{\chi_{a_{ji}}}}~;\overline{\type{\tau_{c_{ji}}}{\chi_{c_{ji}}}} \xrightarrow{pc_{fn_j}} \type{unit}{\bot}}{\bot}}$ for each $act_{a_j} \in \overline{act_{a_j}}$
\item $\ordinaryTyping[pc_{tbl}]{\Gamma_{tbl}}{\Delta}{exp_{a_{ji}}}{\type{\tau_{a_{ji}}}{\chi_{a_{ji}}}~goes~d}$ for each $exp_{a_{ji}} \in \overline{exp_{a_{ji}}}$
\item $val_a =_{tbl} val_b$.
\item $\chi_k \sqsubseteq {pc_{fn_j}}$, for all $j,k$
\item $pc_a \sqsubseteq pc_{fn_j}$, for all $j$
\item ${\chi_k} \sqsubseteq {pc_{tbl}}~\text{for all}~k$.
\item $pc_{tbl} \sqsubseteq pc_a$
\end{enumerate}
From \Cref{tbl-apply-hyp}, we already know that for all $x$ in $\dom{\epsilon}$ and some $\Gamma_{tbl}\subseteq \Gamma$,  if $\Gamma, \Delta \vdash_{pc} x: \tau_{tbl}$,  $\mu(\epsilon(x)) = table~l~(\epsilon_{c},...)$, and $\Xi \models \epsilon_c: \Gamma_{tbl}$, then $\dom{\epsilon_{c}} \subseteq \dom{\epsilon}$ and $\semanticStoreEnv{\Xi}{\Delta}{\mu}{\epsilon_c}{\Gamma_{tbl}}$.

This implies that $\dom{\epsilon_{c_a}} \subseteq \dom{\epsilon_{a}}$ and $\dom{\epsilon_{c_b}} \subseteq \dom{\epsilon_{b}}$. Also, $\semanticStoreEnv{\Xi_{a1}}{\Delta}{\mu_{a1}}{\epsilon_{c_a}}{\Gamma_{tbl}}$ and $\semanticStoreEnv{\Xi_{b1}}{\Delta}{\mu_{b1}}{\epsilon_{c_b}}{\Gamma_{tbl}}$.
Since closure values do not change across $\mu_a$, $\mu_{a1}$, and $\mu_b$, $\mu_{b1}$, the variable that would have evaluated to the table closure value under $\mu_a$ will have the same value under $\mu_{a1}$. By using the property of closures implied by \Cref{tbl-apply-ih-1-1}, we conclude $\semanticBelowPCState{l}{\Xi_{a}}{\Xi_{b}}{\Delta}{\mu_{a1}}{\epsilon_{c_a}}{\mu_{b1}}{\epsilon_{c_b}}{\Gamma_{tbl}}$.
\paragraph*{Evaluating key expression}
By repeatedly applying the induction hypothesis of \Cref{ni-exp} on  $\ordinaryTyping[pc_{tbl}]{\Gamma_{tbl}}{\Delta}{exp_k}{\type{\tau_k}{\chi_k}}$ for each $exp_k \in \overline{exp_k}$, implies that there exists some $\Xi_{a2}$, $\Xi_{b2}$, $\mu_{a2}$ and $\mu_{b2}$ satisfying $\Xi_{a1} \subseteq \Xi_{a2}$, $\Xi_{b1} \subseteq \Xi_{b2}$, $\mu_{a1} \subseteq \mu_{a2}$ and $\mu_{b1} \subseteq \mu_{b2}$ and the following:
\begin{equation}\label{tbl-apply-ih-2-1}
\semanticBelowPCState{l}{\Xi_{a2}}{\Xi_{b2}}{\Delta}{\mu_{a2}}{\epsilon_{c_a}}{\mu_{b2}}{\epsilon_{c_b}}{\Gamma_{tbl}}~,
\end{equation}
\begin{equation}\label{tbl-apply-ih-2-2}
\NIval{l}{\Xi_{a2}}{\Xi_{b2}}{\Delta}{\overline{val_{ka}}}{\overline{val_{kb}}}{\overline{\type{\tau_k}{\chi_k}}}
\end{equation}
This can be read as ``if $\chi_k \sqsubseteq l$ then $val_{ka} = val_{kb}$''.

Also, none of the variables at security label $pc \nsqsubseteq \chi$ are updated between $\mu_{a1}$, $\mu_{a2}$, and $\mu_{b1}$, $\mu_{b2}$.
Similar to the argument used in function call case to prove \Cref{eq-4}, we can also conclude
\begin{equation}\label{tbl-apply-ih-2-3}
\semanticBelowPCState{l}{\Xi_{a2}}{\Xi_{b2}}{\Delta}{\mu_{a2}}{\epsilon_{a}}{\mu_{b2}}{\epsilon_{b}}{\Gamma}
\end{equation}
\paragraph*{Table match}
The $\Downarrow_{match}$ depends on some assumption about the control plane, $\mathcal{C}$ that it will ensure that only well-typed arguments, $\ordinaryTyping[pc_{tbl}]{\Gamma_{tbl}}{\Delta}{exp_{c_{ji}}}{\type{\tau_{c_{ji}}}{\chi_{c_{ji}}}}$ for each $exp_{c_{ji}} \in \overline{exp_{c_{ji}}}$ are passed to partially-applied actions (this is same as Petr4's assumption around the control plane). In addition, considering that the table entries are fixed, matching on a equal $\overline{val_{ka} = val_{kb}}$ will return the same action and arguments, \ie the matched action will be the same $act_{aj} = act_{aj'}$ and $exp_{c_{ji}} = exp_{c_{ji}'}$.
$exp_k$ at security-level $\chi_k \nsqsubseteq l$ might not evaluate to equal values. Therefore, we have two cases for the \textit{match} evaluation, either same actions, $act_{aj} = act_{aj'}$,  with same parameter expressions, $\overline{exp_{c_{ji}} = exp_{c_{ji}'}}$ are returned or $act_{aj} \ne act_{aj'}$ and their parameter expression can also differ.
\paragraph*{Invoking the matched action}
In case $act_{aj} = act_{aj'}$, $\overline{exp_{c_{ji}} = exp_{c_{ji}'}}$ and $\overline{exp_{a_{ji}} = exp_{a_{ji}'}}$, then the last premise of the evaluation rule is equivalent to evaluating a function expression with same parameter expression. By using induction hypothesis of \Cref{ni-stmt} for a well-typed function call statement, we arrive at a final state involving $\Xi_{a3}$, $\Xi_{b3}$, $\mu_{a3}$, $\mu_{b3}$, $\epsilon_{c_a}'$, $\epsilon_{c_b}'$ satisfying $\Xi_{a2} \subseteq \Xi_{a3}$, $\Xi_{b3} \subseteq \Xi_{b3}$, $\dom{\mu_{a2}} \subseteq \dom{\mu_{a3}}$ and $\dom{\mu_{b2}} \subseteq \dom{\mu_{b3}}$, $\dom{\epsilon_{c_a}} \subseteq \dom{\epsilon_{c_a}'}$, and $\dom{\epsilon_{c_b}} \subseteq \dom{\epsilon_{c_b}'}$ and the following:
\begin{equation}\label{tbl-apply-ih-3-1}
\semanticBelowPCState{l}{\Xi_{a3}}{\Xi_{b3}}{\Delta}{\mu_{a3}}{\epsilon_{c_a}'}{\mu_{b3}}{\epsilon_{c_b}'}{\Gamma_{tbl}}~,
\end{equation}
In case of a function call statement, $\epsilon_{c_a} = \epsilon_{c_a}'$, and $\epsilon_{c_b} = \epsilon_{c_b}'$.

Similar to the argument used in function call case to prove \Cref{eq-4}, since we have \Cref{tbl-apply-ih-2-3} we can also conclude
\begin{equation}
\semanticBelowPCState{l}{\Xi_{a3}}{\Xi_{b3}}{\Delta}{\mu_{a3}}{\epsilon_{a}}{\mu_{b3}}{\epsilon_{b}}{\Gamma}
\end{equation}
This proves \Cref{tbl-apply-tp-2}.

In case $act_{aj} \neq act_{aj'}$, $\overline{exp_{c_{ji}} \neq exp_{c_{ji}'}}$ and $\overline{exp_{a_{ji}} \neq exp_{a_{ji}'}}$, then there exists some $val_{ka} \neq val_{kb}$. This implies $\chi_{k} \nsqsubseteq l$. Since $\chi_k \sqsubseteq pc_{tbl}$, we can conclude $pc_{tbl} \nsqsubseteq l$. Although, the function call statements are different in the two cases, we know that both the function call statements are well-typed at $pc_{tbl}$. This implies that when the function call statement in $\mu_{a2}$ and $\epsilon_{c_a}$ is evaluated, then variables at $pc \nsqsubseteq \chi$ will have unchanged value in $\mu_{a3}$. Similarly, the  other function call statement despite being different guarantees that the values of variables at $pc \nsqsubseteq \chi$ in $\mu_{b2}$ and will have unchanged value in $\mu_{b3}$.
We already know that
\[\semanticBelowPCState{l}{\Xi_{a2}}{\Xi_{b2}}{\Delta}{\mu_{a2}}{\epsilon_{a}}{\mu_{b2}}{\epsilon_{b}}{\Gamma}\]
By using the fact that none of the variables at $\chi \sqsubseteq l$ are updated between $\mu_{a2}$ and $\mu_{a3}$, and similarly $\mu_{b2}$ and $\mu_{b3}$, we can conclude that
\[\semanticBelowPCState{l}{\Xi_{a3}}{\Xi_{b3}}{\Delta}{\mu_{a3}}{\epsilon_{a}}{\mu_{b3}}{\epsilon_{b}}{\Gamma}\]
A consistent state requires that any variables at $\chi \sqsubseteq l$ are indistinguishable; this holds in
\[\semanticBelowPCState{l}{\Xi_{a2}}{\Xi_{b2}}{\Delta}{\mu_{a2}}{\epsilon_{a}}{\mu_{b2}}{\epsilon_{b}}{\Gamma}\]
and with no  changes to the variables at $\chi \sqsubseteq l$, it will continue to hold in the final memory store.
\end{enumerate}

\paragraph*{Proof of \Cref{ni-decl}}
By induction on typing derivation of declaration statements.
\begin{enumerate}
\item \textbf{\textsc{T-VarDecl}}
\[\inferrule*[]
{
  \Delta \vdash \tau \rightsquigarrow \tau'
}
{
\Gamma; \Delta \vdash_{pc} \type{\tau}{\chi}~ x \dashv \Gamma [x: \langle \tau', \chi \rangle]; \Delta
}
\]

Given the above typing judgement holds for, $ \type{\tau}{\chi}~ x$, we need to show that $\Xi_a$, $\Xi_b$, $\mu_{a}$, $\mu_{b}$, $\epsilon_{a}$, $\epsilon_{b}$, $\mu_{a}'$, $\mu_{b}'$, $\epsilon_{a}'$, $\epsilon_{b}'$ satisfying
\begin{equation} \label{decl-basic-hyp}
\semanticBelowPCState{l}{\Xi_a}{\Xi_b}{\Delta}{\mu_{a}}{\epsilon_{a}}{\mu_{b}}{\epsilon_{b}}{\Gamma},
\end{equation}
if the declaration, $\type{\tau}{pc }~ x$ is evaluated under two different initial configurations $\langle \mu_a, \epsilon_a \rangle$ and $\langle \mu_b, \epsilon_b \rangle$ as follows,
\begin{mathpar}
\inferrule*[]
{l_{a}~ \text{fresh} \\
\langle \Delta, \mu_a, \epsilon_a, \tau \rangle \Downarrow_{\tau} \tau'
}
{
 \langle \mathcal{C}, \Delta_{}, \mu_{a}, \epsilon_{a}, \type{\tau}{\chi}~ x \rangle  \Downarrow    \langle \Delta_{}, \mu_{a}[l_a:= init_{\Delta} \tau'], \epsilon_{a}[x \mapsto l_a], cont \rangle
}

\inferrule*[]
{
l_{b}~ \text{fresh}  \\
\langle \Delta, \mu_b, \epsilon_b, \tau \rangle \Downarrow_{\tau} \tau'
}
{
 \langle \mathcal{C}, \Delta_{}, \mu_{b}, \epsilon_{b}, \type{\tau}{\chi}~ x \rangle \Downarrow   \langle \Delta_{}, \mu_{b}[l_b:= init_{\Delta} \tau'], \epsilon_{b}[x \mapsto l_b], cont \rangle
}
\end{mathpar}
then there exists $\Xi_a'$, $\Xi_b'$ such that
\begin{enumerate}
\item $\declTyping{pc}{\Gamma}{\Delta}{\type{\tau}{\chi}~ x}{\Gamma'}{\Delta}$. This is already the hypothesis of the theorem.
\item \label{decl-basic-tp-1} $\semanticBelowPCState{l}{\Xi_a'}{\Xi_b'}{\Delta} {\mu_{a}'}{\epsilon_{a}'}{\mu_{b}'}{\epsilon_{b}'}{\Gamma'}$, where $\mu_a' = \mu_{a}[l_a:= init_{\Delta} \tau']$, $\mu_b' = \mu_{b}[l_b:= init_{\Delta} \tau']$, $\epsilon_a' = \epsilon_{a}[x \mapsto l_a]$, $\epsilon_b' = \epsilon_{b}[x \mapsto l_b]$. Also, $\semanticBelowPCState{l}{\Xi_a'}{\Xi_b'}{\Delta} {\mu_{a}'}{\epsilon_{a}}{\mu_{b}'}{\epsilon_{b}}{\Gamma}$,
\item \label{decl-basic-tp-2} For any $l_a \in \dom{\mu_a}$ such that $\ordinaryTyping[]{\Xi_a}{\Delta}{\mu_a(l_a)}{\tau_{clos}}$, where $\tau_{clos} \in \{\tau_{fn}, \tau_{tbl}\}$, then $\mu_a'(l_a) = \mu_a(l_a)$. Similarly for any $l_b \in \dom{\mu_b}$ such that $\ordinaryTyping[]{\Xi_b}{\Delta}{\mu_b(l_b)}{\tau_{clos}}$, where $\tau_{clos} \in \{\tau_{fn}, \tau_{tbl}\}$, then $\mu_b'(l_b) = \mu_b(l_b)$,
\item \label{decl-basic-tp-3} $\Xi_a \subseteq \Xi_a'$, $\Xi_b \subseteq \Xi_b'$, $\dom{\mu_a} \subseteq \dom{\mu_a'}$, $\dom{\mu_b} \subseteq \dom{\mu_b'}$, $\dom{\epsilon_a} \subseteq \dom{\epsilon_a'}$, and $\dom{\epsilon_b} \subseteq \dom{\epsilon_b'}$.
\item \label{decl-basic-tp-4}For any $l_a \in \dom{\mu_a}$ and $l_b \in \dom{\mu_b}$ such that $\ordinaryTyping[]{\Xi_a}{\Delta}{\mu_a(l_a)}{\type{\tau}{\chi}}$ and $\ordinaryTyping[]{\Xi_b}{\Delta}{\mu_b(l_b)}{\type{\tau}{\chi}}$ and $pc \nsqsubseteq \chi$, we have $\mu_{a}'(l_a) = \mu_{a}(l_a)$ and $\mu_{b}'(l_b) = \mu_{b}(l_b)$.
\end{enumerate}
With $\Xi_a' = \Xi_a \cup \{l_a \mapsto \type{\tau'}{\chi}\}$,
$\Xi_b' = \Xi_b \cup \{l_b \mapsto \type{\tau'}{\chi}\}$
the equation in \Cref{decl-basic-tp-3} is evident.
To show \Cref{decl-basic-tp-1}, we need to show the following:
\begin{equation} \label{decl-basic-tp-1-1}
\semanticStoreEnv{\Xi_a'}{\Delta}{\mu_{a}'}{\epsilon_{a}'}{\Gamma'}
\end{equation}
\begin{equation} \label{decl-basic-tp-1-2}
\semanticStoreEnv{\Xi_b'}{\Delta}{\mu_{b}'}{\epsilon_{b}'}{\Gamma'}
\end{equation}
\begin{equation} \label{decl-basic-tp-1-3}
\dom{\epsilon_a'} = \dom{\epsilon_b'}
\end{equation}
\begin{equation} \label{decl-basic-tp-1-4}
\text{For any ~} x \in \dom{\epsilon_a'} = \dom{\epsilon_b'} \text{~we have~} \NIval{l}{\Xi_a'}{\Xi_b'}{\Delta}{\mu_a'(\epsilon_{a}'(x))}{\mu_b'(\epsilon_{b}'(x))}{\Gamma'(x)}
\end{equation}
For all $x$ in $\dom{\epsilon_a'} = \dom{\epsilon_b'}$ and some $\Gamma_{fn} \subseteq \Gamma$ and any $pc$,  if $\Gamma', \Delta \vdash_{pc} x: \tau_{fn}$, $\mu_a'(\epsilon_a'(x)) = clos(\epsilon_{c_a},...)$, $\mu_b'(\epsilon_b'(x)) = clos(\epsilon_{c_b},...)$, $\Xi_a' \models \epsilon_{c_a}: \Gamma_{fn}$, and $\Xi_b' \models \epsilon_{c_b}: \Gamma_{fn}$ , then $\semanticBelowPCState{l}{\Xi_a'}{\Xi_b'}{\Delta}{\mu_{a}'}{\epsilon_{c_a}}{\mu_{b}'}{\epsilon_{c_b}}{\Gamma_{fn}}$.

For all $x$ in $\dom{\epsilon_a'} = \dom{\epsilon_b'}$ and some $\Gamma_{tbl} \subseteq \Gamma'$ and any $pc$,  if $\Gamma', \Delta \vdash_{pc} x: \tau_{tbl}$, $\mu_a'(\epsilon_a'(x)) = table~l_a~(\epsilon_{c_a},...)$, $\mu_b'(\epsilon_b'(x)) = table~l_b~(\epsilon_{c_b},...)$, $\Xi_a' \models \epsilon_{c_a}: \Gamma_{tbl}$, and $\Xi_b' \models \epsilon_{c_b}: \Gamma_{tbl}$ , then $\semanticBelowPCState{l}{\Xi_a'}{\Xi_b'}{\Delta}{\mu_{a}'}{\epsilon_{c_a}}{\mu_{b}'}{\epsilon_{c_b}}{\Gamma_{tbl}}$.

\Cref{decl-basic-tp-1-3} is evident from the definitions of $\epsilon_a'$, and $\epsilon_b'$ and the given fact that $\dom{\epsilon_a} = \dom{\epsilon_b}$.
First, we begin by showing \Cref{decl-basic-tp-1-1}. This requires us to in turn prove the following:
\begin{enumerate}
\item $\semanticStore{\Xi_a'}{\Delta}{\mu_a'}$. This is shown in \Cref{lem:store-typing-extended}.
\item $\Xi_a' \vdash \epsilon_a': \Gamma'$.
We are given $\Xi_a \vdash \epsilon_a: \Gamma$.
Using \Cref{lem-env-subtyping}, we can say that $\Xi_a' \vdash \epsilon_a: \Gamma$.
Since $\Gamma' = \Gamma[x \mapsto \type{\tau'}{\chi}]$, $\epsilon_a' = \epsilon_a[x \mapsto l_a]$, $\Xi_a' = \Xi_a \cup \{l_a \mapsto \type{\tau'}{\chi}\}$, by using the typing judgements for $\Xi \vdash \epsilon: \Gamma$, we can show that
\[\inferrule*[]
{	\Xi_a' \vdash \epsilon_a: \Gamma \qquad
\Xi_a'(l_a) = \type{\tau'}{\chi}
}
{
\Xi_a' \vdash \epsilon_a[x \mapsto l_a]: \Gamma[x \mapsto \type{\tau'}{\chi}]
}
\]
This gives us the proof for $\Xi_a' \vdash \epsilon_a': \Gamma'$.

\item For all $x$ in $\dom{\epsilon_a'}$ and some $\Gamma_{fn} \subseteq \Gamma'$ and any $pc$,  if $\Gamma', \Delta \vdash_{pc} x: \tau_{fn}$, $\mu_a'(\epsilon_a'(x)) = clos(\epsilon_{c},...)$, and $\Xi_a' \models \epsilon_c: \Gamma_{fn}$, then $\dom{\epsilon_{c}} \subseteq \dom{\epsilon_a'}$ and $\semanticStoreEnv{\Xi_a'}{\Delta}{\mu_a'}{\epsilon_{c}}{\Gamma_{fn}}$. Here $\tau_{fn}$ is the function type. We elide the full view of the closures in this definition. Observe that the function closure variables in $\dom{\epsilon_a'}$ are variables that were also in $\dom{\epsilon_a}$ and are not shadowed by the new declaration $x$. We already know for such closure variables that $\dom{\epsilon_{c}} \subseteq \dom{\epsilon_a}$. Therefore, we can conclude that $\dom{\epsilon_{c}} \subseteq \dom{\epsilon_a'}$. Also, we know that $\semanticStoreEnv{\Xi_a}{\Delta}{\mu_a}{\epsilon_{c}}{\Gamma_{fn}}$. Using the proof in \Cref{lem:only-mem-superset} we can conclude that $\semanticStoreEnv{\Xi_a'}{\Delta}{\mu_a'}{\epsilon_{c}}{\Gamma_{fn}}$.
\item For all $x$ in $\dom{\epsilon_a'}$ and some $\Gamma_{tbl} \subseteq \Gamma'$ and any $pc$,  if $\Gamma', \Delta \vdash_{pc} x: \tau_{tbl}$,  $\mu_a'(\epsilon_a'(x)) = table~l~(\epsilon_{c},...)$, and $\Xi_a' \models \epsilon_c: \Gamma_{tbl}$, then $\dom{\epsilon_{c}} \subseteq \dom{\epsilon_a'}$ and $\semanticStoreEnv{\Xi_a'}{\Delta}{\mu_a'}{\epsilon_c}{\Gamma_{tbl}}$. Here $\tau_{tbl}$ is the table type. Proof for this is similar to the function closures case.
\end{enumerate}
Proof of \Cref{decl-basic-tp-1-2} follows similarly.

To show \Cref{decl-basic-tp-1-4}, we again use the fact that any $y \in \dom{\epsilon_a'}$ will be either in $\dom{\epsilon_a}$ or be the new variable.
The new variable already satisfies $\NIval{l}{\Xi_a'}{\Xi_b'}{\Delta}{\mu_a'(\epsilon_{a}'(x))}{\mu_b'(\epsilon_{b}'(x))}{\Gamma'(x)}$, since the value is $init_{\Delta} \tau'$ which is not a function closure.
For the other case where $y \in \dom{\epsilon_a}$, we already know that $\text{for any ~} y \in \dom{\epsilon_a} = \dom{\epsilon_b} \text{~we have~} \NIval{l}{\Xi_a}{\Xi_b}{\Delta}{\mu_a(\epsilon_{a}(x))}{\mu_b(\epsilon_{b}(x))}{\Gamma(x)}$. This concludes $\NIval{l}{\Xi_a'}{\Xi_b'}{\Delta}{\mu_a'(\epsilon_{a}'(x))}{\mu_b'(\epsilon_{b};(x))}{\Gamma'(x)}$ because for variables in $\epsilon_a$ not equal to this new variable the memory store remains unchanged.

The last requirement is to prove that for all $x$ in $\dom{\epsilon_a'} = \dom{\epsilon_b'}$ and some $\Gamma_{fn} \subseteq \Gamma$ and any $pc$,  if $\Gamma', \Delta \vdash_{pc} x: \tau_{fn}$, $\mu_a'(\epsilon_a'(x)) = clos(\epsilon_{c_a},...)$, $\mu_b'(\epsilon_b'(x)) = clos(\epsilon_{c_b},...)$, $\Xi_a' \models \epsilon_{c_a}: \Gamma_{fn}$, and $\Xi_b' \models \epsilon_{c_b}: \Gamma_{fn}$ , then $\semanticBelowPCState{l}{\Xi_a'}{\Xi_b'}{\Delta}{\mu_{a}'}{\epsilon_{c_a}}{\mu_{b}'}{\epsilon_{c_b}}{\Gamma_{fn}}$.  This holds true because we already know that these closures satisfied $\semanticBelowPCState{l}{\Xi_a}{\Xi_b}{\Delta}{\mu_{a}}{\epsilon_{c_a}}{\mu_{b}}{\epsilon_{c_b}}{\Gamma_{fn}}$ and because the memory stores haven't changed for any of the locations in $\dom{\mu_a}$ or $\dom{\mu_b}$, we can conclude that $\semanticBelowPCState{l}{\Xi_a'}{\Xi_b'}{\Delta}{\mu_{a}'}{\epsilon_{c_a}}{\mu_{b}'}{\epsilon_{c_b}}{\Gamma_{fn}}$ is also true. Similarly, we can show this for table closures as well.

This proves that $\semanticBelowPCState{l}{\Xi_a'}{\Xi_b'}{\Delta} {\mu_{a}'}{\epsilon_{a}'}{\mu_{b}'}{\epsilon_{b}'}{\Gamma'}$, where $\mu_a' = \mu_{a}[l_a:= init_{\Delta} \tau']$, $\mu_b' = \mu_{b}[l_b:= init_{\Delta} \tau']$, $\epsilon_a' = \epsilon_{a}[x \mapsto l_a]$, $\epsilon_b' = \epsilon_{b}[x \mapsto l_b]$.
We also have $\semanticBelowPCState{l}{\Xi_a'}{\Xi_b'}{\Delta} {\mu_{a}'}{\epsilon_{a}}{\mu_{b}'}{\epsilon_{b}}{\Gamma}$, which follows from the definition (using induction in a similar manner as shown in \Cref{lem:only-mem-superset}).
\Cref{decl-basic-tp-2} and \Cref{decl-basic-tp-4} is satisfied because only the value of the new variable $x$ is updated across $\mu_a$ and $\mu_a'$ and similarly $\mu_b$ and $\mu_b'$.

\item \textbf{\textsc{T-VarInit}}
\[\inferrule*[]
{
\ordinaryTyping[pc]{\Gamma}{\Delta}{exp}{\type{\tau}{\chi}} \\
\Delta \vdash \tau' \rightsquigarrow \tau
}
{
\Gamma; \Delta \vdash_{pc} \type{\tau'}{\chi}~ x:= exp \dashv \Gamma [x: \langle \tau, \chi\rangle]; \Delta
}
\]

Given the above typing judgement holds for $\type{\tau}{\chi}~ x:= exp$ declaration, we need to show that for any
$\Xi_a$, $\Xi_b$, $\mu_{a}$, $\mu_{b}$, $\epsilon_{a}$, $\epsilon_{b}$, $\mu_{a}'$, $\mu_{b}'$, $\epsilon_{a}'$, $\epsilon_{b}'$ satisfying
\begin{equation} \label{decl-init-hyp}
\semanticBelowPCState{l}{\Xi_a}{\Xi_b}{\Delta}{\mu_{a}}{\epsilon_{a}}{\mu_{b}}{\epsilon_{b}}{\Gamma}
\end{equation}
if the declaration, $\type{\tau}{\chi}~ x:= exp$ is evaluated under two different initial configurations $\langle \mu_a, \epsilon_a \rangle$ and $\langle \mu_b, \epsilon_b \rangle$ as follows,

\begin{mathpar}
\inferrule*[]
{
 \langle \mathcal{C}, \Delta_{ }, \mu_{a}, \epsilon_{a}, exp \rangle \Downarrow \langle \mu_{a1}, val_{a}\rangle \\
l_{a}~ \text{fresh}
}
{
 \langle \mathcal{C}, \Delta_{ }, \mu_{a}, \epsilon_{a}, \tau~x\coloneq exp \rangle \Downarrow   \langle \Delta_{ }, \mu_{a1}[l_a:=val_a], \epsilon_{a}[x \mapsto l_a], cont \rangle
}

\inferrule*[]
{
 \langle \mathcal{C}, \Delta_{ }, \mu_{b}, \epsilon_{b}, exp \rangle \Downarrow \langle \mu_{b1}, val_{b}\rangle \\
l_{b}~ \text{fresh}
}
{
 \langle \mathcal{C}, \Delta_{ }, \mu_{b}, \epsilon_{b}, \type{\tau}{\chi}~x\coloneq exp \rangle \Downarrow   \langle \Delta_{ }, \mu_{b1}[l_b:=val_b], \epsilon_{b}[x \mapsto l_b], cont \rangle
}
\end{mathpar}
then there exists some $\Xi_a'$, $\Xi_b'$ such that
\begin{enumerate}
\item $\declTyping{pc}{\Gamma}{\Delta}{\type{\tau}{\chi}~x\coloneq exp}{\Gamma'}{\Delta}$. This is already the hypothesis of the theorem.
\item \label{decl-int-tp-1} $\semanticBelowPCState{l}{\Xi_a'}{\Xi_b'}{\Delta} {\mu_{a}'}{\epsilon_{a}'}{\mu_{b}'}{\epsilon_{b}'}{\Gamma'}$, where $\mu_a' = \mu_{a1}[l_a:=val_a]$, $\mu_b' =  \mu_{b1}[l_b:=val_b]$, $\epsilon_a' = \epsilon_{a}[x \mapsto l_a]$, and $\epsilon_b' = \epsilon_{b}[x \mapsto l_b]$. Also, $\semanticBelowPCState{l}{\Xi_a'}{\Xi_b'}{\Delta} {\mu_{a}'}{\epsilon_{a}}{\mu_{b}'}{\epsilon_{b}}{\Gamma}$,
\item \label{decl-int-tp-2} For any $l_a \in \dom{\mu_a}$ such that $\ordinaryTyping[]{\Xi_a}{\Delta}{\mu_a(l_a)}{\tau_{clos}}$, where $\tau_{clos} \in \{\tau_{fn}, \tau_{tbl}\}$, then $\mu_a'(l_a) = \mu_a(l_a)$. Similarly for any $l_b \in \dom{\mu_b}$ such that $\ordinaryTyping[]{\Xi_b}{\Delta}{\mu_b(l_b)}{\tau_{clos}}$, where $\tau_{clos} \in \{\tau_{fn}, \tau_{tbl}\}$, then $\mu_b'(l_b) = \mu_b(l_b)$,
\item \label{decl-int-tp-3} $\Xi_a \subseteq \Xi_a'$, $\Xi_b \subseteq \Xi_b'$, $\dom{\mu_a} \subseteq \dom{\mu_a'}$, $\dom{\mu_b} \subseteq \dom{\mu_b'}$, $\dom{\epsilon_a} \subseteq \dom{\epsilon_a'}$, and $\dom{\epsilon_b} \subseteq \dom{\epsilon_b'}$.
\item For any $l_a \in \dom{\mu_a}$ and $l_b \in \dom{\mu_b}$ such that $\ordinaryTyping[]{\Xi_a}{\Delta}{\mu_a(l_a)}{\type{\tau}{\chi}}$ and $\ordinaryTyping[]{\Xi_b}{\Delta}{\mu_b(l_b)}{\type{\tau}{\chi}}$ and $pc \nsqsubseteq \chi$, we have $\mu_{a}'(l_a) = \mu_{a}(l_a)$ and $\mu_{b}'(l_b) = \mu_{b}(l_b)$.
\end{enumerate}
We can show the last three requirements similar to the previous case. In this case, we additionally know using the induction hypothesis of \Cref{ni-exp} that $exp$ evaluates to values that satisfy NI for values.

\item \textbf{\textsc{T-FuncDecl}}
\[
\inferrule*[right=T-FuncDecl]
{
\Gamma_1 = \Gamma[\overline{x_{i}: \type{\tau_{i}'}{\chi_{i}}}, \terminal{return}: \type{\tau_{ret}'}{\chi_{ret}}]\\
\declTyping{pc_{fn}}{\Gamma_{1}}{\Delta}{stmt}{\Gamma_{2}} \\\\
\Delta \vdash \tau_i \rightsquigarrow \tau_i'~\text{for each}~\tau_i \\
\Delta \vdash \tau_{ret} \rightsquigarrow \tau_{ret}' \\
\Gamma' = \Gamma[x: \type{\overline{d~\type{\tau_{i}'}{\chi_{i} }} \xrightarrow[]{pc_{fn}} \type{\tau_{ret}'}{\chi_{ret}}}{\bot}]
}
{
\declTyping{pc}{\Gamma}{\Delta}{\terminal{function}~\type{\tau_{ret}}{\chi_{ret}}~x~(\overline{d~ x_{i}: \type{\tau_{i}}{\chi_{i}}}) \{stmt\}}{\Gamma'}{\Delta}
}\]
Given the above typing judgement holds for function declaration, we need to show that for any
$\Xi_a$, $\Xi_b$, $\mu_{a}$, $\mu_{b}$, $\epsilon_{a}$, $\epsilon_{b}$, $\mu_{a}'$, $\mu_{b}'$, $\epsilon_{a}'$, $\epsilon_{b}'$, $\Delta$ satisfying
\begin{equation} \label{decl-fndecl-hyp}
\semanticBelowPCState{l}{\Xi_a}{\Xi_b}{\Delta}{\mu_{a}}{\epsilon_{a}}{\mu_{b}}{\epsilon_{b}}{\Gamma}
\end{equation}
if the function declaration is evaluated under two different initial configurations $\langle \mu_a, \epsilon_a \rangle$ and $\langle \mu_b, \epsilon_b \rangle$ as follows,
\begin{mathpar}
\inferrule*[]
{
l_{a}~ fresh \\
val_{a} =  clos(\epsilon_{a},  \overline{d~ x_{i}: \type{\tau_{i}'}{\chi_i}}, \type{\tau_{ret}'}{\chi_{ret}}, stmt) \\\\
\langle \Delta, \mu_a, \epsilon_a, \overline{\tau_i}\rangle \Downarrow_{\tau} \overline{\tau_i'}\\
\langle \Delta, \mu_a, \epsilon_a, \tau_{ret}\rangle \Downarrow_{\tau} \overline{\tau_{ret}'}
}
{
 \langle \mathcal{C}, \Delta, \mu_{a}, \epsilon_{a}, \textsf{function}~\type{\tau_{ret}}{\chi_{ret}}~x~(\overline{d~ x_{i}: \type{\tau_{i}}{\chi_{i}}}) \{\overline{decl} ~stmt\} \rangle \Downarrow        	 \langle \Delta, \mu_{a}[l_{a} \mapsto val_{a}], \epsilon_{a}[x \mapsto l_{a}], cont \rangle
}

\inferrule*[]
{
l_{b}~ fresh \\
val_{b} =  clos(\epsilon_{b},  \overline{d~ x_{i}: \type{\tau_{i}'}{\chi_i}}, \type{\tau_{ret}'}{\chi_{ret}},stmt)
\\\\
\langle \Delta, \mu_a, \epsilon_a, \overline{\tau_i}\rangle \Downarrow_{\tau} \overline{\tau_i'}\\
\langle \Delta, \mu_a, \epsilon_a, \tau_{ret}\rangle \Downarrow_{\tau} \overline{\tau_{ret}'}
}
{
 \langle \mathcal{C}, \Delta, \mu_{b}, \epsilon_{b}, \textsf{function}~\type{\tau_{ret}}{\chi_{ret}}~x~(\overline{d~ x_{i}: \type{\tau_{i}}{\chi_{i}}}) \{\overline{decl} ~stmt\} \rangle \Downarrow       	\langle \Delta, \mu_{b}[l_{b} \mapsto val_{b}], \epsilon_{b}[x \mapsto l_{b}], cont \rangle
}
\end{mathpar}
then there exists some $\Xi_a'$, $\Xi_b'$ such that:
\begin{enumerate}
\item $\declTyping{pc}{\Gamma}{\Delta}{function~declaration}{\Gamma'}{\Delta}$. This is already the hypothesis of the theorem.
\item \label{fndecl-tp-1} $\semanticBelowPCState{l}{\Xi_a'}{\Xi_b'}{\Delta} {\mu_{a}'}{\epsilon_{a}'}{\mu_{b}'}{\epsilon_{b}'}{\Gamma'}$ and $\semanticBelowPCState{l}{\Xi_a'}{\Xi_b'}{\Delta} {\mu_{a}'}{\epsilon_{a}}{\mu_{b}'}{\epsilon_{b}}{\Gamma}$, where $\mu_a' = \mu_{a}[l_{a} \mapsto val_{a}]$, $\mu_b' = \mu_{b}[l_{b} \mapsto val_{b}]$, $\epsilon_a' = \epsilon_{a}[x \mapsto l_{a}]$ and $\epsilon_b' = \epsilon_{b}[x \mapsto l_{b}]$.
\item \label{fndecl-tp-4}For any $l_a \in \dom{\mu_a}$ such that $\ordinaryTyping[]{\Xi_a}{\Delta}{\mu_a(l_a)}{\tau_{clos}}$, where $\tau_{clos} \in \{\tau_{fn}, \tau_{tbl}\}$, then $\mu_a'(l_a) = \mu_a(l_a)$. Similarly for any $l_b \in \dom{\mu_b}$ such that $\ordinaryTyping[]{\Xi_b}{\Delta}{\mu_b(l_b)}{\tau_{clos}}$, where $\tau_{clos} \in \{\tau_{fn}, \tau_{tbl}\}$, then $\mu_b'(l_b) = \mu_b(l_b)$.
\item  \label{fndecl-tp-3}$\Xi_a \subseteq \Xi_a'$, $\Xi_b \subseteq \Xi_b'$, $\dom{\mu_a} \subseteq \dom{\mu_a'}$, $\dom{\mu_b} \subseteq \dom{\mu_b'}$, $\dom{\epsilon_a} \subseteq \dom{\epsilon_a'}$, and $\dom{\epsilon_b} \subseteq \dom{\epsilon_b'}$, $\Delta \subseteq \Delta_1$.
\item \label{fndecl-tp-2} For any $l_a \in \dom{\mu_a}$ and $l_b \in \dom{\mu_b}$ such that $\ordinaryTyping[]{\Xi_a}{\Delta}{\mu_a(l_a)}{\type{\tau}{\chi}}$ and $\ordinaryTyping[]{\Xi_b}{\Delta}{\mu_b(l_b)}{\type{\tau}{\chi}}$ and $pc \nsqsubseteq \chi$, we have $\mu_{a}'(l_a) = \mu_{a}(l_a)$ and $\mu_{b}'(l_b) = \mu_{b}(l_b)$,
\end{enumerate}
With $\Xi_a' = \Xi_a \cup \{l_a \mapsto \tau_{fn}\}$,
$\Xi_b' = \Xi_b \cup \{l_b \mapsto \tau_{fn}\}$,
$\mu_a' = \mu_a \cup \{ l_a \mapsto val_a \}$,
$\mu_b' = \mu_b \cup \{ l_b \mapsto val_b \}$,
$\dom{\epsilon_a'} = \dom{\epsilon_a} \cup \{x\}$
$\dom{\epsilon_b'} = \dom{\epsilon_b} \cup \{x\}$ the equation in \Cref{decl-basic-tp-3} is evident.

To prove $\semanticBelowPCState{l}{\Xi_a'}{\Xi_b'}{\Delta} {\mu_{a}'}{\epsilon_{a}}{\mu_{b}'}{\epsilon_{b}}{\Gamma}$, we need to show the following:
\begin{equation} \label{fndecl-tp-0-1}
\semanticStoreEnv{\Xi_a'}{\Delta}{\mu_{a}'}{\epsilon_{a}}{\Gamma}
\end{equation}
\begin{equation} \label{fndecl-tp-0-2}
\semanticStoreEnv{\Xi_b'}{\Delta}{\mu_{b}'}{\epsilon_{b}}{\Gamma}
\end{equation}
\begin{equation} \label{fndecl-tp-0-3}
\dom{\epsilon_a} = \dom{\epsilon_b}
\end{equation}
\begin{equation} \label{fndecl-tp-0-4}
\text{For any ~} x \in \dom{\epsilon_a} = \dom{\epsilon_b} \text{~we have~} \NIval{l}{\Xi_a'}{\Xi_b'}{\Delta}{\mu_a'(\epsilon_{a}(x))}{\mu_b'(\epsilon_{b}(x))}{\Gamma(x)}
\end{equation}
and
for all $x$ in $\dom{\epsilon_a} = \dom{\epsilon_b}$ and some $\Gamma_{fn} \subseteq \Gamma$ and any $pc$,  if $\Gamma, \Delta \vdash_{pc} x: \tau_{fn}$, $\mu_a'(\epsilon_a(x)) = clos(\epsilon_{c_a},...)$, $\mu_b'(\epsilon_b(x)) = clos(\epsilon_{c_b},...)$, $\Xi_a' \models \epsilon_{c_a}: \Gamma_{fn}$, and $\Xi_b' \models \epsilon_{c_b}: \Gamma_{fn}$ , then $\semanticBelowPCState{l}{\Xi_a'}{\Xi_b'}{\Delta}{\mu_{a}'}{\epsilon_{c_a}}{\mu_{b}'}{\epsilon_{c_b}}{\Gamma_{fn}}$.

First, we begin by showing \Cref{fndecl-tp-0-1}. This requires us to in turn prove the following:
\begin{enumerate}
\item $\semanticStore{\Xi_a'}{\Delta}{\mu_a'}$. This is shown in \Cref{lem:store-typing-extended}.
\item $\Xi_a' \vdash \epsilon_a: \Gamma$. This follows from $\Xi_a \vdash \epsilon_a: \Gamma$ and weakening of store typing context.
\item Next, we need to show that any closure value has $\dom{\epsilon_{clos}} \subseteq \dom{\epsilon}$  (already known from \Cref{decl-fndecl-hyp}) and $\semanticStoreEnv{\Xi_a'}{\Delta}{\mu_a'}{\epsilon_{clos}}{\Gamma_{clos}}$, where $\epsilon_{clos}$ is the environment bound to the closure.
We already know that $\semanticStoreEnv{\Xi_a}{\Delta}{\mu_a}{\epsilon_{clos}}{\Gamma_{clos}}$. Using \Cref{lem:only-mem-superset}, we have $\semanticStoreEnv{\Xi_a'}{\Delta}{\mu_a'}{\epsilon_{clos}}{\Gamma_{clos}}$.
\end{enumerate}
Similarly, we can prove \Cref{fndecl-tp-0-2}.
Since value of no location besides the fresh $l_a$ and $l_b$ changes between $\mu_a$ and $\mu_a'$ and $\mu_b$ and $\mu_b'$, we can show \Cref{fndecl-tp-0-4} and the one following it using the results from \Cref{decl-fndecl-hyp}. All the variables referenced by closures that were declared until $\epsilon_a$ or $\epsilon_b$ have unchanged values.

To prove $\semanticBelowPCState{l}{\Xi_a'}{\Xi_b'}{\Delta} {\mu_{a}'}{\epsilon_{a}'}{\mu_{b}'}{\epsilon_{b}'}{\Gamma}$, we need to show the following:
\begin{equation} \label{fndecl-tp-1-1}
\semanticStoreEnv{\Xi_a'}{\Delta}{\mu_{a}'}{\epsilon_{a}'}{\Gamma'}
\end{equation}
\begin{equation} \label{fndecl-tp-1-2}
\semanticStoreEnv{\Xi_b'}{\Delta}{\mu_{b}'}{\epsilon_{b}'}{\Gamma'}
\end{equation}
\begin{equation} \label{fndecl-tp-1-3}
\dom{\epsilon_a'} = \dom{\epsilon_b'}
\end{equation}
\begin{equation} \label{fndecl-tp-1-4}
\text{For any ~} x \in \dom{\epsilon_a'} = \dom{\epsilon_b'} \text{~we have~} \NIval{l}{\Xi_a'}{\Xi_b'}{\Delta}{\mu_a'(\epsilon_{a}'(x))}{\mu_b'(\epsilon_{b}'(x))}{\Gamma'(x)}
\end{equation}
and for all $x$ in $\dom{\epsilon_a'} = \dom{\epsilon_b'}$ and some $\Gamma_{fn} \subseteq \Gamma$ and any $pc$,  if $\Gamma', \Delta \vdash_{pc} x: \tau_{fn}$, $\mu_a'(\epsilon_a'(x)) = clos(\epsilon_{c_a},...)$, $\mu_b'(\epsilon_b'(x)) = clos(\epsilon_{c_b},...)$, $\Xi_a' \models \epsilon_{c_a}: \Gamma_{fn}$, and $\Xi_b' \models \epsilon_{c_b}: \Gamma_{fn}$ , then $\semanticBelowPCState{l}{\Xi_a'}{\Xi_b'}{\Delta}{\mu_{a}'}{\epsilon_{c_a}}{\mu_{b}'}{\epsilon_{c_b}}{\Gamma_{fn}}$.

First, we begin by showing \Cref{fndecl-tp-1-1}. This requires us to in turn prove the following:
\begin{enumerate}
\item $\semanticStore{\Xi_a'}{\Delta}{\mu_a'}$. This is shown in \Cref{lem:store-typing-extended}.
\item $\Xi_a' \vdash \epsilon_a': \Gamma'$. Since we are given $\Xi_a \vdash \epsilon_a: \Gamma$, by using \Cref{lem-env-subtyping}, we can say that $\Xi_a' \vdash \epsilon_a: \Gamma$.
Since $\Gamma' = \Gamma[x \mapsto \type{\overline{d\type{\tau_{i}}{\chi_{i} }} \xrightarrow[]{pc_{fn}} \type{\tau_{ret}}{\chi_{ret}}}{\bot}]$, $\epsilon_a' = \epsilon_a[x \mapsto l_a]$, $\Xi_a' = \Xi_a \cup \{l_a \mapsto \type{\tau}{pc}\}$, by using the rules for $\Xi \vdash \epsilon: \Gamma$, we can show that $\Xi_a' \vdash \epsilon_a': \Gamma'$.
\[\inferrule*[]
{	\Xi_a' \vdash \epsilon_a: \Gamma \qquad
\Xi_a'(l_a) = \type{\overline{d~\type{\tau_{i}}{\chi_{i} }} \xrightarrow[]{pc_{fn}} \type{\tau_{ret}}{\chi_{ret}}}{\bot}
}
{
\Xi_a' \vdash \epsilon_a[x \mapsto l_a]: \Gamma[x \mapsto \type{\overline{d~\type{\tau_{i}}{\chi_{i} }} \xrightarrow[]{pc_{fn}} \type{\tau_{ret}}{\chi_{ret}}}{\bot}]
}
\]
\item Next, we need to show that any closure value has $\dom{\epsilon_{clos}} \subseteq \dom{\epsilon'}$ and $\semanticStoreEnv{\Xi_a'}{\Delta}{\mu_a'}{\epsilon_{clos}}{\Gamma_{clos}}$, where $\epsilon_{clos}$ is the environment bound to the closure.
Since $\epsilon_a' = \epsilon_a[x \mapsto l_a]$, where $\Gamma', \Delta \vdash_{pc} x: \tau_{fn}$, we need to show the above property for $y \in \dom{\epsilon_a}$ that is not equal to $x$ and the new closure variable $x$.
Since
\[
\text{for any~} y \in \dom{\Gamma}~\text{such that~} y \ne x ,\text{we have ~}		\Gamma'(y) = \Gamma(y)
\]
\[
\text{for any~} y \in \dom{\epsilon_a}~\text{such that~} y \ne x ,\text{we have ~}		\epsilon_a'(y) = \epsilon_a(y)
\]
and from \Cref{decl-fndecl-hyp}, we already know that for all $\Gamma, \Delta \vdash_{pc} y: \tau_{clos} \in \dom{\epsilon_a}$ such that $y\ne x$ and  $\mu_a'(\epsilon_a'(y)) = val_{clos}$, we have $\dom{\epsilon_{c}} \subseteq \dom{\epsilon_a}$, $\semanticStoreEnv{\Xi_a}{\Delta}{\mu_a}{\epsilon_{c}}{\Gamma_{clos}}$ and $\semanticStoreEnv{\Xi_b}{\Delta}{\mu_b}{\epsilon_{c}}{\Gamma_{clos}}$. Using \Cref{lem:only-mem-superset}, we can conclude that $\semanticStoreEnv{\Xi_a'}{\Delta}{\mu_a'}{\epsilon_{c}}{\Gamma_{clos}}$ and $\semanticStoreEnv{\Xi_b'}{\Delta}{\mu_b'}{\epsilon_{c}}{\Gamma_{clos}}$.

Since $\mu_a'(\epsilon_a'(x)) = clos(\epsilon_{a},  \overline{d~ x_{i}: \type{\tau_{i}'}{\chi_i}}, \type{\tau_{ret}}{\chi_{ret}}, stmt)$, $\dom{\epsilon_a} \subseteq \dom{\epsilon_a'}$, and $\semanticStoreEnv{\Xi_a'}{\Delta}{\mu_a'}{\epsilon_{a}}{\Gamma}$ we can conclude that for all closure values, we have $\dom{\epsilon_{c}} \subseteq \dom{\epsilon_a'}$, $\semanticStoreEnv{\Xi_a'}{\Delta}{\mu_a'}{\epsilon_{c}}{\Gamma_{clos}}$ and $\semanticStoreEnv{\Xi_b'}{\Delta}{\mu_b'}{\epsilon_{c}}{\Gamma_{clos}}$.
\end{enumerate}
This proves \Cref{fndecl-tp-1-1}.
Proof of \Cref{fndecl-tp-1-2} follows similarly.

To show \Cref{fndecl-tp-1-4}, we again use the fact that any $y \in \dom{\epsilon_a'}$ will be either in $\dom{\epsilon_a}$ or be the function name, $x$.
For the case where $y\ne x \in \dom{\epsilon_a}$, we already know that $\NIval{l}{\Xi_a}{\Xi_b}{\Delta}{\mu_a(\epsilon_{a}(x))}{\mu_b(\epsilon_{b}(x))}{\Gamma(x)}$.
This implies that $\text{for any ~} y\ne x \in \dom{\epsilon_a} = \dom{\epsilon_b}$, we have $\NIval{l}{\Xi_a'}{\Xi_b'}{\Delta}{\mu_a'(\epsilon_{a}'(x))}{\mu_b'(\epsilon_{b}'(x))}{\Gamma'(x)}$, since such $y$ satisfies $\mu_a'(\epsilon_a'(x)) = \mu_a(\epsilon_a(x))$, $\mu_b'(\epsilon_b'(x)) = \mu_b(\epsilon_b(x))$, $\Gamma(x) = \Gamma'(x)$ and $\Xi_a'(x) = \Xi_a(x)$, and $\Xi_b'(x) =\Xi_b(x)$.

For the case when $y = x$, we need to show the following , where $\tau_{fn} =  \type{\overline{d~\type{\tau_{i}}{\chi_{i} }} \xrightarrow[]{pc_{fn}} \type{\tau_{ret}}{\chi_{ret}}}{\bot}$.
\begin{equation}
\NIval{l}{\Xi_a'}{\Xi_b'}{\Delta}{clos(\epsilon_{a},  \overline{d~ x_{i}: \type{\tau_{i}'}{\chi_i}}, \type{\tau_{ret}}{\chi_{ret}},stmt)}{clos(\epsilon_{b},  \overline{d~ x_{i}: \type{\tau_{i}'}{\chi_i}}, \type{\tau_{ret}}{\chi_{ret}},stmt)}{\tau_{fn}}
\end{equation}

To show this, we need to first prove that $\ordinaryTyping[]{\Xi_a'}{\Delta}{clos(\epsilon_{a},  \overline{d~ x_{i}: \type{\tau_{i}'}{\chi_i}}, \type{\tau_{ret}}{\chi_{ret}}, \overline{decl}~stmt)}{ \type{\overline{d~\type{\tau_{i}}{\chi_{i} }} \xrightarrow[]{pc_{fn}} \type{\tau_{ret}}{\chi_{ret}}}{\bot}}$.
For this we need to look at the value typing rule for function closures.
\[\inferrule[TV-Clos]
{
\Xi \vdash \epsilon: \Gamma \\
\Gamma[\overline{x: \type{\tau}{\chi}}, \textsf{return} = \type{\tau_{ret}}{\chi_{ret}}], \Delta \vdash_{pc_{fn}} stmt \dashv \Gamma'
}
{\Xi, \Delta \vdash_{} clos(\epsilon, \overline{d~x: \type{\tau}{\chi}}, \type{\tau_{ret}}{\chi_{ret}}, \overline{decl}~stmt): \langle \langle \overline{d~\tau, \chi}\rangle \xrightarrow{pc_{fn}} \langle \tau_{ret}, \chi_{ret}\rangle, \bot \rangle}
\]
These premises are:
\begin{enumerate}
\item We need to show $\Xi_a' \vdash \epsilon_a: \Gamma$.
We already know from \Cref{decl-fndecl-hyp} that $\Xi_a \vdash \epsilon_a: \Gamma$. Using \Cref{lem-env-subtyping}, we can conclude $\Xi_a' \vdash \epsilon_a: \Gamma$.
\item We need to show $\declTyping{pc_{fn}}{\Gamma_{}[\overline{x_{i}: \type{\tau_{i}}{\chi_{i}}}, return: \type{\tau_{ret}}{\chi_{ret}}]}{\Delta}{stmt}{\Gamma_{1}}$
\end{enumerate}
This is satisfied as a part of the premise in the typing rule for function declaration.
This concludes \[\ordinaryTyping[]{\Xi_a'}{\Delta}{clos(\epsilon_{a},  \overline{d~ x_{i}: \type{\tau_{i}'}{\chi_i}}, \type{\tau_{ret}}{\chi_{ret}},stmt)}{ \type{\overline{d~\type{\tau_{i}}{\chi_{i} }} \xrightarrow[]{pc_{fn}} \type{\tau_{ret}}{\chi_{ret}}}{\bot}}\]
\[\ordinaryTyping[]{\Xi_b'}{\Delta}{clos(\epsilon_{b},  \overline{d~ x_{i}: \type{\tau_{i}'}{\chi_i}}, \type{\tau_{ret}}{\chi_{ret}}, stmt)}{ \type{\overline{d~\type{\tau_{i}}{\chi_{i} }} \xrightarrow[]{pc_{fn}} \type{\tau_{ret}}{\chi_{ret}}}{\bot}}\]

Next, we show that for $\Gamma$, the following properties hold:
\begin{enumerate}
\item $\Xi_a' \vdash \epsilon_{a}: \Gamma$, $\Xi_b' \vdash \epsilon_{b}: \Gamma$. Already shown above.
\item $\ordinaryTyping[pc]{\Gamma}{\Delta}{clos(\epsilon_{a},  \overline{d~ x_{i}: \type{\tau_{i}'}{\chi_i}}, \type{\tau_{ret}}{\chi_{ret}}, stmt)}{\type{\overline{d~\type{\tau_{i}}{\chi_{i} }} \xrightarrow[]{pc_{fn}} \type{\tau_{ret}}{\chi_{ret}}}{\bot}}$. We already know this by the typing derivation.
Similarly, 	$\ordinaryTyping[pc]{\Gamma}{\Delta}{clos(\epsilon_{b},  \overline{d~ x_{i}: \type{\tau_{i}'}{\chi_i}}, \type{\tau_{ret}}{\chi_{ret}}, stmt)}{\type{\overline{d~\type{\tau_{i}}{\chi_{i} }} \xrightarrow[]{pc_{fn}} \type{\tau_{ret}}{\chi_{ret}}}{\bot}}$.
\item  $\Gamma[\overline{x:\type{\tau}{\chi}}, return: \type{\tau_{ret}}{\chi_{ret}}], \Delta \vdash_{pc_{fn}} stmt \dashv \Gamma'$
\item Also, we have $val_a =_{clos} val_b$.
\end{enumerate}

We also need to show that for all $y \in \dom{\epsilon_a'} = \dom{\epsilon_b'}$ and some $\Gamma_{clos} \subseteq \Gamma'$ and any $pc$,  if $\Gamma', \Delta \vdash_{pc} y: \tau_{clos}$, $\mu_a'(\epsilon_a'(x)) = val_{clos}$ with environment $\epsilon_{c_a}$, $\mu_b'(\epsilon_b'(y)) = val_{clos}$ with environment $\epsilon_{c_b}$, $\Xi_a' \models \epsilon_{c_a}: \Gamma_{clos}$, and $\Xi_b' \models \epsilon_{c_b}: \Gamma_{clos}$ , then $\semanticBelowPCState{l}{\Xi_a'}{\Xi_b'}{\Delta}{\mu_{a}'}{\epsilon_{c_a}}{\mu_{b}'}{\epsilon_{c_b}}{\Gamma_{clos}}$.
For $y\ne x$, we know that a closure value would satisfy $\semanticBelowPCState{l}{\Xi_a}{\Xi_b}{\Delta}{\mu_{a}}{\epsilon_{c_a}}{\mu_{b}}{\epsilon_{c_b}}{\Gamma_{clos}}$. Since value of no variable referenced by any of the closure defined until $\epsilon_a$ is updated between $\mu_a'$  and $\mu_b'$, we can say that $\semanticBelowPCState{l}{\Xi_a}{\Xi_b}{\Delta}{\mu_{a}'}{\epsilon_{c_a}}{\mu_{b}'}{\epsilon_{c_b}}{\Gamma_{clos}}$. By weakening the store typing context, we can also say $\semanticBelowPCState{l}{\Xi_a'}{\Xi_b'}{\Delta}{\mu_{a}'}{\epsilon_{c_a}}{\mu_{b}'}{\epsilon_{c_b}}{\Gamma_{clos}}$.
For the new closure variable $x$, we already have $\semanticBelowPCState{l}{\Xi_a'}{\Xi_b'}{\Delta}{\mu_{a}'}{\epsilon_{a}}{\mu_{b}'}{\epsilon_{b}}{\Gamma}$.

This finally proves \Cref{fndecl-tp-1}.
\Cref{fndecl-tp-2} is satisfied because only the value of the location pointed by the function name, $x$ is updated in the memory store. \Cref{fndecl-tp-4} is trivial since no location  besides the fresh $l_a$ and $l_b$ are updated.

\item \textbf{\textsc{T-TblDecl}}
		\[      \inferrule*[right=T-TblDecl]
    {
    \newordinaryTyping[pc_{tbl}]{\Gamma}{\Delta}{\overline{exp_k:\type{\tau_k}{\chi_k}}}{} \\
    \newordinaryTyping[pc_{tbl}]{\Gamma}{\Delta}{\overline{x_k:\type{match\_kind}{\bot}}}{} \\\\
    \newordinaryTyping[pc_{tbl}]{\Gamma}{\Delta}{act_{a_j}: \type{\overline{d\type{\tau_{a_{ji}}}{\chi_{a_{ji}}}}~;\overline{\type{\tau_{c_{ji}}}{\chi_{c_{ji}}}} \xrightarrow{pc_{fn_j}} \type{unit}{\bot}}{\bot}}{},~\text{for all}~ j\\\\
    \newordinaryTyping[pc_{tbl}]{\Gamma}{\Delta}{\overline{exp_{a_{ji}}:\type{\tau_{a_{ji}}}{\chi_{a_{ji}}}goes~d}}{}  \\
    {\chi_k} \sqsubseteq {pc_{fn_j}}~\text{for all}~j,k \\
    pc_a \sqsubseteq pc_{fn_j}, \text{for all}~j \\
    {\chi_k} \sqsubseteq {pc_{tbl}}~\text{for all}~k\\
    pc_{tbl} \sqsubseteq pc_a
    }
    {
    \declTyping{pc}{\Gamma}{\Delta}{\text{table}~x~ \{\overline{exp_k: x_k}~ \overline{act_{a_j}(\overline{exp_{a_{ji}}})}\}}{\Gamma[x: \type{table(pc_{tbl})}{\bot}]}{\Delta}
    }\]

Given the above typing judgement holds for table declaration, we need to show that for any $\Xi_a$, $\Xi_b$, $\mu_{a}$, $\mu_{b}$, $\epsilon_{a}$, $\epsilon_{b}$, $\mu_{a}'$, $\mu_{b}'$, $\epsilon_{a}'$, $\epsilon_{b}'$ satisfying
\begin{equation} \label{decl-tbl-decl-hyp}
\semanticBelowPCState{l}{\Xi_a}{\Xi_b}{\Delta}{\mu_{a}}{\epsilon_{a}}{\mu_{b}}{\epsilon_{b}}{\Gamma}
\end{equation}
if the table declaration is evaluated under two different initial configurations $\langle \mu_a, \epsilon_a \rangle$ and $\langle \mu_b, \epsilon_b \rangle$ as follows,
\begin{mathpar}
\inferrule*[]
{
l_{a}~ fresh \\
val_{a} =  table~l_a~(\epsilon_{a}, \overline{exp_k: x_k}, \overline{act_{a_j}(\overline{exp_{a_{ji}}}, \overline{y_{c_{ji}}: \type{\tau_{c_{ji}}}{\chi_{c_{ji}}}})})
}
{
 \langle \mathcal{C}, \Delta, \mu_{a}, \epsilon_{a}, \text{table~}~x~\{\overline{exp_k: x_k}~ \overline{act_{a_j}(\overline{exp_{a_{ji}}})}\} \rangle \Downarrow         	 \langle \Delta, \mu_{a}[l_{a} \mapsto val_{a}], \epsilon_{a}[x \mapsto l_{a}], cont \rangle
}

\inferrule*[]
{
l_{b}~ fresh \\
val_{b} =  table~l_b~(\epsilon_{b}, \overline{exp_k: x_k}, \overline{act_{a_j}(\overline{exp_{a_{ji}}}, \overline{y_{c_{ji}}: \type{\tau_{c_{ji}}}{\chi_{c_{ji}}}})})
}
{
 \langle \mathcal{C}, \Delta, \mu_{b}, \epsilon_{b}, \text{table~}~x~\{\overline{exp_k: x_k}~ \overline{act_{a_j}(\overline{exp_{a_{ji}}})}\} \rangle \Downarrow           	\langle \Delta, \mu_{b}[l_{b} \mapsto val_{b}], \epsilon_{b}[x \mapsto l_{b}], cont \rangle
}
\end{mathpar}
then there exists some $\Xi_a'$, $\Xi_b'$ such that
\begin{enumerate}
\item $\declTyping{pc}{\Gamma}{\Delta}{table~declaration}{\Gamma'}{\Delta}$. This is already the hypothesis of the theorem.
\item \label{tbldecl-tp-1} $\semanticBelowPCState{l}{\Xi_a'}{\Xi_b'}{\Delta} {\mu_{a}'}{\epsilon_{a}'}{\mu_{b}'}{\epsilon_{b}'}{\Gamma'}$ and $\semanticBelowPCState{l}{\Xi_a'}{\Xi_b'}{\Delta} {\mu_{a}'}{\epsilon_{a}}{\mu_{b}'}{\epsilon_{b}}{\Gamma}$, where $\mu_a' = \mu_{a}[l_{a} \mapsto val_{a}]$, $\mu_b' = \mu_{b}[l_{b} \mapsto val_{b}]$, $\epsilon_a' = \epsilon_{a}[x \mapsto l_{a}]$ and $\epsilon_b' = \epsilon_{b}[x \mapsto l_{b}]$, $\Gamma' = \Gamma[x\mapsto \type{\text{table}(pc_{tbl})}{\bot}]$.
\item \label{tbldecl-tp-2} For any $l_a \in \dom{\mu_a}$ such that $\ordinaryTyping[]{\Xi_a}{\Delta}{\mu_a(l_a)}{\tau_{clos}}$, where $\tau_{clos} \in \{\tau_{fn}, \tau_{tbl}\}$, then $\mu_a'(l_a) = \mu_a(l_a)$. Similarly for any $l_b \in \dom{\mu_b}$ such that $\ordinaryTyping[]{\Xi_b}{\Delta}{\mu_b(l_b)}{\tau_{clos}}$, where $\tau_{clos} \in \{\tau_{fn}, \tau_{tbl}\}$, then $\mu_b'(l_b) = \mu_b(l_b)$,
\item \label{tbldecl-tp-3} $\Xi_a \subseteq \Xi_a'$, $\Xi_b \subseteq \Xi_b'$, $\dom{\mu_a} \subseteq \dom{\mu_a'}$, $\dom{\mu_b} \subseteq \dom{\mu_b'}$, $\dom{\epsilon_a} \subseteq \dom{\epsilon_a'}$, and $\dom{\epsilon_b} \subseteq \dom{\epsilon_b'}$.
\item For any $l_a \in \dom{\mu_a}$ and $l_b \in \dom{\mu_b}$ such that $\ordinaryTyping[]{\Xi_a}{\Delta}{\mu_a(l_a)}{\type{\tau}{\chi}}$ and $\ordinaryTyping[]{\Xi_b}{\Delta}{\mu_b(l_b)}{\type{\tau}{\chi}}$ and $pc \nsqsubseteq \chi$, we have $\mu_{a}'(l_a) = \mu_{a}(l_a)$ and $\mu_{b}'(l_b) = \mu_{b}(l_b)$.
\end{enumerate}
With $\mu_{a}' = \mu_{a}[l_{a} \mapsto val_{a}]$, $\epsilon_{a}' = \epsilon_{a}[x \mapsto l_{a}]$, $\textsc{dom}(\epsilon_{a}') = \textsc{dom}(\epsilon_{a}) \cup \{x\}$,  $\mu_{b}' = \mu_{b}[l_{b} \mapsto val_{b}]$, $\epsilon_{b}' = \epsilon_{b}[x \mapsto l_{b}]$, $\textsc{dom}(\epsilon_{b}') = \textsc{dom}(\epsilon_{b}) \cup \{x\}$, $\Xi_a' = \Xi_a[l_a \mapsto \type{text{table}(pc_{tbl})}{\bot}]$, and $\Xi_b' = \Xi_b[l_b \mapsto \type{\text{table}(pc_{tbl})}{\bot}]$ \Cref{tbldecl-tp-3} is evident.

Proof of  \Cref{tbldecl-tp-1} follows similar to the function declaration case.
The interesting bit is to show that the freshly added table name $x$ satisfies the following property. For the case when $y = x$, we need to show that
\[\NIval{l}{\Xi_a'}{\Xi_b'}{\Delta}{table~l_a~(\epsilon_{a}, \overline{exp_k: x_k}, \overline{act_{a_j}(\overline{exp_{a_{ji}}}, \overline{y_{c_{ji}}: \type{\tau_{c_{ji}}}{\chi_{c_{ji}}}})})}{table~l_b~(\epsilon_{a}, \overline{exp_k: x_k}, \overline{act_{a_j}(\overline{exp_{a_{ji}}}, \overline{y_{c_{ji}}: \type{\tau_{c_{ji}}}{\chi_{c_{ji}}}})})}{\tau_{tbl}}\]
where $\tau_{tbl} =  \type{table(pc_{tbl})}{\bot}$.
For this, we need to first show that
\[\ordinaryTyping[]{\Xi_a'}{\Delta}{table~l_a~(\epsilon_{a}, \overline{exp_k: x_k}, \overline{act_{a_j}(\overline{exp_{a_{ji}}}, \overline{y_{c_{ji}}: \type{\tau_{c_{ji}}}{\chi_{c_{ji}}}})})}{\type{table(pc_{tbl})}{\bot}}\]
\[\ordinaryTyping[]{\Xi_b'}{\Delta}{table~l_b~(\epsilon_{b}, \overline{exp_k: x_k}, \overline{act_{a_j}(\overline{exp_{a_{ji}}}, \overline{y_{c_{ji}}: \type{\tau_{c_{ji}}}{\chi_{c_{ji}}}})})}{\type{table(pc_{tbl})}{\bot}}\]
To show this, we need to prove that the premises of the following value typing rule are satisfied,
\[
\inferrule[TV-Tbl]
{
\Xi \vdash \epsilon: \Gamma \\
    \newordinaryTyping[pc_{tbl}]{\Gamma}{\Delta}{\overline{exp_k:\type{\tau_k}{\chi_k}}}{} \\
    \newordinaryTyping[pc_{tbl}]{\Gamma}{\Delta}{\overline{x_k:\type{match\_kind}{\bot}}}{} \\\\
    \newordinaryTyping[pc_{tbl}]{\Gamma}{\Delta}{act_{a_j}: \type{\overline{d~\type{\tau_{a_{ji}}}{\chi_{a_{ji}}}}~;\overline{\type{\tau_{c_{ji}}}{\chi_{c_{ji}}}} \xrightarrow{pc_{fn_j}} \type{unit}{\bot}}{\bot}}{},~\text{for all}~ j\\\\
    \newordinaryTyping[pc_{tbl}]{\Gamma}{\Delta}{\overline{exp_{a_{ji}}:\type{\tau_{a_{ji}}}{\chi_{a_{ji}}}goes~d}}{}  \\
    {\chi_k} \sqsubseteq {pc_{fn_j}}~\text{for all}~j,k \\
    pc_a \sqsubseteq pc_{fn_j}, \text{for all}~j \\
    {\chi_k} \sqsubseteq {pc_{tbl}}~\text{for all}~k\\
    pc_{tbl} \sqsubseteq pc_{a}
}
{{\Xi},{\Delta}\vdash_{}{\text{table}~l~ \{\epsilon, \overline{exp_k: x_k}~ \overline{act_{a_j}(\overline{exp_{a_{ji}}})}\}}: \type{table(pc_{tbl})}{\bot}}
\]
\begin{enumerate}
\item We need to show $\Xi_a', \Delta \vdash \epsilon_a: \Gamma$.
We already know from \Cref{decl-fndecl-hyp} that $\Xi_a \vdash \epsilon_a: \Gamma$.
Using \Cref{lem-env-subtyping}, we can conclude 	$\Xi_a' \vdash \epsilon_a: \Gamma$.
\item We need to show $\newordinaryTyping[pc_{tbl}]{\Gamma}{\Delta}{\overline{exp_k:\type{\tau_k}{\chi_k}}}{}$
\item We need to show $\newordinaryTyping[pc_{tbl}]{\Gamma}{\Delta}{\overline{x_k:\type{match\_kind}{\bot}}}{}$
\item We need to show $\newordinaryTyping[pc_{tbl}]{\Gamma}{\Delta}{\overline{act_{a_j}: \type{\overline{d~\type{\tau_{a_{ji}}}{\chi_{a_{ji}}}}~;\overline{\type{\tau_{c_{ji}}}{\chi_{c_{ji}}}} \xrightarrow{pc_{fn_j}} \type{unit}{\bot}}{\bot}}}{}$
\item We need to show $ \newordinaryTyping[pc_{tbl}]{\Gamma}{\Delta}{\overline{exp_{a_{ji}}:\type{\tau_{a_{ji}}}{\chi_{a_{ji}}}~goes~d}}{}$
\end{enumerate}
The last four properties are satisfied as a part of the premise for the typing rule for table declaration.
This concludes:
\[
  \ordinaryTyping[]{\Xi_a'}{\Delta}{table~l_a~(\epsilon_{a}, \overline{exp_k:
  x_k}, \overline{act_{a_j}(\overline{exp_{a_{ji}}}, \overline{y_{c_{ji}}:
\type{\tau_{c_{ji}}}{\chi_{c_{ji}}}})})}{\type{table(pc_{tbl})}{\bot}}
\]
\[
  \ordinaryTyping[]{\Xi_b'}{\Delta}{table~l_b~(\epsilon_{a}, \overline{exp_k: x_k}, \overline{act_{a_j}(\overline{exp_{a_{ji}}}, \overline{y_{c_{ji}}: \type{\tau_{c_{ji}}}{\chi_{c_{ji}}}})})}{ \type{table(pc_{tbl})}{\bot}}
\]

Next we show that for $\Gamma$ the following properties hold (by expanding the definition of NI for table closures)
\begin{enumerate}
\item $\Xi_a'\vdash \epsilon_{a}: \Gamma$, $\Xi_b' \vdash \epsilon_{b}: \Gamma$. Already shown.
\item $\ordinaryTyping[pc]{\Gamma}{\Delta}{table~l_a~(\epsilon_{a}, \overline{exp_k: x_k}, \overline{act_{a_j}(\overline{exp_{a_{ji}}}, \overline{y_{c_{ji}}: \type{\tau_{c_{ji}}}{\chi_{c_{ji}}}})})}{\type{table(pc_{tbl})}{\bot}}$. We already know this by the typing derivation.
Similarly, 	$\ordinaryTyping[pc]{\Gamma}{\Delta}{table~l_b~(\epsilon_{a}, \overline{exp_k: x_k}, \overline{act_{a_j}(\overline{exp_{a_{ji}}}, \overline{y_{c_{ji}}: \type{\tau_{c_{ji}}}{\chi_{c_{ji}}}})})}{\type{table(pc_{tbl})}{\bot}}$.
\item  $\ordinaryTyping[pc_{tbl}]{\Gamma}{\Delta}{x_k}{\type{match\_kind}{\bot}}$ for each $x_k \in \overline{x_k}$.
%
\item $\newordinaryTyping[pc_{tbl}]{\Gamma}{\Delta}{\overline{exp_k:\type{\tau_k}{\chi_k}}}{}$. for each $exp_k \in \overline{exp_k}$.
%
\item $\ordinaryTyping[pc_{tbl}]{\Gamma}{\Delta}{act_{aj}}{\type{\overline{d~\type{\tau_{a_{ji}}}{\chi_{a_{ji}}}}~;\overline{\type{\tau_{c_{ji}}}{\chi_{c_{ji}}}} \xrightarrow{pc_{fn_j}} \type{unit}{\bot}}{\bot}}$ for each $act_{a_j} \in \overline{act_{a_j}}$.
\item $\ordinaryTyping[pc_{tbl}]{\Gamma}{\Delta}{exp_{a_{ji}}}{\type{\tau_{a_{ji}}}{\chi_{a_{ji}}}~goes~d}$ for each $exp_{a_{ji}} \in \overline{exp_{a_{ji}}}$.
%
\item $val_a =_{tbl} val_b$. 
\item $\chi_k \sqsubseteq {pc_{fn_j}}$, for all $j,k$
\item $pc_a \sqsubseteq pc_{fn_j}$, for all $j$
\item $\chi_k \sqsubseteq {pc_{tbl}}~\text{for all}~k$
\item $pc_{tbl} \sqsubseteq pc_a$
\end{enumerate}
These properties are can be shown using the premise in the typing derivation.

\item  \textbf{\textsc{T-Typedef}}
\[ \inferrule*[right=T-Typedef]
{
}
{
\declTyping{pc}{\Gamma}{\Delta}{\terminal{typedef}~\tau~X}{\Gamma}{\Delta[X=\tau]}
}
\]
\[
\inferrule*
{~}
{ \langle \mathcal{C}, \Delta; \mu_{a}; \epsilon_{a};\terminal{typedef}~\tau~X \rangle \Downarrow \langle \Delta[X=\tau], \mu_a, \epsilon_a, cont \rangle}
\]
The proof of this case is trivial. The only interesting part is to show that $\semanticBelowPCState{l}{\Xi_a'}{\Xi_b'}{\Delta[X=\tau]}{\mu_{a}}{\epsilon_{a}}{\mu_{b}}{\epsilon_{b}}{\Gamma}$. We already know that $\semanticBelowPCState{l}{\Xi_a}{\Xi_b}{\Delta}{\mu_{a}}{\epsilon_{a}}{\mu_{b}}{\epsilon_{b}}{\Gamma}$. By definition of this judgement for a pair of consistent state, we can observe that we can prove this for the extended $\Delta$ as it is a case of weakening the context.
\item \textbf{\textsc{T-MatchKind}}
\[ \inferrule*[right=T-MemHdr]
{
}
{
\declTyping{pc}{\Gamma}{\Delta}{match\_kind\{\overline{f}\}}{\Gamma}{\Delta[match\_kind=match\_kind\{\overline{f}\}]}
}
\]
Evaluation rule is
\[
\inferrule*[right=Eval 1]
{~}
{ \langle \mathcal{C}, \Delta; \mu_{a}; \epsilon_{a}; match\_kind\{\overline{f}\} \rangle \Downarrow \langle \Delta[match\_kind=match\_kind\{\overline{f}\}], \mu_a, \epsilon_a, cont \rangle}
\]
The proof is similar to the typedef case.

\item \textbf{\textsc{T-Seq-2}}

\[ \inferrule*[right=T-Seq]
{
\Gamma, \Delta \vdash_{pc} decl \dashv \Gamma_{1}, \Delta_1 \qquad
\Gamma_{1}, \Delta_1 \vdash_{pc} stmt \dashv \Gamma_{2}
}
{
\Gamma, \Delta \vdash_{pc} decl~stmt \dashv \Gamma_2, \Delta_1
}
\]

\[
\inferrule*[]
{
         \langle \mathcal{C}, \Delta, \mu_{a}, \epsilon_{a}, decl \rangle \Downarrow  \langle \Delta_1, \mu_{a1}, \epsilon_{a1}, cont \rangle \\
         \langle \mathcal{C}, \Delta_1, \mu_{a1}, \epsilon_{a1}, stmt \rangle \Downarrow   \langle \mu_{a2}, \epsilon_{a2}, sig_a \rangle
}
{
 \langle \mathcal{C}, \Delta, \mu_{a}, \epsilon_{a}, decl~stmt \rangle  \Downarrow  \langle \Delta_1, \mu_{a2}, \epsilon_{a2}, sig_a \rangle
}
\]

\[
\inferrule*[]
{
         \langle \mathcal{C}, \Delta, \mu_{b}, \epsilon_{b}, decl \rangle \Downarrow  \langle \Delta_1, \mu_{b1}, \epsilon_{b1}, cont \rangle \\
         \langle \mathcal{C}, \Delta_1, \mu_{b1}, \epsilon_{b1}, stmt \rangle \Downarrow   \langle \mu_{b2}, \epsilon_{b2}, sig_b \rangle
}
{
 \langle \mathcal{C}, \Delta, \mu_{b}, \epsilon_{b}, decl~stmt \rangle  \Downarrow  \langle \Delta_1, \mu_{b2}, \epsilon_{b2}, sig_b \rangle
}
\]
Given the above typing judgement holds for the statement,$decl~stmt$, we need to show that for any $\Xi_a$, $\Xi_b$, $\mu_{a}$, $\mu_{b}$, $\epsilon_{a}$, $\epsilon_{b}$, $\mu_a'$, $\mu_b'$, $\epsilon_a'$, $\epsilon_b'$ satisfying
\begin{equation} \label{stmt-seq-hyp}
\semanticBelowPCState{l}{\Xi_a}{\Xi_b}{\Delta}{\mu_{a}}{\epsilon_{a}}{\mu_{b}}{\epsilon_{b}}{\Gamma}
\end{equation}
If the statement, $decl~stmt$ is evaluated under two different initial configurations $\langle \mu_a, \epsilon_a \rangle$ and $\langle \mu_b, \epsilon_b \rangle$, then there exists some $\Xi_a'$ and $\Xi_b'$, such that the following hold:
\begin{enumerate}
\item $\declTyping{pc}{\Gamma}{\Delta}{decl~stmt}{\Gamma_2}{\Delta_1}$. This is already the theorem's hypothesis.
\item \label{stmt-seq-tp-2} We have$\Xi_a \subseteq \Xi_a'$, $\Xi_b \subseteq
  \Xi_b'$, $\dom{\mu_a} \subseteq \dom{\mu_a'}$, $\dom{\mu_b} \subseteq
  \dom{\mu_b'}$, $\dom{\epsilon_a} \subseteq \dom{\epsilon_a'}$, and
  $\dom{\epsilon_b} \subseteq \dom{\epsilon_b'}$, $\Delta \subseteq \Delta_1$,
and $\semanticBelowPCState{l}{\Xi_a'}{\Xi_b'}{\Delta_1} {\mu_{a}'}{\epsilon_{a}'}{\mu_{b}'}{\epsilon_{b}'}{\Gamma'}$. In this case $\mu_a' = \mu_{a2}$, $\mu_b' = \mu_{b2}$, $\epsilon_a' = \epsilon_{a2}$, $\epsilon_b' = \epsilon_{b2}$.
We also need to show that $\semanticBelowPCState{l}{\Xi_a'}{\Xi_b'}{\Delta} {\mu_{a}'}{\epsilon_{a}}{\mu_{b}'}{\epsilon_{b}}{\Gamma}$.
\item \label{stmt-seq-tp-4} For any $l_a \in \dom{\mu_a}$ such that $\ordinaryTyping[]{\Xi_a}{\Delta}{\mu_a(l_a)}{\tau_{clos}}$, where $\tau_{clos} \in \{\tau_{fn}, \tau_{tbl}\}$, then $\mu_a'(l_a) = \mu_a(l_a)$. Similarly for any $l_b \in \dom{\mu_b}$ such that $\ordinaryTyping[]{\Xi_b}{\Delta}{\mu_b(l_b)}{\tau_{clos}}$, where $\tau_{clos} \in \{\tau_{fn}, \tau_{tbl}\}$, then $\mu_b'(l_b) = \mu_b(l_b)$.
\item For any $l_a' \in \dom{\mu_{a}}$ and $l_b' \in \dom{\mu_{b}}$ such that $\ordinaryTyping[]{\Xi_{a}}{\Delta}{\mu_{a}(l_a')}{\type{\tau}{\chi}}$ and $\ordinaryTyping[]{\Xi_{b}}{\Delta}{\mu_{b}(l_b')}{\type{\tau}{\chi}}$ and $pc \nsqsubseteq \chi$, we have $\mu_{a}(l_a') = \mu_{a}'(l_a')$ and $\mu_{b}(l_b') = \mu_{b}'(l_b')$,
\item \label{stmt-seq-tp-3} $sig$ in any two evaluations are of the same form.
\end{enumerate}
The proof is direct by applying induction hypothesis on the $decl$ and $stmt$. We will highlight the most interesting part.
By applying induction hypothesis of \Cref{ni-decl} on $decl$, we conclude that NI decl.
This implies $\semanticBelowPCState{l}{\Xi_{a1}}{\Xi_{b1}}{\Delta_1} {\mu_{a1}}{\epsilon_{a1}}{\mu_{b1}}{\epsilon_{b1}}{\Gamma}$ and $\semanticBelowPCState{l}{\Xi_{a1}}{\Xi_{b1}}{\Delta_1} {\mu_{a1}}{\epsilon_{a}}{\mu_{b1}}{\epsilon_{b}}{\Gamma}$
By applying induction hypothesis of \Cref{ni-stmt} on $stmt$, we conclude that $\semanticBelowPCState{l}{\Xi_{a2}}{\Xi_{b2}}{\Delta_1}{\mu_{a2}}{\epsilon_{a2}}{\mu_{b2}}{\epsilon_{b2}}{\Gamma}$ and $\semanticBelowPCState{l}{\Xi_{a2}}{\Xi_{b2}}{\Delta_1}{\mu_{a2}}{\epsilon_{a1}}{\mu_{b2}}{\epsilon_{b1}}{\Gamma}$. To prove $\semanticBelowPCState{l}{\Xi_{a2}}{\Xi_{b2}}{\Delta_1}{\mu_{a2}}{\epsilon_{a}}{\mu_{b2}}{\epsilon_{b}}{\Gamma}$, we use the same approach from \textsc{T-Seq-1} case (\Cref{seq}).
\end{enumerate}

\section{Value Typing Rule} \label{vtyping}
\begin{mathpar}
\inferrule[TV-Rec]
{\Xi, \Delta \vdash_{} \overline{val: \type{\tau}{\chi}}}
{\Xi, \Delta \vdash_{} \{\overline{f=val}\}: \type{\{\overline{f: \type{\tau}{\chi}}\}}{\bot}}

\inferrule[TV-Hdr]
{\Xi, \Delta \vdash_{} \overline{val: \type{\tau}{\chi}}}
{\Xi, \Delta \vdash_{} header\{valid, \overline{f: \type{\tau}{\chi}=val}\}: \type{header\{\overline{f: \type{\tau}{\chi}}\}}{\bot}}

\inferrule[TV-Stack]
{len(\overline{val}) = n\\
\Xi, \Delta \vdash_{} \overline{val}: \type{\tau}{\chi}}
{\Xi, \Delta \vdash_{} stack~ \type{\tau}{\chi}~\{\overline{val}\}: ~\type{\type{\tau}{\chi}[n]}{\bot}}

\inferrule[TV-Clos]
{
\Xi \vdash \epsilon: \Gamma \\
\Gamma[\overline{x: \type{\tau}{\chi}}, \textsf{return} = \type{\tau_{ret}}{\chi_{ret}}], \Delta \vdash_{pc_{fn}} stmt \dashv \Gamma'
}
{\Xi, \Delta \vdash_{} clos(\epsilon, \overline{d~x: \type{\tau}{\chi}}, \type{\tau_{ret}}{\chi_{ret}}, \overline{decl}~stmt): \langle \langle \overline{d~\tau, \chi}\rangle \xrightarrow{pc_{fn}} \langle \tau_{ret}, \chi_{ret}\rangle, \bot \rangle}

\inferrule[TV-Tbl]
{
\Xi \vdash \epsilon: \Gamma \\
    \newordinaryTyping[pc_{tbl}]{\Gamma}{\Delta}{\overline{exp_k:\type{\tau_k}{\chi_k}}}{} \\
    \newordinaryTyping[pc_{tbl}]{\Gamma}{\Delta}{\overline{x_k:\type{match\_kind}{\bot}}}{} \\\\
    \newordinaryTyping[pc_{tbl}]{\Gamma}{\Delta}{act_{a_j}: \type{\overline{d~\type{\tau_{a_{ji}}}{\chi_{a_{ji}}}}~;\overline{\type{\tau_{c_{ji}}}{\chi_{c_{ji}}}} \xrightarrow{pc_{fn_j}} \type{unit}{\bot}}{\bot}}{},~\text{for all}~ j\\\\
    \newordinaryTyping[pc_{tbl}]{\Gamma}{\Delta}{\overline{exp_{a_{ji}}:\type{\tau_{a_{ji}}}{\chi_{a_{ji}}}goes~d}}{}  \\
    {\chi_k} \sqsubseteq {pc_{fn_j}}~\text{for all}~j,k \\
    pc_a \sqsubseteq pc_{fn_j}, \text{for all}~j \\
    {\chi_k} \sqsubseteq {pc_{tbl}}~\text{for all}~k\\
    pc_{tbl} \sqsubseteq pc_{a}
}
{{\Xi},{\Delta}\vdash_{}{\text{table}~l~ \{\epsilon, \overline{exp_k: x_k}~ \overline{act_{a_j}(\overline{exp_{a_{ji}}})}\}}: \type{table(pc_{tbl})}{\bot}}

\inferrule[TV-PartialApp]
{
\Gamma, \Delta \vdash_{pc} x_{act}: \type{\overline{d~\type{\tau_{a_{ji}}}{\chi_{a_{ji}}}}~;\overline{\type{\tau_{c_{ji}}}{\chi_{c_{ji}}}} \xrightarrow{pc_{fn}} \type{unit}{\bot}}{\bot} \\
\Gamma, \Delta \vdash_{pc} \overline{exp:\type{\tau}{\chi} goes~ d}
}
{{\Xi},{\Delta}\vdash_{}x_{act}(\overline{exp}, \overline{x_c: \type{\tau}{\chi}}): \type{\overline{d~\type{\tau_{a_{ji}}}{\chi_{a_{ji}}}}~;\overline{\type{\tau_{c_{ji}}}{\chi_{c_{ji}}}} \xrightarrow{pc_{fn}} \type{unit}{\bot}}{\bot}}

\inferrule[Match]{\Delta(match\_kind) = match\_kind \{\overline{f}\}}{\Xi, \Delta \vdash_{} match\_kind.f: match\_kind}

\inferrule[TV-SubType]
{\Xi, \Delta \vdash_{} val: \type{\tau}{\chi} \\ \chi \sqsubseteq \chi'}
{\Xi, \Delta \vdash_{} val: \type{\tau}{\chi'} }
\end{mathpar}
\fi

\end{document}